\documentclass{sigkddExp}

\usepackage{subfigure}
\usepackage{graphicx}
\usepackage{amsmath}
\usepackage{bm}
\usepackage{url}
\usepackage{stmaryrd}
\usepackage{balance}
\usepackage{amssymb}
\usepackage{pgfplots}
\usepackage{pifont}
\usepackage{epsfig}
\usepackage{booktabs}
\usepackage{pgffor}
\usepackage{tikz}
\usepackage{multirow}
\usepackage{amsmath}
\usepackage{bbm}

\makeatletter
\newcommand\xleftrightarrow[2][]{%
  \ext@arrow 9999{\longleftrightarrowfill@}{#1}{#2}}
\newcommand\longleftrightarrowfill@{%
  \arrowfill@\leftarrow\relbar\rightarrow}
\makeatother

\newcommand{\mb}{\mathbf}
\newcommand{\mc}{\mathcal}

\newtheorem{theorem}{\textsc{Theorem}}
\newtheorem{defn}{\textsc{Definition}}
\newtheorem{lemma}{\textsc{Lemma}}

\newcommand{\ouruma}{\textsc{Uma}}

\newcommand{\ourmli}[0]{\textsc{Mli}}

\newcommand{\ourmcd}[0]{\textsc{Mcd}}
\newcommand{\ourcad}[0]{\textsc{Cad}}
\newcommand{\ourscalable}[0]{\textsc{Spmn}}

\newcommand{\diffusion}[0]{\textsc{Tlt}}
\newcommand{\sic}[0]{\textsc{MFC}}

\newcommand{\muse}{\textbf{M}\textsc{use}}
\newcommand{\ipath}{\textsc{IPath}}

\newcommand{\dime}{\textsc{DIME}}
\newcommand{\dimesh}{\textsc{DIME-SH}}

\usepackage{algorithm}
\usepackage{algorithmic}

\begin{document}

\title{Social Network Fusion and Mining: A Survey}

\numberofauthors{1}
\author{
\alignauthor Jiawei Zhang \\
       \affaddr{IFM Lab}\\
       \affaddr{Florida State University, Tallahassee, FL 32311, USA}\\
       \email{jiawei@ifmlab.org}
       }

\maketitle

\begin{abstract}

Looking from a global perspective, the landscape of online social networks is highly fragmented. A large number of online social networks have appeared, which can provide users with various types of services. Generally, the information available in these online social networks is of diverse categories, which can be represented as heterogeneous social networks (HSN) formally. Meanwhile, in such an age of online social media, users usually participate in multiple online social networks simultaneously to enjoy more social networks services, who can act as bridges connecting different networks together. So multiple HSNs not only represent information in single network, but also fuse information from multiple networks. 

Formally, the online social networks sharing common users are named as the aligned social networks, and these shared users who act like anchors aligning the networks are called the anchor users. The heterogeneous information generated by users' social activities in the multiple aligned social networks provides social network practitioners and researchers with the opportunities to study individual user's social behaviors across multiple social platforms simultaneously. This paper presents a comprehensive survey about the latest research works on multiple aligned HSNs studies based on the broad learning setting, which covers $5$ major research tasks, i.e., \textit{network alignment}, \textit{link prediction}, \textit{community detection}, \textit{information diffusion} and \textit{network embedding} respectively. 

\end{abstract}


\section{Introduction}\label{sec:introduction}

In the real world, on the same information entities, e.g., products, movies, POIs (points-of-interest) and even human beings, a large amount of information can actually be collected from various sources. These sources are usually of different varieties, like Walmart vs Amazon for commercial products; IMDB vs Rotten Tomatoes for movies; Yelp vs Foursquare for POIs; and various online social medium websites vs diverse offline shopping, traveling, living service providers for human beings. Each information source provides a specific signature of the same entity from a unique underlying aspect. However, in many cases, these information sources are usually separated in difference places, and an effective fusion of these different information sources provides an opportunity for researchers and practitioners to understand the entities more comprehensively, which renders \textit{broad learning} \cite{icdm17, cikm17, cikm17_2} an extremely important learning task.

Broad learning introduced in \cite{icdm17, cikm17, cikm17_2} is a new type of learning task, which focuses on fusing multiple large-scale information sources of diverse varieties together and carrying out synergistic data mining tasks across these fused sources in one unified analytic. Fusing and mining multiple information sources of large volumes and diverse varieties are also the fundamental problems in big data studies. Broad learning investigates the principles, methodologies and algorithms for synergistic knowledge discovery across multiple aligned information sources, and evaluates the corresponding benefits. Great challenges exist in broad learning for the effective fusion of relevant knowledge across different aligned information sources depends upon not only the relatedness of these information sources, but also the target application problems. Broad learning aims at developing general methodologies, which will be shown to work for a diverse set of applications, while the specific parameter settings can be learned for each application from the training data.

Broad learning is a challenging problem. We categorize its main challenges into two main groups as follows:

\begin{itemize}

\item \textit{How to Fuse}: The data fusion strategy is highly dependent on the data types, and different data categories of data may required different fusion methods. For instance, for the fusion of image sources about the same entities, a necessary entity recognition step is required; to combine multiple online social networks, inference of the potential anchor link mappings the shared users across networks will be key task; meanwhile, to fuse diverse textual data, concept entity extraction or topic modeling can both be the potential options. In many cases, the fusion strategy is also correlated with the specific applications to be studied, which may pose extract constraints or requirements on the fusion results. More information about related data fusion strategies of online social networks will be introduced later in Section~\ref{sec:alignment}.

\item \textit{How to Mine}: To mine the fused data sources, there also exist many great challenges. In many of the cases, not all the data sources will be helpful for certain application tasks. For instance, in social community detection, the fused information about the users' credit card transaction will have less correlation with the social communities formed by the users. On the other hand, the information diffusion among users is regarded as irrelevant with the information sources depicting the daily commute routes of people in the real world. Among all these fused data sources, picking the useful ones is not an easy task. Several strategies, like feature selection \cite{kdd14}, meta path weighting \cite{cikm16, bigdata15}, network sampling \cite{icdm13} and information source embedding \cite{icde17_2, icdm17}, will be described in the application tasks to be introduced in Sections~\ref{sec:link_prediction}-\ref{sec:embedding} respectively.

\end{itemize}

In this paper, we will focus on introducing the broad learning research works done based on online social media data. Nowadays, to enjoy more social network services, people are usually involved in multiple online social networks simultaneously, such as Facebook, Twitter and Foursquare \cite{kdd14, cikm13}. Individuals usually have multiple separate accounts in different social networks, and discovering the correspondence between accounts of the same user (i.e., network alignment or user anchoring) \cite{icdm15, www16, cikm13, iri15, ijcai15, wsdm17_2} will be an interesting problem. What's more, network alignment is also the crucial prerequisite step for many interesting inter-network synergistic knowledge discovery applications, like (1) inter-network link prediction/recommendation \cite{ZY14, kdd14, icdm13, wsdm14, ijcai15, wsdm17_2, www16_hu, kdd15, cikm15}, (2) mutual community detection \cite{sdm15, bigdata14, bigdata15, ijcnn16, cikm17, icde17}, (3) cross-platform information diffusion \cite{pakdd15, iri16, cikm16}, and (4) multiple networks synergistic embedding \cite{icde17_2, icdm17}. These application tasks are fundamental problems in social network studies, which together with the network alignment problem will form the backbone of the multiple social network broad learning ecosystem.


This paper will cover five strongly correlated research directions in the study of broad learning on multiple online social networks:
\begin{itemize}
\item \textbf{Network Alignment}: users nowadays are usually involved in multiple online social networks simultaneously. Identifying the common users shared by different online social networks can effectively combine these networks together, which will also provide the opportunity to study users' social behaviors from a more comprehensive perspective. Many research works have proposed to align the online social networks together by inferring the mappings of the shared users between different networks, which will be introduced in great detail in this paper.

\item \textbf{Link Prediction}: users' friendship connections in different networks have strong correlations. With the social activity data across multiple aligned social networks, we can acquire more comprehensive knowledge about users and their personal social preferences and habbits. We will introduce the existing research works on the socail link prediction problem across multiple aligned social sites simultaneously.

\item \textbf{Community Detection}: information available across multiple aligned social networks provides more complete signals revealing the social community structures formed by people in the real world. We will introduce the existing research works on community detection with knowledge fused from multiple aligned heterogeneous social networks as the third task.

\item \textbf{Information Diffusion}: the formulation of multiple aligned heterogeneous social network provides researchers with the opportunity to study the information diffusion process across different social sites. The latest research papers on information diffusion problem across multiple aligned networks will be illustrated as well.

\item \textbf{Network Embedding}, information from other aligned networks can provide complimentary information for refining the feature representations of users effectively. In recent years, some research papers introduce the synergistic network embedding across aligned social networks, where knowledge from other external networks can effectively be utilized in their representation learning process mutually.
\end{itemize}

The remainder parts of this paper will be organized as follows. We will first provide the basic terminology definitions in Section~\ref{sec:terminology}. Via the anchor links, we will introduce the inter-network meta path concept in Section~\ref{sec:meta_path}, which will be extensively used in the following sections. The network alignment research papers will be introduced in Section~\ref{sec:alignment}. Inter-network link prediction and friend recommendation will be talked about in Section~\ref{sec:link_prediction}. A detailed review about cross-network community detection will be provided in Section~\ref{sec:clustering}. Broad learning based information diffusion is introduced in Section~\ref{sec:diffusion} and network embedding works are available in Section~\ref{sec:embedding}. Finally, we will illustrate several potential future development directions about broad learning and conclude this paper in Section~\ref{sec:future_works}.

\section{Terminology Definition} \label{sec:terminology}

Online social networks (OSNs) denote the online platforms which allow people to build social connections with other people, who share similar personal or career interests, backgrounds, and real-life connections. Online social networking sites vary greatly and each category of online social networks can provide a specific type of featured services. For instance, Facebook\footnote{https://www.facebook.com} allows users to socialize with each other via making friends, posting text, sharing photos/videos; Twitter\footnote{https://twitter.com} focuses on providing micro-blogging services for users to write/read the latest news and messages; Foursquare\footnote{https://foursquare.com} is a location-based social network offering location-oriented services; and Instagram\footnote{http://instagram.com} is a photo and video sharing social site among friends or to the public. To enjoy different kinds of social networks services simultaneously, users nowadays are usually involved in many of these online social networks aforementioned at the same time, in each of which they will form separate social connections and generate a large amount of social information.

Generally, the online social networks can be represented as graphs in mathematics. Besides the users, there usually exist many other types of information entities, like posts, photos, videos and comments, generated by users' online social activities. Information entities in online social networks are extensively connected, and the connections among different types of nodes usually have different physical meanings. The diverse nodes and connections render the online social networks a very complex graph structure. Meanwhile, depending on categories of information entities and connections involved, the online social networks can be divided into different types, like homogeneous network, bipartite network and heterogeneous network. To model the phenomenon that users are involved multiple networks, a new concept called ``multiple aligned heterogeneous social networks'' \cite{kdd14, cikm13} has been proposed in recent years.

For the networks with simple structures, like the homogeneous networks merely involving users and friendship links, the social patterns in them are usually easy to study. However, for the networks with complex structures, like the heterogeneous networks, the nodes can be connected by different types of link, which will have totally different physical meanings. One general technique for heterogeneous network studies is ``meta path'' \cite{SHYYW11, kdd14}, which specifically depicts certain link-sequence structures connecting node defined based on the network schema. The meta path concept can also been extended to the multiple aligned social network scenario as well, which can connect the node across different social networks.

Given a network $G = (\mathcal{V}, \mathcal{E})$, we can represent the set of node and link types involved in the network as sets $\mathcal{N}$ and $\mathcal{R}$ respectively. Based on such information, the \textit{social network} concept can be formally defined based on the graph concept by adding the mappings indicating the node and link type information.
\begin{defn}
(\textit{Social Networks}): Formally, a heterogeneous social network can be represented as $G = (\mathcal{V}, \mathcal{E}, \phi, \psi)$, where $\mathcal{V}$, $\mathcal{E}$ are the sets of nodes and links in the network, and mappings $\phi: \mathcal{V} \to \mathcal{N}$, $\psi: \mathcal{E} \to \mathcal{R}$ project the nodes and links to their specific types respectively. In many cases, the mappings $\phi$, $\psi$ are omitted assuming that the node and link types are known by default.
\end{defn}

In the following parts of this paper, depending on the categories of information involved in the online social networks, we propose to categorize the online social networks into three groups: \textit{homogeneous social networks}, \textit{heterogeneous social networks} and \textit{aligned heterogeneous social networks}. Several important concepts about social networks that will be used throughout this paper will be introduced as follows.

\subsection{Homogeneous Social Network}

\begin{defn}
(\textit{Homogeneous Social Network}): For a online social network $G$, if there exists one single type of nodes and links in the network (i.e., $|\mathcal{N}| = |\mathcal{R}| = 1$), then the network is called a \textit{homogeneous social network}.
\end{defn}

Besides the online social networks involving users and friendship links only, many different types of network structures can also be represented as the \textit{homogeneous networks} actually. Several representative examples include company internal organizational network involving employees and management relationships, and computer networks involving PCs and their networking connections. \textit{Homogeneous networks} are one of the simplest network structure, analysis of which can provide many basic knowledge for studying networks with more complex structures.

Given a \textit{homogeneous social network} $G = (\mathcal{V}, \mathcal{E})$ with user set $\mathcal{V}$ and social relationship set $\mathcal{E}$, depending on whether the links in $G$ are directed or undirected, the social link can denote either the \textit{follow} links or \textit{friendship} links among individuals. Given an individual user $u \in \mathcal{V}$ in a undirected friendship social network, the set of users connected to $u$ can be represented as the friends of user $u$ in the network $G$, denoted as $\Gamma(u) \subset \mathcal{V} = \{v | v \in \mathcal{V} \land (u, v) \in \mathcal{E}\}$. The number of friends that user $u$ has in the network is also called the degree of node $u$, i.e., $|\Gamma(u)|$.

Meanwhile, in a directed network $G$, the set individuals followed by $u$ (i.e., $\Gamma_{out}(u) = \{v | v \in \mathcal{V} \land (u, v) \in \mathcal{E}\}$) are called the set of followees of $u$; and the set of individuals that follow $u$ (i.e., $\Gamma_{out}(u) = \{v | v \in \mathcal{V} \land (v, u) \in \mathcal{E}\}$) are called the set of followers of $u$. The number of users who follow $u$ is called the in-degree of $u$, and the number of users followed by $u$ is called the out-degree of $u$ in the network. For the users with large out-degrees, they are called the \textit{hubs} \cite{K99} in the network; while those with large in-degrees, they are called the \textit{authorities} \cite{K99} in the network.

\subsection{Heterogeneous Social Network}\label{chap1_subsec:sec3_subsec2_heterogeneous_network}

\begin{defn}
(\textit{Heterogeneous Social Network}): For a online social network $G$, if there exists multiple types of nodes or links in the network (i.e., $|\mathcal{N}| > 1$, or $|\mathcal{R}| > 1$), then the network is called a \textit{heterogeneous social network}.
\end{defn}

Most of the graph-structured networks in the real world may contain very complex information involving multiple types of nodes and connections. Representative examples include \textit{heterogeneous social networks} involving users, posts, check-ins, words and timestamps, as well as the friendship links, write links and contain links among these nodes; \textit{bibliographic network} including authors, papers, conferences and the write, cite, and publish-in links among them; and \textit{movie knowledge libraries} containing movies, casts, reviewers, reviews and ratings, as well as the complex links among these nodes. The \textit{neighbor}, \textit{degree}, \textit{hub} and \textit{authority} concepts introduced before for the \textit{homogeneous networks} can be applied to the \textit{heterogeneous networks} as well.

Formally, the online social network mentioned above can be defined as $G = (\mathcal{V}, \mathcal{E})$, where $\mathcal{V}$ denotes the set of nodes and $\mathcal{E}$ represent the set of links in $G$. The node set $\mathcal{V}$ can be divided into several subsets $\mathcal{V} = \mathcal{U} \cup \mathcal{P} \cup \mathcal{L} \cup \mathcal{T} \cup \mathcal{W}$ involving the user nodes, post nodes, location nodes, word nodes and timestamp nodes respectively. The link set $\mathcal{E}$ can be divided into several subsets as well, $\mathcal{E} = \mathcal{E}_{u,u} \cup \mathcal{E}_{u,p} \cup \mathcal{E}_{p,l} \cup \mathcal{E}_{p,w} \cup \mathcal{E}_{p,t}$, containing the links among users, the links between users and posts, and those between posts with location checkins, words, and timestamps.

In the \textit{heterogeneous social networks}, each node can be connected with a set of nodes belonging to different categories via various type of connections. For example, given a user $u \in \mathcal{U}$, the set of user node incident to $u$ via the friend links can be represented as the online friends of $u$, denoted as set $\{v | v \in \mathcal{U}, (u, v) \in \mathcal{E}_{u,u}\}$; the set of post node incident to $u$ via the write links can be represented as the posts written by $u$, denoted as set $\{w | w \in \mathcal{P}, (u, w) \in \mathcal{E}_{u,p}\}$. The location check-in nodes, word nodes and timestamp nodes are not directly connected to the user node, while via the post nodes, we can also obtain the set of locations/words/timestamps that are visited/used/active-at by user $u$ in the network. Such a indirect connection can be described more clearly by the \textit{meta path} concept more clearly in Section~\ref{sec:meta_path}.

\subsection{Aligned Heterogeneous Social Networks}

\begin{defn}
(\textit{Multiple Aligned Heterogeneous Networks}): Formally, the \textit{multiple aligned heterogeneous networks} involving $n$ networks can be defined as $\mathcal{G} = ((G^{(1)}, G^{(2)},  \cdots, G^{(n)}), \\(\mathcal{A}^{(1, 2)}, \mathcal{A}^{(1, 3)}, \cdots, \mathcal{A}^{(n-1, n)}))$, where $G^{(1)}, G^{(2)},  \cdots, G^{(n)}$ denote these $n$ heterogeneous social networks and the sets $\mathcal{A}^{(1, 2)}, \\\mathcal{A}^{(1, 3)}, \cdots, \mathcal{A}^{(n-1, n)}$ represent the undirected \textit{anchor links} aligning these networks respectively.
\end{defn}

\textit{Anchor links} actually refer to the mappings of information entities across different sources, which correspond to the the same information entity in the real world, e.g., users in online social networks, authors in different bibliographic networks, and movies in the movie knowledge libraries.

\begin{defn}
(\textit{Anchor Link}): Given two heterogeneous networks $G^{(i)}$ and $G^{(j)}$ which share some common information entities, the set of \textit{anchor links} connecting $G^{(i)}$ and $G^{(j)}$ can be represented as set $\mathcal{A}^{(i,j)} = \{(u^{(i)}_m, u^{(j)}_n) | u^{(i)}_m \in \mathcal{V}^{(i)} \land u^{(j)}_n \in \mathcal{V}^{(j)} \land u^{(i)}_m, u^{(j)}_n$ {denote the same information entity}$\}$.
\end{defn}

The \textit{anchor links} depict a transitive relationship among the information entities across different networks. Given $3$ information entities $u^{(i)}_m$, $u^{(j)}_n$, $u^{(k)}_o$ from networks $G^{(i)}$, $G^{(j)}$ and $G^{(k)}$ respectively, if $u^{(i)}_m$, $u^{(j)}_n$ are connected by an anchor link and $u^{(j)}_n$, $u^{(k)}_o$ are connected by an anchor link, then the user pair $u^{(i)}_m$, $u^{(k)}_o$ will be connected by an anchor link by default. For more detailed definitions about other related terms, like \textit{anchor users}, \textit{non-anchor users}, \textit{full alignment}, \textit{partial alignment} and \textit{non-alignment}, please refer to \cite{kdd14}.

\section{Meta Path}\label{sec:meta_path}

To deal with the social networks, especially the heterogeneous social networks, a very useful tool is \textit{meta paths} \cite{SHYYW11, kdd14}. \textit{Meta path} is a concept defined based on the network schema, outlining the connections among nodes belonging to different categories. For the nodes which are not directly connected, their relationships can be depicted with the meta path concept. In this part, we will define the meta path concept, and introduce a set of meta paths within and across real-world heterogeneous social networks respectively.

\subsection{Network Schema}

Given a network $G = (\mathcal{V}, \mathcal{E})$, we can define its corresponding \textit{network schema} to describe the categories of nodes and links involved in $G$. 

\begin{defn}
(\textit{Network Schama}): Formally, the network schema of network $G$ can be represented as $S_G = (\mathcal{N}, \mathcal{R})$, where $\mathcal{N}$ and $\mathcal{R}$ denote the node type set and link type set of network $G$ respectively. 
\end{defn}

Network schema provides a meta level description of networks. Meanwhile, if a network $G$ can be outlined by the network schema $S_G$, $G$ is also called a \textit{network instance} of the network schema. For a given node $u \in \mathcal{V}$, we can represent its corresponding node type as $\phi(u) = N \in \mathcal{N}$, and call $u$ is an instance of node type $N$, which can also be denoted as $u \in N$ for simplicity. Similarly, for a link $(u, v)$, we can denotes its link type as $\psi((u, v)) = R \in \mathcal{R}$, or $(u, v) \in R$ for short. The inverse relation $R^{-1}$ denotes a new link type with reversed direction. Generally, $R$ is not equal to $R^{-1}$, unless $R$ is symmetric.

\subsection{Meta Path in Heterogeneous Social Networks}

Meta path is a concept defined based on the network schema denoting the correlation of nodes based on the heterogeneous information (i.e., different types of nodes and links) in the networks.

\begin{defn}
(\textit{Meta Path}): A meta path $P$ defined based on the network schema $S_G = (\mathcal{N}, \mathcal{R})$ can be represented as $P = N_1 \xrightarrow{R_1} N_2 \xrightarrow{R_2} \cdots N_{k-1} \xrightarrow{R_{k-1}} N_{k}$, where $N_i \in \mathcal{N}, i \in \{1, 2, \cdots, k\}$ and $R_i \in \mathcal{R}, i \in \{1, 2, \cdots, k-1\}$.
\end{defn}

Furthermore, depending on the categories of node and link types involved in the meta path, we can specify the meta path concept into several more refined groups, like \textit{homogeneous meta path} and \textit{heterogeneous meta path}, or \textit{social meta path} and other \textit{meta paths}.

\begin{defn}
(\textit{Homogeneous/Heterogeneous Meta Path}): Let $P = N_1 \xrightarrow{R_1} N_2 \xrightarrow{R_2} \cdots N_{k-1} \xrightarrow{R_{k-1}} N_{k}$ denote a meta path defined based on the network schema $S_G = (\mathcal{N}, \mathcal{R})$. If all the node types and link types involved in $P$ are of the same category, $P$ is called a \textit{homogeneous meta path}; otherwise, $P$ is called a \textit{heterogeneous meta path}.
\end{defn}

The meta paths can connect any kinds of node type pairs, and specifically, for the meta paths starting and ending with the user node types, those meta paths are called the \textit{social meta paths}.
\begin{defn}
(\textit{Social Meta Path}): Let $P = N_1 \xrightarrow{R_1} N_2 \xrightarrow{R_2} \cdots N_{k-1} \xrightarrow{R_{k-1}} N_{k}$ denote a meta path defined based on the network schema $S_G = (\mathcal{N}, \mathcal{R})$. If the starting and ending node types $N_1$ and $N_k$ are both the user node type, $P$ is called a \textit{social meta path}. 
\end{defn}

Users are usually the focus in social network studies, and the \textit{social meta paths} are frequently used in both research and real-world applications and services. If all the node types in the meta paths are all user node type and the link types are also of a common category, then the meta path is called the \textit{homogeneous social meta path}. The number of path segments in the meta path is called the meta path length. For instance, the length of meta path $P = N_1 \xrightarrow{R_1} N_2 \xrightarrow{R_2} \cdots N_{k-1} \xrightarrow{R_{k-1}} N_{k}$ is $k-1$. Meta paths can also been concatenated together with the \textit{meta path composition operator}.

\begin{defn}
(\textit{Meta Path Composition}): Meta paths $P^1 = N^1_1 \xrightarrow{R^1_1} N^1_2 \xrightarrow{R^1_2} \cdots N^1_{k-1} \xrightarrow{R^1_{k-1}} N^1_{k}$, and $P^2 = N^2_1 \xrightarrow{R^2_1} N^2_2 \xrightarrow{R^2_2} \cdots N^2_{l-1} \xrightarrow{R^2_{l-1}} N^1_{l}$ can be concatenated together to form a longer meta path $P = P^1 \circ P^2 = N^1_1 \xrightarrow{R^1_1} \cdots \xrightarrow{R^1_{k-1}} N^1_{k} \xrightarrow{R^2_1} N^2_2 \xrightarrow{R^2_2} \cdots N^2_{l-1} \xrightarrow{R^2_{l-1}} N^1_{l}$, if the ending node type of $P^1$ is the same as the starting node type of $P^2$, i.e., $N^1_{k} = N^2_1$. The new composed meta path is of length $k+l -2$.
\end{defn}

Meta path $P = N_1 \xrightarrow{R_1} N_2 \xrightarrow{R_2} \cdots N_{k-1} \xrightarrow{R_{k-1}} N_{k}$ can also been treated as the concatenation of simple meta paths $N_1 \xrightarrow{R_1} N_2$, $N_2 \xrightarrow{R_2} N_3$, $\cdots$, $N_{k-1} \xrightarrow{R_{k-1}} N_{k}$, which can be represented as $P  = R_1 \circ R_2 \circ \cdots \circ R_{k-1} \circ R_k$. 

\subsection{Meta Path across Aligned Heterogeneous Social Networks}

Besides the meta paths within one single heterogeneous network, the meta paths can also be defined across multiple aligned heterogeneous networks via the \textit{anchor meta paths}.

\begin{defn}
(\textit{Anchor Meta Path}): Let $G^{(1)}$ and $G^{(2)}$ be two aligned heterogeneous networks sharing the common anchor information entity of types $N^{(1)} \in \mathcal{N}^{(1)}$ and $N^{(2)} \in \mathcal{N}^{(2)}$ respectively. The anchor meta path between the schemas of networks $G^{(1)}$ and $G^{(2)}$ can be represented as meta path $\Phi = N^{(1)} \xleftrightarrow{Anchor} N^{(2)}$ of length $1$. 
\end{defn}

The \textit{anchor meta path} is the simplest meta path across aligned networks, and a set of inter-network meta paths can be defined based on the intra-network meta paths and the anchor meta path.

\begin{defn}
(\textit{Inter-Network Meta Path}): A meta path $\Psi = N_1 \xrightarrow{R_1} N_2 \xrightarrow{R_2} \cdots N_{k-1} \xrightarrow{R_{k-1}} N_{k}$ is called an \textit{inter-network meta path} between networks $G^{(1)}$ and $G^{(2)}$ iff $\exists m \in \{1, 2, \cdots, k-1\}, R_m = Anchor$.
\end{defn}

The \textit{inter-network meta paths} can be viewed as a composition of \textit{intra-network meta paths} and the \textit{anchor meta path} via the user node types. An \textit{inter-network meta path} can be a meta path starting with an \textit{anchor meta path} followed by the \textit{intra-network meta paths}, or those with \textit{anchor meta paths} in the middle. Here, we would like to introduce several categories \textit{inter-network meta paths} involving the anchor meta paths at different positions as defined in \cite{kdd14}:
\begin{itemize}
\item $\Psi(G^{(1)}, G^{(2)}) = \Phi(G^{(1)}, G^{(2)})$, which denotes the simplest \textit{inter-network meta path} composed of the anchor meta path only between networks $G^{(1)}$ and $G^{(2)}$.

\item $\Psi(G^{(1)}, G^{(2)}) = \Phi(G^{(1)}, G^{(2)}) \circ P(G^{(1)})$, which denotes the \textit{inter-network meta path} starting with an \textit{anchor meta path} and followed by the \textit{intra-network social meta path} in network $G^{(2)}$.

\item $\Psi(G^{(1)}, G^{(2)}) = P(G^{(1)}) \circ \Phi(G^{(1)}, G^{(2)})$, which denotes the \textit{inter-network meta path} starting with the \textit{intra-network social meta path} in network $G^{(1)}$ followed by an \textit{anchor meta path} between networks $G^{(1)}$ and $G^{(2)}$.

\item $\Psi(G^{(1)}, G^{(2)}) = P(G^{(1)}) \circ \Phi(G^{(1)}, G^{(2)}) \circ P(G^{(2)})$, which denotes the \textit{inter-network meta path} starting and ending with the \textit{intra-network social meta path} in networks $G^{(1)}$ and $G^{(2)}$ respectively connected by an \textit{anchor meta path} between networks $G^{(1)}$ and $G^{(2)}$.

\item $\Psi(G^{(1)}, G^{(2)}) = P(G^{(1)}) \circ \Phi(G^{(1)}, G^{(2)}) \circ P(G^{(2)}) \circ \Phi(G^{(2)}, G^{(1)})$, which denotes the \textit{inter-network meta path} starting and ending with node types in network $G^{(1)}$ and traverse across the networks twice via the \textit{anchor meta path}.

\item $\Psi(G^{(1)}, G^{(2)}) = P(G^{(1)}) \circ \Phi(G^{(1)}, G^{(2)}) \circ P(G^{(2)}) \circ \Phi(G^{(2)}, G^{(1)}) \circ P(G^{(1)})$, which denotes the \textit{inter-network meta path} starting and ending with the \textit{intra-network social meta paths} in network $G^{(1)}$ and traverse across the networks twice via the \textit{anchor meta path} between them.
\end{itemize} 

These meta path concepts introduced in this section will be widely used in various social network broad learning tasks to be introduced later. 

\section{Network Alignment}\label{sec:alignment}

Network alignment is an important research problem and dozens of papers have been published on this topic in the past decades. Depending on specific disciplines, the studied networks can be social networks in data mining \cite{icdm15, www16, cikm13, iri15, ijcai15, wsdm17_2} protein-protein interaction (PPI) networks and gene regulatory networks in bioinformatics \cite{KBS08, SP12, LLBSB09, SXB07}, chemical compound in chemistry \cite{SHL08}, data schemas in data warehouse \cite{MGR02}, ontology in web semantics \cite{DMDH04}, graph matching in combinatorial mathematics \cite{MH14}, as well as graphs in computer vision \cite{CFSV04, BGGSW09}.

In bioinformatics, the network alignment problem aims at predicting the best mapping between two biological networks based on the similarity of the molecules and their interaction patterns. By studying the cross-species variations of biological networks, network alignment problem can be applied to predict conserved functional modules \cite{SSKKMUSKI05} and infer the functions of proteins \cite{PSBLB11}. Graemlin \cite{FNSMB06} conducts pairwise network alignment by maximizing an objective function based on a set of learned parameters. Some works have been done on aligning multiple network in bioinformatics. IsoRank proposed in \cite{SXB08} can align multiple networks greedily based on the pairwise node similarity scores calculated with spectral graph theory. IsoRankN \cite{LLBSB09} further extends IsoRank by exploiting a spectral clustering scheme in the alignment model.

In recent years, with rapid development of online social networks, researchers' attention starts to shift to the alignment of social networks. Enlightened by the homogeneous network alignment method in \cite{U88}, Koutra et al. \cite{KTL13} propose to align two bipartite graphs with a fast alignment algorithm. Zafarani et al. \cite{ZL13} propose to match users across social networks based on various node attributes, e.g., username, typing patterns and language patterns etc. Kong et al. formulate the heterogeneous social network alignment problem as an anchor link prediction problem. A two-step supervised method MNA is proposed in \cite{cikm13} to infer potential anchor links across networks with heterogeneous information in the networks. However, social networks in the real world are mostly partially aligned actually and lots of users are not anchor users. Zhang et al. have proposed a partial network alignment method specifically in \cite{iri15}. 

In the social network alignment model building, the anchor links are very expensive to label manually, and achieving a large-sized anchor link training set can be extremely challenging. In \cite{ijcai15}, Zhang et al. propose to study the network alignment problem based on the PU (Positive and Unlabeled) learning setting instead, where the model is built based on a small amount of positive set and a large unlabeled set. Furthermore, in the case when no training data is available, via inferring the potential anchor user mappings across networks, Zhang et al. have introduced an unsupervised network alignment models for multiple (more than $2$) social networks in \cite{icdm15} and an unsupervised network concurrent alignment model via multiple shared information entities simultaneously in \cite{www16}.

In this section, we will introduce the social network alignment methods based on the \textit{supervised learning}, \textit{unsupervised learning} and \textit{semi-supervised learning} settings respectively.

\subsection{Supervised Network Alignment}

Formally, let $G^{(1)} = (\mathcal{V}^{(1)}, \mathcal{E}^{(1)})$ and $G^{(2)} = (\mathcal{V}^{(2)}, \mathcal{E}^{(2)})$ denote two online social networks, where $\mathcal{V}^{(1)}$/$\mathcal{V}^{(2)}$ and $\mathcal{E}^{(1)}$/$\mathcal{E}^{(2)}$ denote the sets of nodes and links involved in these two networks respectively. Let set $\mathcal{A}_{train}$ denotes the set of labeled anchor links connecting networks $G^{(1)}$ and $G^{(2)}$, we can represent the set of anchor links without known labels as the test set $\mathcal{A}_{test} \subseteq \mathcal{U}^{(1)} \times \mathcal{U}^{(2)} \setminus \mathcal{A}_{train}$. 

In the supervised network alignment problem, a set of features will be extracted for the anchor links with the heterogeneous information available across the social networks. Meanwhile, the existing and non-existing anchor links will be labeled as positive and negative instances respectively. Based on the training set $\mathcal{A}_{train}$, we can represent the feature vectors and labels of links in the set as a group of tuples $\{(\mb{x}_l, y_l)\}_{l \in \mathcal{A}_{train}}$, where $\mb{x}_l$ represents the feature vector extracted for anchor link $l$ and $y_l \in \{-1, +1\}$ denotes its label. Based on the training set, we aim at building a mapping $f: \mathcal{A}_{test} \to \{-1, +1\}$ to determine the labels of the anchor links in the test set. To address the problem, we will take the supervised network alignment model proposed in \cite{cikm13} as an example to illustrate the problem setting and potential solutions.

\subsubsection{Anchor Link Feature Extraction}

The supervised network alignment model proposed in \cite{cikm13} involves three main phases: (1) feature extraction, (2) classification model building, and (3) network matching. One of the main goal in supervised network alignment is to extract discriminative social features for a pair of user accounts between two disjoint social networks. Intuitively, the social neighbors of each user account can only involve users from the same social network, which will have no common neighbors actually. For example, the neighbors for a Facebook user will only involve the other users in Facebook, which has no overlap with his neighbors in Twitter (which contains the Twitter users only). However, in anchor link prediction problem, we need to extract a set of features for the anchor links between two different networks, which can be a challenging problem. In the following, we will introduce several social features proposed in \cite{cikm13} for the multi-network settings specifically.

Let $(u_i^{(1)}, u_j^{(2)})$ be a potential anchor link between these two networks, and $\mathcal{A}_{train}^+ \subset \mathcal{A}_{train}$ be the set of positively labeled anchor links in the training set. \cite{cikm13} proposes to extend the definition of some commonly used social features in link prediction, i.e., ``common neighbors'', ``Jaccard's coefficient'' and ``Adamic/Adar measure'', to extract effective features for these anchor links based on the known anchor links in set $\mathcal{A}_{train}^+$.

\noindent \textbf{Extended Common Neighbor}

The {extended common neighbor} (ECN) $CN(u^{(1)}_i, u^{(2)}_j)$ represents the number of `common' neighbors between $u^{(1)}_i$ in network $G^{(1)}$ and $u^{(2)}_j$ in network $G^{(2)}$. We denote the neighbors of $u^{(1)}_i$ in network $G^{(1)}$ as $\Gamma(u^{(1)}_i)$, and the neighbors of $u^{(2)}_j$ in network $G^{(2)}$ as $\Gamma(u^{(2)}_j)$. It is easy to identify that the sets $\Gamma(u^{(1)}_i)$ and $\Gamma(u^{(2)}_j)$ contain the users from two different networks respectively, which are isolated without any common entries. 

Meanwhile, based on the existing anchor links $\mathcal{A}_{train}^+$, some of the users in $\Gamma(u^{(1)}_i)$ and $\Gamma(u^{(2)}_j)$ can correspond to the accounts of the same users in these two networks, who are actually connected by the anchor links in $\mathcal{A}_{train}^+$. Based on such an intuition, \cite{cikm13} defines the {extended common neighbor} measure between these two users as the number of shared anchor users in their neighbor sets respectively.

\begin{defn}
(\textit{Extended Common Neighbor}): The measure of \textit{extended common neighbor} is defined as the number of known anchor links between  $\Gamma(u^{(1)}_i)$ and $\Gamma(u^{(2)}_j)$. \begingroup\makeatletter\def\f@size{8}\check@mathfonts
\begin{align}
ECN(u^{(1)}_i, u^{(2)}_j) = & \Big|\{ (u^{(1)}_p , u^{(2)}_q) | (u^{(1)}_p , u^{(2)}_q) \in \mathcal{A}_{train}^+, \\ 
&\ \ u^{(1)}_p\in \Gamma(u^{(1)}_i), u^{(2)}_q \in \Gamma(u^{(2)}_j)  \}\Big| \\
		= & \left| \Gamma(u^{(1)}_i) \bigcap_{\mathcal{A}_{train}^+}  \Gamma(u^{(2)}_j) \right|.
\end{align}\endgroup
\end{defn}

\noindent \textbf{Extended Jaccard's Coefficient}

\cite{cikm13} also extends the measure of Jaccard's coefficient to multi-network setting using similar method of extending common neighbor. $EJC(u^{(1)}_i, u^{(2)}_j)$ is a normalized version of common neighbors, i.e., $ECN(u^{(1)}_i, u^{(2)}_j)$  divided by the total number of distinct users in $\Gamma(u^{(1)}_i) \cup \Gamma(u^{(2)}_j)$

\begin{defn}
(\textit{Extended Jaccard's Coefficient}): Given the neighborhood set of users $u^{(1)}_i$ and $u^{(2)}_j$ in networks $G^{(1)}$ and $G^{(2)}$ respectively, the {Extended Jaccard's Coefficient} of user pair $u^{(1)}_i$ and $u^{(2)}_j$ can be represented as \begingroup\makeatletter\def\f@size{7}\check@mathfonts
\begin{equation}
EJC(u^{(1)}_i, u^{(2)}_j) =  \frac{ \left| \Gamma(u^{(1)}_i) \bigcap_{\mc{A}_{train}^+}  \Gamma(u^{(2)}_j) \right| } { \left| \Gamma(u^{(1)}_i) \bigcup_{\mc{A}_{train}^+}  \Gamma(u^{(2)}_j) \right|},
\end{equation}
where
\begin{align}
&\left| \Gamma(u^{(1)}_i) \bigcup_{\mc{A}_{train}^+}  \Gamma(u^{(2)}_j) \right| \\
&= | \Gamma(u^{(1)}_i) | + | \Gamma(u^{(2)}_j) | -   \left| \Gamma(u^{(1)}_i) \bigcap_{\mc{A}_{train}^+}  \Gamma(u^{(2)}_j) \right|.
\end{align}\endgroup
\end{defn}

\noindent \textbf{Extended Adamic/Adar Index}

Similarly, \cite{cikm13} also extends the Adamic/Adar Measure into multi-network settings, where the common neighbors are weighted by their average degrees in both social networks.

\begin{defn}
(\textit{Extended Adamic/Adar Index}): The {Extended Adamic/Adar Index} of the user pairs $u^{(1)}_i$ and $u^{(2)}_j$ across networks can be represented as \begingroup\makeatletter\def\f@size{6}\check@mathfonts 
\begin{align}
\hspace{-5pt}&EAA(u^{(1)}_i, u^{(2)}_j) \\
\hspace{-5pt}&= \sum_{ (u^{(1)}_p, u^{(2)}_q) \in  \Gamma(u^{(1)}_i) \bigcap_{\mathcal{A}_{train}^+}   \Gamma(u^{(2)}_j)} \hspace{-15pt} \log^{-1} \left(\frac{  | \Gamma(u^{(1)}_p) |  +| \Gamma(u^{(2)}_q) | }{2}\right).
\end{align}\endgroup
\end{defn}

In the EAA definition, for the common neighbor shared by $u^{(1)}_i$ and $u^{(2)}_j$, their degrees are defined as the average of their degrees in networks $G^{(1)}$ and $G^{(2)}$. Considering that different networks are of different scales, like Twitter if far larger than Twitter, the node degree measure can be dominated by the degree of the larger networks. Some other weighted form of the degree measure, like $\alpha \cdot  | \Gamma(u^{(1)}_p) |  + (1 - \alpha) \cdot | \Gamma(u^{(2)}_q)$ ($\alpha \in [0, 1]$), can be applied to replace $\frac{  | \Gamma(u^{(1)}_p) |  +| \Gamma(u^{(2)}_q) | }{2}$ in the definition.

In addition to the social features mentioned above, heterogeneous social networks also involve abundant information about: where, when and what. A number of features extracted by exploiting the spatial, temporal and text content information can also be extracted to facilitate anchor link prediction, which have been introduced in detail in \cite{cikm13}.

\subsubsection{Anchor Link Model Building}

Given the multiple aligned social networks, via manual labeling, the sets of identified existing and non-existing anchor links can be denoted as $\mathcal{A}_{train}^+$ and $\mathcal{A}_{train}^-$ respectively. The anchor links in sets $\mathcal{A}_{train}^+$ and $\mathcal{A}_{train}^-$ are assigned with the positive and negative labels respectively, i.e., $\{-1, +1\}$, depending on whether they exist or not. For instance, given a link $l \in \mathcal{A}_{train}^+$, it will be associated with a positive label, i.e., $y_l = +1$; while if link $l \in \mathcal{A}_{train}^-$, it will be associated with a negative label, $y_l = -1$. With the information in these aligned heterogeneous social networks, a set of features introduced in the previous subsection can be extracted for the links in sets $\mathcal{A}_{train}^+$ and $\mathcal{A}_{train}^-$. For instance, for a link $l$ in the training set $\mathcal{A}_{train}^+$ (or $\mathcal{A}_{train}^-$), we can represent its feature vector as $\mb{x}_l$, which will be called an {anchor link instance} and each feature is an {attribute} of the anchor link. With these anchor link instances and their labels, a classification model, like SVM (support vector machine), Decision Tree, or neural networks, can be trained. Meanwhile, in its test procedure, for each link $l$ in the test set $\mathcal{A}_{test}$, a similar set of features (or attributes) can be extracted, which can be denoted as its feature vector as $\mb{x}_l$. However, without knowledge about its label, the main objective of Step (2) is to determine whether the potential anchor links in set $\mathcal{A}_{test}$ exists or not (its label is positive or negative). By applying the trained to the feature vector of the anchor link, we will obtain a prediction label, which will be returned as the result of Step (2).

\subsubsection{Network Matching}

\begin{figure}[t]
\centering
\subfigure[input/ranking scores]{ \label{fig_chap4_sec3_network_matching_1}
    \begin{minipage}[l]{0.4\columnwidth}
      \centering
      \includegraphics[width=.8\textwidth]{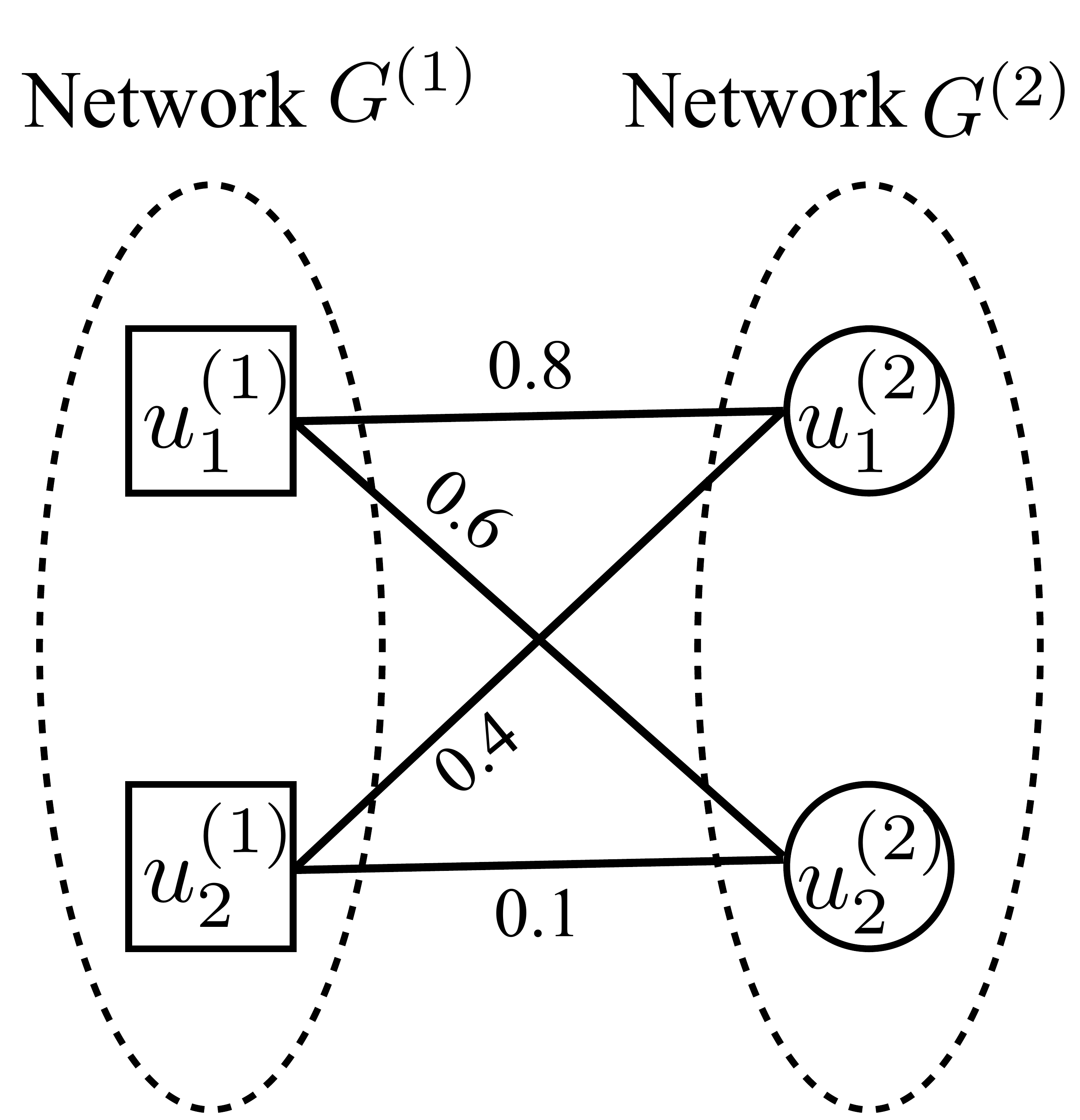}
    \end{minipage}
}
\subfigure[link prediction]{ \label{fig_chap4_sec3_network_matching_2}
    \begin{minipage}[l]{0.4\columnwidth}
      \centering
      \includegraphics[width=.8\textwidth]{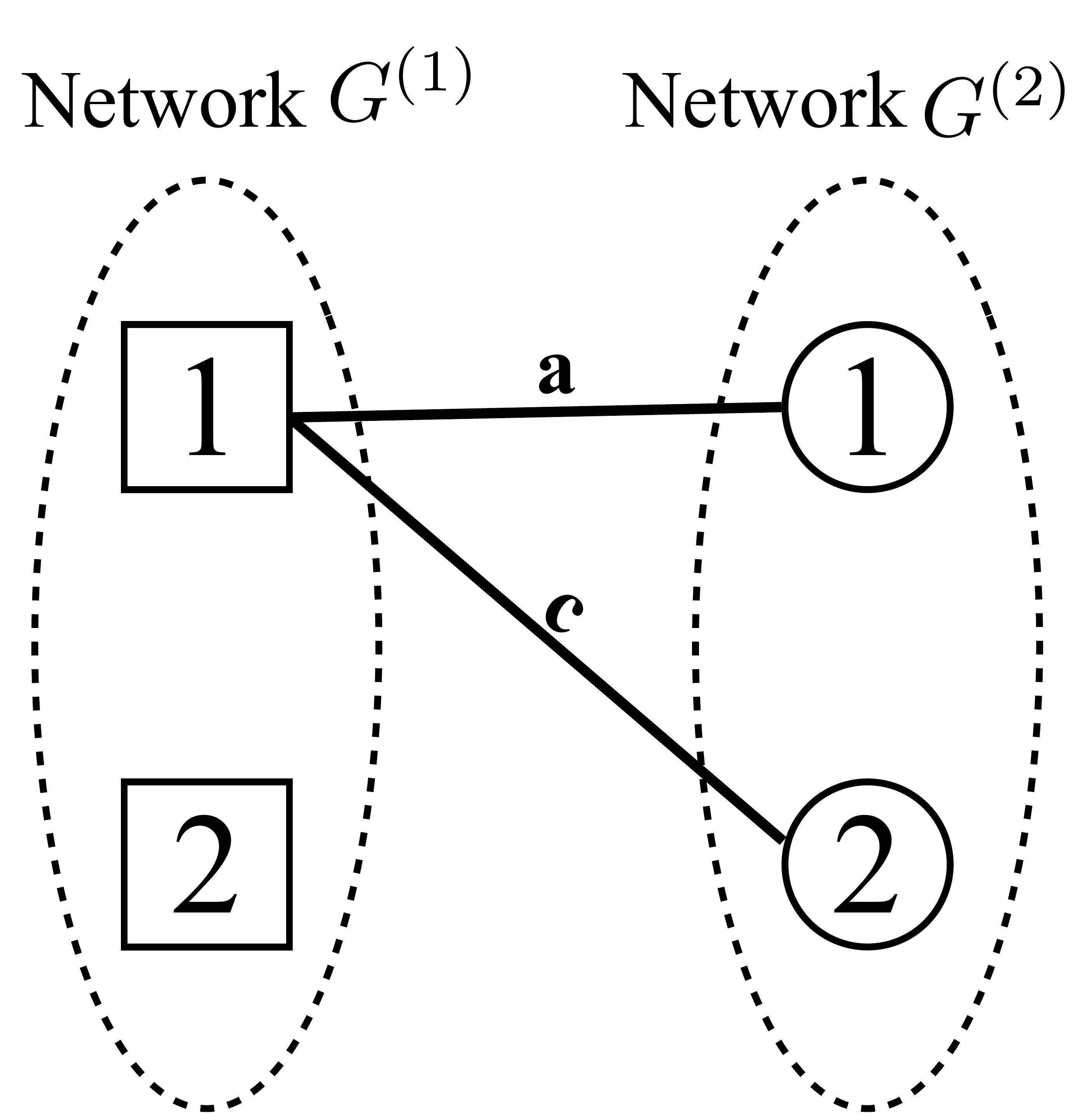}
    \end{minipage}
}
\subfigure[maximize sum of weights (1:1 constrained)]{ \label{fig_chap4_sec3_network_matching_3}
    \begin{minipage}[l]{0.4\columnwidth}
      \centering
      \includegraphics[width=.8\textwidth]{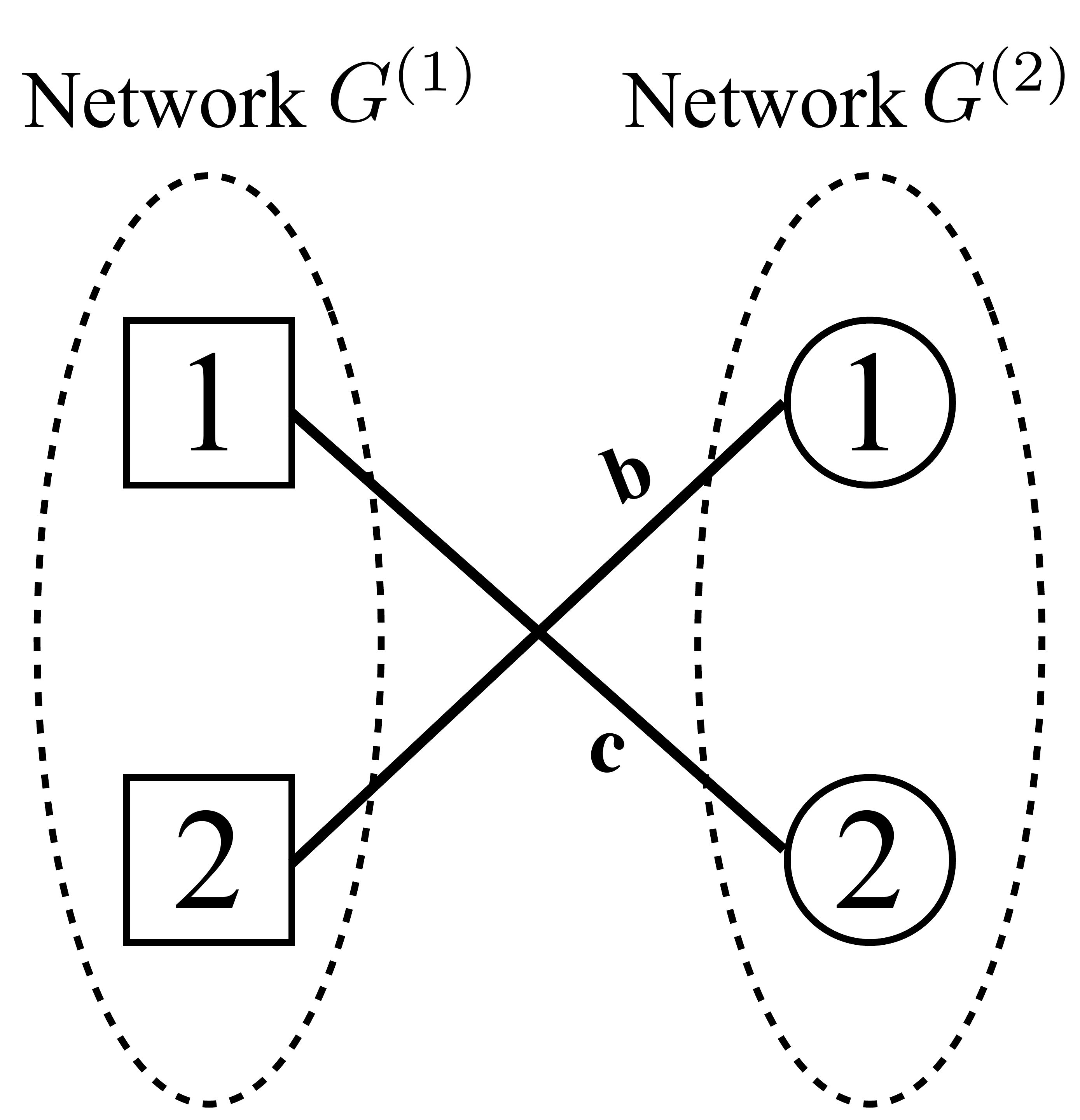}
    \end{minipage}
}
\subfigure[MNA method ]{ \label{fig_chap4_sec3_network_matching_4}
    \begin{minipage}[l]{0.4\columnwidth}
      \centering
      \includegraphics[width=.8\textwidth]{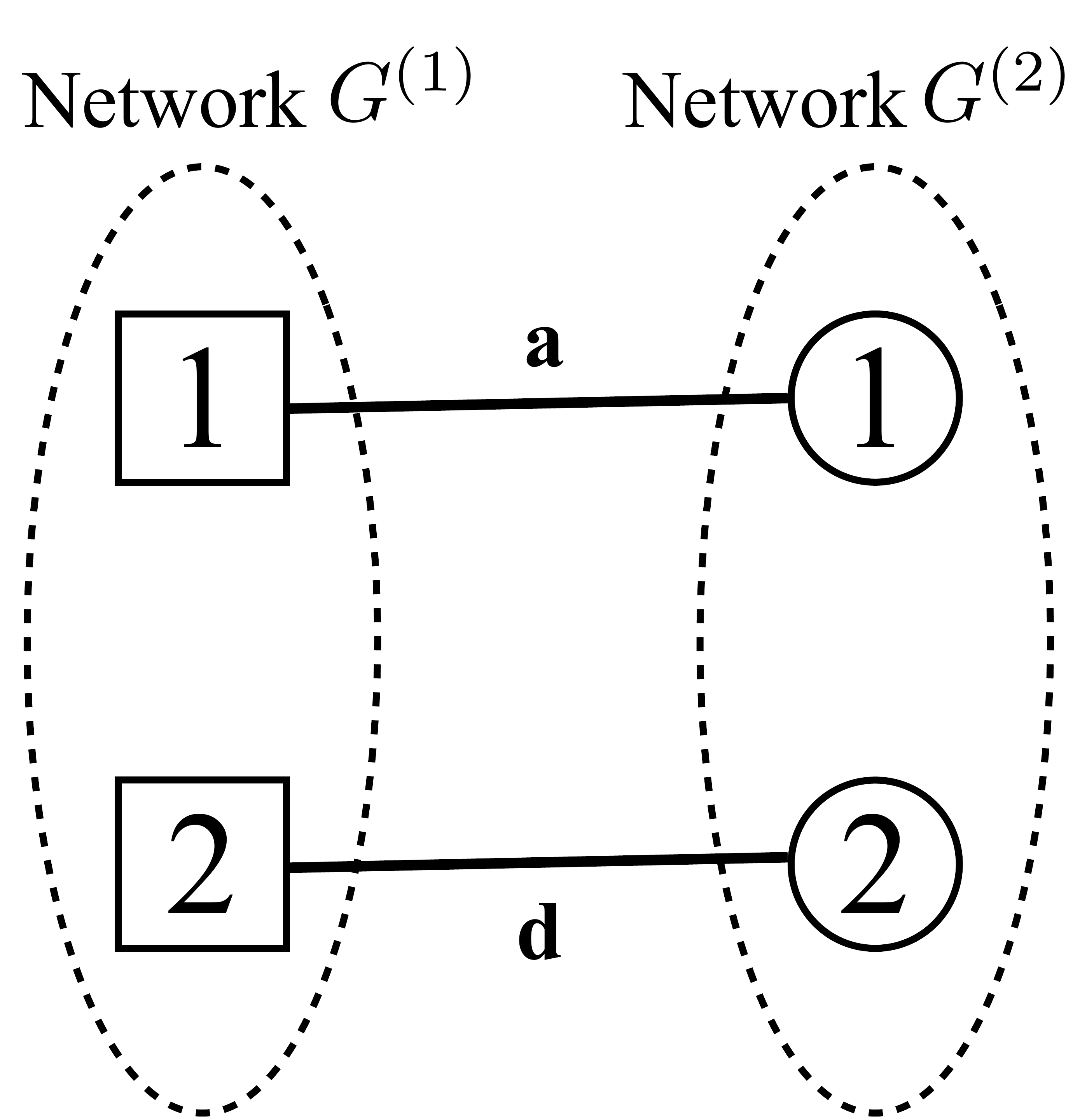}
    \end{minipage}
}
\caption{An example of anchor link inference by different methods. (a) is the input, ranking scores. (b)-(d) are the results of different methods for anchor link inference.}\label{fig_chap4_sec3_network_matching}
\end{figure}

However, in the inference process, the predictions of the binary classifier cannot be directly used as anchor links due to the following issues:
\begin{itemize}
	\item The inference of conventional classifiers are designed for constraint-free settings, and the one-to-one constraint \cite{cikm13, wsdm17_2} on anchor links may not necessarily hold in the label prediction of the classifier (SVM).
	\item Most classifiers also produce output scores, which can be used to rank the data points in the test set. 
	However, these ranking scores are uncalibrated in scale to anchor link prediction task. 
	Previous classifier calibration methods \cite{ZE02} apply only to classification problems without any constraint.
\end{itemize}

In order to tackle the above issues, \cite{cikm13} introduces an inference process, called {MNA} (Multi-Network Anchoring), to infer anchor links based upon the ranking scores of the classifier. This model is motivated by the \textit{stable marriage problem} \cite{DF81} in mathematics. 

We first use a toy example in Figure~\ref{fig_chap4_sec3_network_matching} to illustrate the main idea of {MNA}.
Suppose in Figure~\ref{fig_chap4_sec3_network_matching_1}, we are given the ranking scores from the classifiers, between the $4$ user pairs two networks (i.e., network $G^{(1)}$ and network $G^{(2)}$).
We can see in Figure~\ref{fig_chap4_sec3_network_matching_2} that link prediction  methods with a fixed threshold may not be able to predict well, because the predicted links do not satisfy the constraint of one-to-one relationship. Thus one user account in network $G^{(1)}$ can be linked with multiple accounts in network $G^{(2)}$.
In Figure~\ref{fig_chap4_sec3_network_matching_3}, \textit{weighted maximum matching} methods can find a set of links with maximum sum of weights.
However, it is worth noting that the input scores are uncalibrated, so the maximum weight matching may not be a good solution for anchor link prediction problems. 
The input scores only indicate the ranking of different user pairs, i.e., the preference relationship among different user pairs.

Here we say `node $x$ prefers node $y$ over node $z$', if the score of pair $(x,y)$ is larger than the score of pair $(x,z)$.
For example, in Figure~\ref{fig_chap4_sec3_network_matching_3}, the weight of pair $a$, {i.e.,} Score$(a) = 0.8$, is larger than Score$(c) = 0.6$.
It shows that user $u^{(1)}_1$ (the first user in network $G^{(1)}$) \textit{prefers} $u^{(2)}_1$ over $u^{(2)}_2$.
The problem with the prediction result in Figure~\ref{fig_chap4_sec3_network_matching_3} is that, the pair $(u^{(1)}_1, u^{(2)}_1)$ should be more likely to be an anchor link due to the following reasons: (1) $u^{(1)}_1$ prefers $u^{(2)}_1$ over $u^{(2)}_2$; (2)  $u^{(2)}_1$ also prefers $u^{(1)}_1$ over $u^{(1)}_2$.

By following such an intuition, we can obtain the final stable matching result in Figure~\ref{fig_chap4_sec3_network_matching_4}, where anchor links $(u^{(1)}_1, u^{(2)}_1)$ and $(u^{(1)}_2, u^{(2)}_2)$ are selected in the matching process.

\begin{algorithm}[t]
\caption{Multi-Network Stable Matching}
\label{alg:matching_framwork}
\begin{algorithmic}[1]
	\REQUIRE two heterogeneous social networks, $\mc{G}^s$ and $\mc{G}^t$. \\
\qquad  a set of known anchor links $\mc{A}$\\
\ENSURE a set of inferred anchor links $\mc{A}'$ \\

\STATE Construct a training set of user account pairs with known labels using $\mc{A}$.
\STATE	For each pair $(u^s_i, u^t_j)$, extract four types of features.
\STATE	Training classification model $C$ on the training set.
\STATE	Perform classification using model $C$ on the test set.
\STATE	For each unlabeled user account, sort the ranking scores into a preference list of the matching accounts.
\STATE	Initialize all unlabeled $u^s_i$ in $\mc{G}^s$ and $u^t_j$ in $\mc{G}^t$ as free
\STATE	$\mc{A}' = \emptyset$
\WHILE{$\exists$ free $u^s_i$ in $\mc{G}^s$ and  $u^s_i$'s preference list is non-empty}
\STATE	Remove the top-ranked account $u^t_j$ from  $u^s_i$'s preference list
\IF{ $u^t_j$ is free} 
\STATE	$\mc{A}' = \mc{A}' \cup \{ (u^s_i, u^t_j)\}$
\STATE	Set $u^s_i$ and $u^t_j$ as occupied
\ELSE
\STATE	$\exists u^s_p$ that $u^t_j$ is occupied with.
\IF{ $u^t_j$ prefers $u^s_i$ to $u^s_p$}
\STATE $\mc{A}' = (\mc{A}' - \{ (u^s_p,u^t_j) \}) \cup \{ (u^s_i,u^t_j) \}$
\STATE Set $u^s_p$ as free and $u^s_i$ as occupied
\ENDIF
\ENDIF
\ENDWHILE
\end{algorithmic}
\end{algorithm}

\begin{defn}
(\textit{Matching}): Mapping $\mu: \mathcal{U}^{(1)} \cup \mathcal{U}^{(2)} \to \mathcal{U}^{(1)} \cup \mathcal{U}^{(2)}$ is defined to be a \textit{matching} iff (1) $|\mu(u_i)| = 1, \forall u_i \in \mathcal{U}^{(1)}$ and $\mu(u_i) \in \mathcal{U}^{(2)}$; (2) $|\mu(v_j)| = 1, \forall v_j \in \mathcal{U}^{(2)}$ and $\mu(v_j) \in \mathcal{U}^{(1)}$; (3) $\mu(u_i) = v_j$ iff $\mu(v_j) = u_i$.
\end{defn}

\begin{defn} 
(\textit{Blocking Pair}): A  pair  $(u^{(1)}_i, u^{(2)}_j)$ is a blocking pair iff $u^{(1)}_i$ and $ u^{(2)}_j$ both prefer each other over their current assignments  respectively in the predicted set of  anchor links $\mc{A}'$.
\end{defn}

\begin{defn} 
(\textit{Stable Matching}): An inferred anchor link set $\mc{A}'$ is stable if there is no blocking pair.
\end{defn}

Based on the result from the previous step, the MNA method introduced in \cite{cikm13} formulates the anchor link pruning problem as a stable matching problem between user accounts in network $G^{(1)}$ and accounts in network $G^{(2)}$.
Assume that we have two sets of unlabeled user accounts, {i.e.,} $ \mc{U}^{(1)}$ in network $G^{(1)}$ and $ \mc{U}^{(2)}$ in network $G^{(2)}$. 
Each user $u^{(1)}_i$ has a ranking list or preference list  $P(u^{(1)}_i)$ over all the user accounts in network $G^{(2)}$ ($u^{(2)}_j \in \mc{U}^{(2)}$) based upon the input scores of different pairs. 
For example, in Figure~\ref{fig_chap4_sec3_network_matching_1}, the preference list of node $u^{(1)}_1$ is $P(u^{(1)}_1)= ( u^{(2)}_1 > u^{(2)}_2)$, indicating that node $ u^{(2)}_1$ is preferred by $u^{(1)}_1$ over  $u^{(2)}_2$.  The preference list of node $u^{(1)}_2$ is also $P(u^{(1)}_2)=( u^{(2)}_1 > u^{(2)}_2)$.
Similarly, we also build a preference list for each user account in network $G^{(2)}$.
In Figure~\ref{fig_chap4_sec3_network_matching_1},  $P(u^{(2)}_1)= P(u^{(2)}_2)=( u^{(1)}_1 > u^{(1)}_2)$.

The proposed {MNA} method for anchor link prediction is shown in Algorithm~\ref{alg:matching_framwork}. 
In each iteration, {MNA} first randomly selects a free user account $u^{(1)}_i$ from network $G^{(1)}$. 
Then {MNA} gets the most preferred user node $u^{(2)}_j$ by $u^{(1)}_i$ in its preference list $P(u^{(1)}_i)$.
The most preferred user $u^{(2)}_j$ will be removed from the preference list, {i.e.,} $P(u^{(1)}_i) = P(u^{(1)}_i) - u^{(2)}_j$.
If  $u^{(2)}_j$ is also a free account, {MNA} will add the pair of accounts $(u^{(1)}_i, u^{(2)}_j)$ into the current solution set $\mc{A}'$.
Otherwise, $u^{(2)}_j$ is already occupied with $u^{(1)}_p$ in $\mc{A}'$.
{MNA} then examines the preference of  $u^{(2)}_j$. 
If  $u^{(2)}_j$ also prefers  $u^{(1)}_i$ over $u^{(1)}_p$, it means that the pair $ (u^{(1)}_i,u^{(2)}_j)$ is a blocking pair. 
{MNA} removes the blocking pair by replacing the pair $ (u^{(1)}_p,u^{(2)}_j)$ in the solution set $\mc{A}'$ with the pair $ (u^{(1)}_i,u^{(2)}_j)$.
Otherwise, if $u^{(2)}_j$ prefers  $u^{(1)}_p$ over $u^{(1)}_i$, {MNA} will start the next iteration to reach out the next free node in network $G^{(1)}$. 
The algorithm stops when all the users in network $G^{(1)}$ are occupied, or all the preference lists of free accounts in network $G^{(1)}$ are empty.

Finally, the selected anchor links in set $\mc{A}'$ will be returned as the final positive instances, while the remaining ones in the test set $\mathcal{A}_{test}$ will be labeled as the negative instances. Another variant of the supervised network alignment model has been proposed in \cite{iri15}, which adds an extra threshold on the user preference list to make the matching algorithm applicable to handle the \textit{non-anchor users} as well.


\subsection{Pairwise Unsupervised Homogeneous Network Alignment}

In this part, we will study the network alignment problem based on unsupervised learning setting, which needs no labeled training data. Given two heterogeneous online social networks, which can be represented as $G^{(1)} = (\mathcal{V}^{(1)}, \mathcal{E}^{(1)})$ and $G^{(2)} = (\mathcal{V}^{(2)}, \mathcal{E}^{(2)})$ respectively, the {unsupervised network alignment} problem aims at inferring the anchor links between networks $G^{(1)}$ and $G^{(2)}$. Let $\mathcal{U}^{(1)} \subset \mathcal{V}^{(1)}$ and $\mathcal{U}^{(2)} \subset \mathcal{V}^{(2)}$ be the user set in these two networks respectively, we can represent the set of potential anchor links between networks $G^{(1)}$ and $G^{(2)}$ as $\mathcal{A} = \mathcal{U}^{(1)} \times \mathcal{U}^{(2)}$. In the {unsupervised network alignment} problem, among all the potential anchor links in set $\mathcal{A} = \mathcal{U}^{(1)} \times \mathcal{U}^{(2)}$, we want to infer which ones in set $\mathcal{A}$ exist in the real world.

Given two homogeneous networks $G^{(1)}$ and $G^{(2)}$, mapping the nodes between them is an extremely challenging task, which is also called the {graph isomorphism} problem \cite{RC06, F96}. The {graph isomorphism} has been shown to be NP, but it is still not known whether it also belongs to P or NP-complete yet. So far, no efficient algorithm exists that can address the problem in polynomial time. In this part, we will introduce several heuristics based methods to solve the {pairwise homogeneous network alignment} problem.

\subsubsection{Heuristic Measure based Network Alignment Model}

The information generated by users' online social activities can indicate their personal characteristics. The features introduced in the previous subsection, like ECN, EJC and EAA based on social connection information, similarity/distance measures based on location checkin information, temporal activity closeness, and text word usage similarity can all be used as the predictors indicating whether the cross-network user pairs are the same user or not. Besides these measures, in this part, we will introduce a category new measures, {Relative Centrality Difference} (RCD), which can also be applied to solve the {unsupervised network alignment} problem. 

The {centrality} concept can denote the importance of users in the online social networks. Here, we assume that important users in one social network (like celebrities, movie stars and politicians) will be important as well in other networks. Based on such an assumption, the {centrality} of users in different networks can be an important signal for inferring the anchor links across networks. 

\begin{defn}
({Relative Centrality Difference}): Given two users $u^{(1)}_i$, $u^{(2)}_j$ from networks $G^{(1)}$ and $G^{(2)}$ respectively, let $C(u^{(1)}_i)$ and $C(u^{(2)}_j)$ denote the {centrality} scores of the users, we can define the {relative centrality difference} (RCD) as \begingroup\makeatletter\def\f@size{8}\check@mathfonts
\begin{equation}
RCD(u^{(1)}_i, u^{(2)}_j) = \left(1 + \frac{|C(u^{(1)}_i) - C(u^{(2)}_j)|}{\left(C(u^{(1)}_i) + C(u^{(2)}_j) \right)/2} \right)^{-1}.
\end{equation}\endgroup
\end{defn}

Depending on the {centrality} measures applied, different types of \textit{relative centrality difference} measures can be defined. For instance, if we use node degree as the {centrality} measure, the \textit{relative degree difference} can be represented as \begingroup\makeatletter\def\f@size{8}\check@mathfonts
\begin{equation}
RDD(u^{(1)}_i, u^{(2)}_j) = \left(1 + \frac{|D(u^{(1)}_i) - D(u^{(2)}_j)|}{\left(D(u^{(1)}_i) + D(u^{(2)}_j) \right)/2} \right)^{-1}.
\end{equation}\endgroup

Meanwhile, if the PageRank scores of the nodes are used to define their {centrality}, we can represent the {relative centrality difference} measure as \begingroup\makeatletter\def\f@size{8}\check@mathfonts
\begin{equation}
RCD(u^{(1)}_i, u^{(2)}_j) = \left(1 + \frac{|S(u^{(1)}_i) - S(u^{(2)}_j)|}{\left(S(u^{(1)}_i) + S(u^{(2)}_j) \right)/2} \right)^{-1}.
\end{equation}\endgroup

In the above equations, $D(u)$ and $S(u)$ denote the \textit{node degree} and \textit{page rank score} of node $u$ within each network respectively.

\subsubsection{IsoRank}

Model IsoRank \cite{SXB08} initially proposed to align the biomedical networks, like protein protein interaction (PPI) networks and gene expression networks, can be used to solve the {unsupervised social network alignment} problem as well. The IsoRank algorithm has two stages. It first associates a score with each possible anchor links between nodes of the two networks. For instance, we can denote $r(u^{(1)}_i, u^{(2)}_j)$ as the reliability score of an potential anchor link $(u^{(1)}_i, u^{(2)}_j)$ between the networks $G^{(1)}$ and $G^{(2)}$, and all such scores can be organized into a vector $\mb{r}$ of length $|\mathcal{U}^{(1)}| \times |\mathcal{U}^{(2)}|$. In the second stage of IsoRank, it constructs the mapping for the networks by extracting from $\mb{r}$.

\begin{defn}
({Reliability Score}): The reliability score $r(u^{(1)}_i, u^{(2)}_j)$ of anchor link $(u^{(1)}_i, u^{(2)}_j)$ is highly correlated with the support provided by the mapping scores of the neighborhoods of users $u^{(1)}_i$ and $u^{(2)}_j$. Therefore, we can define the score $r(u^{(1)}_i, u^{(2)}_j)$ as \begingroup\makeatletter\def\f@size{7}\check@mathfonts
\begin{align}\label{equ:isorank_unweighted}
&r(u^{(1)}_i, u^{(2)}_j) \\
&= \sum_{u^{(1)}_m \in \Gamma(u^{(1)}_i)} \sum_{u^{(2)}_n \in \Gamma(u^{(2)}_i)} \frac{1}{|\Gamma(u^{(1)}_i)| |\Gamma(u^{(2)}_j)|} r(u^{(1)}_m, u^{(2)}_n),
\end{align}\endgroup
where sets $\Gamma(u^{(1)}_i)$ and $\Gamma(u^{(2)}_i)$ represent the neighborhoods of users $u^{(1)}_i$ and $u^{(1)}_i$ respectively in networks $G^{(1)}$ and $G^{(2)}$.
\end{defn}

If the networks are weighted, and all the intra-network connections like $(u^{(1)}_i, u^{(1)}_m)$ will be associated with a weight $w(u^{(1)}_i, u^{(1)}_m)$, we can represented the {reliability} measure of $r(u^{(1)}_i, u^{(2)}_j)$ in the weighted network as \begingroup\makeatletter\def\f@size{6}\check@mathfonts
\begin{align}\label{equ:isorank_weighted}
\hspace{-10pt}&r(u^{(1)}_i, u^{(2)}_j) = \hspace{-5pt} \sum_{u^{(1)}_m \in \Gamma(u^{(1)}_i)} \sum_{u^{(2)}_n \in \Gamma(u^{(2)}_i)}  \hspace{-5pt}  w(u_i^{(1)}, u_j^{(2)}) r(u^{(1)}_m, u^{(2)}_n),
\end{align}
where the weight term 
\begin{align}
&w(u_i^{(1)}, u_j^{(2)}) \\
&= \frac{w(u^{(1)}_i, u^{(1)}_m) w(u^{(2)}_j, u^{(2)}_n)}{\sum_{u^{(1)}_p \in \Gamma(u^{(1)}_i)} w(u^{(1)}_i, u^{(1)}_p) \sum_{u^{(2)}_q \in \Gamma(u^{(2)}_i)} w(u^{(2)}_j, u^{(2)}_q}.
\end{align}\endgroup

As we can see, Equation~\ref{equ:isorank_unweighted} is a special case of Equation~\ref{equ:isorank_weighted} with link weight $w(u^{(1)}_i, u^{(1)}_j) = 1$ for $u^{(1)}_i \in \mathcal{U}^{(1)}$ and $u^{(2)}_j \in \mathcal{U}^{(2)}$. Equation~\ref{equ:isorank_unweighted} can also be rewritten with linear algebra
\begin{equation}
\mb{r} = \mb{A} \mb{r},
\end{equation}
where matrix $\mb{A} \in \mathbb{R}^{|\mathcal{U}^{(1)}||\mathcal{U}^{(2)}| \times |\mathcal{U}^{(1)}||\mathcal{U}^{(2)}|}$ with entry  \begingroup\makeatletter\def\f@size{6}\check@mathfonts
\begin{align}
&A\big((i,j), (p,q) \big) \\
&= \begin{cases}
\frac{1}{|\Gamma(u^{(1)}_i)| |\Gamma(u^{(2)}_j)|}, \hspace{-5pt}& \mbox{if } (u^{(1)}_i, u^{(1)}_p) \in \mathcal{E}^{(1)}, (u^{(2)}_j, u^{(2)}_q) \in \mathcal{E}^{(2)},\\
0, \hspace{-5pt}& \mbox{otherwise}.
\end{cases}
\end{align}\endgroup

The matrix $\mb{A}$ is of dimension $|\mathcal{U}^{(1)}||\mathcal{U}^{(2)}| \times |\mathcal{U}^{(1)}||\mathcal{U}^{(2)}|$, where the row and column indexes correspond to different potential anchor links across the networks. The entry $A\big((i,j), (p,q) \big)$ corresponds the anchor links $(u^{(1)}_i, u^{(2)}_j)$ and $(u^{(1)}_p, u^{(2)}_q)$. As we can see, the above equation denotes a {random walk} across the graphs $G^{(1)}$ and $G^{(2)}$ via the social links and anchor links in them. The solution to the above equation denotes the principal eigenvector of the matrix $\mb{A}$ corresponding to the eigenvalue $1$. For more information about the {random walk} model, please refer to \cite{SXB08}.

\subsubsection{IsoRankN}

IsoRankN \cite{LLBSB09} algorithm is an extension to IsoRank. Based on the learning results of IsoRank, IsoRankN further adopts the spectral clustering method on the induced graph of pairwise alignment scores to achieve the final alignment results. The new approach provides significant advantages not only over the original IsoRank but also over other methods. IsoRankN has $4$ main steps: (1) initial network alignment with IsoRank, (2) star spread, (3) spectral partition, and (4) star merging, where steps (3) and (4) will repeat until all the nodes are assigned to a cluster.

\noindent \textbf{Initial Network Alignment}: Given $k$ isolated networks $G^{(1)}, G^{(2)}, \cdots, G^{(k)}$, IsoRankN computes the local alignment scores of node pairs across networks with IsoRank algorithm. For instance, if the networks are unweighted, the alignment score between nodes $u_l^{(i)}$ and $u_m^{(j)}$ between networks $G^{(i)}$, $G^{(j)}$ can be denoted as. \begingroup\makeatletter\def\f@size{6}\check@mathfonts
\begin{align}
&r(u^{(1)}_i, u^{(2)}_j) \\
&= \sum_{u^{(1)}_m \in \Gamma(u^{(1)}_i)} \sum_{u^{(2)}_n \in \Gamma(u^{(2)}_i)} \frac{1}{|\Gamma(u^{(1)}_i)| |\Gamma(u^{(2)}_j)|} r(u^{(1)}_m, u^{(2)}_n),
\end{align}\endgroup

It will lead to a weighted k-partite graph, where the links denotes the anchor links across networks weighted by the scores calculated above. If the networks $G^{(1)}, \cdots G^{(k)}$ are all complete graphs, the alignment results will be the maximum weighted cliques. However, in the real world, such an assumption can hardly met, and IsoRankN proposes to use ``Star Spread'' technique to select a subgraph with high weights. 

\noindent \textbf{Star Spread}: For each node in a network, e.g., $u_l^{(i)}$ in network $G^{(i)}$, the set of nodes connected with $u_l^{(i)}$ via potential anchor links can be denoted as set $\Gamma(u_l^{(i)})$. The nodes in $\Gamma(u_l^{(i)})$ can be further pruned by removing the nodes connected with \textit{weak} anchor links. Here, the ``weak'' denotes the anchor links with a low score calculated with IsoRank. Formally, among all the nodes in $\Gamma(u_l^{(i)})$, we can denote the node connected to $u_l^{(i)}$ with the strongest link as $v^* =\arg_{v \in \Gamma(u_l^{(i)})} \max r(u^{(i)}_l, v)$. For all the nodes with weights lower than $\beta \cdot r(u^{(i)}_l, v^*)$ will be removed from $\Gamma(u_l^{(i)})$ (where $\beta$ is a threshold parameter), and the remaining nodes together with $u^{(i)}_l$ will form a star structured graph $S_{u^{(i)}_l}$. 

\noindent \textbf{Spectral Partition}: For each node $u_l^{(i)}$, IsoRankN aims at selecting a subgraph $S^*_{u^{(i)}_l}$ from $S_{u^{(i)}_l}$, which contains the highly weighted neighbors of $u_l^{(i)}$. To achieve such a objective, IsoRankN proposes to identify a subgraph with low \textit{conductance} from $S_{u^{(i)}_l}$ instead. Formally, given a network $G = (\mathcal{V}, \mathcal{E})$, let $\mathcal{S} \subset \mathcal{V}$ denote a subset of $G$. The \textit{conductance} of the subgraph involving $\mathcal{S}$ can be represented as
\begin{equation}
\phi(\mathcal{S}) = \frac{\sum_{u \in \mathcal{S}} \sum_{v \in \bar{\mathcal{S}}} w_{u,v}}{ \min (\mbox{vol}(\mathcal{S}), \mbox{vol}(\bar{\mathcal{S}}) )},
\end{equation}
where $\bar{\mathcal{S}} = \mathcal{V} \setminus \mathcal{S}$, and $\mbox{vol}(\mathcal{S}) = \sum_{u \in \mathcal{S}} \sum_{v \in \mathcal{V}} w_{u,v}$.
IsoRankN points out that a node subset $\mathcal{S}$ containing node $u_l^{(i)}$ can be computed effectively and efficiently with the personalized PageRank algorithm starting from node $u_l^{(i)}$.

\noindent \textbf{Star Merging}: Considering that links in the star graph $S^*_{u^{(i)}_l}$ are all the anchor links across networks, there exist no intra-network links at all in $S^*_{u^{(i)}_l}$, e.g., the links in network $G^{(i)}$ only. However, in many cases, there may exist multiple nodes corresponding to the same entity inside the network as well. To solve such a problem, IsoRankN proposes a star merging step to combine several star graphs together, e.g., $S^*_{u^{(i)}_l}$ and $S^*_{u^{(j)}_m}$.

Formally, given two star graphs $S^*_{u^{(i)}_l}$ and $S^*_{u^{(j)}_m}$, if the following conditions both hold, $S^*_{u^{(i)}_l}$ and $S^*_{u^{(j)}_m}$ can be merged into one star graph. \begingroup\makeatletter\def\f@size{8}\check@mathfonts
\begin{align}
&\forall v \in S^*_{u^{(j)}_m} \setminus \{u^{(j)}_m\}, r(v, u^{(i)}_l) \ge \beta \cdot \max_{v' \in \Gamma(u_l^{(i)})} r(v', u^{(i)}_l),\\
&\forall v \in S^*_{u^{(i)}_l} \setminus \{u^{(i)}_l\}, r(v, u^{(j)}_m) \ge \beta \cdot \max_{v' \in \Gamma(u_m^{(j)})} r(v', u^{(j)}_m).
\end{align}\endgroup

\subsubsection{Matrix Inference based Network Alignment}

Formally, given a homogeneous network $G^{(1)}$, its structure can be organized as the adjacency matrix $\mb{A}_{G^{(1)}} \in \mathbb{R}^{|\mathcal{U}^{(1)}| \times |\mathcal{U}^{(1)} |}$. If network $G^{(1)}$ is unweighted, then matrix $\mb{A}_{G^{(1)}}$ will be a binary matrix and entry ${A}_{G^{(1)}}(i, p) = 1$ (or ${A}_{G^{(1)}}(u^{(1)}_i, u^{(1)}_p) = 1$) iff the correspond social link $(u^{(1)}_i, u^{(1)}_p)$ exists. In  the case that the network is weighted, the entries like ${A}_{G^{(1)}}(i, p) = 1$ denotes the weight of link $(u^{(1)}_i, u^{(1)}_p)$ and $0$ if $(u^{(1)}_i, u^{(1)}_p)$ doesn't exist. In a similar way, we can also represent the social adjacency matrix $\mb{A}_{G^{(2)}}$ for network $G^{(2)}$ as well.

The network alignment problem aims at inferring an {one-to-one} node mapping function, that can project nodes from one network to the other networks. For instance, we can denote the mapping between networks $G^{(1)}$ to $G^{(2)}$ as $f: \mathcal{U}^{(1)} \to \mathcal{U}^{(2)}$. Via the mapping $f$, besides the nodes, the network structure can be projected across networks as well. For instance, given a social connection $(u^{(1)}_i, u^{(1)}_p)$ in $G^{(1)}$, we can represent its corresponding connection in $G^{(2)}$ as $(f(u^{(1)}_i), f(u^{(1)}_p))$. 

Via the mapping $f$, we can denote the network structure differences between $G^{(1)}$ and $G^{(2)}$ as the summation of the link projection difference between them \begingroup\makeatletter\def\f@size{6}\check@mathfonts
\begin{align}
&L(G^{(1)}, G^{(2)}, f) =\\
& \sum_{u^{(1)}_i \in \mathcal{U}^{(1)}} \sum_{u^{(1)}_p \in \mathcal{U}^{(1)}} \left({A}_{G^{(1)}}(u^{(1)}_i, u^{(1)}_p) - {A}_{G^{(1)}}(f(u^{(1)}_i), f(u^{(1)}_p))\right)^2.
\end{align}\endgroup

Formally, the {one-to-one} projection can be represented as a matrix $\mb{P}$ as well, where entry $P(i, j) = 1$ iff anchor link $(u^{(1)}_i, u^{(2)}_j)$ exists between networks $G^{(1)}$ and $G^{(2)}$. Via the matrix $\mb{P}$, we can represent the above loss term as \begingroup\makeatletter\def\f@size{8}\check@mathfonts
\begin{equation}
L(\mb{A}_{G^{(1)}}, \mb{A}_{G^{(2)}}, \mb{P}) = \left\| \mb{P}^\top \mb{A}_{G^{(1)}} \mb{P} - \mb{A}_{G^{(2)}} \right\|^2.
\end{equation}\endgroup

If there exists a perfect mapping of users across networks, we can obtain a mapping matrix $\mb{P}$ introducing zero loss in the above function, i.e.,  $L(\mb{A}_{G^{(1)}}, \mb{A}_{G^{(2)}}, \mb{P}) = 0$. Inferring the optimal mapping matrix $\mb{P}$ which can introduce the minimum loss can be represented as the following objective function
\begin{equation}
\mb{P}^* = \arg \min_{\mb{P}} \left\| \mb{P}^\top \mb{A}_{G^{(1)}} \mb{P} - \mb{A}_{G^{(2)}} \right\|^2,
\end{equation}
where the matrix $\mb{P}$ is usually subject to some constraint, like $\mb{P}$ is binary and each row and column should contain at most one entry being filled with value $1$.

In general, it is not easy to find the optimal solution to the above objective function, as it is a purely combinatorial problem. Identifying the optimal solution requires the enumeration of all the potential user mapping across different networks. In \cite{U88}, Umeyama provides an algorithm that can solve the function with a nearly optimal solution.


\subsection{Global Unsupervised Alignment of Multiple Social Networks}

The works introduced in the previous section are all about pairwise network alignment, which focus on the alignment of two networks only. However, in the real-world, people are normally involved in multiple (usually more than two) social networks simultaneously. In this section, we will focus on the simultaneous alignment problem of multiple (more than two) networks, which is called the ``multiple anonymized social networks alignment'' problem formally \cite{icdm15}.

\begin{figure}[t]
\centering
    \begin{minipage}[l]{0.75\columnwidth}
      \centering
      \includegraphics[width=\textwidth]{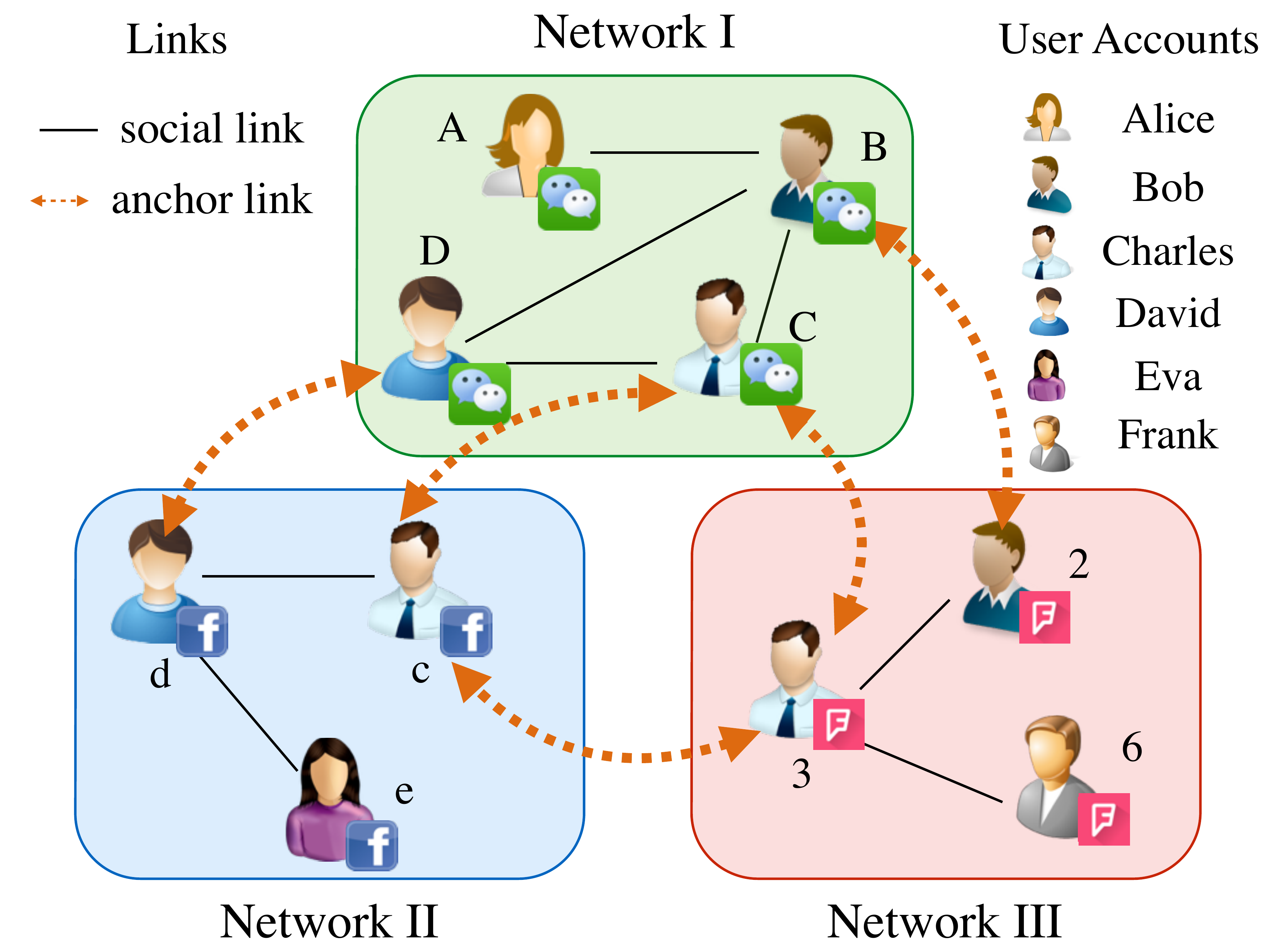}
    \end{minipage}
  \caption{An example of multiple anonymized partially aligned social networks.}\label{fig:chap5_sec3_multiple_alignment_example}
\end{figure}

To help illustrate the multi-network alignment problem more clearly, we also give an example in Figure~\ref{fig:chap5_sec3_multiple_alignment_example}, which involves $3$ different social networks (i.e., networks I, II and III). Users in these $3$ networks are all anonymized and their names are replaced with randomly generated identifiers. Each pair of these $3$ anonymized networks can actually share some common users, e.g., ``David'' participates in both networks I and II simultaneously, ``Bob'' is using networks I and III concurrently, and ``Charles'' is involved in all these $3$ networks at the same time. Besides these shared anchor users, in these $3$ partially aligned networks, some users are involved in one single network only (i.e., the non-anchor users \cite{kdd14}), e.g., ``Alice'' in network I, ``Eva'' in network II and ``Frank'' in network III. The problem studied in this part aims at discovering the anchor links (i.e., the dashed bi-directional red lines) connecting anchor users across these $3$ social networks respectively.

The significant difference of the studied problem from existing \textit{two} network alignment problems is due to the ``\textit{transitivity law}'' that anchor links follow. In traditional set theory, a relation $\mathcal{R}$ is defined to be a \textit{transitive relation} in domain $\mathcal{X}$ iff $\forall a, b, c \in \mathcal{X}, (a, b) \in \mathcal{R} \land (b, c) \in \mathcal{R} \to (a, c) \in \mathcal{R}$. If we treat the union of user account sets of all these social networks as the target domain $\mathcal{X}$ and treat anchor links as the relation $\mathcal{R}$, then anchor links depict a ``\textit{transitive relation}'' among users across networks. We can take the networks shown in Figure~\ref{fig:chap5_sec3_multiple_alignment_example} as an example. Let $u$ be a user involved in networks I, II and III simultaneously, whose accounts in these networks are $u^I$, $u^{II}$ and $u^{III}$ respectively. If anchor links $(u^{I}, u^{II})$ and $(u^{II}, u^{III})$ are identified in aligning networks (I, II) and networks (II, III) respectively (i.e., $u^{I}$, $u^{II}$ and $u^{III}$ are discovered to be the same user), then anchor link $(u^{I}, u^{III})$ should also exist in the alignment result of networks (I, III) as well. In the multi-network alignment problem, we need to guarantee the inferred anchor links can meet the \textit{transitivity law}. Formally, the multi-network alignment problem can be represented as follows. 

Given the $n$ isolated social networks $\{G^{(1)}, G^{(2)}, \cdots, G^{(n)}\}$, the multi-network alignment problem aims at discovering the anchor links among these $n$ networks, i.e., the anchor link sets $\mathcal{A}^{(1,2)}, \mathcal{A}^{(1,3)}, \cdots, \mathcal{A}^{(n-1,n)}$. These $n$ social etworks $G^{(1)}, G^{(2)}, \cdots, G^{(n)}$ are partially aligned and the constraint on anchor links in $\mathcal{A}^{(1,2)}, \mathcal{A}^{(1,3)}, \cdots,\mathcal{A}^{(n-1,n)}$ is \textit{one-to-one}, which also follow the \textit{transitivity law}.

To solve the {multi-network alignment} problem, a novel network alignment framework {\ouruma} (\underline{U}nsupervised \underline{M}ulti-network \underline{A}lignment) is proposed in \cite{icdm15}. {\ouruma} addresses the {multi-network alignment} problem with two steps: (1) unsupervised transitive anchor link inference across multi-networks, and (2) transitive multi-network matching to maintain the \textit{one-to-one constraint}.

\subsubsection{Unsupervised Network Alignment Loss Function}

Anchor links between any two given networks $G^{(i)}$ and $G^{(j)}$ actually define an \textit{one-to-one} mapping (of users and social links) between $G^{(i)}$ and $G^{(j)}$. To evaluate the quality of different inferred mapping (i.e., the inferred anchor links), {\ouruma} introduces the concepts of cross-network \textit{Friendship Consistency/Inconsistency} concept in \cite{icdm15}. The optimal inferred anchor links are those which can maximize the \textit{Friendship Consistency} (or minimize the \textit{Friendship Inconsistency}) across networks. Formally, given two partially aligned social networks $G^{(i)} = (\mathcal{U}^{(i)}, \mathcal{E}^{(i)})$ and $G^{(j)} =  (\mathcal{U}^{(j)}, \mathcal{E}^{(j)})$, we can represent their corresponding \textit{social adjacency} matrices to be $\mb{S}^{(i)} \in \mathbb{R}^{|\mathcal{U}^{(i)}| \times |\mathcal{U}^{(i)}|}$ and $\mb{S}^{(j)} \in \mathbb{R}^{|\mathcal{U}^{(j)}| \times |\mathcal{U}^{(j)}|}$ respectively.

Meanwhile, given anchor link set $\mathcal{A}^{(i,j)} \subset \mathcal{U}^{(i)} \times \mathcal{U}^{(j)}$ between networks $G^{(i)}$ and $G^{(j)}$, the \textit{binary transitional matrix} from $G^{(i)}$ to $G^{(j)}$ can be represented as $\mb{T}^{(i,j)} \in \{0,1\}^{|\mathcal{U}^{(i)}| \times |\mathcal{U}^{(j)}|}$, where $\mb{T}^{(i,j)}(l,m) = 1$ iff link $(u^{(i)}_l, u^{(j)}_m) \in \mathcal{A}^{(i,j)}$, $u^{(i)}_l \in \mathcal{U}^{(i)}$, $u^{(j)}_m \in \mathcal{U}^{(j)}$. The \textit{binary transitional matrix} from $G^{(j)}$ to $G^{(i)}$ can be defined in a similar way, which can be represented as $\mb{T}^{(j,i)} \in \{0,1\}^{|\mathcal{U}^{(j)}| \times |\mathcal{U}^{(i)}|}$, where $(\mb{T}^{(i,j)})^\top = \mb{T}^{(j,i)}$ as the anchor links between $G^{(i)}$ and $G^{(j)}$ are undirected. Considering that anchor links have an inherent \textit{one-to-one} constraint, each row and each column of the \textit{binary transitional matrices} $\mb{T}^{(i,j)}$ and $\mb{T}^{(j,i)}$ should have at most one entry filled with $1$, which will constrain the inference space of potential \textit{binary transitional matrices} $\mb{T}^{(i,j)}$ and $\mb{T}^{(j,i)}$ greatly.

{\ouruma} defines the \textit{friendship inconsistency} as the number of non-shared social links between those mapped from $G^{(i)}$ and those in $G^{(j)}$. Based on the inferred \textit{anchor transitional matrix} $\mb{T}^{(i,j)}$, the introduced \textit{friendship inconsistency} between matrices $(\mb{T}^{(i,j)})^\top\mb{S}^{(i)} \mb{T}^{(i,j)}$ and $\mb{S}^{(j)}$ can be represented as:
\begin{equation}
\left \| (\mb{T}^{(i,j)})^\top\mb{S}^{(i)} \mb{T}^{(i,j)} -  \mb{S}^{(j)}\right \|^2_F,
\end{equation}
where $\left \| \cdot \right \|_F$ denotes the Frobenius norm. And the optimal \textit{binary transitional matrix} $\bar{\mb{T}}^{(i,j)}$, which can lead to the minimum \textit{friendship inconsistency} can be represented as 
\begin{align}
\bar{\mb{T}}^{(i,j)} &={\arg \min}_{\mb{T}^{(i,j)}} \left \| (\mb{T}^{(i,j)})^\top \mb{S}^{(i)} \mb{T}^{(i,j)} - \mb{S}^{(j)} \right \|^2_F\\
s.t.\ \ \ \ &\mb{T}^{(i,j)} \in \{0,1\}^{|\mathcal{U}^{(i)}| \times |\mathcal{U}^{(j)}|},\\
& \mb{T}^{(i,j)} \mb{1}^{|\mathcal{U}^{(j)}| \times 1} \preccurlyeq \mb{1}^{|\mathcal{U}^{(i)}| \times 1},\\
& (\mb{T}^{(i,j)})^\top \mb{1}^{|\mathcal{U}^{(i)}| \times 1} \preccurlyeq \mb{1}^{|\mathcal{U}^{(j)}| \times 1},
\end{align}
where the last two equations are added to maintain the \textit{one-to-one} constraint on anchor links and $\mb{X} \preccurlyeq \mb{Y}$ iff $\mb{X}$ is of the same dimensions as $\mb{Y}$ and every entry in $\mb{X}$ is no greater than the corresponding entry in $\mb{Y}$.

\subsubsection{Transitivity Constraint on Alignment Results}

Isolated network alignment can work well in addressing the alignment problem of two social networks. However, in the {multi-network alignment} problem studied in this part, multiple social networks (more than two) social networks are to be aligned simultaneously. Besides minimizing the \textit{friendship inconsistency} between each pair of networks, the \textit{transitivity} property of anchor links also needs to be preserved in the transitional matrices inference.

The \textit{transitivity} property should holds for the alignment of any $n$ networks, where the minimum of $n$ is $3$. To help illustrate the \textit{transitivity property} more clearly, here we will use $3$ network alignment as an example to introduce the {multi-network alignment} problem and the {\ouruma} model, which can be easily generalized to the case of $n$ networks alignment. Let $G^{(i)}$, $G^{(j)}$ and $G^{(k)}$ be $3$ social networks to be aligned concurrently. To accommodate the alignment results and preserve the \textit{transitivity} property, {\ouruma} introduces the following \textit{alignment transitivity penalty}:

\begin{defn} 
(Alignment Transitivity Penalty): Formally, let $\mb{T}^{(i,j)}$, $\mb{T}^{(j,k)}$ and $\mb{T}^{(i,k)}$ be the inferred binary transitional matrices from $G^{(i)}$ to $G^{(j)}$, from $G^{(j)}$ to $G^{(k)}$ and from $G^{(i)}$ to $G^{(k)}$ respectively among these $3$ networks. The \textit{alignment transitivity penalty} $C(\{G^{(i)}, G^{(j)}, G^{(k)}\})$ introduced by the inferred transitional matrices can be quantified as the number of inconsistent social links being mapped from $G^{(i)}$ to $G^{(k)}$ via two different alignment paths $G^{(i)} \to G^{(j)} \to G^{(k)}$ and $G^{(i)} \to G^{(k)}$, i.e., \begingroup\makeatletter\def\f@size{6}\check@mathfonts
\begin{align}
&C(\{G^{(i)}, G^{(j)}, G^{(k)}\}) = \\
&\left \| (\mb{T}^{(j,k)})^\top (\mb{T}^{(i,j)})^\top \mb{S}^{(i)} \mb{T}^{(i,j)}\mb{T}^{(j,k)} - (\mb{T}^{(i,k)})^\top \mb{S}^{(i)} \mb{T}^{(i,k)} \right \|^2_F.
\end{align}\endgroup
\end{defn}

Alignment transitivity penalty is a general penalty concept and can be applied to $n$ networks $\{G^{(1)}, G^{(2)}, \cdots, G^{(n)}\},\\ n\ge3$ as well, which can be defined as the summation of penalty introduced by any three networks in the set, i.e., \begingroup\makeatletter\def\f@size{6}\check@mathfonts
\begin{align}
&C(\{G^{(1)}, G^{(2)}, \cdots, G^{(n)}\}) \\
&= \sum_{\forall \{G^{(i)}, G^{(j)}, G^{(k)}\} \subset \{G^{(1)}, G^{(2)}, \cdots, G^{(n)}\}} C(\{G^{(i)}, G^{(j)}, G^{(k)}\}).
\end{align}\endgroup

The optimal \textit{binary transitional matrices} $\bar{\mb{T}}^{(i,j)}$, $\bar{\mb{T}}^{(j,k)}$ and $\bar{\mb{T}}^{(k,i)}$ which can minimize friendship inconsistency and the \textit{alignment transitivity penalty} at the same time can be represented to be \begingroup\makeatletter\def\f@size{5}\check@mathfonts
\begin{align}
&\bar{\mb{T}}^{(i,j)},\bar{\mb{T}}^{(j,k)}, \bar{\mb{T}}^{(k,i)} \\
&={\arg \min}_{\mb{T}^{(i,j)}, \mb{T}^{(j,k)}, \mb{T}^{(k,i)}} \left \| (\mb{T}^{(i,j)})^\top \mb{S}^{(i)} \mb{T}^{(i,j)} - \mb{S}^{(j)} \right \|^2_F +\\
&\left \| (\mb{T}^{(j,k)})^\top \mb{S}^{(j)} \mb{T}^{(j,k)} - \mb{S}^{(k)} \right \|^2_F + \left \| (\mb{T}^{(k,i)})^\top \mb{S}^{(k)} \mb{T}^{(k,i)} - \mb{S}^{(i)} \right \|^2_F\\
&+ \alpha \left \| (\mb{T}^{(j,k)})^\top (\mb{T}^{(i,j)})^\top \mb{S}^{(i)} \mb{T}^{(i,j)}\mb{T}^{(j,k)} - \mb{T}^{(k,i)}  \mb{S}^{(i)} (\mb{T}^{(k,i)})^\top \right \|^2_F\\
&s.t. \ \mb{T}^{(i,j)} \in \{0,1\}^{|\mathcal{U}^{(i)}| \times |\mathcal{U}^{(j)}|}, \mb{T}^{(j,k)} \in \{0,1\}^{|\mathcal{U}^{(j)}| \times |\mathcal{U}^{(k)}|}\\
&\ \ \ \ \ \mb{T}^{(k,i)} \in \{0,1\}^{|\mathcal{U}^{(k)}| \times |\mathcal{U}^{(i)}|}\\
&\ \ \ \ \  \mb{T}^{(i,j)} \mb{1}^{|\mathcal{U}^{(j)}| \times 1} \preccurlyeq \mb{1}^{|\mathcal{U}^{(i)}| \times 1},(\mb{T}^{(i,j)})^\top \mb{1}^{|\mathcal{U}^{(i)}| \times 1} \preccurlyeq \mb{1}^{|\mathcal{U}^{(j)}| \times 1},\\
&\ \ \ \ \  \mb{T}^{(j,k)} \mb{1}^{|\mathcal{U}^{(k)}| \times 1} \preccurlyeq \mb{1}^{|\mathcal{U}^{(j)}| \times 1},(\mb{T}^{(j,k)})^\top \mb{1}^{|\mathcal{U}^{(j)}| \times 1} \preccurlyeq \mb{1}^{|\mathcal{U}^{(k)}| \times 1},\\
&\ \ \ \ \  \mb{T}^{(k,i)} \mb{1}^{|\mathcal{U}^{(i)}| \times 1} \preccurlyeq \mb{1}^{|\mathcal{U}^{(k)}| \times 1},(\mb{T}^{(k,i)})^\top \mb{1}^{|\mathcal{U}^{(k)}| \times 1} \preccurlyeq \mb{1}^{|\mathcal{U}^{(i)}| \times 1},
\end{align}\endgroup
where parameter $\alpha$ denotes the weight of the alignment transitivity penalty term, which is set as $1$ by default.

The above objective function aims at obtaining the \textit{hard} mappings among users across different networks and entries in all these \textit{transitional matrices} are binary, which can lead to a fatal drawback: \textit{hard assignment} can be neither possible nor realistic for networks with star structures as proposed in \cite{KTL13} and the hard subgraph isomorphism \cite{LHKL12} is NP-hard. To address the function, {\ouruma} proposes to relax the hard binary constraints on the variables first and solve the function with gradient descent. Furthermore, based on the learning results {\ouruma} keeps the one-to-one constraint on anchor links by selecting those which can maximize the overall existence probabilities while maintaining the \textit{matching transitivity} property at the same time.


\subsection{Semi-Supervised Network Alignment}

As mentioned before, in the real-world online social networks, the anchor links are extremely difficult to label manually. The training set we can obtain are usually of a small size compared with the network scale. For instance, given the Facebook and Twitter networks containing billions and millions of users respectively, identifying a training set with thousands correct anchor links is not an easy task. Meanwhile, between Facebook and Twitter, the total number of potential anchor links could be of the scale $10^{15}$. Therefore, besides the small sized identified anchor links, there usually exist a very large number of unlabeled anchor links, which are extremely hard to predict.

In this part, we will be focused on the network alignment problem based on the semi-supervised learning setting. Besides these identified anchor links, we also try to make utilize of the unlabeled anchor links in the model building. Given two heterogeneous online social networks $G^{(1)}$ and $G^{(2)}$, and a set of labeled anchor link instances $\mathcal{A}_{train}$ as well as a large number of unlabeled anchor link instances $\mathcal{A}_{unlabeled} = \mathcal{U}^{(1)} \times \mathcal{U}^{(2)} \setminus \mathcal{A}_{train}$, we aim at building a model with the labeled and unlabeled anchor link sets $\mathcal{A}_{train}$ and $\mathcal{A}_{unlabeled}$. In our network alignment task, the test set is a subset of or equal to the unlabeled set, i.e., $\mathcal{A}_{test} \subseteq \mathcal{A}_{unlabeled}$. The built model will be further applied to the test set to infer the potential labels of these anchor links.

To address the problem, in this part, we will introduce the semi-supervised network alignment model introduced in \cite{wsdm17_2}, which solves the problem as an optimization problem and models \textit{one-to-one} cardinality constraint on the anchor links as a mathematical constraint.

\subsubsection{Loss Function for Anchor Links}

Let set $\mathcal{L} = \mathcal{U}^{(1)} \times \mathcal{U}^{(2)}$ denote all the potential anchor links between networks $G^{(1)}$ and $G^{(2)}$, where $\mathcal{L} = \mathcal{A}_{train} \cup \mathcal{A}_{unlabeled}$. Based on the whole link set $\mathcal{L}$, as introduced in the previous sections, a set of features can be extracted for these links with the information available in the information network $G$, which can be represented as set $\mathcal{X} = \{\mb{x}_l\}_{l \in \mathcal{L}}$ ($\mb{x}_l \in \mathbb{R}^m, \forall l \in \mathcal{L}$). Given the link existence label set $\mathcal{Y} = \{0, 1\}$, the objective of the problem studied in this part is to achieve a general link inference function $f: \mathcal{X} \to \mathcal{Y}$ to map the link feature vectors to their corresponding labels. Here, $0$ denotes the label of the negative class. Depending on the specific application setting and information available in the networks, the feature vectors extracted for links in $\mathcal{L}$ can be very diverse. 

Formally, the loss introduced in the mapping $f(\cdot)$ can be represented as function $L: \mathcal{X} \times \mathcal{Y} \to \mathbb{R}$ over the link feature vector/label pairs. Meanwhile, for one certain input feature vector $\mb{x}_l$ for link $l \in \mathcal{L}$, we can denote its inferred label introducing the minimum loss as $\hat{y_l}$:
\begin{equation}
\hat{y_l} = \arg \min_{y_l \in \mathcal{Y}, \mb{w}} L(\mb{x}_l, y_l; \mb{w}),
\end{equation} 
where vector $\mb{w}$ denotes the parameters involved in the mapping function $f(\cdot)$.

Therefore, given the pre-defined loss function $L(\cdot)$, the general form of the objective mapping $f: \mathcal{X} \to \mathcal{Y}$ parameterized by vector $\mb{w}$ can be represented as:
\begin{align}
f(\mb{x}; \mb{w}) = \arg \min_{y_l \in \mathcal{Y}} L(\mb{x}, y; \mb{w}).
\end{align}

In many cases (e.g., when the links are not linearly separable), the feature vector $\mb{x}_l$ of link $l$ needs to be transformed as $g(\mb{x}_l) \in \mathbb{R}^k$ ($k$ is the transformed feature number) and the transformation function $g(\cdot)$ can be different \textit{kernel projections} depending on the separability of instances. Here we assume loss function $L(\cdot)$ to be linear in some combined representation of the transformed link feature vector $g(\mb{x}_l)^\top$ and label $y_l$, i.e.,
\begin{equation}
L(\mb{x}_l, y_l; \mb{w}) = (\left \langle \mb{w}, g(\mb{x}_l) \right \rangle - y_l)^2 = (\mb{w}^\top g(\mb{x}_l) - y_l)^2.
\end{equation} 

Furthermore, based on all the links in the network $\mathcal{L}$, we can represent the extracted feature vectors for these links to be matrix $\mb{X} = [g(\mb{x}_{l_1}), g(\mb{x}_{l_2}), \cdots, g(\mb{x}_{l_{|\mathcal{L}|}})]^\top \in \mathbb{R}^{|\mathcal{L}| \times k}$ (for simplicity, linear kernel projection is used here, and $g(\mb{x}_l) = \mb{x}_l$). Meanwhile, their existence labels can be represented as vector $\mb{y} = [y_{l_1}, y_{l_2}, \cdots, y_{l_{|\mathcal{L}|}}]^\top$, where $y_{l} \in \{0, 1\}, \forall l \in \mathcal{L}$. Specifically, for the existing links in $\mathcal{E}$, we know their labels to be positive in advance, i.e., $y_l = 1, \forall l \in \mathcal{E}$. According to the above loss function definition, based on $\mb{X}$ and $\mb{y}$, the loss introduced by all links in $\mathcal{L}$ can be represented to be
\begin{align}
L(\mb{X}, \mb{y}; \mb{w}) = \left\| \mb{X}\mb{w} - \mb{y} \right\|^2_2.
\end{align}

To learn the parameter vector $\mb{w}$ and infer the potential label vector $\mb{y}$, \cite{wsdm17_2} proposes to minimize the loss term introduced by all the links in $\mathcal{L}$. Meanwhile, to avoid overfitting the training set, besides minimizing the loss function $L(\mb{X}, \mb{y}; \mb{w})$, a regularization term $\left\| \mb{w} \right\|^2_2$ about the parameter vector $\mb{w}$ is added to the objective function:
\begin{align}
&\min_{\mb{w}, \mb{y}} \frac{1}{2} \left\| \mb{w} \right\|^2_2 + \frac{c}{2} \left\| \mb{X}\mb{w} - \mb{y} \right\|^2_2,\\
&s.t.\ \  \mb{y} \in \{0, 1\}^{|\mathcal{L}| \times 1}, \mbox{ and } y_l = 1, \forall l \in \mathcal{E},
\end{align}
where constant $c$ denotes the weight of the loss term in the function.

\subsubsection{Cardinality Constraint on Anchor Links}

The \textit{cardinality constraints} define both the limit on link cardinality and the limit on node degrees that those links are incident to. To be general, the links studied here can be either uni-directed or bi-directed, where undirected links are treated as bi-directed. For each node $u \in \mathcal{V}$ in the network, we can represent the potential links going-out from $u$ as set $\Gamma^{out}(u) = \{l | l \in \mathcal{L},\exists v \in \mathcal{V}, l = (u, v)\}$, and those going-into $u$ as set $\Gamma^{in}(u) = \{l | l \in \mathcal{L}, \exists v \in \mathcal{V}, l = (v, u)\}$. Furthermore, with the link label variables $\{y_l\}_{l \in \mathcal{L}}$, we can represent the out-degree and in-degree of node $u \in \mathcal{V}$ as $D^{out}(u) = \sum_{l \in \Gamma^{out}(u)} y_l$ and $D^{in}(u) = \sum_{l \in \Gamma^{in}(u)} y_l$ respectively. Considering that the node degrees cannot be negative, besides the upper bounds introduced by the \textit{cardinality constraints}, a lower bound ``$\ge 0$'' is also added to guarantee validity of node degrees by default.

\noindent \textbf{One-to-One Cardinality Constraint}

For the bi-directed anchor links with $1:1$ \textit{cardinality constraint}, the nodes in the information networks can be attached with at most one such kind of link. In other words, for all the nodes (e.g., $u \in \mathcal{V}$) in the network, its in-degree and out-degree can not exceed $1$, i.e.,
\begin{align}
&0 \le \sum_{l \in \Gamma^{out}(u)} y_l \le 1, \mbox{ and } 0 \le \sum_{l \in \Gamma^{in}(u)} y_l \le 1, \forall u \in \mathcal{V}.
\end{align}

\noindent \textbf{One-to-Many Cardinality Constraint}

Meanwhile, for the uni-directed supervision links with the $N:1$ \textit{cardinality constraint}, the manager nodes can have multiple ($N$) links going out from them while the subordinate nodes should have exactly one link going into them (except the CEO). In other words, for all the nodes (e.g., $u \in \mathcal{V}$) in the network, its \textit{out-degree} cannot exceed $N$ and the \textit{in-degree} should be exactly $1$, i.e.,
\begin{align}
&0 \le \sum_{l \in \Gamma^{out}(u)} y_l \le N, \mbox{ and } 1 \le \sum_{l \in \Gamma^{in}(u)} y_l \le 1, \forall u \in \mathcal{V}.
\end{align}

\noindent \textbf{Many-to-Many Cardinality Constraint}

In many cases, there usually exist no specific \textit{cardinality constraints} on links, and nodes can be connected with each other freely. Simply, we can assume the node \textit{in-degrees} and \textit{out-degrees} to be limited by the maximum degree parameter $N = |\mathcal{V}| - 1$, i.e., \begingroup\makeatletter\def\f@size{8}\check@mathfonts
\begin{align}
&0 \le \sum_{l \in \Gamma^{out}(u)} y_l \le N, \mbox{ and } 0 \le \sum_{l \in \Gamma^{in}(u)} y_l \le N, \forall u \in \mathcal{V}.
\end{align}\endgroup

The \textit{cardinality constraint} on links can be generally represented with the linear algebra equations. The relationship between nodes $\mathcal{V}$ and links $\mathcal{L}$ can actually be represented as matrices $\mb{T}^{out} \in \{0, 1\}^{|\mathcal{V}| \times |\mathcal{L}|}$ and $\mb{T}^{in} \in \{0, 1\}^{|\mathcal{V}| \times |\mathcal{L}|}$, where entry $\mb{T}^{out}(u, l) = 1$ iff $l \in \Gamma^{out}(u)$ and  $\mb{T}^{in}(u, l) = 1$ iff $l \in \Gamma^{in}(u)$. Based on the link label vector $\mb{y}$, the node out-degrees and in-degrees can be formally represented as vectors $\mb{T}^{out} \cdot \mb{y}$ and $\mb{T}^{in} \cdot \mb{y}$ respectively. The general representation of the \textit{cardinality constraints} introduced above can be rewritten as follows:
\begin{align}
&\underline{\mb{b}}^{out} \le \mb{T}^{out} \cdot \mb{y} \le \overline{\mb{b}}^{out}, \mbox{ and } \underline{\mb{b}}^{in} \le \mb{T}^{in} \cdot \mb{y} \le \overline{\mb{b}}^{in},
\end{align}
where vectors $\underline{\mb{b}}^{out}$, $\overline{\mb{b}}^{out}$, $\underline{\mb{b}}^{in}$ and $\overline{\mb{b}}^{in}$ can take different values depending on the cardinality constraint on the links (e.g., for the $1:1$ constraint, we have $\underline{\mb{b}}^{out} = \underline{\mb{b}}^{in} = \mb{0}$ and $\overline{\mb{b}}^{out} = \overline{\mb{b}}^{in} = \mb{1}$).

\subsubsection{Joint Objective Function Solution}

\setlength{\textfloatsep}{0pt}
\begin{algorithm}[t]
\caption{Greedy Link Selection}
\label{alg:chap4_sec5_greedy}
\begin{algorithmic}[1]
	\REQUIRE link estimate result $\hat{\mb{y}}$, parameter $k$\\
\ENSURE link label vector $\mb{y}$ 

\STATE 	{initialize link label vector $\mb{y} = \mb{0}$}
\FOR	{$l \in \mathcal{E}$}
\STATE	{$y_l = 1$}
\ENDFOR
\FOR	{$l \in \mathcal{L} \setminus \mathcal{E}$ and $\hat{y_l} < 0.5$}
\STATE	{$y_l = 0$}
\ENDFOR
\STATE	{Let $\tilde{\mathcal{L}} = \{l | l \in \mathcal{L} \setminus \mathcal{E}, \hat{y}_l \ge 0.5\}$}
\WHILE	{$\tilde{\mathcal{L}} \neq \emptyset$}
\STATE	{select $l \in \tilde{\mathcal{L}}$ with the highest estimation score}
\IF		{add $l$ as positive instance violates the \textit{cardinality constraint} or more than $k$ links have been selected}
\STATE	{$y_l = 0$}
\ELSE	
\STATE	{$y_l = 1$}
\ENDIF
\ENDWHILE
\STATE	{\textbf{return} $\mb{y}$}
\end{algorithmic}
\end{algorithm}

\setlength{\textfloatsep}{0pt}
\begin{algorithm}[t]
\caption{Cardinality Constrained Anchor Link Prediction Framework}
\label{alg:chap4_sec5_framework}
\begin{algorithmic}[1]
\REQUIRE link feature vector $\mb{X}$ \\
	\qquad  weight parameter $c$\\
\ENSURE parameter vector $\mb{w}$, link label vector $\mb{y}$ 

\STATE 	{Initialize label vector $\mb{y} = \frac{1}{2} \cdot \mb{1}$}
\STATE	{For links in $\mathcal{E}$, assign their label as $1$}
\STATE	{Initialize parameter vector $\mb{w} = \mb{0}$}
\STATE	{Initialize convergence-tag = False}
\WHILE	{convergence-tag == False}
\STATE	{Update vector $\mb{w}$ with equation $\mb{w} = c(\mb{I} + c\mb{X}^\top\mb{X})^{-1}\mb{X}^\top \mb{y}$}
\STATE	{Calculate link estimation result $\mb{\hat{y}} = \mb{X}\mb{w}$}
\STATE	{Update vector $\mb{y}$ with Algorithm Greedy($\mb{\hat{y}}$)}
\IF		{$\mb{w}$ and $\mb{y}$ both converge}
\STATE	{convergence-tag = True}
\ENDIF
\ENDWHILE
\end{algorithmic}
\end{algorithm}

For simplicity, we assume the weight scalars $c_1$ and $c_2$ both to be $c$, i.e., all the links in the networks are assumed to be of similar importance in training. And the new loss term of all the links in $\mathcal{E}$, $\mathcal{U}$ can be simplified as
\begin{align}
\frac{c}{2} \left\| \mb{w} \mb{X} - \mb{y} \right\|_2^2,
\end{align}
where matrix $\mb{X} = [\mb{x}_{l_1}^\top, \mb{x}_{l_2}^\top, \cdots, \mb{x}_{l_|\mathcal{L}|}^\top]^T$ denotes the feature matrix of all the links in $\mathcal{L}$.

Based on the above remarks, the constrained optimization objective function of the problem can be represented as
\begin{align}
&\min_{\mb{w}, \mb{y}} \frac{1}{2} \left\| \mb{w} \right\|_2^2 + \frac{c}{2} \left\| \mb{X}\mb{w} - \mb{y}\right\|_2^2,\\
&s.t.\ \  \mb{y} \in \{0, 1\}^{|\mathcal{L}| \times 1}, y_l = 1, \forall l \in \mathcal{E},\\
&\ \ \ \ \ \ \ \ \underline{\mb{b}}^{out} \le \mb{T}^{out} \cdot \mb{y} \le \overline{\mb{b}}^{out}, \underline{\mb{b}}^{in} \le \mb{T}^{in} \cdot \mb{y} \le \overline{\mb{b}}^{in}.
\end{align}

The above objective function involves variables $\mb{w}$ and $\mb{y}$ at the same time, which is actually not jointly convex and can be very challenging to solve. In \cite{wsdm17_2}, the proposed model solves the function with an alternative updating framework by fixing one variable and updating the other one iteratively. The framework involves two steps:

\vspace{5pt}
\noindent \textbf{Step 1}: Fix $\mb{y}$ and Update $\mb{w}$

By fixing $\mb{y}$ (i.e., treating $\mb{y}$ as a constant vector), the objective function about $\mb{w}$ can be simplified as
\begin{align}
\min_{\mb{w}} \frac{1}{2} \left\| \mb{w} \right\|^2_2 + \frac{c}{2} \left\| \mb{X}\mb{w} - \mb{y}\right\|^2_2.
\end{align}

Let $h(\mb{w}) = \frac{1}{2} \left\| \mb{w} \right\|^2_2 + \frac{c}{2} \left\| \mb{X}\mb{w} - \mb{y}\right\|^2_2$. By taking the derivative of the function $h(\mb{w})$ regarding $\mb{w}$ we can have
\begin{equation}
\frac{\mathrm{d} h(\mb{w})}{\mathrm{d} \mb{w}} = \mb{w} + c\mb{X}\mb{w}\mb{X}^\top - c\mb{y}\mb{X}^\top.
\end{equation}
By making the derivation to be zero, the optimal vector $\mb{w}$ can be represented to be
\begin{equation}
\mb{w} = c(\mb{I} + c\mb{X}^\top\mb{X})^{-1}\mb{X}^\top \mb{y},
\end{equation} 
and the minimum value of the function will be $\frac{c}{2} \mb{y}^\top \mb{y} - \frac{c^2}{2} \mb{y}^\top \mb{X} (\mb{I} + c\mb{X}^\top\mb{X})^{-1}\mb{X}^\top \mb{y}$.

\vspace{5pt}
\noindent \textbf{Step 2}: Fix $\mb{w}$ and Update $\mb{y}$

When fixing $\mb{w}$ and treating it as a constant vector, the objective function about $\mb{y}$ can be represented as
\begin{align}
&\min_{\mb{y}} \frac{c}{2} \left\| \mb{\hat{y}}-\mb{y}\right\|^2_2,\\
&s.t.\ \  \mb{y} \in \{0, 1\}^{|\mathcal{L}| \times 1}, y_l = 1, \forall l \in \mathcal{E},\\
&\ \ \ \ \ \ \ \ \underline{\mb{b}}^{out} \le \mb{T}^{out} \cdot \mb{y} \le \overline{\mb{b}}^{out}, \underline{\mb{b}}^{in} \le \mb{T}^{in} \cdot \mb{y} \le \overline{\mb{b}}^{in},
\end{align}
where $\mb{\hat{y}} = \mb{X}\mb{w}$ denotes the inference results of the links in $\mathcal{L}$ with the updated parameter vector $\mb{w}$ from Step 1. The objective function is an constrained non-linear integer programming problem about variable $\mb{y}$. Formally, the above optimization sub-problem is named as the ``Cardinality Constrained Link Selection'' problem. The problem is shown to be NP-hard (we will analyze it in the next subsection), and achieving the optimal solution to it is very time consuming. To preserve the \textit{cardinality constraints} on the variables and minimize the loss term, one brute-force way to achieve the optimal solution $\mb{y}$ is to enumerate all the feasible combination of links candidates to be selected as the positive instances, which will lead to very high time complexity. In \cite{wsdm17_2}, a greedy link selection algorithm is adopted to resolve the problem, and the pseudo-code of the greedy link selection method is available in Algorithm~\ref{alg:chap4_sec5_greedy}. Meanwhile, the framework is illustrated with the pseudo-code available in Algorithm~\ref{alg:chap4_sec5_framework}. The framework updates vectors $\mb{w}$ and $\mb{y}$ alternatively until both of them converge, where vector $\mb{y}$ will be returned as the final prediction results.



\section{Link Prediction}\label{sec:link_prediction}

Given a screenshot of an online social network, the problem of inferring the missing links or the links to be formed in the future is called the \textit{link prediction} problem. Link prediction problem has concrete applications in the real world, and many social network services can be cast to the link prediction problem. For instance, the friend recommendations problem in online social networks can be modeled as the social link prediction problem among users. Users' trajectory prediction problem can be formulated as the prediction task of potential checkin links between users and offline POIs (point of interest) in the location based social networks. The user identifier resolution problem across networks (i.e., the network alignment problem introduced in the previous section) can be modeled as the anchor link prediction problem of user accounts across different online social networks.

In this section, we will introduce the general link prediction problems in the online social networks. Formally, given the training set $\mathcal{T}_{train}$ involving links belong to different classes ($\mathcal{Y} = \{+1, -1\}$ denoting the links that have been/will be formed and those will never be formed) and the test set $\mathcal{T}_{test}$ (with unknown labels for the links), the link prediction problem aims at building a mapping $f: \mathcal{T}_{test} \to \mathcal{Y}$ to infer the potential labels of links in the test set $\mathcal{T}_{test}$.

Depending on the scenarios of the link prediction problems, the existing links prediction works can be divided into several different categories. Traditional link prediction problems are mainly focused on inferring the links in one single homogeneous network, like inferring the friendship links among users in online social networks or co-author links in bibliographic networks. As the network structures are becoming more and more complicated, many of them are modeled as the heterogeneous networks involving different types of nodes and complex connections among them. The heterogeneity of the networks leads to many new link prediction problems, like predicting the links between nodes belonging to different categories and the concurrent inference of multiple types of links in the heterogeneous networks. In recent years, many online social networks have appeared, and lots of new research opportunities exist for researchers and practitioners to study the link prediction problem from the cross-network perspective.

Meanwhile, depending on the learning settings used in the link prediction problem formulation and models, the existing link prediction works can be categorized in another way. For some of the link prediction models, they calculate the user-pair closeness as the prediction result without needing any training data, which are referred to as the \textit{unsupervised link prediction models}. For some other models, they will label the known links into different classes, and use them as the training set to learn a supervised classification models as the base model instead. These models are called the \textit{supervised link prediction models}. Usually, manual labeling of the links is very expensive and tedious. In recent years, many of the works have proposed to apply semi-supervised learning techniques in the link prediction problem to utilize the links without labels. 

In this part, we will introduce the link prediction problems in online social networks, including the \textit{traditional homogeneous link prediction}, \textit{cold start link prediction}, and \textit{cross-network link prediction}, which covers the \textit{PU link prediction} and \textit{sparse and low rank matrix estimation based link prediction}.


\subsection{Traditional Homogeneous Network Link Prediction}~\label{sec:chap7_sec2_homogeneous}

Traditional link prediction problems are mainly studied based on one homogeneous network, involving one single type of nodes and links. In this section, we will first briefly introduce how to use the social closeness measures for link prediction tasks. To integrate different social closeness measures together in the link prediction task, we will talk about the supervised link prediction model. Finally, we will introduce some models which formulate the link prediction task as a recommendation problem, and apply the matrix factorization method to address the problem.

\subsubsection{Unsupervised Link Prediction}\label{subsec:closeness_measures}

Given a screenshot of a homogeneous network $G = (\mathcal{V}, \mathcal{E})$, the unsupervised link prediction methods \cite{LK07} aims at inferring the potential links that will be formed in the future. Usually, the unsupervised link prediction models will calculate some scores for the links, which will be used as the predicted confidence scores of these links. Depending on the specific scenario and the link formation assumptions applied, different measures have been proposed for the link prediction models.

\noindent \textbf{Local Neighbor based Predicators}: Local neighbor based predicators are based on regional social network information, i.e., neighbors of users in the network. Consider, for example, given a social link $(u, v)$ in network $G$, where $u$ and $v$ are both users in $G$, the neighbor sets of $u, v$ can be represented as $\Gamma(u)$ and $\Gamma(v)$ respectively. Based on $\Gamma(u)$ and $\Gamma(v)$, the following predicators measuring the proximity of users $u$ and $v$ in network $G$ can be obtained.

\begin{enumerate}

\item \textit{Preferential Attachment Index} (PA) \cite{BJNRSV02}:
\begin{equation}
PA(u, v) = \left | \Gamma(u) \right | \left | \Gamma(v) \right |.
\end{equation}

$PA(u, v)$ uses the product of the degrees of users $u$ and $v$ in the network as the proximity measure, considering that new links are more likely to appear between users who have large number of social connections.

\item \textit{Common Neighbor} (CN) \cite{HZ11}:
\begin{equation}
CN(u, v) = \left | \Gamma(u) \cap \Gamma(v) \right |.
\end{equation}

$CN(u, v)$ uses the number of shared neighbor as the proximity score of user $u$ and $v$. The larger $CN(u, v)$ is, the closer user $u$ and $v$ are in the network.

\item \textit{Jaccard's Coefficient} (JC) \cite{HZ11}:
\begin{equation}
JC(u, v) = \frac{ \left | \Gamma(u) \cap \Gamma(v) \right |}{ \left | \Gamma(u) \cup \Gamma(v) \right |}.
\end{equation}

$JC(u, v)$ takes the total number of neighbors of $u$ and $v$ into account, considering that $CN(u, v)$ can be very large because each one has a lot of neighbors rather than they are strongly related to each other.

\item \textit{Adamic/Adar Index} (AA) \cite{AA01}:
\begin{equation}
AA(u, v) = \sum_{w \in (\Gamma(u) \cap \Gamma(v))} \frac{1}{\log \left | \Gamma(w) \right |}.
\end{equation}

Different from $JC(u, v)$, $AA(u, v)$ further gives each common neighbor of user $u$ and $v$ a weight, $ \frac{1}{\log \left | \Gamma(w) \right |}$, to denote its importance.

\item \textit{Resource Allocation Index} (RA) \cite{ZLZ09}:
\begin{equation}
RA(u, v) = \sum_{w \in (\Gamma(u) \cap \Gamma(v))} \frac{1}{\left | \Gamma(w) \right |}.
\end{equation}

$RA(u, v)$ gives each common neighbor a weight $\frac{1}{\left | \Gamma(w) \right |}$ to represent its importance, where those with larger degrees will have a less weight number.

\end{enumerate}

All these predicators are called \textit{local neighbor based predicators} as they are all based on users' local social network information.

\noindent \textbf{Global Path based Predicators}: In addition to the local neighbor based predicators, many other predicators based on paths in the network have also been proposed to measure the proximity among users.

\begin{enumerate}

\item \textit{Shortest Path} (SP) \cite{HCSZ06}:
\begin{equation}
SP(u, v) = \min \{\left | p_{u \rightsquigarrow v} \right |\},
\end{equation}
where $p_{u \rightsquigarrow v}$ denotes a path from $u$ to $v$ in the network and $\left | p \right |$ represents the length of path $p$.

\item \textit{Katz} \cite{K53}:
\begin{equation}
Katz(u, v)  = \sum_{l = 1}^{\infty} \beta^l \left | p_{u \rightsquigarrow v}^l \right |,
\end{equation}
where $p_{u \rightsquigarrow v}^l$ is the set of paths of length $l$ from $u$ to $v$ and parameter $\beta \in [0,1]$ is a regularizer of the predicator. Normally, a small $\beta$ favors shorter paths as $\beta^l$ can decay very quickly when $\beta$ is small, in which case $Katz(u, v)$ will be behave like the predicators based on local neighbors.

\end{enumerate}

\noindent \textbf{Random Walk based Link Prediction}: In addition to the unsupervised link predicators which can be obtained from the networks directly, there exists another category link prediction methods which can calculate the proximity scores among users based on \textit{random walk} \cite{GKR98, FPRS07, KSJ09, BL11, TFP06, LZ11, HZ11}. In this part, we will introduce the concept of random walk at first. Next, we will introduce the proximity measures based on random walk, which include the \textit{commute time} \cite{FPRS07, LZ11, HZ11}, \textit{hitting time} \cite{FPRS07, LZ11, HZ11} and \textit{cosine similarity} \cite{FPRS07, LZ11, HZ11}.

Let matrix $\mb{A}$ be the adjacency matrix of network $G$, where $A(i,j) = 1$ iff social link $(u_i, u_j) \in \mathcal{E}$, where $u_i, u_j \in \mathcal{V}$. The normalized matrix of $\mb{A}$ by rows will be $\mb{P} = \mb{D}^{-1} \mb{A}$, where diagonal matrix $\mb{D}$ of $\mb{A}$ has value $D(i,i) = \sum_j A(i,j)$ on its diagonal and $P(i,j)$ stores the probability of stepping on node $u_j \in \mathcal{V}$ from node $u_i \in \mathcal{V}$. Let entries in vector $\mb{x}^{(\tau)}(i)$ denote the probabilities that a random walker is at user node $u_i \in \mathcal{V}$ at time $\tau$. Then we have the updating equation of entry $\mb{x}^{(\tau)}(i)$ via the random walk as follows:
\begin{equation}
\mb{x}^{(\tau + 1)}(i) = \sum_j \mb{x}^{(\tau)}(j) \mb{P}(j,i).
\end{equation}

In other words, the updating equation of vector $\mb{x}$ will be represented as:
\begin{equation}
\mb{x}^{(\tau + 1)} = \mb{P} \mb{x}^{(\tau)}.
\end{equation}

By keeping updating $\mb{x}$ according to the following equation until convergence, we can have the stationary vector $\mb{x}^{(\tau + 1)}$ as
\begin{equation}
\begin{cases}
\mb{x}^{(\tau + 1)} = \mb{P}^T\mb{x}^{(\tau)}, \\
\mb{x}^{(\tau + 1)} = \mb{x}^{(\tau)}.
\end{cases}
\end{equation}
The above equation is equivalent to 
\begin{equation}
\mb{v} = \mb{P}^T\mb{v},
\end{equation}
where vector $\mb{v}$ denotes the stationary random walk probability vector.

The above equation denotes that the final stationary distribution vector $\mb{v}$ is actually a eigenvector of matrix $\mb{P}^T$ corresponding to eigenvalue $1$. Some existing works have pointed out that if a markov chain is \textit{irreducible} \cite{FPRS07} and \textit{aperiodic} \cite{FPRS07} then the largest eigenvalue of the transition matrix will be equal to $1$ and all the other eigenvalues will be strictly less than $1$. In addition, in such a condition, there will exist one single unique stationary distribution which is vector $\mb{v}$ obtained at convergence of the updating equations.

\begin{defn}
(Irreducible): Network $G$ is \textit{irreducible} if there exists a path from every node to every other nodes in $G$ \cite{FPRS07}.
\end{defn}

\begin{defn} 
(Aperiodic): Network $G$ is \textit{aperiodic} if the greatest common divisor of the lengths of its cycles in $G$ is $1$, where the greatest common divisor is also called the \textit{period} of $G$ \cite{FPRS07}.
\end{defn}

\noindent \textbf{Proximity Measures based on Random Walk}

\begin{enumerate}
\item \textit{Hitting Time} (HT): \begingroup\makeatletter\def\f@size{8}\check@mathfonts
\begin{equation}
\hspace{-4pt} HT(u, v) = \mathbb{E}\left( \min \{\tau | \tau \in \mathbb{N}^+, X^{(\tau)} = v \land X^0 = u \} \right),
\end{equation}\endgroup
where variable $X^{(\tau)} = v$ denotes that a random walker is at node $v$ at time $\tau$.

$HT(u, v)$ counts the average steps that a random walker takes to reach node $v$ from node $u$. According to the definition, the hitting time measure is usually asymmetric, $HT(u, v) \ne HT(v, u)$. Based on matrix $\mb{P}$ defined before, the definition of $HT(u, v)$ can be redefined as \cite{FPRS07}:
\begin{equation}
HT(u, v) = 1 + \sum_{w \in \Gamma(u)} P_{u,w} HT(w, v).
\end{equation}

\item \textit{Commute Time} (CT): 
\begin{equation}
CT(u, v) = HT(u, v) + HT(v, u).
\end{equation}

$CT(u, v)$ counts the expectation of steps used to reach node $u$ from $v$ and those needed to reach node $v$ from $u$. According to existing works, the commute time, $CT(u, v)$, can be obtained as follows
\begin{equation}
CT(u, v) = 2m(L^{\dagger}_{u,u} + L^{\dagger}_{v,v} - 2L^{\dagger}_{u,v}),
\end{equation}
where $\mb{L}^{\dagger}$ is the pseudo-inverse of matrix $\mb{L} = \mb{D}_A - \mb{A}$.

\item \textit{Cosine Similarity based on $\mb{L}^\dagger$} (CS):
\begin{equation}
CS(u,v) = \frac{\mb{x}^T_u \mb{x}_v}{\sqrt{(\mb{x}^T_u \mb{x}_u)(\mb{x}^T_v \mb{x}_v)}},
\end{equation}
where, $\mb{x}_u = (\mb{L}^\dagger)^{\frac{1}{2}} \mb{e}_u$ and vector $\mb{e}_u$ is a vector of $0$s except the entries corresponding to node $u$ that is filled with $1$. According to existing works \cite{FPRS07, LZ11}, the cosine similarity based on $\mb{L}^\dagger$ , $CS(u,v)$, can be obtained as follows,
\begin{equation}
CS(u, v) = \frac{L^\dagger_{u,v}}{\sqrt{L^\dagger_{u,u} L^\dagger_{v,v}}}.
\end{equation}

\item \textit{Random Walk with Restart} (RWR):
Based on the definition of random walk, if the walker is allowed to return to the starting point with a probability of $1 - c$, where $c \in [0,1]$, then the new random walk method is formally defined as \textit{random walk with restart}, whose updating equation is shown as follows:
\begin{equation}
\begin{cases}
\mb{x}^{(\tau + 1)}_u = c \mb{P}^T \mb{x}^{(\tau)}_u + (1-c)\mb{e}_u, \\
\mb{x}^{(\tau + 1)}_u = \mb{x}^{(\tau)}_u.
\end{cases}
\end{equation}

Keep updating $\mb{x}$ until convergence, the stationary distribution vector $\mb{x}$ can meet
\begin{equation}
\mb{x}_u = (1 - c)(\mb{I} - c\mb{P}^T)^{-1}\mb{e}_u.
\end{equation}

The proximity measure based on random walk with restart between user $u$ and $v$ will be
\begin{equation}
RWR(u, v) = \mb{x}_u(v),
\end{equation}
where $\mb{x}_u(v)$ denotes the entry corresponding to $v$ in vector $\mb{x}_u$.

\end{enumerate}

\subsubsection{Supervised Link Prediction}

In some cases, links in the networks are explicitly categorized into different groups, like links denoting friends vs those representing enemies, friends (formed connections) vs strangers (no connections). Given a set of labeled links, e.g., set $\mathcal{E}$, containing links belonging to different classes, the \textit{supervised link prediction} \cite{HCSZ06} problem aims at building a supervised learning model with the labeled set. The learnt model will be applied to determine the labels of links in the test set. In this part, we still take the link formation problem as an example to illustrate the supervised link prediction model.

To represent each of the social links, like link $l = (u, v) \in \mathcal{E}$ between nodes $u$ and $v$, a set of features representing the characteristics of the link $l$ or nodes $u$, $v$ will be extracted in the model building. Normally, the features can be extracted for links in the prediction task can be divided into two categories:

\noindent \textbf{Link Feature Extraction}
\begin{itemize}
\item \textit{Features of Nodes}: The characteristics of the nodes can be denoted by various measures, like these various node centrality measures. For instance, for the link $(u, v)$, based on the known links in the training set, the centrality measures can be computed based on degree, normalized degree, eigen-vector, Katz, PageRank, Betweenness of nodes $u$ and $v$ as part of the features for link $(u, v)$.

\item \textit{Features of Links}: The characteristics of the links in the networks can be calculated by computing the closeness between the nodes composing the nodes. For instance, for link $(u, v)$, based on the known links in the training set, the closeness measures can be computed based on reciprocity, common neighbor, Jaccard's coefficient, Adamic/Adar, shortest path, Katz, hitting time, commute time, etc. between nodes $u$ and $v$ as the features for link $(u, v)$.
\end{itemize}

We can append the features for nodes $u$, $v$ and those for link $(u, v)$ together and represent the extracted feature vector for link $l = (u, v)$ as vector $\mb{x}_l \in \mathbb{R}^{k \times 1}$, whose length is $k$ in total.

\noindent \textbf{Link Prediction Model}

With the training set $\mathcal{L}_{train}$, the feature vectors and labels for the links in $\mathcal{L}_{train}$ can be represented as the training data $\{(\mb{x}_l, y_l)\}_{l \in \mathcal{L}_{train}}$. Meanwhile, with the testing set $\mathcal{L}_{test}$, the features extracted for the links in it can be represented as $\{\mb{x}_l\}_{l \in \mathcal{L}_{train}}$. Different classification models can be used as the base model for the link prediction task, like the Decision Tree, Artificial Neural Network and Support Vector Machine (SVM). The model can be trained with the training data, and the labels of links in the test can be determined by applying models to the test set.

Depending on the specific model being applied, the output of the link prediction result can include (1) the predicted labels of the links, and (2) the prediction confidence scores/probability scores of links in the test set.

\subsubsection{Matrix Factorization based Link Prediction}

Besides unsupervised link predicators and the classification based supervised link prediction models, many other methods based on matrix factorization can also be applied to solve the link prediction task in homogeneous networks \cite{MAE11, TGHL13, DKA11}. 

Given a homogeneous social network $G = (\mathcal{V}, \mathcal{E})$ and the existing social links among users in set $\mathcal{E}$, the links can represented with the social adjacency matrix $\mb{A} \in \{0, 1\}^{|\mathcal{V}| \times |\mathcal{V}|}$. Given the adjacency matrix $\mb{A}$ of network G, \cite{icde17_2} proposes to use a low-rank compact representation, $\mb{U} \in \mathbb{R}^{|\mathcal{V}| \times d}, d < |\mathcal{V}|$, to store social information for each user in the network. Matrix $\mb{U}$ can be obtained by solving the following optimization objective function:
\begin{equation}
\min_{\mb{U}, \mb{V}} \left \| \mb{A} - \mb{U}\mb{V}\mb{U}^T \right \|^2_F,
\end{equation}
where $\mb{U}$ is the low rank matrix and matrix $\mb{V}$ saves the correlation among the rows of $\mb{U}$, $\left \| \mb{X} \right\|_F$ is the Frobenius norm of matrix $\mb{X}$.

To avoid overfitting, regularization terms $\left \| \mb{U} \right \|^2_F$ and $\left \| \mb{V} \right \|^2_F$ are added to the object function as follows:
\begin{align}
&\min_{\mb{U}, \mb{V}} \left \| \mb{A} - \mb{U}\mb{V}\mb{U}^T \right \|^2_F + \alpha \left \| \mb{U} \right \|^2_F + \beta \left\| \mb{V} \right \|^2_F,\\
&s.t.,	\mb{U} \ge \mb{0}, \mb{V} \ge \mb{0},
\end{align}
where $\alpha$ and $\beta$ are the weight of terms $\left \| \mb{U} \right \|^2_F$, $\left \| \mb{V} \right \|^2_F$ respectively.

This object function is very hard to achieve the global optimal result for both $\mb{U}$ and $\mb{V}$. A alternative optimization schema can be used here, which can update $\mb{U}$ and $\mb{V}$ alternatively. The Lagrangian function of the object equation should be: \begingroup\makeatletter\def\f@size{8}\check@mathfonts
\begin{align}
\mathcal{F} &= Tr(\mb{A}\mb{A}^T) - Tr(\mb{A} \mb{U} \mb{V}^T \mb{U}^T)\\ 
&- Tr(\mb{U}\mb{V}\mb{U}^T\mb{A}^T) +Tr(\mb{U}\mb{V}\mb{U}^T\mb{U}\mb{V}^T\mb{U}^T) \\
&+ \alpha Tr(\mb{U}\mb{U}^T) + \beta Tr(\mb{V}\mb{V}^T) - Tr({\Theta}\mb{U}) - Tr({\Omega}\mb{V}),
\end{align}\endgroup
where ${\Theta}$ and ${\Omega}$ are the multiplier for the constraint of $\mb{U}$ and $\mb{V}$ respectively.

By taking derivatives of $\mathcal{F}$ with regarding to $\mb{U}$ and $\mb{V}$ respectively, the partial derivatives of $\mathcal{F}$ will be \begingroup\makeatletter\def\f@size{8}\check@mathfonts
\begin{align}
\frac{\partial \mathcal{F}}{\partial \mb{U}} =&-2 \mb{A}^T\mb{U}\mb{V} -2 \mb{A}\mb{U}\mb{V}^T +2\mb{U}\mb{V}^T\mb{U}^T\mb{U}\mb{V}^T \\
&+ 2\mb{U}\mb{V}\mb{U}^T\mb{U}\mb{V}^T + 2\alpha \mb{U} - {\Theta}^T,\\
\frac{\partial \mathcal{F}}{\partial \mb{V}} =& -2\mb{U}^T\mb{A}\mb{U} +2\mb{U}^T\mb{U}\mb{V}\mb{U}^T\mb{U} + 2\beta\mb{V} - {\Omega}^T
\end{align}\endgroup
Let $\frac{\partial \mathcal{F}}{\partial \mb{U}} = \mb{0}$ and $\frac{\partial \mathcal{F}}{\partial \mb{V}} = \mb{0}$ and use the KKT complementary condition, we can get: \begingroup\makeatletter\def\f@size{8}\check@mathfonts
\begin{equation}
\begin{cases}
\mb{U}(i,j) \hspace{-5pt} &\gets \mb{U}(i,j)  \sqrt{\frac{\left(\mb{A}^T\mb{U}\mb{V} + \mb{A}\mb{U}\mb{V}^T \right)(i,j)}{\left(\mb{U}\mb{V}^T\mb{U}^T\mb{U}\mb{V} + \mb{U}\mb{V}\mb{U}^T\mb{U}\mb{V}^T + \alpha\mb{U} \right)(i,j)}},\\
\mb{V}(i,j) \hspace{-5pt} &\gets \mb{V}(i,j) \sqrt{\frac{\left(\mb{U}^T\mb{A}\mb{U} \right)(i,j)}{\left(\mb{U}^T\mb{U}\mb{V}\mb{U}^T\mb{U} + \beta\mb{V} \right)(i,j)}}.
\end{cases}
\end{equation}\endgroup

The low-rank matrix $\mb{U}$ captures the information of each users from the adjacency matrix. The matrix $\mb{U}$ can be used in different ways. For instance, each row of $\mb{U}$ represents the \textit{latent feature vectors} of users in the network, which can be used in many link prediction models, e.g., supervised link prediction models. Meanwhile, based on the matrix $\mb{V}$ learnt from the model, the predicted score of link $(u, v)$ can be represented as $\mb{U}_u \mb{V} \mb{U}_v^\top$, where notations $\mb{U}_u$ and $\mb{U}_v$ represent the rows in matrix $\mb{U}$ corresponding to users $u$ and $v$ respectively.

\subsection{Cold Start Link Prediction for New Users}\label{sec:chap7_sec4_new}

These previous works on link prediction focus on predicting potential links that will appear among all the users, based upon a snapshot of the social network. These works treat all users equally and try to predict social links for all users in the network. However, in real-world social networks, many new users are joining in the service every day. Predicting social links for new users are more important than for those existing active users in the network as it will leave the first impression on the new users. First impression often has lasting impact on a new user and may decide whether he will become an active user. A bad first impression can turn a new user away. So it is important to make meaningful recommendation to  a new user to create a good first impression and attract him to participate more. For simplicity, we refer users that have been actively using the the network for a long time as ``old users''. It has been shown in previous works that there is a negative correlation between the age of nodes in the network and their link attachment rates. The distribution of linkage formation probability follows a power-law decay with the age of nodes \cite{KE02}. So, new users are more likely to accept the recommended links compared with existing old users and predicting links for new users could lead to more social connections. In this part, we will introduce a recent research work on link prediction for new users, which is based on \cite{icdm13}.

A natural challenge inherent in the usage of the historical links in social networks to predict social links for new users is the differences in information distributions of new users and old users as mentioned before. To address this problem, \cite{icdm13} propose a method to accommodate old users' and new users' sub-network by using a within-network personalized sampling method to process old users' information. By sampling the old users' sub-network, we want to meet the following objectives:
\begin{itemize}
	\item \textit{Maximizing Relevance}: We aim at maximizing the relevance of the old users' sub-network and the new users' sub-network to accommodate differences in information distributions of new users and old users in the heterogeneous target network.
	\item \textit{Information Diversity}: Diversity of old users' information after sampling is still of great significance and should be preserved. 
	\item \textit{Structure Maintenance}: Some old users possessing sparse social links should have higher probability to survive after sampling to maintain their links so as to maintain the network structure.
\end{itemize}

Let the target network be $G = (\mathcal{V}, \mathcal{E})$, and $\mathcal{V} = \mathcal{V}_{old} \cup \mathcal{V}_{new}$ is the set of user nodes (i.e., set of old users and new users) in the target network. Personalized sampling is conducted on the old users' part: $G_{old} = (\mathcal{V}_{old}, \mathcal{E}_{old})$, in which each node is sampled independently with the sampling rate distribution vector $\boldsymbol{\delta}$ = ($\delta_1$, $\delta_2$, $\cdots$, $\delta_n$), where $n = \left| \mathcal{V}_{old} \right|$, $\sum_{i=1}^n \delta_i = 1$ and $\delta_i \ge 0$. Old users' sub-network after sampling is denoted as $\bar{G}_{old} = (\bar{\mathcal{V}}_{old}, \bar{\mathcal{V}}_{old})$.

We aim at making the old users' sub-network as relevant to new users' as possible. To measure the similarity score of a user $u_i$ and a heterogeneous network $G$, we define a relevance function as follows:
\begin{equation}
R(u_i, G) = \frac{1}{\left| \mathcal{V} \right|}\sum_{u_j \in \mathcal{V}}S(u_i, u_j)
\end{equation}
where set $\mathcal{V}$ is the user set of network $G$ and $S(u_i, u_j)$ measures the similarity between user $u_i$ and $u_j$ in the network. Each user has social relationships as well as other heterogeneous auxiliary information and $S(u_i, u_j)$ is defined as the average of similarity scores of these two parts:
\begin{equation}
S(u_i, u_j) = \frac{1}{2}(S_{aux}(u_i, u_j) + S_{social}(u_i, u_j))
\end{equation}

There are many different methods measuring the similarities of these auxiliary information in different aspects, e.g. cosine similarity. As to the social similarity, Jaccard's Coefficient can be used to depict how similar two users are in their social relationships. 

The relevance between the sampled old users' network and the new users' network could be defined as the expectation value of function $R(\bar{u}_{old}, G_{new})$: \begingroup\makeatletter\def\f@size{8}\check@mathfonts
\begin{align}
R(\bar{G}_{old}, G_{new}) &= \mathbb{E}(R(\bar{u}_{old}, G_{new})) \\
			&= \frac{1}{\left| \mathcal{V}_{new} \right|} \sum_{j=1}^{\left | \mathcal{V}_{new} \right |} \mathbb{E}(S(\bar{u}_{old}, u_{new,j})) \\
			&= \frac{1}{\left| \mathcal{V}_{new} \right|} \sum_{j=1}^{\left | \mathcal{V}_{new} \right |} \sum_{i=1}^{\left | \mathcal{V}_{old} \right |} \delta_i \cdot S(\bar{u}_{old,i},u_{new,j}) \\
			&= \boldsymbol{\delta}^\top \boldsymbol{s}
\end{align}\endgroup
where vector $\boldsymbol{s}$ equals: \begingroup\makeatletter\def\f@size{7}\check@mathfonts
\begin{equation}
\frac{1}{\left| \mathcal{V}_{new} \right|}\Big[ \hspace{-2pt} \sum_{j=1}^{\left | \mathcal{V}_{new} \right |} \hspace{-5pt} S(\bar{u}_{old,1},u_{new,j}), \cdots,  \hspace{-5pt}  \sum_{j=1}^{\left | \mathcal{V}_{new} \right |} \hspace{-5pt} S(\bar{u}_{old,n},u_{new,j}) \Big]^\top
\end{equation}\endgroup
and ${\left | \mathcal{V}_{old} \right |} = n$. Besides the relevance, we also need to ensure that the diversity of information in the sampled old users' sub-network could be preserved. Similarly, it also includes diversities of the auxiliary information and social relationships. The diversity of auxiliary information is determined by the sampling rate $\delta_i$, which could be define with the averaged \textit{Simpson Index} \cite{S49} over the old users' sub-network.
\begin{equation}
D_{aux}(\bar{G}_{old}) = \frac{1}{\left| \mathcal{V}_{old} \right|} \cdot \sum_{i=1}^{\left| \mathcal{V}_{old} \right|}{\delta_i^2}.
\end{equation}
As to the diversity in the social relationship, we could get the existence probability of a certain social link $(u_i, u_j)$ after sampling to be proportional to $\delta_i \cdot \delta_j$. So, the diversity of social links in the sampled network could be defined as average existence probabilities of all the links in the old users' sub-network.
\begin{equation}
D_{social}(\bar{G}_{old}) =\frac{1}{{\left| S_{old} \right|}} \cdot \sum_{i=1}^{\left| \mathcal{V}_{old} \right|}\sum_{j=1}^{\left| \mathcal{V}_{old} \right|}\delta_i \cdot \delta_j \times I(u_i, u_j)
\end{equation}
where $\left| S_{old} \right|$ is the size of social link set of old users' sub-network and  $I(u_i, u_j)$ is an indicator function $I: (u_i, u_j) \to \{0, 1\}$ to show whether a certain social link exists or not originally before sampling. For example, if link $(u_i, u_j)$ is a social link in the target network originally before sampling, then $I(u_i, u_j) = 1$, otherwise it equals to $0$.

Considering these two terms simultaneously, we could have the diversity of information in the sampled old users' sub-network to be the average diversities of these two parts: \begingroup\makeatletter\def\f@size{8}\check@mathfonts
\begin{align}
D(\bar{G}_{old}) &= \frac{1}{2}(D_{social}(\bar{G}_{old}) + D_{aux}(\bar{G}_{old}))\\
			&= \frac{1}{2}(\sum_{i=1}^{\left| \mathcal{V}_{old} \right|}\sum_{j=1}^{\left| \mathcal{V}_{old} \right|}\frac{1}{{\left| S_{old} \right|}} \cdot \delta_i \cdot \delta_j \times I(u_i, u_j)\\
			&+ \sum_{i=1}^{\left| \mathcal{V}_{old} \right|}\frac{1}{\left| \mathcal{V}_{old} \right|} \cdot {\delta_i^2} )\\
			&= \boldsymbol{\delta}^\top \cdot (\frac{1}{2{\left| S_{old} \right|}} \cdot \mathbf{A_{old}} + \frac{1}{2\left| \mathcal{V}_{old} \right|} \cdot \mathbf{I_{\left| \mathcal{V}_{old} \right|}}) \cdot \boldsymbol{\delta}
\end{align}\endgroup
where matrix $\mathbf{I_{\left| \mathcal{V}_{old} \right|}}$ is the diagonal identity matrix of size ${\left| \mathcal{V}_{old} \right|} \times {\left| \mathcal{V}_{old} \right|}$ and $\mathbf{A_{old}}$ is the adjacency matrix of old users' sub-network. 

To ensure that the structure of the original old users' subnetwork is not destroyed, we need to ensure that users with few links could also preserve their links. So, we could add a regularization term to increase the sampling rate for these users as well as their neighbors by maximizing the following terms:
\begin{align}
	Reg(\bar{G}_{old}) &= \min\{\mathcal{N}_i, \min_{u_j \in \mathcal{N}_i}\{\mathcal{N}_j\}\} \times \delta_i^2 = \boldsymbol{\delta}^\top \cdot \mathbf{M} \cdot \boldsymbol{\delta}
\end{align}
where matrix $\mathbf{M}$ is a diagonal matrix with element $\mathbf{M}_{i,i} = \min\{\mathcal{N}_i, \min_{u_j \in \mathcal{N}_i}\{\mathcal{N}_j\}\} = \min\{\mathcal{N}_i, \{\mathcal{N}_i| u_j \in \mathcal{N}_i\}$ and $\mathcal{N}_j = |\Gamma(u_j)|$ is the size of user $u_j$'s neighbor set. So, if a user or his/her neighbors have few links, then this user as well as his/her neighbors should have higher sampling rate so as to preserve the links between them.

\begin{figure}[!t]
 \centering    
 \begin{minipage}[l]{0.8\columnwidth}
  \centering
    \includegraphics[width=1.0\textwidth]{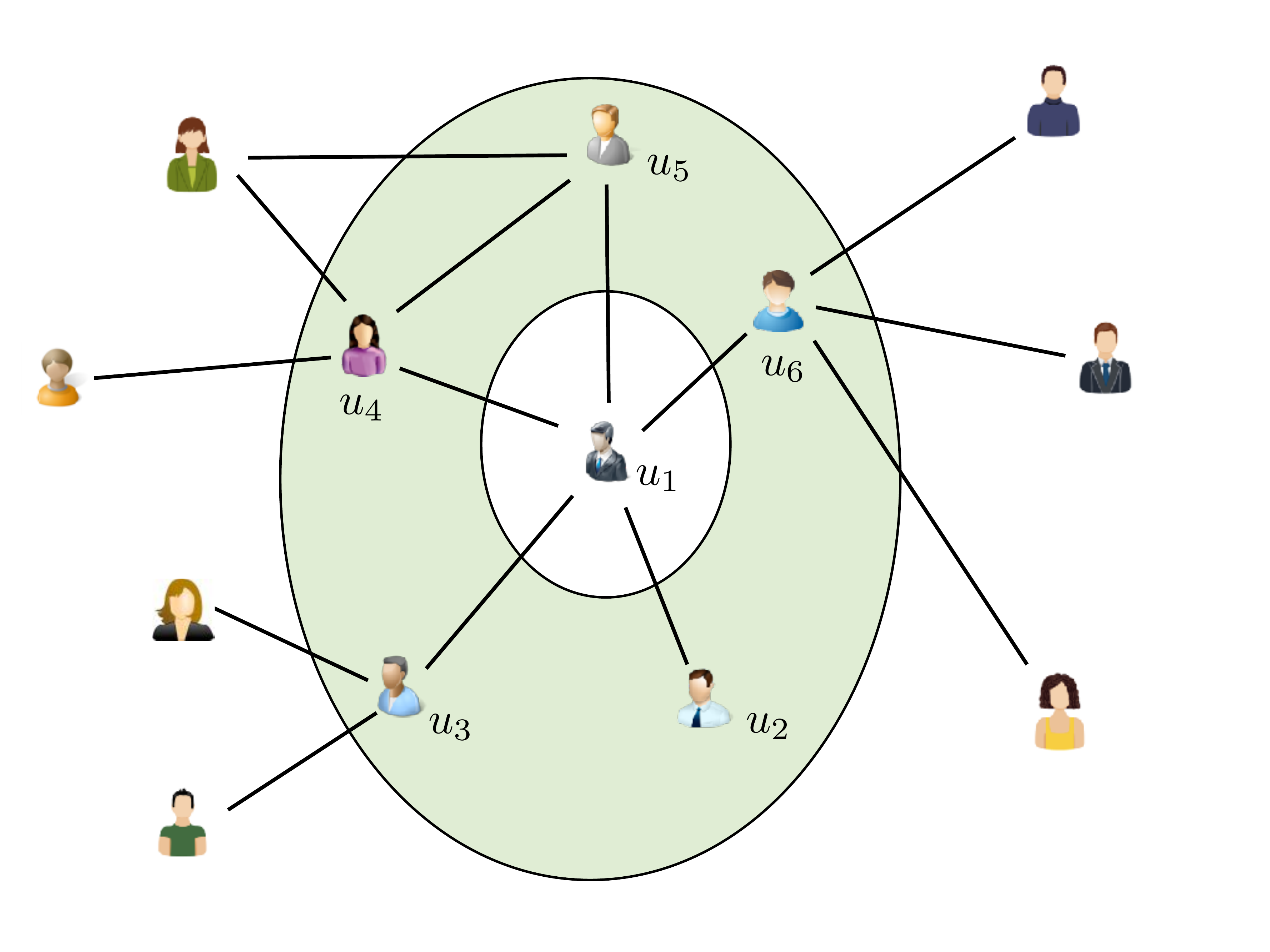}
 \end{minipage}
\caption{Personalized Network Sampling with Preservation of the Network Structure Properties.}\label{fig:chap7_sec4_fig_eg4}
\end{figure}

For example, in Figure~\ref{fig:chap7_sec4_fig_eg4}, we have $6$ users. To decide the sampling rate of user $u_1$, we need to consider his/her social structure. We find that since $u_1$'s neighbor $u_2$ has no other neighbor except $u_1$. To preserve the social link between $u_1$ and $u_2$ we need to increase the sampling rate of $u_2$. However, the existence probability of link $(u_1, u_2)$ is also decided by the sampling rate of user $u_1$, which also needs to be increased too. Combining the diversity term and the structure preservation term, we could define the regularized diversity of information after sampling to be
\begin{align}
D_{Reg}(\bar{G}_{old}) &= D(\bar{G}_{old}) + Reg(\bar{G}_{old}) = \boldsymbol{\delta}' \cdot \mathbf{N} \cdot \boldsymbol{\delta}
\end{align}
where $\mathbf{N} = \mathbf{ \frac{1}{2\left| \mathcal{V}_{old} \right|} \cdot I_{\left| \mathcal{V}_{old} \right|}} + \mathbf{\frac{1}{{2\left| S_{old} \right|}} \cdot \mathbf{A_{old}}}+  \mathbf{M}$.

The optimal value of $\boldsymbol{\delta}$ should be able to maximize the relevance of new users' sub-network and old users' as well as the regularized diversity of old users' information in the target network \begingroup\makeatletter\def\f@size{8}\check@mathfonts
\begin{align}
\boldsymbol{\delta}^* &= \underset{\boldsymbol{\delta}}{\arg \max}\  R(\bar{G}_{old}, G_{new}) + \theta \cdot D_{Reg}(\bar{G}_{old}) \\
	&= \underset{\boldsymbol{\delta}}{\arg \max}\  \boldsymbol{\delta}^\top \boldsymbol{s} + \theta \cdot \boldsymbol{\delta}^\top \cdot \mathbf{N} \cdot \boldsymbol{\delta}\\
	&s.t. \sum_{i=1}^{\left| \mathcal{V}_{old} \right|}\delta_i = 1\ and\ \delta_i \ge 0,
\end{align}\endgroup
where parameter $\theta$ is used to weight the importance of term regularized information diversity. The learned sampling rate can be applied to randomly sampled the old users' historical information, so as to utilize their information for model building in predicting social links for the new users. 

\subsection{Link Prediction across Multiple Aligned Social Networks}\label{sec:chap7_sec4_pu}

Besides the link prediction problems in one single target network, some research works have been done on simultaneous link prediction in multiple aligned online social networks concurrently. In the supervised link prediction model introduced before, among all the non-existing social links, a subset of the links can be identified and labeled as the negative instances. However, in the real world, labeling the links which will never be formed can be extremely hard and almost impossible, since new links are keeping being formed. In this section, we will introduce the cross-network concurrent link prediction problem with PU learning setting.

Let $G^{(i)}, i \in \{1, 2, \cdots, n\}$ be a \textit{heterogeneous online social network} in the multiple \textit{aligned networks}. The user set and existing social link set of $G^{(i)}$ can be represented as $U^{(i)}$ and $E^{(i)}_{u,u}$ respectively. In network $G^{(i)}$, all the existing links are the formed links and, as a result, the formed links of $G^{(i)}$ can be represented as the positive set $\mathcal{P}^{(i)}$, where $\mathcal{P}^{(i)} = E^{(i)}_{u,u}$. Furthermore, a large set of unconnected user pairs are referred to as the unconnected links, $\mathcal{U}^{(i)}$, and can be extracted from network $G^{(i)}$: $\mathcal{U}^{(i)} = U^{(i)} \times U^{(i)} \setminus \mathcal{P}^{(i)}$. However, no information about links that will never be formed can be obtained from the network. With $\mathcal{P}^{(i)}$ and $\mathcal{U}^{(i)}$, we formulate the \textit{link formation prediction} as the {PU (Positive and Unlabeled) link prediction} problem.

Formally, let the notations $\{\mathcal{P}^{(1)}, \cdots, \mathcal{P}^{(n)}\}$, $\{\mathcal{U}^{(1)}, \cdots, \mathcal{U}^{(n)}\}$ and $\{\mathcal{L}^{(1)}, \cdots, \mathcal{L}^{(n)}\}$ be the sets of formed links, unconnected links, and links to be predicted of networks $G^{(1)}, G^{(2)},\\ \cdots, G^{(n)}$ respectively. With the formed and unconnected links of $G^{(1)}, G^{(2)}, \cdots, G^{(n)}$, the \textit{multi-network link prediction} problem can be formulated as a \textit{multi-PU link prediction} problem.

In this part, we will introduce the {\ourmli} model proposed in \cite{kdd14} to solve the \textit{multi-network link prediction} problem. The {\ourmli} model includes 3 parts: (1) social meta path based feature extraction and selection; (2) PU link prediction; (3) multi-network link prediction framework, where the feature extraction is done based on the \textit{inter-network meta paths} defined in Section~\ref{sec:meta_path}. Next, we will mainly focus on introducing the Steps (2) and (3) of the {\ourmli} model respectively.

\subsubsection{PU Link Prediction} \label{subsec:chap7_sec5_pulinkprediction}

In this subsection, we will introduce a method to solve the \textit{PU link prediction} problem in one single network. As introduced in the problem formulation at the beginning, from a given network, e.g., $G$, two disjoint sets of links: connected (i.e., formed) links $\mathcal{P}$ and unconnected links $\mathcal{U}$, can be obtained. To differentiate these links, {\ourmli} uses a new concept ``\textit{connection state}'', $z$, to show whether a link is connected (i.e., formed) or unconnected in network $G$. For a given link $l$, if $l$ is connected in the network, then $z(l) = +1$; otherwise, $z(l) = -1$. As a result, {\ourmli} can have the ``\textit{connection states}'' of links in $\mathcal{P}$ and $\mathcal{U}$ to be: $z(\mathcal{P}) = \mb{+1}$ and $z(\mathcal{U}) = \mb{-1}$.

Besides the ``\textit{connection state}'', links in the network can also have their own ``\textit{labels}'', $y$, which can represent whether a link is to be formed or will never be formed in the network. For a given link $l$, if $l$ has been formed or to be formed, then $y(l) = +1$; otherwise, $y(l) = -1$. Similarly, {\ourmli} can have the ``\textit{labels}'' of links in $\mathcal{P}$ and $\mathcal{U}$ to be: $y(\mathcal{P}) = \mb{+1}$ but $y(\mathcal{U})$ can be either $+1$ or $-1$, as $\mathcal{U}$ can contain both links to be formed and links that will never be formed.

By using $\mathcal{P}$ and $\mathcal{U}$ as the positive and negative training sets, {\ourmli} can build a \textit{link \underline{c}onnection prediction model} $\mathcal{M}_{c}$, which can be applied to predict whether a link exists in the original network, i.e., the \textit{connection state} of a link. Let $l$ be a link to be predicted, by applying $\mathcal{M}_{c}$ to classify $l$, the \textit{connection probability} of $l$ can be represented to be:
\begin{defn}
(Connection Probability): The probability that link $l$'s \textit{connection states} is predicted to be \textit{connected} (i.e., $z(l) = +1$) is formally defined as the \textit{connection probability} of link $l$: $p(z(l) = +1 | \mb{x}(l))$, where $\mb{x}(l)$ denotes the feature vector extracted for link $l$ based on meta path.
\end{defn}

Meanwhile, if we can obtain a set of links that ``will never be formed'', i.e., ``-1'' links, from the network, which together with $\mathcal{P}$ (``+1'' links) can be used to build a \textit{link \underline{f}ormation prediction model}, $\mathcal{M}_{f}$, which can be used to get the \textit{formation probability} of $l$ to be:

\begin{defn}
(Formation Probability): The probability that link $l$'s \textit{label} is predicted to be \textit{formed or will be formed} (i.e., $y(l) = +1$) is formally defined as the \textit{formation probability} of link $l$: $p(y(l) = +1 | \mb{x}(l))$.
\end{defn}


\begin{figure}[t]
\centering
    \begin{minipage}[l]{0.65\columnwidth}
      \centering
      \includegraphics[width=\textwidth]{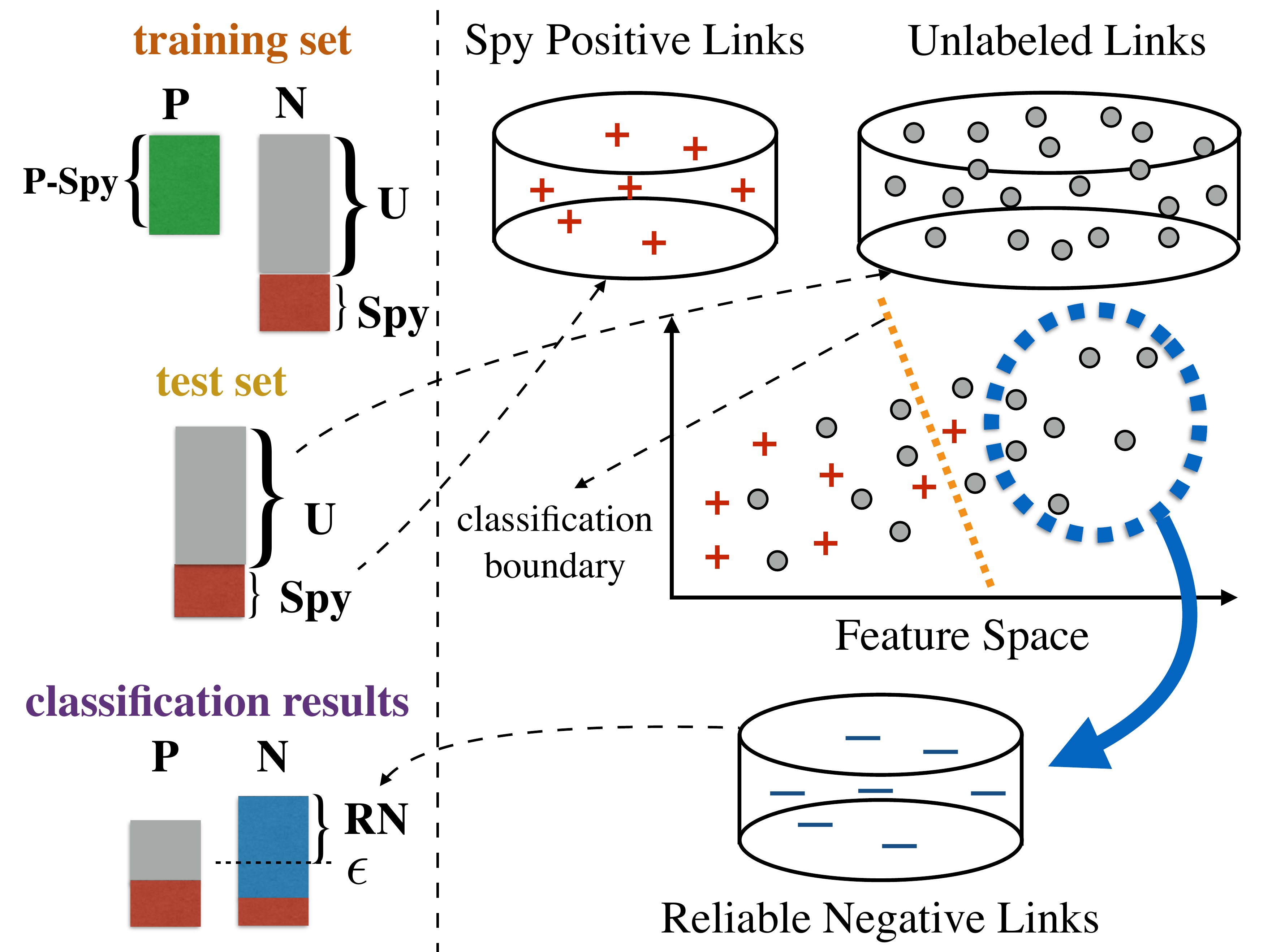}
    \end{minipage}
  \caption{PU Link Prediction.}\label{fig:chap7_sec5_pulinkprediction}
\end{figure}

However, from the network, we have no information about ``links that will never be formed'' (i.e., ``-1'' links). As a result, the \textit{formation probabilities} of potential links that we aim to obtain can be very challenging to calculate. Meanwhile, the correlation between link $l$'s \textit{connection probability} and \textit{formation probability} has been proved in existing works \cite{EN08} to be:
\begin{equation}
p(y(l) = +1 | \mb{x}(l)) \propto p(z(l) = +1 | \mb{x}(l)).
\end{equation}

In other words, for links whose \textit{connection probabilities} are low, their \textit{formation probabilities} will be relatively low as well. This rule can be utilized to extract links which can be more likely to be the reliable ``-1'' links from the network.  {\ourmli} proposes to apply the the \textit{link connection prediction model} $\mathcal{M}_{c}$ built with $\mathcal{P}$ and $\mathcal{U}$ to classify links in $\mathcal{U}$ to extract the \textit{reliable negative link set}.

\begin{defn}
(Reliable Negative Link Set): The \textit{reliable negative links} in the \textit{unconnected link} set $\mathcal{U}$ are those whose \textit{connection probabilities} predicted by the \textit{link connection prediction model}, $\mathcal{M}_c$, are lower than threshold $\epsilon \in [0,1]$:
\begin{equation}
\mathcal{RN} = \{l | l \in \mathcal{U}, p(z(l) = +1 | \mb{x}(l)) < \epsilon \}.
\end{equation}
\end{defn}

Some Heuristic methods have been proposed to set the optimal threshold $\epsilon$, e.g., the \textit{spy technique} proposed in \cite{LDLLY03}. As shown in Figure~\ref{fig:chap7_sec5_pulinkprediction}, {\ourmli} proposes randomly select a subset of links in $\mathcal{P}$ as the spy, $\mathcal{SP}$, whose proportion is controlled by $s\%$. $s\% = 15\%$ is used as the default sample rate in \cite{kdd14}. Sets $(\mathcal{P} - \mathcal{SP})$ and $(\mathcal{U} \cup \mathcal{SP})$ are used as positive and negative training sets to the \textit{\underline{s}py prediction} model, $\mathcal{M}_s$. By applying $\mathcal{M}_s$ to classify links in $(\mathcal{U} \cup \mathcal{SP})$, their \textit{connection probabilities} can be represented to be:
\begin{equation}
p(z(l) = +1 | \mb{x}(l)), l \in (\mathcal{U} \cup \mathcal{SP}),
\end{equation}
and parameter $\epsilon$ is set as the minimal \textit{connection probability} of spy links in $\mathcal{SP}$:
\begin{equation}
\epsilon = \min_{l \in \mathcal{SP}} p(z(l) = +1| \mathbf{x}(l)).
\end{equation}

With the extracted \textit{reliable negative link set} $\mathcal{RN}$, {\ourmli} can solve the \textit{PU link prediction} problem with \textit{classification based link prediction methods}, where $\mathcal{P}$ and $\mathcal{RN}$ are used as the positive and negative training sets respectively. Meanwhile, when applying the built model to predict links in $\mathcal{L}^{(i)}$, the optimal labels, $\hat{\mathcal{Y}}^{(i)}$, of $\mathcal{L}^{(i)}$, should be those which can maximize the following \textit{formation probabilities}: \begingroup\makeatletter\def\f@size{8}\check@mathfonts
\begin{align}
\hspace{-10pt}\hat{\mathcal{Y}}^{(i)} &= \arg \max_{\mathcal{Y}^{(i)}} p(y(\mathcal{L}^{(i)}) = \mathcal{Y}^{(i)} | G^{(1)}, G^{(2)}, \cdots, G^{(n)})\\
\hspace{-10pt}&= \arg \max_{\mathcal{Y}^{(i)}} p(y(\mathcal{L}^{(i)}) = \mathcal{Y}^{(i)} | \left[\mathbf{\bar{x}}_{\Phi}(\mathcal{L}^{(i)})^T, \mathbf{\bar{x}}_{\Psi}(\mathcal{L}^{(i)})^T \right]^T),
\end{align}\endgroup
where $y(\mathcal{L}^{(i)}) = \mathcal{Y}^{(i)}$ represents that links in $\mathcal{L}^{(i)}$ have labels $\mathcal{Y}^{(i)}$.

\begin{figure}[t]
\centering
    \begin{minipage}[l]{0.9\columnwidth}
      \centering
      \includegraphics[width=\textwidth]{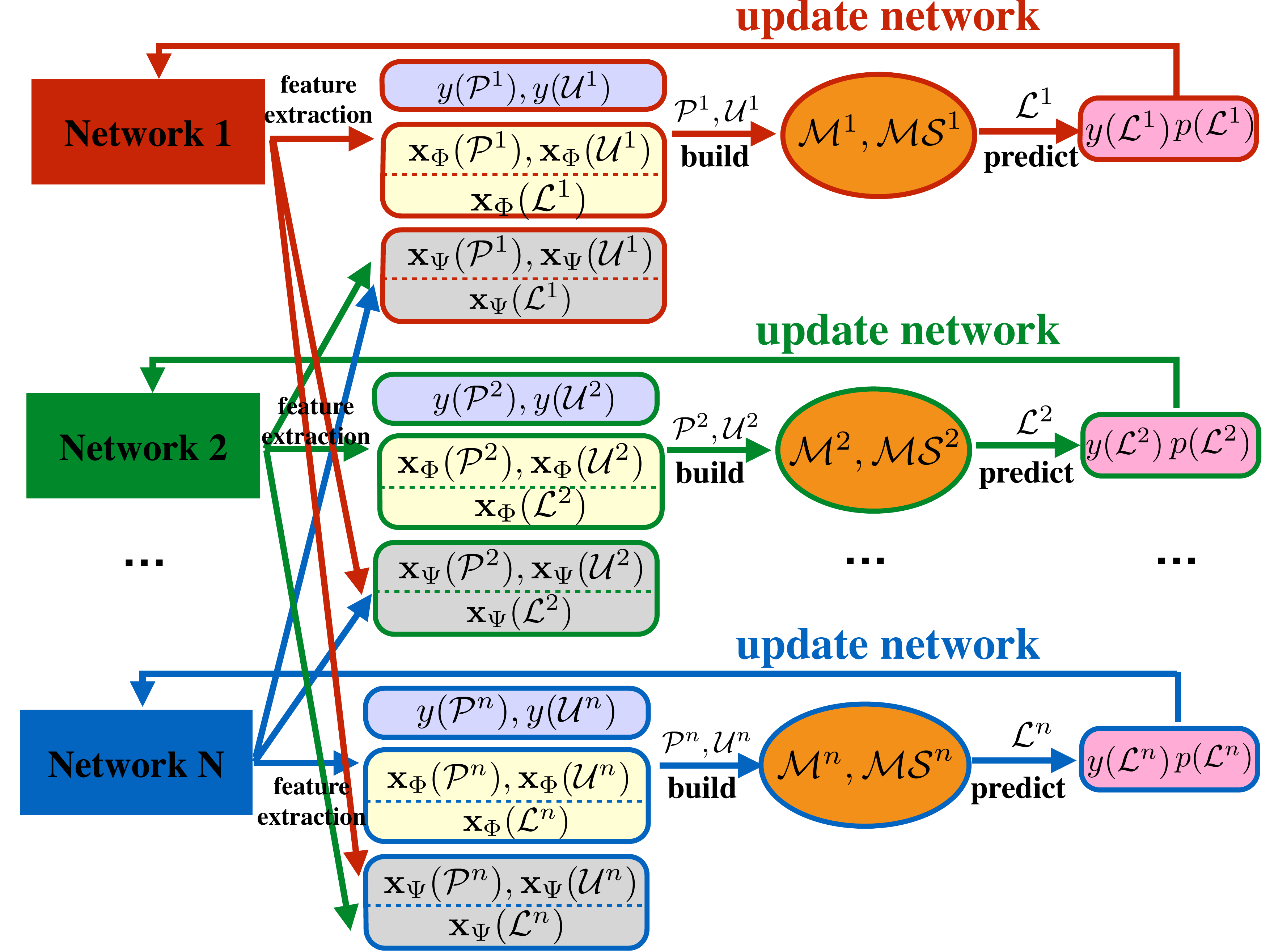}
    \end{minipage}
  \caption{Multi-PU Link Prediction Framework.}\label{fig:chap7_sec5_multipulinkprediction}
\end{figure}

\subsubsection{Multi-Network Link Prediction Framework}

Method {\ourmli} proposed in \cite{kdd14} is a general link prediction framework and can be applied to predict social links in $n$ \textit{partially aligned networks} simultaneously. When it comes to $n$ partially aligned network, the optimal labels of potential links $\{\mathcal{L}^{(1)}, \mathcal{L}^{(2)}, \cdots, \mathcal{L}^{(n)}\}$ of networks $G^{(1)}, \cdots, G^{(n)}$ will be: \begingroup\makeatletter\def\f@size{6}\check@mathfonts
\begin{align}
\hspace{-7pt}&\hat{\mathcal{Y}}^{(1)}, \hat{\mathcal{Y}}^{(2)}, \cdots, \hat{\mathcal{Y}}^{(n)} = \arg \max_{\mathcal{Y}^{(1)}, \cdots, \mathcal{Y}^{(n)}} \\
\hspace{-7pt}&\ \ \ \ \ \ \ \ \ \  p\Big(y(\mathcal{L}^{(1)}) = \mathcal{Y}^{(1)} , \cdots, y(\mathcal{L}^{(n)}) = \mathcal{Y}^n | G^{(1)}, \cdots, G^{(n)} \Big).
\end{align} \endgroup

The above target function is very complex to solve and, in \cite{kdd14}, {\ourmli} proposes to obtain the solution by updating one variable, e.g., $\mathcal{Y}^{(1)}$, and fix other variables, e.g., $\mathcal{Y}^{(2)}, \cdots, \mathcal{Y}^{(n)}$, alternatively with the following equation \cite{wsdm14}: \begingroup\makeatletter\def\f@size{6}\check@mathfonts
\begin{equation}
\begin{cases}
(\hat{\mathcal{Y}}^{(1)})^{(\tau)} \hspace{-10pt} &= \arg \max_{\mathcal{Y}^{(1)}} p\Big(y(\mathcal{L}^{(1)}) = \mathcal{Y}^{(1)} | G^{(1)}, G^{(2)}, \cdots, G^{(n)}, \\
&\ \ \ \ \ \ \ \ \ \ \ \ \ \ \ \ (\hat{\mathcal{Y}}^{2})^{(\tau - 1)}, (\hat{\mathcal{Y}}^{3})^{(\tau - 1)}, \cdots, (\hat{\mathcal{Y}}^{n})^{(\tau - 1)} \Big),\\
(\hat{\mathcal{Y}}^{(2)})^{(\tau)}  \hspace{-10pt} &= \arg \max_{\mathcal{Y}^{(2)}} p \Big(y(\mathcal{L}^{(2)}) = \mathcal{Y}^{(2)} | G^{(1)}, G^{(2)}, \cdots, G^{(n)},\\
&\ \ \ \ \ \ \ \ \ \ \ \ \ \ \ \  (\hat{\mathcal{Y}}^{(1)})^{(\tau)}, (\hat{\mathcal{Y}}^{(3)})^{(\tau - 1)}, \cdots, (\hat{\mathcal{Y}}^{(n)})^{(\tau - 1)} \Big),\\
&\cdots \cdots \\
(\hat{\mathcal{Y}}^{(n)})^{(\tau)} \hspace{-10pt} &= \arg \max_{\mathcal{Y}^{(n)}} p \Big(y(\mathcal{L}^{(n)}) = \mathcal{Y}^{(n)} | G^{(1)}, G^{(2)}, \cdots, G^{(n)}, \\
&\ \ \ \ \ \ \ \ \ \ \ \ \ \ \ \ (\hat{\mathcal{Y}}^{(1)})^{(\tau)}, (\hat{\mathcal{Y}}^{(2)})^{(\tau)}, \cdots, (\hat{\mathcal{Y}}^{(n-1)})^{(\tau)} \Big).
\end{cases} 
\end{equation}\endgroup

The structure of framework {\ourmli} is shown in Figure~\ref{fig:chap7_sec5_multipulinkprediction}. When predicting social links in network $G^{(i)}$, {\ourmli} can extract features based on the \textit{intra-network social meta path} extracted from $G^{(i)}$ and those extracted based on the \textit{inter-network social meta path} across $G^{(1)}$, $G^{(2)}$, $\cdots$, $G^{(i - 1)}$, $G^{(i + 1)}$, $\cdots$, $G^{(n)}$ for links in $\mathcal{P}^{(i)}$, $\mathcal{U}^{(i)}$ and $\mathcal{L}^{(i)}$. Feature vectors $\mathbf{x}(\mathcal{P})$, $\mathbf{x}(\mathcal{U})$ as well as the labels, $y(\mathcal{P})$, $y(\mathcal{U})$, of links in $\mathcal{P}$ and $\mathcal{U}$ are passed to the PU link prediction model $\mathcal{M}^{(i)}$ and the meta path selection model $\mathcal{MS}^{(i)}$. The formation probabilities of links in $\mathcal{L}^{(i)}$ predicted by model $\mathcal{M}^{(i)}$ will be used to update the network by replace the weights of $\mathcal{L}^{(i)}$ with the newly predicted formation probabilities. The initial weights of these potential links in $\mathcal{L}^{(i)}$ are set as $0$. After finishing these steps on $G^{(i)}$, we will move to conduct similar operations on $G^{(i+1)}$. {\ourmli} iteratively predicts links in $G^{(1)}$ to $G^{(n)}$ alternatively in a sequence until the results in all of these networks converge.

\subsection{Sparse and Low Rank Matrix Estimation based Inter-Network Link Prediction}\label{sec:chap7_sec5_matrix}

Different online social networks usually have different functions, and information in them follows totally different distributions. When predicting the links across multiple aligned online social networks, the link prediction models aforementioned didn't address the domain difference problem at all. In this section, we will introduce a new cross-network link prediction model introduced in \cite{icde17_2}, which embeds the feature vectors of links from aligned networks into a shared feature space. Via the shared feature space, knowledge from the source networks will be effectively transferred to the target network.

\subsubsection{Link Prediction Objective Function}

\noindent \textbf{Link Prediction Loss Term}

Give the target network $G^t$ involving users $\mathcal{U}^t$, the observed social connection among the users can be represented with the binary social adjacency matrix $\mb{A}^t \in \{0,1\}^{|\mathcal{U}^t| \times |\mathcal{U}^t|}$, where entry ${A}^t(i,j) = 1$ iff the corresponding social link $(u^t_i, u^t_j)$ exists between users $u^t_i$ and $u^t_j$ in $G^t$. In the studied problem here, our objective is to infer the potential unobserved social links for the target network, which can be achieved by finding a sparse and low-rank predictor matrix $\mb{S} \in \mathcal{S}$ from some convex admissible set $\mathcal{S} \subset \mathbb{R}^{|\mathcal{U}^t| \times |\mathcal{U}^t|}$. Meanwhile, the inconsistency between the inferred matrix $\mb{S}$ and the observed social adjacency matrix $\mb{A}^t$ can be represented as the loss function $l(\mb{S}, \mb{A}^t)$. The optimal social link predictor for the target network can be achieved by minimizing the loss term, i.e.,
\begin{equation}
\arg\min_{\mb{S} \in \mathcal{S}} l(\mb{S}, \mb{A}^t).
\end{equation}

The loss function $l(\mb{S}, \mb{A}^t)$ can be defined in many different ways, and, in \cite{icde17_2}, the \textit{loss function} is approximated by counting the loss introduced by the existing social links in $\mathcal{E}_u^t$, i.e., \begingroup\makeatletter\def\f@size{8}\check@mathfonts
\begin{equation}
l(\mb{S}, \mb{A}^t) = \frac{1}{|\mathcal{E}_u^t|} \sum_{(u^t_i, u^t_j) \in \mathcal{E}_u^t} \mathbbm{1}\Big(\big(A^t(i,j) - \frac{1}{2}\big)\cdot S(i,j) \le 0\Big).
\end{equation}\endgroup

\noindent \textbf{Intra-Network Attribute based Intimacy Term}

Besides the connection information, there also exists a large amount of attribute information available in the target network, e.g., \textit{location checkin records}, \textit{online social activity temporal patterns}, and \textit{text usage patterns}, etc. Based on the attribute information, a set of features can be extracted  for all the potential user pairs to denote their closeness, which are called the \textit{intimacy features} formally. For instance, given user pair $(u^t_i, u^t_j)$ in the target network, its \textit{intimacy features} can be represented as vector $\mb{x}^t_{i,j} \in \mathbb{R}^{d^t}$ ($d^t$ denotes the extracted intimacy feature number).

More generally, the feature vectors extracted for user pairs can be represented as a 3-way tensor $\mb{X}^t \in \mathbb{R}^{d^t \times |\mathcal{U}^t| \times |\mathcal{U}^t|}$, where slice $\mb{X}^t(k,:,:)$ denote all the $k_{th}$ intimacy features among all the user pairs. In online social networks, \textit{homophily} principle \cite{MSC01} has been observed to widely structure the users' online social connections, and users who are close to each other are more likely to be friends. Based on such an intuition, the potential social connection matrix $\mb{S}$ can be inferred by maximizing the overall intimacy scores of the inferred new social connections, i.e.,
\begin{align}
\arg \max_{\mb{S} \in \mathcal{S}} int(\mb{S}, \mb{X}^t).
\end{align}

In \cite{icde17_2}, the introduced model proposes to define the intimacy score term $int(\mb{S}, \mb{X}^t)$ by enumerating and summing the \textit{intimacy scores} of the inferred social connections, i.e., 
\begin{equation}
int(\mb{S}, \mb{X}^t) = \sum_{k=1}^{d^t} \left\| \mb{S} \circ \mb{X}^t(k,:,:)\right\|_1,
\end{equation}
where operator $\circ$ denotes the Hadamard product (i.e., entrywise product) of matrices.

\noindent \textbf{Intra-Network Attribute based Intimacy Term}

Furthermore, with the information from the external source networks, more knowledge can be obtained about the users and their social patterns. By projecting the link instances to a shared feature space as introduced in \cite{icde17_2}, the the adapted features from the target network and external sources can be represented as tensors $\mb{\hat{X}}^t, \mb{\hat{X}}^1, \cdots, \mb{\hat{X}}^K$. Formally, the intimacy scores of the potential social links based on these adapted features from the external source networks can be represented as
\begin{align}
int(\mb{S}, \mb{\hat{X}}^1, \cdots, \mb{\hat{X}}^K) =  \sum_{k=1}^K \alpha^i \cdot int(\mb{S}, \mb{\hat{X}}^k),
\end{align}
where term $int(\mb{S}, \mb{\hat{X}}^k) = \left\| \mb{S} \circ \mb{\hat{X}}^k \right\|_1$, and users in $\mb{\hat{X}}^k$ are organized in the same order as $\mb{{X}}^t$. Parameters $\alpha^i$ denotes the importance of the information transferred from the source network $G^i$. 

\noindent \textbf{Joint Objective Function}

By adding the intimacy terms about the source networks into the objective function, the equation can be rewriten as follows: \begingroup\makeatletter\def\f@size{8}\check@mathfonts
\begin{align}
\arg \min_{\mb{S} \in \mathcal{S}}\ \  & l(\mb{S}, \mb{A}^t) - \alpha^t \cdot int(\mb{S}, \mb{\hat{X}}^t) - \sum_{k=1}^K \alpha^i \cdot int(\mb{S}, \mb{\hat{X}}^k))\\
&+ \gamma \cdot \left\| \mb{S} \right\|_1 + \tau \cdot \left\| \mb{S} \right\|_*,
\end{align}\endgroup
where $\left\| \mb{S} \right\|_1$ and $\left\| \mb{S} \right\|_*$ denote the $L_1$-norm and trace-norm of matrix $\mb{S}$ respectively.

\subsubsection{Proximal Operator based CCCP Algorithm}

By studying the objective function, we observe that the intimacy terms are convex while the empirical loss term $l(\mb{S}, \mb{A}^t)$ is non-convex. In \cite{icde17_2}, the introduced model proposes to approximate it with other classical loss functions (e.g., the hinge loss and the Frobenius norm) instead, and the convex squared Frobenius norm loss function is used in \cite{icde17_2} (i.e., $l(\mb{S}, \mb{A}^t) = \left\| \mb{S} - \mb{A}^t \right\|_F^2$). Therefore, the above objective function can be represented as a convex loss term minus another convex term together with two convex non-differentiable regularizers, which actually renders the objective function non-trivial. According to the existing works \cite{YR03, SL09}, this kind of objective function can be addressed with the concave-convex procedure (CCCP). CCCP is a majorization-minimization algorithm that solves the difference of convex functions problems as a sequence of convex problems. Meanwhile, the regularization terms can be effectively handled with the proximal operators in each iteration of the CCCP process.

\noindent \textbf{CCCP Algorithm}

Formally, the objective function can be decomposed into two convex functions:
\begin{align}
&u(\mb{S}) = l(\mb{S}, \mb{A}^t) + \gamma \cdot \left\| \mb{S} \right\|_1 + \tau \cdot \left\| \mb{S} \right\|_*,\\
&v(\mb{S}) = \alpha^t \cdot int(\mb{S}, \mb{\hat{X}}^t) + \sum_{k=1}^K \alpha^i \cdot int(\mb{S}, \mb{\hat{X}}^k).
\end{align}

With $u(\mb{S})$ and $v(\mb{S})$, the objective function can be rewritten as
\begin{equation}
\arg \min_{\mb{S} \in \mathcal{S}} u(\mb{S}) - v(\mb{S}).
\end{equation}

The CCCP algorithm can address the objective function with an iterative procedure that solves the following sequence of convex problems:
\begin{align}
\mb{S}^{(h + 1)} &= \arg \min_{\mb{S} \in \mathcal{S}} u(\mb{S}) - \mb{S}^\top \nabla v(\mb{S}^{(h)}).
\end{align}

It is easy to show that function $v(\mb{S})$ differentiable, and the derivative of function $v(\mb{S})$ is actually a constant term
\begin{equation}
\nabla v(\mb{S})= \sum_{k =t}^K \alpha^i \sum_{i = 1}^c \mb{\hat{X}}^k(i,:,:).
\end{equation}

By relying on the Zangwill's global convergence theory \cite{Z69} of iterative algorithms, it is theoretically proven in \cite{SL09} that as such a procedure continues, the generated sequence of the variables $\{\mb{S}^{(h)}\}_{h=0}^{\infty}$ will converge to some stationary points $\mb{S}_*$ in the inference space $\mathcal{S}$.

\noindent \textbf{Proximal Operators}

Meanwhile, in each iteration of the CCCP updating process, objective function is not easy to address due to the non-differentiable regularizers. Some works have been done to deal with the objective function involving non-smooth functions. The Forward-Backward splitting method proposed in \cite{CW05} can handle such a kind of optimization function with one single non-smooth regularizer based on the introduced proximal operators. More specifically, as introduced in \cite{CW05}, the proximal operators for the trace norm and $L_1$ norm can be represented as follows
\begin{align}
&\mbox{prox}_{\tau \left\| \cdot \right\|_*}(\mb{S}) = \mb{U} \mbox{diag}((\sigma_i - \tau)_+)_i \mb{V}^\top,\\
&\mbox{prox}_{\gamma \left\| \cdot \right\|_1}(\mb{S}) = \mbox{sgn}(\mb{S}) \circ (|\mb{S}| - \gamma)_+,
\end{align}
where $\mb{S} =\mb{U} \mbox{diag}(\sigma_i)_i \mb{V}^\top$ denotes the singular decomposition of matrix $\mb{S}$, and $\mbox{diag}(\sigma_i)_i$ represents the diagonal matrix with values $\sigma_i$ on the diagonal.

Recently, some works have proposed the generalized Forward-Backward algorithm to tackle the case with $q (q \ge 2)$ non-differentiable convex regularizers \cite{RFP13}. These methods alternate the gradient step and the proximal steps to update the variables. For instance, given the above objective function in iteration $h$ of the CCCP, the alternative updating equations in step $k$ to address the objective function can be represented as follows:
\begin{equation}\hspace{-1pt}
\begin{cases}
\mb{S}^{(k)} &= \mb{S}^{(k-1)} - \theta \cdot \nabla_{\mb{S}}\left(l(\mb{S}, \mb{A}) - \mb{S}^\top \nabla v(\mb{S}^{(h)})\right),\\
\mb{S}^{(k)} &= \mbox{prox}_{\theta \tau \left\| \cdot \right\|_*}(\mb{S}^{(k)}),\\
\mb{S}^{(k)} &= \mbox{prox}_{\theta \gamma \left\| \cdot \right\|_1}(\mb{S}^{(k)}),
\end{cases}
\end{equation}
where the parameter $\theta$ denotes the learning rate and it is assigned with a very small value to ensure the converge of the above functions \cite{SRV12}. The pseudo-code of the Proximal Operators based CCCP algorithm is available in Algorithm~\ref{alg:cccp_proximal}.

\setlength{\textfloatsep}{0pt}
\begin{algorithm}[t]
\small
\caption{Proximal Operator Based CCCP Algorithm}
\label{alg:cccp_proximal}
\begin{algorithmic}[1]
	\REQUIRE social adjacency matrix $\mb{A}$ \\
\qquad  projected feature tensors $\mb{\hat{X}}^t$, $\mb{\hat{X}}^1$, $\cdots$, $\mb{\hat{X}}^K$\\
\ENSURE link predictor matrix $\mb{S}$ 

\STATE 	{Initialize matrix $\mb{S}_{cccp} = \mb{A}$}
\STATE	{Initialize CCCP convergence CCCP-tag = False}
\WHILE	{CCCP-tag == False}
\STATE	{Initialize Proximal convergence Proximal-tag = False}
\STATE	{Solve optimization function $\min_{\mb{S} \in \mathcal{S}} u(\mb{S}) - \mb{S}^\top \nabla v(\mb{S}_{cccp})
$}
\STATE	{Initialize $\mb{S}_{po} = \mb{S}_{cccp}$}
\WHILE	{Proximal-tag == False}
\STATE	{$\mb{S}_{po} = \mb{S}_{po} - \theta \nabla_{\mb{S}}\left(l(\mb{S}_{po}, \mb{A}) - \mb{S}_{po}^\top \nabla v(\mb{S}_{cccp}) \right)$}
\STATE	{$\mb{S}_{po} = \mbox{prox}_{\theta \tau \left\| \cdot \right\|_*}(\mb{S}_{po})$}
\STATE	{$\mb{S}_{po} = \mbox{prox}_{\theta \gamma \left\| \cdot \right\|_1}(\mb{S}_{po})$}
\IF		{$\mb{S}_{po}$ converges}
\STATE	{Proximal-tag = True}
\STATE	{$\mb{S}_{cccp} = \mb{S}_{po}$}
\ENDIF
\ENDWHILE
\IF		{$\mb{S}_{cccp}$ converges}
\STATE	{CCCP-tag = True}
\ENDIF
\ENDWHILE
\STATE	{\mbox{Return} $\mb{S}_{cccp}$}
\end{algorithmic}
\end{algorithm}




\section{Community Detection}\label{sec:clustering}

In the real-world online social networks, users tend to form different social groups \cite{Arenas2004}. Users belonging to the same groups usually have more frequent interactions with each other, while those in different groups will have less interactions on the other hand \cite{ZLZ11}. Formally, such social groups form by users in online social networks are called the online social communities \cite{bigdata15}. Online social communities will partition the network into a number of connected components, where the intra-community social connections are usually far more dense compared with the inter-community social connections \cite{bigdata15}. Meanwhile, from the mathematical representation perspective, due to these online social communities, the social network adjacency matrix tend to be not only sparse but also low-rank \cite{icde17}. 

Identifying the social communities formed by users in online social networks is formally defined as the \textit{community detection} problem \cite{bigdata15, sdm15, bigdata14}. Community detection is a very important problem for online social network studies, as it can be crucial prerequisite for numerous concrete social network services: (1) better organization of users' friends in online social networks (e.g., Facebook and Twitter), which can be achieved by applying community detection techniques to partition users' friends into different categories, e.g., schoolmates, family, celebrities, etc. \cite{EAMK14}; (2) better recommender systems for users with common shopping preference in e-commerce social sites (e.g., Amazon and Epinions), which can be addressed by grouping users with similar purchase records into the same clusters prior to recommender system building \cite{R16}; and (3) better identification of influential users \cite{TBB10} for advertising campaigns in online social networks, which can be attained by selecting the most influential users in each community as the seed users in the viral marketing \cite{RD02}.

In this section, we will focus on introducing the \textit{social community detection} problem in online social networks. Given a heterogeneous network $G$ with node set $\mathcal{V}$, the involved user nodes in network $G$ can be represented as set $\mathcal{U} \subset \mathcal{V}$. Based on both the social structures among users as well as the diverse attribute information from the network $G$, the \textit{social community detection} problem aims at partitioning the user set $\mathcal{U}$ into several subsets $\mathcal{C} = \{\mathcal{U}_1, \mathcal{U}_2, \cdots, \mathcal{U}_k\}$, where each subset $\mathcal{U}_i, i \in \{1, 2, \cdots, k\}$ is called a social community. Term $k$ formally denotes the total number of partitioned communities, which is usually provided as a hyper-parameter in the problem.

Depending on whether the users are allowed to be partitioned into multiple communities simultaneously or not, the \textit{social community detection} problem can actually be categorized into two different types:
\begin{itemize}

\item \textit{Hard Social Community Detection}: In the \textit{hard social community detection} problem, each user will be partitioned into one single community, and all the social communities are disjoint without any overlap. In other words, given the communities $\mathcal{C} = \{\mathcal{U}_1, \mathcal{U}_2, \cdots, \mathcal{U}_k\}$ detected from network $G$, we have $\mathcal{U} = \bigcup_i \mathcal{U}_i$ and $\mathcal{U}_i \cap \mathcal{U}_j = \emptyset, \forall i, j \in \{1, 2, \cdots, k\} \land i \neq j$.

\item \textit{Soft Social Community Detection}: In the \textit{soft social community detection} problem, users can belong to multiple social communities simultaneously. For instance, if we apply the \textit{Mixture-of-Gaussian Soft Clustering} algorithm as the base community detection model \cite{ZY12, YML13}, each user can belong to multiple communities with certain probabilities. In the \textit{soft social community detection} result, the communities are no longer disjoint and will share some common users with other communities. 

\end{itemize}


Meanwhile, depending on the network connection structures, the \textit{community detection} problem can be categorized as \textit{directed network community detection} \cite{MV13} and \textit{undirected network community detection} \cite{ZLZ11}. Based on the heterogeneity of the network information, the \textit{community detection} problem can be divided into the \textit{homogeneous network community detection} \cite{WC03} and \textit{heterogeneous network community detection} \cite{ijcnn16, SYH09, cikm17, icde17}. Furthermore, according to the number of networks involved, the \textit{community detection} problem involves \textit{single network community detection} \cite{LLM10} and \textit{multiple network community detection} \cite{bigdata15, sdm15, bigdata14, cikm17, icde17}. In this section, we will take the \textit{hard community detection problem} as an example to introduce the existing models proposed for conventional (one single) \textit{homogeneous social network}, and especially the recent broad learning based (multiple aligned) \textit{heterogeneous social networks} \cite{cikm13, icdm13, wsdm14, kdd14} respectively.

This section is organized as follows. At the beginning, in Section~\ref{sec:chap8_sec2_single}, we will introduce the community detection problem and the existing methods proposed for traditional one single homogeneous networks. After that, we will talk about the latest research works on social community detection across multiple aligned heterogeneous networks. The cold start community detection \cite{sdm15} is introduced in Section~\ref{sec:chap8_sec3_cold}, in which we will talk about a new information transfer algorithm to propagate information from other developed source networks to the emerging target network. In Section~\ref{sec:chap8_sec4_mutual}, we will be focused on the concurrent mutual community detection \cite{bigdata15} across multiple aligned heterogeneous networks simultaneously, where information from other aligned networks will be applied to refine their community detection results mutually. Finally, in Section~\ref{sec:chap8_sec5_large}, we talk about the synergistic community detection across multiple large-scale networks based on the distributed computing platform \cite{bigdata14}.


\subsection{Traditional Homogeneous Network Community Detection}\label{sec:chap8_sec2_single}

Social community detection problem has been studied for a long time, and many community detection models have been proposed based on different types of techniques. In this section, we will talk about the social community detection problem for one single homogeneous network $G$, whose objective is to partition the user set $\mathcal{U}$ in network $G$ into $k$ disjoint subsets $\mathcal{C} = \{\mathcal{U}_1, \mathcal{U}_2, \cdots, \mathcal{U}_k\}$, where $\mathcal{U} = \bigcup_i \mathcal{U}_i$ and $\mathcal{U}_i \cap \mathcal{U}_j = \emptyset, \forall i, j \in \{1, 2, \cdots, k\}$. Several different community detection methods will be introduced, which include \textit{node proximity based community detection}, \textit{modularity maximization based community detection}, and \textit{spectral clustering based community detection}.


\subsubsection{Node Proximity based Community Detection}\label{sec:chap8_sec2_proximity}

The \textit{node proximity based community detection} method assumes that ``close nodes tend to be in the same communities, while the nodes far away from each other will belong to different communities''. Therefore, the \textit{node proximity based community detection} model partition the nodes into different clusters based on the node proximity measures \cite{LK07}. Various node proximity measures can be used here, including the node \textit{structural equivalence} to be introduced as follows, as well as various node closeness measures as introduced in Section~\ref{subsec:closeness_measures}.

In a homogeneous network $G$, the proximity of nodes, like $u$ and $v$, can be calculated based on their positions and connections in the network structure. 

\begin{defn}
(Structural Equivalence): Given a network $G = (\mathcal{V}, \mathcal{E})$, two nodes $u, v \in \mathcal{V}$ are said to be \textit{structural equivalent} iff \begin{enumerate}
\item Nodes $u$ and $v$ are not connected and $u$ and $v$ share the same set of neighbors (i.e., $(u, v) \notin \mathcal{E} \land \Gamma(u) = \Gamma(v)$), \\ 
\item Or $u$ and $v$ are connected and excluding themselves, $u$ and $v$ share the same set of neighbors (i.e., $(u, v) \in \mathcal{E} \land \Gamma(u) \setminus \{v\} = \Gamma(v) \setminus \{u\}$).
\end{enumerate}
\end{defn}

For the nodes which are \textit{structural equivalent}, they are \textit{substitutable} and switching their positions will not change the overall network structure. The \textit{structural equivalence} concept can be applied to partition the nodes into different communities. For the nodes which are \textit{structural equivalent}, they can be grouped into the same communities, while for the nodes which are not equivalent in their positions, they will be partitioned into different groups. However, the \textit{structural equivalence} can be too restricted for practical application in detecting the communities in real-world social networks. Computing the \textit{structural equivalence} relationships among all the node pairs in the network can lead to very high time cost. What's more, the \textit{structural equivalence} relationship will partition the social network structure into lots of small-sized fragments, since the users will have different social patterns in making friends online and few user will have identical neighbors actually.

To avoid the weakness mentioned above, some other measures are proposed to measure the proximity among nodes in the networks. For instance, as introduced in Section~\ref{subsec:closeness_measures}, the node closeness measures based on the social connections can all be applied here to compute the node proximity, e.g., ``common neighbor'', ``Jaccard's coefficient''. Here, if we use ``common neighbor'' as the proximity measure, by applying the ``common neighbor'' measure to the network $G$, the network $G$ can be transformed into a set of instances $\mathcal{V}$ with mutual closeness scores $\{c(u, v)\}_{u, v \in \mathcal{V}}$. Some existing similarity/distance based clustering algorithms, like k-Medoids, can be applied to partition the users into different communities.


\subsubsection{Modularity Maximization based Community Detection}

Besides the pairwise proximity of nodes in the network, the connection strength of a community is also very important in the community detection process. Different measures have been proposed to compute the strength of a community, like the \textit{modularity} measure \cite{N06} to be introduced in this part. 

The \textit{modularity} measure takes account of the node degree distribution. For instance, given the network $G$, the expected number of links existing between nodes $u$ and $v$ with degrees $D(u)$ and $D(v)$ can be represented as $\frac{D(u) \cdot D(v)} {2 |\mathcal{E}|}$. Meanwhile, in the network, the real number of links existing between $u$ and $v$ can be denoted as entry $A[u, v]$ in the social adjacency matrix $\mb{A}$. For the user pair $(u, v)$ with a low expected connection confidence score, if they are connected in the real world, it indicates that $u$ and $v$ have a relatively strong relationship with each other. Meanwhile, if the community detection algorithm can partition such user pairs into the same group, it will be able to identify very strong social communities from the network.

Based on such an intuition, the strength of a community, e.g., $\mathcal{U}_i \in \mathcal{C}$ can be defined as  \begingroup\makeatletter\def\f@size{8}\check@mathfonts
\begin{equation}
\sum_{u, v \in \mathcal{U}_i} \left(A[u, v] - \frac{D(u) \cdot D(v)} {2 |\mathcal{E}|} \right).
\end{equation}\endgroup

Furthermore, the strength of the overall community detection result $\mathcal{C} = \{\mathcal{U}_1, \mathcal{U}_2, \cdots, \mathcal{U}_k\}$ can be defined as the \textit{modularity} of the communities as follows.

\begin{defn}
(Modularity): Given the community detection result $\mathcal{C} = \{\mathcal{U}_1, \mathcal{U}_2, \cdots, \mathcal{U}_k\}$, the modularity of the community structure is defined as  \begingroup\makeatletter\def\f@size{8}\check@mathfonts
\begin{equation}
Q(\mathcal{C}) = \frac{1}{2|\mathcal{E}|} \sum_{\mathcal{U}_i \in \mathcal{C}} \sum_{u, v \in \mathcal{U}_i} \left(A[u, v] - \frac{D(u) \cdot D(v)} {2 |\mathcal{E}|} \right).
\end{equation}\endgroup
\end{defn}
The \textit{modularity} concept effectively measures the strength of the detected community structure. Generally, for a community structure with a larger \textit{modularity} score, it indicates a good community detection result. 

Another way to explain the \textit{modularity} is from the number of links within and across communities. By rewriting the above \textit{modularity} equation, we can have \begingroup\makeatletter\def\f@size{6}\check@mathfonts
\begin{align}
&Q(\mathcal{C}) \\
&= \frac{1}{2|\mathcal{E}|} \sum_{\mathcal{U}_i \in \mathcal{C}} \sum_{u, v \in \mathcal{U}_i} \left(A[u, v] - \frac{D(u) \cdot D(v)} {2 |\mathcal{E}|} \right)\\
&= \frac{1}{2|\mathcal{E}|} \left( \sum_{\mathcal{U}_i \in \mathcal{C}} \sum_{u, v \in \mathcal{U}_i} A[u, v] - \sum_{\mathcal{U}_i \in \mathcal{C}} \sum_{u, v \in \mathcal{U}_i} \frac{D(u) \cdot D(v)} {2 |\mathcal{E}|} \right)\\
&= \frac{1}{2|\mathcal{E}|} \left( \sum_{\mathcal{U}_i \in \mathcal{C}} \sum_{u, v \in \mathcal{U}_i} A[u, v] - \frac{1}{2 |\mathcal{E}|} \sum_{\mathcal{U}_i \in \mathcal{C}} \sum_{u \in \mathcal{U}_i} D(u) \sum_{u \in \mathcal{U}_i} D(v) \right)\\
&= \frac{1}{2|\mathcal{E}|} \left( \sum_{\mathcal{U}_i \in \mathcal{C}} \sum_{u, v \in \mathcal{U}_i} A[u, v] - \frac{1}{2 |\mathcal{E}|} \sum_{\mathcal{U}_i \in \mathcal{C}} (\sum_{u \in \mathcal{U}_i} D(u) )^2\right).
\end{align}\endgroup

In the above equation, term $\sum_{u, v \in \mathcal{U}_i} A[u, v]$ denotes the number of links connecting users within the community $\mathcal{U}_i$ (which will be $2$ times the intra-community links for undirected networks, as each link will be counted twice). Term $\sum_{u \in \mathcal{U}_i} D(u)$ denotes the sum of node degrees in community $\mathcal{U}_i$, which equals to the number of intra-community and inter-community links connected to nodes in community $\mathcal{U}_i$. If there exist lots of inter-community links, then the \textit{modularity} measure will have a smaller value. On the other hand, if the inter-community links are very rare, the \textit{modularity} measure will have a larger value. Therefore, maximizing the community \textit{modularity} measure is equivalent to minimizing the inter-community link numbers. 

The \textit{modularity} measure can also be represented with linear algebra equations. Let matrix $\mb{A}$ denote the adjacency matrix of the network, and vector $\mb{d} \in \mathbb{R}^{|\mathcal{V}| \times 1}$ denote the degrees of nodes in the network. The \textit{modularity matrix} can be defined as
\begin{equation}
\mb{B} = \mb{A} - \frac{\mb{d}\mb{d}^\top}{2 |\mathcal{E}|}.
\end{equation}

Let matrix $\mb{H} \in \{0, 1\}^{|\mathcal{V}| \times k}$ denotes the communities that users in $\mathcal{V}$ belong to. In real application, such a binary constraint can be relaxed to allow real value solutions for matrix $\mb{H}$. The optimal community detection result can be obtained by solving the following objective function
\begin{align}
&\max \frac{1}{2 |\mathcal{E}|} \mbox{Tr}(\mb{H}^\top \mb{B} \mb{H})\\
&s.t. \ \ \mb{H}^\top\mb{H} = \mb{I},
\end{align}
where constraint $\mb{H}^\top\mb{H} = \mb{I}$ ensures there are not overlap in the community detection result.

The above objective function looks very similar to the objective function of \textit{spectral clustering} to be introduced in the next section. After obtaining the optimal $\mb{H}$, the communities can be obtained by applying the K-Means algorithm to $\mb{H}$ to determine the cluster labels of each node in the network.


\subsubsection{Spectral Clustering based Community Detection}\label{subsec:chap8_sec2_spectral}

In the community detection process, besides maximizing the proximity of nodes belonging to the same communities (as introduced in Section~\ref{sec:chap8_sec2_proximity}), minimizing the connections among nodes in different clusters is also an important factor. Different from the previous proximity based community detection algorithms, another way to address the community detection problem is from the cost perspective. Partition the nodes into different clusters will cut the links among the clusters. To ensure the nodes partitioned into different clusters have less connections with each other, the number of links to be cut in the community detection process should be as small as possible \cite{SM00, L07}.

\noindent \textbf{Cut}

Formally, given the community structure $\mathcal{C} = \{\mathcal{U}_1, \mathcal{U}_2, \cdots, \mathcal{U}_k\}$ detected from network $G$. The number of links cut \cite{SM00} between communities $\mathcal{U}_i, \mathcal{U}_j \in \mathcal{C}$ can be represented as
\begin{equation}
cut(\mathcal{U}_i, \mathcal{U}_j) = \sum_{u \in \mathcal{U}_i} \sum_{v \in \mathcal{U}_j} I(u, v),
\end{equation}
where function $I(u, v) = 1$ if $(u, v) \in \mathcal{E}$; otherwise, it will be $0$.

The total number of links cut in the partition process can be represented as
\begin{equation}
cut(\mathcal{C}) = \sum_{\mathcal{U}_i \in \mathcal{C}} cut(\mathcal{U}_i, \bar{\mathcal{U}}_i),
\end{equation}
where set $\bar{\mathcal{U}}_i = \mathcal{C} \setminus \mathcal{U}_i$ denotes the remaining communities except $\mathcal{U}_i$.

By minimizing the cut cost introduced in the partition process, the optimal community detection result can be obtained with the minimum number of cross-community links. However, as introduced in \cite{SM00, L07}, by minimizing the cut of edges across clusters, the results may involve high imbalanced communities, some community may involve one single node. Such a problem will be much more severe when it comes to the real-world social network data. In the following part of this section, we will introduce two other cost measures that can help achieve more balanced community detection results.


\noindent \textbf{Ratio-Cut and Normalized-Cut}

As shown in the example, the minimum cut cost treat all the links in the network equally, and can usually achieve very imbalanced partition results (e.g., a singleton node as a cluster) when applied in the real-world community detection problem. To overcome such a disadvantage, some models have been proposed to take the community size into consideration. The community size can be calculated by counting the number of nodes or links in each community, which will lead to two new cost measures: \textit{ratio-cut} and \textit{normalized-cut} \cite{SM00, L07}.

Formally, given the community detection result $\mathcal{C} = \{\mathcal{U}_1, \mathcal{U}_2,\\ \cdots, \mathcal{U}_k\}$ in network $G$, the \textit{ratio-cut} and \textit{normalized-cut} costs introduced in the community detection result can be defined as follows respectively.

\begin{equation}
ratio-cut(\mathcal{C}) = \frac{1}{k} \sum_{\mathcal{U}_i \in \mathcal{C}} \frac{cut(\mathcal{U}_i, \bar{\mathcal{U}}_i)}{|\mathcal{U}_i|},
\end{equation}
where $|\mathcal{U}_i|$ denotes the number of nodes in community $\mathcal{U}_i$.

\begin{equation}
ncut(\mathcal{C}) = \frac{1}{k} \sum_{\mathcal{U}_i \in \mathcal{C}} \frac{cut(\mathcal{U}_i, \bar{\mathcal{U}}_i)}{vol(\mathcal{U}_i)},
\end{equation}
where $vol(\mathcal{U}_i)$ denotes the degree sum of nodes in community $\mathcal{U}_i$.

As shown in the above example, from the computed costs, we find that the community detected in plot C achieves much lower ratio-cut and ncut costs compared with those in plots B and D. Compared against the regular \textit{cut} cost, both \textit{ratio-cut} and \textit{normalized-cut} prefer a balanced partition of the social network.


\noindent \textbf{Spectral Clustering}

Actually the objective function of both \textit{ratio-cut} and \textit{normalized-cut} can be unified as the following linear algebra equation
\begin{equation}
\min_{\mb{H} \in \{0, 1\}^{|\mathcal{V}| \times k}} \mbox{Tr}(\mb{H}^\top \bar{\mb{L}} \mb{H}),
\end{equation}
where matrix $\mb{H} \in \{0, 1\}^{|\mathcal{V}| \times k}$ denotes the communities that users in $\mathcal{V}$ belong to. 

Let $\mb{A} \in \{0, 1\}^{|\mathcal{V}| \times |\mathcal{V}|}$ denote the social adjacency matrix of the network, and the corresponding diagonal matrix of $\mb{A}$ can be represented as matrix $\mb{D}$, where $\mb{D}$ has value $D(i,i) = \sum_j A(i,j)$ on its diagonal. The Laplacian matrix of the network adjacency matrix $\mb{A}$ can be represented as $\mb{L} = \mb{D} - \mb{A}$. Depending on the specific measures applied, matrix $\bar{\mb{L}}$ can be represented as
\begin{equation}
\bar{\mb{L}} = 
\begin{cases}
\mb{L}, & \mbox{ for ratio-cut measure,}\\
\mb{D}^{\frac{-1}{2}} \mb{L} \mb{D}^{\frac{-1}{2}}, & \mbox{ for normalized-cut measure.}
\end{cases}
\end{equation}

The binary constraint on the variable $\mb{H}$ renders the problem a non-linear integer programming problem, which is very hard to solve. One common practice to learn the variable $\mb{H}$ is to apply spectral relaxation to replace the binary constraint with the orthogonality constraint.
\begin{align}
&\min \mbox{Tr}(\mb{H}^\top \bar{\mb{L}} \mb{H}),\\
&s.t. \mb{H}^\top \mb{H} = \mb{I}.
\end{align}

As proposed in \cite{SM00}, the optimal solution $\mb{H}^*$ to the above objective function equals to the eigen-vectors corresponding to the $k$ smallest eigen-values of matrix $\bar{\mb{L}}$.


\subsection{Emerging Network Community Detection}\label{sec:chap8_sec3_cold}

The community detection algorithms introduced in the previous section are mostly proposed for one single homogeneous network. However, in the real world, most of the online social networks are actually heterogeneous containing very complex information. In recent years, lots of new online social networks have emerged and start to provide services, the information available for the users in these emerging networks is usually very limited. Meanwhile, many of the users are also involved in multiple online social networks simultaneously. For users who are using these emerging networks, they may also be involved in other developed social networks for a long time \cite{sdm15, pakdd15}. The abundant information available in these mature networks can actually be useful for the community detection in the emerging networks. In this section, we will introduce the cross-network community detection for emerging networks with information transferred from other mature social networks \cite{sdm15}.

In this part, we will introduce the social community detection for \textit{emerging networks} with information propagated across multiple \textit{partially aligned social networks}, which is formally defined as the ``\textit{emerging network community detection}'' problem. Especially, when the network is brand new, the problem will be the ``\textit{cold start community detection}'' problem. \textit{Cold start problem} is mostly prevalent in \textit{recommender systems} \cite{icdm13}, where the system cannot draw any inferences for users or items, for which it has not yet gathered sufficient information, but few works have been done on studying the \textit{cold start problem} in clustering/community detection problems. The ``\textit{emerging network community detection}'' problem and ``\textit{cold start community detection}'' problem studied in this section are both novel problems and very different from other existing works on community detection with abundant information.

Networks studied in this section can be formulated as two partially aligned attribute augmented heterogeneous networks: $\mathcal{G} = ((G^t, G^s), (A^{t,s}, A^{s,t}))$, where $G^t$ and $G^s$ are the emerging target network and well-developed source network respectively and $A^{t,s}, A^{s,t}$ are the sets of anchor links between $G^t$ and $G^s$. Both $G^t$ and $G^s$ can be formulated as the attribute augmented heterogeneous social network, e.g., $G^t = (\mathcal{V}^t, \mathcal{E}^t, \mathcal{A}^t)$ (where sets $\mathcal{V}^t$, $\mathcal{E}^t$ and $\mathcal{A}^t$ denote the user nodes, social links and diverse attributes in the network). With information propagated across $\mathcal{G}$, the \textit{intimacy matrix}, $\mb{H}$, among users in $\mathcal{V}^t$ can be computed. \textit{emerging network community detection} problem aims at partitioning user set $\mathcal{V}^t$ of the emerging network $G^t$ into $K$ disjoint clusters, $\mathcal{C} = \{C_1, C_2, \cdots,$ $C_K\}$, based on the \textit{intimacy matrix}, $\mb{H}$, where $\bigcup_i^K C_i = \mathcal{V}^t$ and $C_i \cap C_j = \emptyset, \forall i, j \in \{1, 2, \cdots, K\}, i \neq j$. When the target network $G^t$ is brand new, i.e., $\mathcal{E}^t = \emptyset$ and $\mathcal{A}^t = \emptyset$, the problem will be the \textit{cold start community detection} problem.

To solve all the above challenges, we will introduce a novel community detection method, {\ourcad}, proposed in \cite{sdm15}. {\ourcad} introduces a new concept, \textit{intimacy}, to measure the closeness relationships among users with both link and attribute information in online social networks. Useful information from aligned well-developed networks will be propagated via {\ourcad} to the emerging network to solve the shortage of information problem.


%



\subsubsection{Intimacy Matrix of Homogeneous Network}

The {\ourcad} model is built based on the closeness scores among users, which is formally called the \textit{intimacy scores} in this section. Here, we will introduce the \textit{intimacy scores} and \textit{intimacy matrix} used in {\ourcad} from a information propagation perspective.

For a given homogeneous network, e.g., $G = (\mathcal{V}, \mathcal{E})$, where $\mathcal{V}$ is the set of users and $\mathcal{E}$ is the set of social links among users in $\mathcal{V}$, the adjacency matrix of $G$ can be defined to be $\mb{A} \in \mathbb{R}^{|\mathcal{V}| \times |\mathcal{V}|}$, where $A(i, j) = 1$, iff $(u_i, u_j) \in \mathcal{E}$. Meanwhile, via the social links in $\mathcal{E}$, information can propagate among the users within the network, whose propagation paths can reflect the closeness among users \cite{PNX12}. Formally, term
\begin{equation}
p_{ji} = \frac{A(j, i)}{\sqrt{\sum_m A(j, m) \sum_n A(n, i)}}
\end{equation}
is called the information \textit{transition probability} from $u_j$ to $u_i$, which equals to the proportion of information propagated from $u_j$ to $u_i$ in one step. 

We can use an example to illustrate how information propagates within the network more clearly. Let's assume that user $u_i \in \mathcal{V}$ injects a stimulation into network $G$ initially and the information will be propagated to other users in $G$ via the social interactions afterwards. During the propagation process, users receive stimulation from their neighbors and the amount is proportional to the difference of the amount of information reaching the user and his neighbors. Let vector $\boldsymbol{f}^{(\tau)} \in \mathbb{R}^{|\mathcal{V}|}$ denote the states of all users in $\mathcal{V}$ at time $\tau$, i.e., the proportion of stimulation at users in $\mathcal{V}$ at $\tau$. The change of stimulation at $u_i$ at time $\tau + \Delta t$ is defined as follows:
\begin{equation}
\frac{f^{(\tau + \Delta t)}(i) - f^{(\tau)}(i)}{\Delta t} = \alpha \sum_{u_j \in \mathcal{V}} p_{ji} (f^{(\tau)}(j) - f^{(\tau)}(i)),
\end{equation}
where coefficient $\alpha$ can be set as $1$. The \textit{transition probabilities} $p_{ij}, i, j \in \{1, 2, \cdots, |\mathcal{V}|\}$ can be represented with the \textit{transition matrix} 
\begin{equation}
\mb{X} = (\mb{D}^{-\frac{1}{2}} \mb{A} \mb{D}^{-\frac{1}{2}})
\end{equation} 
of network $G$, where $\mb{X} \in \mathbb{R}^{|\mathcal{V}| \times |\mathcal{V}|}$, $X(i,j) = p_{ij}$ and diagonal matrix $\mb{D} \in \mathbb{R}^{|\mathcal{V}| \times |\mathcal{V}|}$ has value $D(i,i) = \sum_{j = 1}^{|\mathcal{V}|} A(i, j)$ on its diagonal. 

\begin{defn} 
(Social Transition Probability Matrix): The \textit{social transition probability matrix} of network $G$ can be represented as $\mb{Q} = \mb{X} - \mb{D}_{\mb{X}}$, where $\mb{X}$ is the \textit{transition matrix} defined above and diagonal matrix $D_{\mb{X}}$ has value $D_{\mb{X}}(i,i) = \sum_{j = 1}^{|\mathcal{V}|} \mb{X}(i, j)$ on its diagonal.
\end{defn}

Furthermore, by setting $\Delta t = 1$, denoting that stimulation propagates step by step in a discrete time through network, the propagation updating equation can be rewritten as: \begingroup\makeatletter\def\f@size{8}\check@mathfonts
\begin{align}
\boldsymbol{f}^{(\tau)} &=\boldsymbol{f}^{(\tau - 1)} + \alpha (\mb{X} - \mb{D}_{\mb{X}}) \boldsymbol{f}^{(\tau - 1)} = (\mb{I} + \alpha \mb{Q}) f^{(\tau - 1)} \\
&= (\mb{I} + \alpha \mb{Q})^\tau \boldsymbol{f}^{(0)}.
\end{align}\endgroup
Such a propagation process will stop when $\boldsymbol{f}^{(\tau)} = \boldsymbol{f}^{(\tau - 1)}$, i.e., 
\begin{equation}
(\mb{I} + \alpha \mb{Q})^{(\tau)} = (\mb{I} + \alpha \mb{Q})^{(\tau - 1)}.
\end{equation} 
The smallest $\tau$ that can stop the propagation is defined as the \textit{stop step}. To obtain the \textit{stop step} $\tau$, {\ourcad} need to keep checking the powers of $(\mb{I} + \alpha \mb{Q})$ until it doesn't change as $\tau$ increases, i.e., the \textit{stop criteria}.

\begin{defn} 
(Intimacy Matrix): Matrix 
\begin{equation}
\mb{H} = (\mb{I} + \alpha \mb{Q})^\tau \in \mathbb{R}^{|\mathcal{V}| \times |\mathcal{V}|}
\end{equation} 
is defined as the \textit{intimacy matrix} of users in $\mathcal{V}$, where $\tau$ is the \textit{stop step} and $H(i,j)$ denotes the \textit{intimacy score} between $u_i$ and $u_j \in \mathcal{V}$ in the network.
\end{defn}


\subsubsection{Intimacy Matrix of Attributed Heterogeneous Network}

\begin{figure*}[t]
\centering
\subfigure[augmented network]{\label{fig:chap8_sec3_eg_fig_attribute_1}
\begin{minipage}[l]{0.4\columnwidth}
\centering
\includegraphics[width=1.0\textwidth]{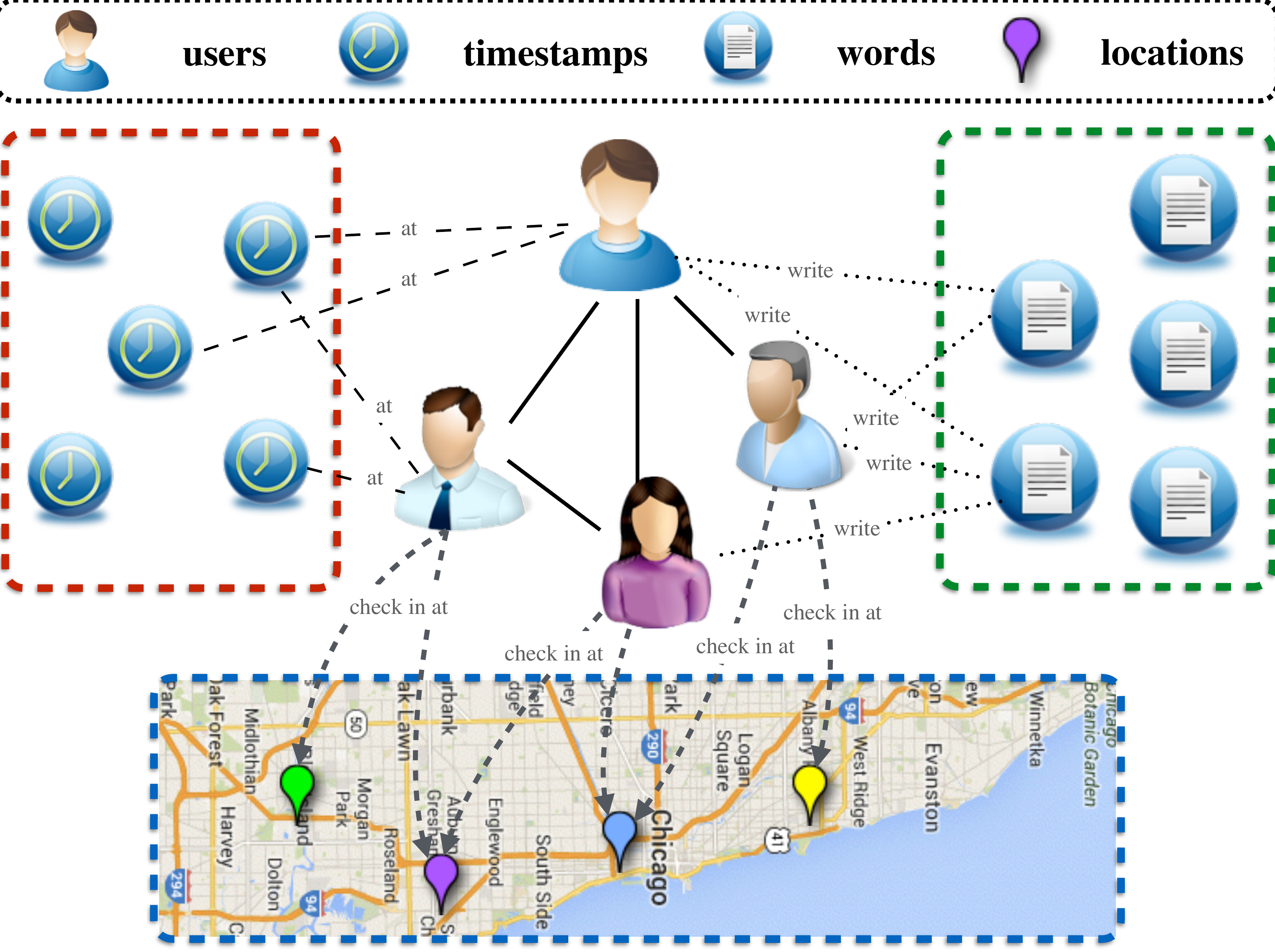}
\end{minipage}
}
\subfigure[timestamp attribute]{\label{fig:chap8_sec3_eg_fig_attribute_2}
\begin{minipage}[l]{0.4\columnwidth}
\centering
\includegraphics[width=1.0\textwidth]{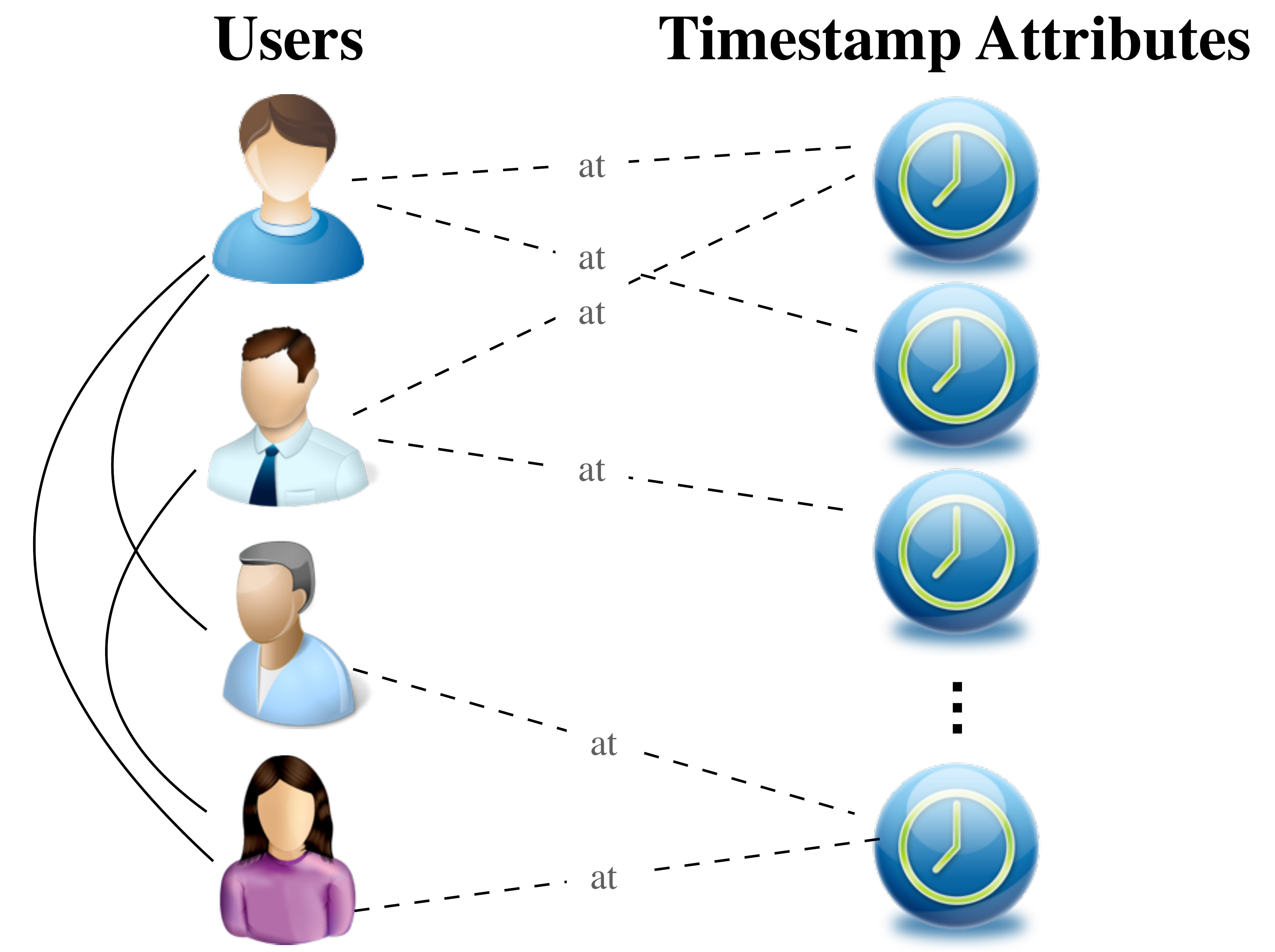}
\end{minipage}
}
\subfigure[text attribute]{\label{fig:chap8_sec3_eg_fig_attribute_3}
\begin{minipage}[l]{0.4\columnwidth}
\centering
\includegraphics[width=1.0\textwidth]{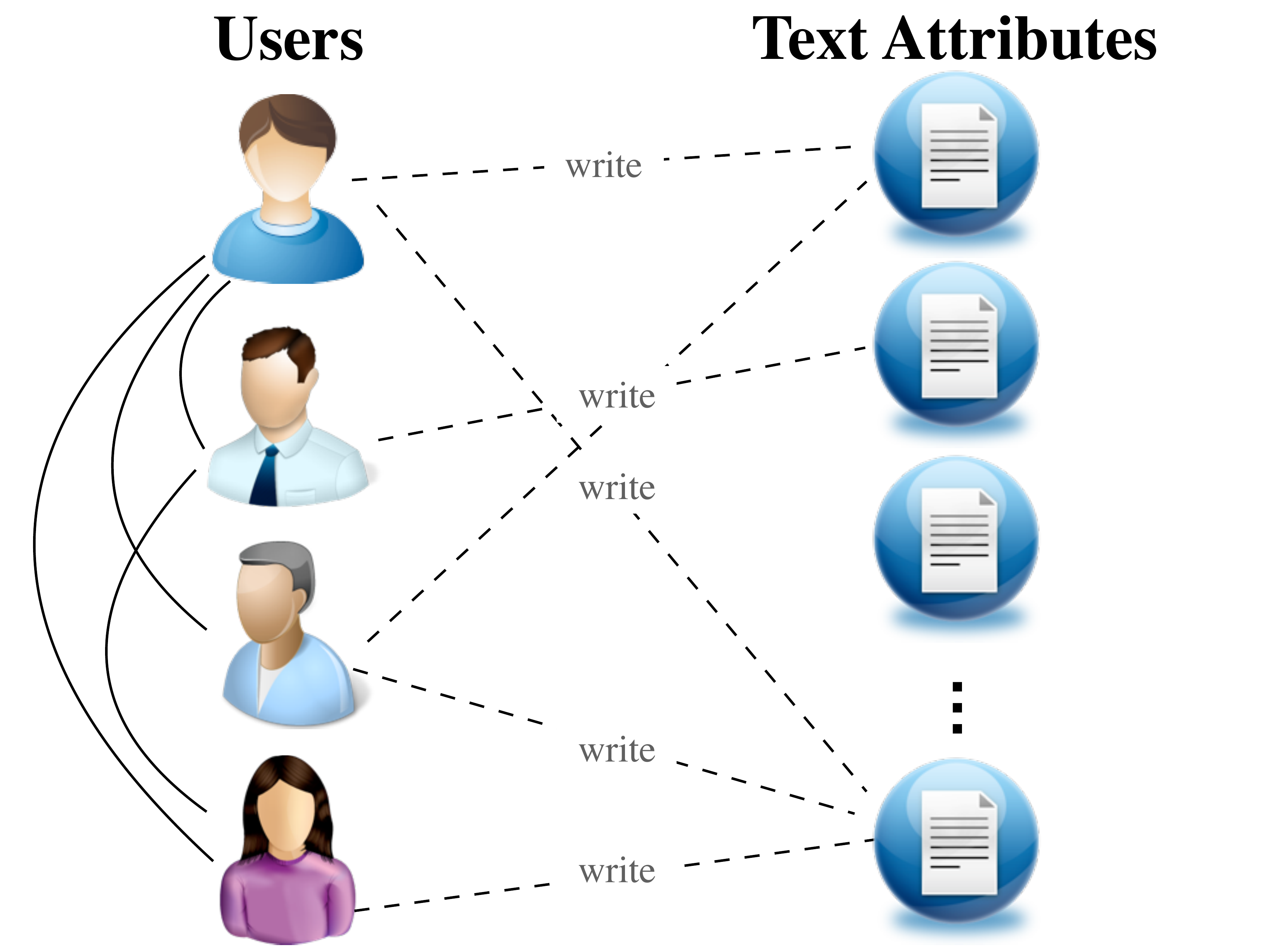}
\end{minipage}
}
\subfigure[checkin attribute]{\label{fig:chap8_sec3_eg_fig_attribute_4}
\begin{minipage}[l]{0.4\columnwidth}
\centering
\includegraphics[width=1.0\textwidth]{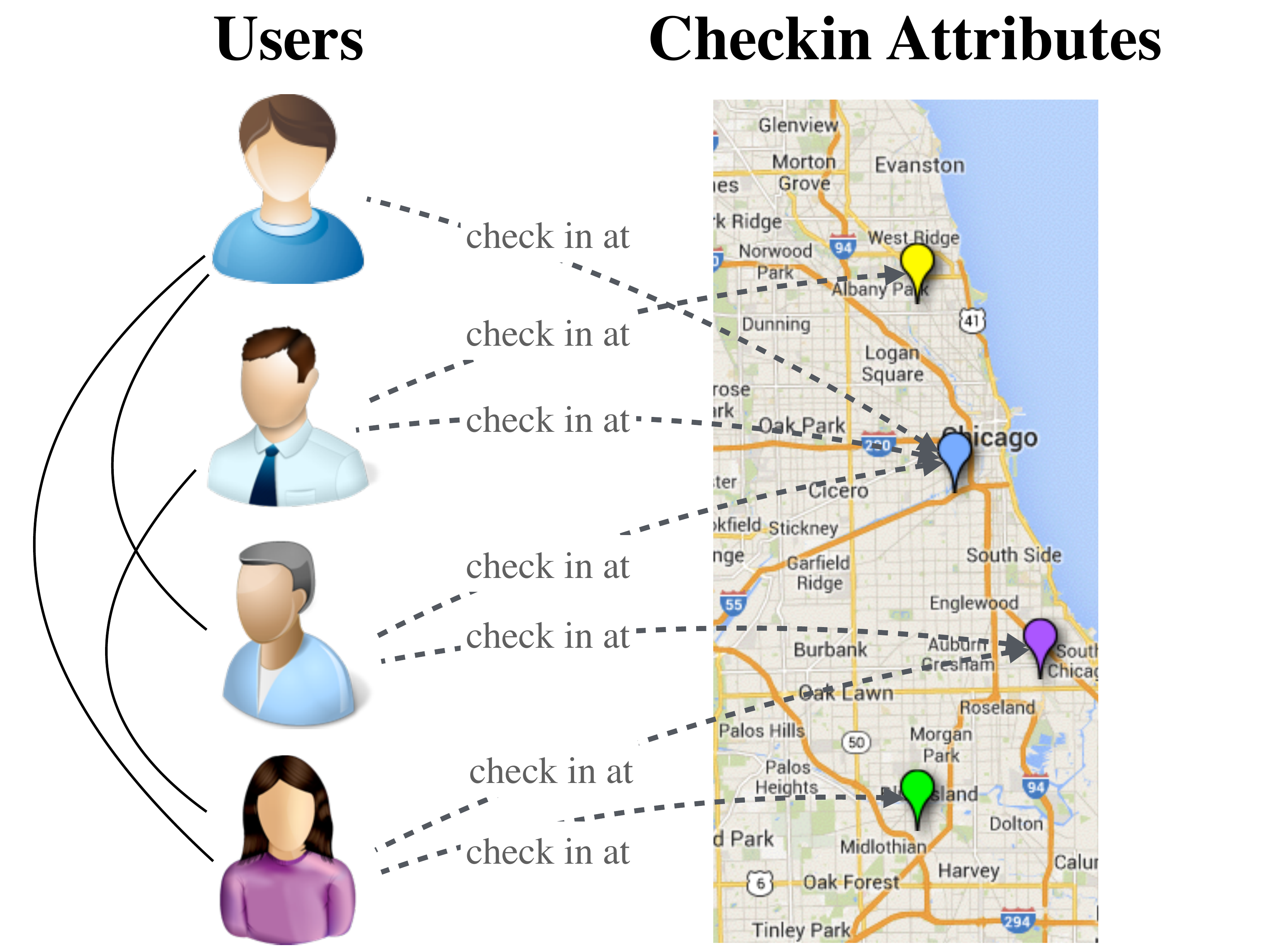}
\end{minipage}
}
\caption{An example of attribute augmented heterogeneous network. (a): attribute augmented heterogeneous network, (b): timestamp attribute, (c): text attribute, (d): location checkin attribute.}\label{fig:chap8_sec3_eg_fig_attribute}
\end{figure*}

Real-world social networks can usually contain various kinds of information, e.g., links and attributes, and can be formulated as $G = (\mathcal{V}, \mathcal{E}, \mathcal{A})$. Attribute set $\mathcal{A} = \{a_1, a_2,$ $\cdots,$ $a_m\}$, $a_i = \{a_{i1}, a_{i2},$$\cdots,$ $a_{in_i}\}$, can have $n_i$ different values for $i \in \{1, 2, \cdots, m\}$. An example of attribute augmented heterogeneous network is given in Figure~\ref{fig:chap8_sec3_eg_fig_attribute}, where Figure~\ref{fig:chap8_sec3_eg_fig_attribute_1} is the input \textit{attribute augmented heterogeneous network}. Figures~\ref{fig:chap8_sec3_eg_fig_attribute_2}-\ref{fig:chap8_sec3_eg_fig_attribute_4} show the attribute information in the network, which include timestamps, text and location checkins. Including the attributes as a special type of nodes in the graph definition provides a conceptual framework to handle social links and node attributes in a unified framework. The effect on increasing the dimensionality of the network will be handled  as in Lemma~\ref{lemma:chap8_sec3_lemma1} in  lower dimensional space.

\begin{defn} 
(Attribute Transition Probability Matrix): The connections between users and attributes, e.g., $a_i$, can be represented as the \textit{attribute adjacency matrix} $\mb{A}_{a_i}\in \mathbb{R}^{|\mathcal{V}| \times n_i}$. Based on $\mb{A}_{a_i}$, {\ourcad} formally defines the \textit{attribute transition probability matrix} from users to attribute $a_i$ to be $\mb{R}_i \in \mathbb{R}^{|\mathcal{V}| \times n_i}$, where  \begingroup\makeatletter\def\f@size{7}\check@mathfonts
\begin{equation}
\mb{R}_i(i, j) = \frac{1}{\sqrt{(\sum_{m = 1}^{n_i} \mb{A}_{a_i}(i, m))(\sum_{n = 1}^{|\mathcal{V}|} \mb{A}_{a_i}(n,j))}} \mb{A}_{a_i}(i, j).
\end{equation}
\end{defn} \endgroup
Similarly, {\ourcad} defines the \textit{attribute transition probability matrix} from attribute $a_i$ to users in $\mathcal{V}$ as $\mb{S}_i = \mb{R}_i^T$.

The importance of different information types in calculating the closeness measure among users can be different. 
To handle the \textit{network heterogeneity problem}, the {\ourcad} model proposes to apply the \textit{micro-level} control by giving different information sources distinct weights to denote their differences: $\mb{\omega} = [\omega_0, \omega_1, \cdots, \omega_m]^\top$, where $\sum_{i = 0}^m \omega_i = 1.0$, $\omega_0$ is the weight of link information and $\omega_i$ is the weight of attribute $a_i$, for $i \in \{1, 2, \cdots, m\}$. 

\begin{defn} 
(Weighted Attribute Transition Probability Matrix): With weights $\mb{\omega}$, {\ourcad} can define matrices  \begingroup\makeatletter\def\f@size{8}\check@mathfonts
\begin{equation}
\mb{\tilde{R}} = \left[\omega_1 \mb{R}_1, \cdots, \omega_n \mb{R}_n \right], \mbox{ and } \mb{\tilde{S}} = \left[\omega_1 \mb{S}_1, \cdots, \omega_n \mb{S}_n \right]^\top
\end{equation}\endgroup
to be the \textit{weighted attribute transition probability matrices} between users and all attributes, where $\mb{\tilde{R}} \in \mathbb{R}^{|\mathcal{V}|\times (n_{aug} - |\mathcal{V}|)}$, $\mb{\tilde{S}} \in \mathbb{R}^{(n_{aug} - |\mathcal{V}|) \times |\mathcal{V}|}$,
$n_{aug} = (|\mathcal{V}| + \sum_{i=1}^m n_i)$ is the number of all user and attribute nodes in the augmented network.
\end{defn} 

\begin{defn} 
(Network Transition Probability Matrix): Furthermore, the \textit{transition probability matrix} of the whole attribute augmented heterogeneous network $G$ is defined as 
\begin{equation}
\mb{\tilde{Q}}_{aug} = \begin{bmatrix}
\mb{\tilde{Q}} &\mb{\tilde{R}} \\
\mb{\tilde{S}} &\mb{0}\\
\end{bmatrix},
\end{equation}
where $\mb{\tilde{Q}}_{aug} \in \mathbb{R}^{n_{aug} \times n_{aug}}$ and block matrix $\mb{\tilde{Q}} = \omega_0 \mb{Q}$ is the \textit{weighted social transition probability matrix} of social links in $\mathcal{E}$.
\end{defn} 

In the real world, heterogeneous social networks can contain large amounts of attributes, i.e., $n_{aug}$ can be extremely large. The \textit{weighted transition probability matrix}, i.e., $\mb{\tilde{Q}}_{aug}$, can be of extremely high dimensions and can hardly fit in the memory. As a result, it will be impossible to update the matrix until the \textit{stop criteria} meets to obtain the \textit{stop step} and the \textit{intimacy matrix}. To solve such problem, {\ourcad} proposes to obtain the \textit{stop step} and the \textit{intimacy matrix} by applying partitioned block matrix operations with the following Lemma~\ref{lemma:chap8_sec3_lemma1}.

\begin{lemma}
$(\mb{\tilde{Q}}_{aug})^k =
\begin{bmatrix}
\mb{\tilde{Q}}_k &\mb{\tilde{Q}}_{k - 1}\mb{\tilde{R}} \\
\mb{\tilde{S}}\mb{\tilde{Q}}_{k - 1} &\mb{\tilde{S}}\mb{\tilde{Q}}_{k - 2}\mb{\tilde{R}}\\
\end{bmatrix}
$, $k \ge 2$, where 
\begin{equation}
\mb{\tilde{Q}}_k =
\begin{cases}
\mb{I}, & \mbox{if }k=0,\\
\mb{\tilde{Q}}, & \mbox{if }k=1, \\
\mb{\tilde{Q}}\mb{\tilde{Q}}_{k - 1} + \mb{\tilde{R}}\mb{\tilde{S}}\mb{\tilde{Q}}_{k - 2}, &\mbox{if } k \ge 2\\
\end{cases}
\end{equation}
and the \textit{intimacy matrix} among users in $\mathcal{V}$ can be represented as
\begin{align}
\mb{\tilde{H}}_{aug} &= \left( \mb{I} + \alpha \mb{\tilde{Q}}_{aug}\right )^\tau(1:|\mathcal{V}|, 1:|\mathcal{V}|)\\
& =  \left( \sum_{t = 0}^\tau \dbinom{\tau}{t} \alpha^t (\mb{\tilde{Q}}_{aug})^t \right)(1:|\mathcal{V}|, 1:|\mathcal{V}|) \\
& =  \left( \sum_{t = 0}^\tau \dbinom{\tau}{t} \alpha^t \left( (\mb{\tilde{Q}}_{aug})^t (1:|\mathcal{V}|, 1:|\mathcal{V}|)\right) \right) \\
&= \left( \sum_{t = 0}^\tau \dbinom{\tau}{t} \alpha^t \mb{\tilde{Q}}_t \right),
\end{align}
where $\mb{X}(1:|\mathcal{V}|, 1:|\mathcal{V}|)$ is a sub-matrix of $\mb{X}$ with indexes in range $[1, |\mathcal{V}|]$, $\tau$ is the \textit{stop step}, achieved when $\mb{\tilde{Q}}_{\tau} = \mb{\tilde{Q}}_{\tau - 1}$, i.e., the \textit{stop criteria}, $\mb{\tilde{Q}}_{\tau}$ is called the \textit{stationary matrix} of the attributed augmented heterogeneous network.
\end{lemma}\label{lemma:chap8_sec3_lemma1}

\begin{proof}
The lemma can be proved by induction on $k$. Considering that $(\mb{\tilde{R}}\mb{\tilde{S}}) \in \mathbb{R}^{|\mathcal{V}| \times |\mathcal{V}|}$ can be precomputed in advance, the space cost of Lemma~\ref{lemma:chap8_sec3_lemma1} is $\textrm O(|\mathcal{V}|^2)$, where $|\mathcal{V}| \ll n_{aug}$.
\end{proof}

Since we are only interested in the \textit{intimacy} and \textit{transition matrices} among user nodes instead of those between the augmented items and users for the community detection task, {\ourcad} creates a reduced dimensional representation only involving users for $\mb{\tilde{Q}}_k$ and $\mb{\tilde{H}}$ such that {\ourcad} can capture the effect of ``user-attribute'' and ``attribute-user'' transition on ``user-user'' transition. $\mb{\tilde{Q}}_k$ is a reduced dimension representation of $\mb{\tilde{Q}}_{aug}^k$, while eliminating the augmented items, it can still capture the ``user-user'' transitions effectively. 


\subsubsection{Intimacy Matrix across Aligned Heterogeneous Networks}

When $G^t$ is new, the \textit{intimacy matrix} $\mb{\tilde{H}}$ among users calculated based on the information in $G^t$ can be very sparse. To solve this problem, {\ourcad} proposes to propagate useful information from other well developed aligned networks to the emerging network. Information propagated from other aligned well-developed networks can help solve the shortage of information problem in the emerging network \cite{icdm13,wsdm14}. However, as proposed in \cite{PY10}, different networks can have different properties and information propagated from other well-developed aligned networks can be very different from that of the emerging network as well. 

To handle this problem, {\ourcad} model proposes to apply the \textit{macro-level control} technique by using weights, $\rho^{s,t}, \rho^{t,s }\in [0, 1]$, to control the proportion of information propagated between developed network $G^s$ and emerging network $G^t$. If information from $G^s$ is helpful for improving the community detection results in $G^t$, {\ourcad} can set a higher $\rho^{s,t}$ to propagate more information from $G^s$. Otherwise, {\ourcad} can set a lower $\rho^{s,t}$ instead. The weights $\rho^{s,t}$ and $\rho^{t,s}$ can be adjusted automatically with method to be introduced in \cite{sdm15}.

\begin{defn} 
(Anchor Transition Matrix): To propagate information across networks, {\ourcad} introduces the \textit{anchor transition matrices} between $G^t$ and $G^s$ to be $\mb{T}^{t,s} \in \mathbb{R}^{|\mathcal{V}^t| \times |\mathcal{V}^s|}$ and $\mb{T}^{s,t}\in \mathbb{R}^{|\mathcal{V}^s| \times |\mathcal{V}^t|}$, where entries $\mb{T}^{t,s}(i, j) = \mb{T}^{s,t}(j, i) = 1$, iff $(u^t_i, u^s_j) \in A^{t,s}, u^t_i \in \mathcal{V}^t, u^s_j \in \mathcal{V}^s$. 
\end{defn}

Meanwhile, with weights $\rho^{s,t}$ and $\rho^{t,s}$, the \textit{weighted network transition probability matrix} of $G^t$ and $G^s$ are represented as \begingroup\makeatletter\def\f@size{7}\check@mathfonts
\begin{equation}
\mb{\bar{Q}}^t_{aug} = (1-\rho^{t,s}) \begin{bmatrix}
\mb{\tilde{Q}}^t &\mb{\tilde{R}}^t \\
\mb{\tilde{S}}^t &\mb{0}\\
\end{bmatrix},
\mb{\bar{Q}}^s_{aug} = (1-\rho^{s,t}) \begin{bmatrix}
\mb{\tilde{Q}}^s &\mb{\tilde{R}}^s \\
\mb{\tilde{S}}^s &\mb{0}\\
\end{bmatrix},
\end{equation} \endgroup
where $\mb{\bar{Q}}^t_{aug} \in \mathbb{R}^{n^t_{aug} \times n^t_{aug}}$ and $\mb{\bar{Q}}^s_{aug} \in \mathbb{R}^{n^s_{aug} \times n^s_{aug}}$, $n^t_{aug}$ and $n^s_{aug}$ are the numbers of all nodes in $G^t$ and $G^s$ respectively. 

Furthermore, to accommodate the dimensions, {\ourcad} introduces the \textit{weighted anchor transition matrices} between $G^s$ and $G^t$ to be \begingroup\makeatletter\def\f@size{7}\check@mathfonts
\begin{equation}
\mb{\bar{T}}^{t,s} = (\rho^{t,s}) \begin{bmatrix}
\mb{T}^{t,s} &\mb{{0}} \\
\mb{{0}} &\mb{0}\\
\end{bmatrix},
\mbox{ and }
\mb{\bar{T}}^{s,t} = (\rho^{s,t}) \begin{bmatrix}
\mb{T}^{s,t} &\mb{{0}} \\
\mb{{0}} &\mb{0}\\
\end{bmatrix},
\end{equation} \endgroup
where $\mb{\bar{T}}^{t,s}$ $\in \mathbb{R}^{n^t_{aug} \times n^s_{aug}}$ and $\mb{\bar{T}}^{s,t} \in \mathbb{R}^{n^s_{aug} \times n^t_{aug}}$. Nodes corresponding to entries in $\mb{\bar{T}}^{t,s}$ and $\mb{\bar{T}}^{s,t}$ are of the same order as those in $\mb{\bar{Q}}^t_{aug}$ and $\mb{\bar{Q}}^s_{aug}$ respectively.

By combining the weighted intra-network transition probability matrices together with the weighted anchor transition matrices, {\ourcad} defines the \textit{transition probability matrix} across \textit{aligned networks} as
\begin{equation}
\mb{\bar{Q}}_{align} = \begin{bmatrix}
\mb{\bar{Q}}^t_{aug} &\mb{\bar{T}}^{t,s} \\
\mb{\bar{T}}^{s,t} &\mb{\bar{Q}}^s_{aug} \\
\end{bmatrix}
\end{equation}
where $\mb{\bar{Q}}_{align} \in \mathbb{R}^{n_{align} \times n_{align}}$, $n_{align} = n^t_{aug} + n^s_{aug}$ is the number of all nodes across the aligned networks.

\begin{defn} 
(Aligned Network Intimacy Matrix): According to the previous remarks, with $\mb{\bar{Q}}_{align}$, {\ourcad} can obtain the the \textit{intimacy matrix}, $\mb{\bar{H}}_{align}$, of users in $G^t$ to be
\begin{equation}
\mb{\bar{H}}_{align} = (\mb{I} + \alpha \mb{\bar{Q}}_{align})^\tau (1:|\mathcal{V}^t|, 1:|\mathcal{V}^t|),
\end{equation}
where $\mb{\bar{H}}_{align} \in \mathbb{R}^{|\mathcal{V}^t| \times |\mathcal{V}^t|}$, $\tau$ is the \textit{stop step}.
\end{defn} 

Meanwhile, the structure of $(\mb{I} + \alpha \mb{\bar{Q}}_{align})$ can not meet the requirements of Lemma~\ref{lemma:chap8_sec3_lemma1} as it doesn't have a zero square matrix at the bottom right corner. As a result, methods introduced in Lemma~\ref{lemma:chap8_sec3_lemma1} cannot be applied. To obtain the \textit{stop step}, there is no other choice but to keep calculating powers of $(\mb{I} + \alpha \mb{\bar{Q}}_{align})$ until the \textit{stop criteria} can meet, which can be very time consuming. In this part, we will introduce with the following Lemma~\ref{lemma:chap8_sec3_lemma2} adopted by {\ourcad} model for efficient computation of the high-order powers of matrix $(\mb{I} + \alpha \mb{\bar{Q}}_{align})$.

\begin{lemma}
For the given matrix $(\mb{I} + \alpha \mb{\bar{Q}}_{align})$, its $k_{th}$ power meets 
\begin{equation}
(\mb{I} + \alpha \mb{\bar{Q}}_{align})^k\mb{P} = \mb{P}\boldsymbol{\Lambda}^k, k \ge 1,
\end{equation} 
matrices $\mb{P}$ and $\boldsymbol{\Lambda}$ contain the eigenvector and eigenvalues of $(\mb{I} + \alpha \mb{\bar{Q}}_{align})$. The $i_{th}$ column of matrix $\mb{P}$ is the eigenvector of $(\mb{I} + \alpha \mb{\bar{Q}}_{align})$ corresponding to its $i_{th}$ eigenvalue $\lambda_i$ and diagonal matrix $\boldsymbol{\Lambda}$ has value $\Lambda(i,i) = \lambda_i$ on its diagonal.
\end{lemma}\label{lemma:chap8_sec3_lemma2}

The Lemma can be proved by induction on $k$ \cite{P12}. The time cost of calculating $\boldsymbol{\Lambda}^k$ is $\textrm O(n_{align})$, which is far less than that required to calculate $(\mb{I} + \alpha \mb{\bar{Q}}_{align})^k$.

\begin{defn} 
(Eigen-decomposition based Aligned Network Intimacy Matrix): In addition, if $\mb{P}$ is invertible, we can have
\begin{equation}
(\mb{I} + \alpha \mb{\bar{Q}}_{align})^k = \mb{P}\boldsymbol{\Lambda}^k\mb{P}^{-1},
\end{equation}
where $\boldsymbol{\Lambda}^k$ has $\Lambda(i,i)^k$ on its diagonal. And the intimacy calculated based on eigenvalue decomposition will be
\begin{equation}
\mb{\bar{H}}_{align} = \left( \mb{P}\boldsymbol{\Lambda}^\tau \mb{P}^{-1}\right )(1:|\mathcal{V}^t|, 1:|\mathcal{V}^t|).
\end{equation}
where the \textit{stop step} $\tau$ can be obtained when $\mb{P}\boldsymbol{\Lambda}^\tau \mb{P}^{-1} = \mb{P}\boldsymbol{\Lambda}^{\tau - 1}\mb{P}^{-1}$, i.e., \textit{stop criteria}.
\end{defn} 

Based on the computed matrix $\mb{\bar{H}}_{align}$, various clustering methods, e.g., KMedoids, can be adopted to identify the clusters of the social community.


\subsection{Mutual Community Detection}\label{sec:chap8_sec4_mutual}

Besides the knowledge transfer from developed networks to the emerging networks to overcome the cold start problem, information in developed networks can also be transferred mutually to help refine the detected community structure detected from each of them. In this section, we will introduce the mutual community detection problem across multiple aligned heterogeneous networks and introduce a new cross-network mutual community detection model {\ourmcd}. To refine the community structures, a new concept named \textit{discrepancy} is introduced to help preserve the consensus of the community detection result of the shared anchor users according to \cite{bigdata15}.

For the given multiple aligned heterogeneous networks $\mathcal{G}$, the \textit{Mutual Community Detection} problem aims to obtain the optimal communities $\{\mathcal{C}^{(1)}, \mathcal{C}^{(2)}, \cdots, \mathcal{C}^{(n)}\}$ for $\{G^{(1)}, G^{(2)}, \cdots,\\ G^{(n)}\}$ simultaneously, where $\mathcal{C}^{(i)} = \{U^{(i)}_1, U^{(i)}_2, \ldots, U^{(i)}_{k^{(i)}}\}$ is a partition of the users set $\mathcal{U}^{(i)}$ in $G^{(i)}$, $k^{(i)} = \left | \mathcal{C}^{(i)} \right |$, $U^{(i)}_l \cap U^{(i)}_m = \emptyset$, $\forall\ l,m \in \{1, 2, \ldots, k^{(i)}\}$ and $\bigcup_{j = 1}^{k^{(i)}} U^{(i)}_j = \mathcal{U}^{(i)}$. Users in each detected social community are more densely connected with each other than with users in other communities. In this section, we focus on studying the hard (i.e., non-overlapping) community detection of users in online social networks, and will illustrate a model proposed in paper \cite{bigdata15}.

Instead of the propagation based social intimacy score computation among users, {\ourmcd} proposes to use the meta paths introduced in Section~\ref{sec:meta_path} to utilize both direct and indirect connections among users in closeness scores calculation. With full considerations of the network characteristics, {\ourmcd} exploits the information in aligned networks to refine and disambiguate the community structures of the multiple networks concurrently. More detailed information about the {\ourmcd} model will be introduced as follows.

\subsubsection{Meta Path based Social Proximity Measure}\label{subsec:metapath_proximity}


\begin{table*}[t]
\scriptsize
\centering
{
\caption{Summary of HNMPs.}\label{tab:chap8_sec4_meta_path}
\begin{tabular}{llll}
\hline
\textbf{ID}
&\textbf{Notation}
& \textbf{Heterogeneous Network Meta Path}
& \textbf{Semantics}\\
\hline
\hline

1
&U $\to$ U
&User $\xrightarrow{follow}$ User
&Follow\\

2
&U $\to$ U $\to$ U
&User $\xrightarrow{follow}$ User $\xrightarrow{follow}$ User
&Follower of Follower\\

3
&U $\to$ U $\gets$ U
&User $\xrightarrow{follow}$ User $\xrightarrow{follow^{-1}}$ User
&Common Out Neighbor\\

4
&U $\gets$ U $\to$ U
&User $\xrightarrow{follow^{-1}}$ User $\xrightarrow{follow}$ User
&Common In Neighbor\\

\hline


5
&U $\to$ P $\to$ W $\gets$ P $\gets$ U
&User $\xrightarrow{write}$ Post $\xrightarrow{contain}$ Word 
&Posts Containing Common Words\\

&&\ \ \ \ \ \ \ \ \ \ \ \ \ \ \ \ \ \ \ \ $\xrightarrow{contain^{-1}}$ Post $\xrightarrow{write^{-1}}$ User&\\

6
&U $\to$ P $\to$ T $\gets$ P $\gets$ U
&User $\xrightarrow{write}$ Post $\xrightarrow{contain}$ Time 
&Posts Containing Common Timestamps\\

&&\ \ \ \ \ \ \ \ \ \ \ \ \ \ \ \ \ \ \ \ $\xrightarrow{contain^{-1}}$ Post $\xrightarrow{write^{-1}}$ User&\\

7
&U $\to$ P $\to$ L $\gets$ P $\gets$ U
&User $\xrightarrow{write}$ Post $\xrightarrow{attach}$ Location
&Posts Attaching Common Location Check-ins\\

&&\ \ \ \ \ \ \ \ \ \ \ \ \ \ \ \ \ \ \ \ $\xrightarrow{attach^{-1}}$ Post $\xrightarrow{write^{-1}}$ User&\\
\hline

\end{tabular}
}
\end{table*}

Many existing similarity measures, e.g., ``Common Neighbor'' \cite{HZ11}, ``Jaccard's Coefficient'' \cite{HZ11}, defined for homogeneous networks cannot capture all the connections among users in heterogeneous networks. To use both direct and indirect connections among users in calculating the similarity score among users in the heterogeneous information network, {\ourmcd} introduces meta path based similarity measure HNMP-Sim, whose information will be introduced as follows.

In heterogeneous networks, pairs of nodes can be connected by different paths, which are sequences of links in the network. Meta paths \cite{SHYYW11, SYH09} in heterogeneous networks, i.e., \textit{heterogeneous network meta paths} (HNMPs), can capture both direct and indirect connections among nodes in a network. The length of a meta path is defined as the number of links that constitute it. Meta paths in networks can start and end with various node types. However, in this section, we are mainly concerned about those starting and ending with users, which are formally defined as the \textit{social HNMPs}. A formal definition of \textit{social HNMPs} is available in \cite{kdd14, bigdata15, cikm16}. The notation, definition and semantics of $7$ different \textit{social HNMPs} used in {\ourmcd} are listed in Table~\ref{tab:chap8_sec4_meta_path}. To extract the social meta paths, prior domain knowledge about the network structure is required.

These $7$ different social HNMPs in Table~\ref{tab:chap8_sec4_meta_path} can cover lots of connections among users in networks. Some meta path based similarity measures have been proposed so far, e.g., the \textit{PathSim} proposed in \cite{SHYYW11}, which is defined for undirected networks and considers different meta paths to be of the same importance. To measure the social closeness among users in directed heterogeneous information networks, we extend \textit{PathSim} to propose a new closeness measure as follows. 
\begin{defn} 
(HNMP-Sim): Let $\mathcal{P}_i(x \rightsquigarrow y)$ and $\mathcal{P}_i(x \rightsquigarrow \cdot)$ be the sets of path instances of HNMP \# $i$ going from $x$ to $y$ and those going from $x$ to other nodes in the network. The HNMP-Sim (HNMP based Similarity) of node pair $(x, y)$ is defined as \begingroup\makeatletter\def\f@size{7}\check@mathfonts
\begin{equation}
\mbox{HNMP-Sim}(x, y) = \sum_i \omega_i \left(\frac{\left| \mathcal{P}_i(x \rightsquigarrow y) \right| + \left| \mathcal{P}_i(y \rightsquigarrow x) \right|} {{\left| \mathcal{P}_i(x \rightsquigarrow \cdot) \right| } + \left| \mathcal{P}_i(y \rightsquigarrow \cdot) \right|}\right),
\end{equation}\endgroup
where $\omega_i$ is the weight of the $i_{th}$ HNMP and $\sum_i \omega_i = 1$. In {\ourmcd}, the weights of different HNMPs can be automatically adjusted by applying a greedy search technique as introduced in \cite{bigdata15, sdm15}. 
\end{defn} 

Let $\mathbf{A}_i$ be the \textit{adjacency matrix} corresponding to the $i_{th}$ HNMP among users in the network and $\mb{A}_i(m, n) = k$ iff there exist $k$ different path instances of the $i_{th}$ HNMP from user $m$ to $n$ in the network. Furthermore, the similarity score matrix among users of HNMP \# $i$ can be represented as $\mb{S}_i = \mb{B}_i \circ \left( \mb{A}_i + \mb{A}_i^T \right)$, where $\mb{A}_i^T$ denotes the transpose of $\mb{A}_i$ and $\mb{B}_i$ represents the sum of the out-degree of user $x$ and $y$ has values $\mb{B}_i{(x,y)} = \frac{1}{\left(\sum_m \mb{A}_i{(x,m)} + \sum_m \mb{A}_i{(y,m)}\right)}$. The $\circ$ symbol represents the Hadamard product of two matrices. The HNMP-Sim matrix of the network which can capture all possible connections among users is represented as follows:
\begin{equation}
\mathbf{S} = \sum_i \omega_i \mb{S}_i = \sum_i \omega_i \left( \mb{B}_i \circ \left( \mb{A}_i + \mb{A}_i^T \right) \right).
\end{equation}


\subsubsection{Network Characteristic Preservation Clustering}
Clustering each network independently can preserve each networks characteristics effectively as no information from external networks will interfere with the clustering results. Partitioning users of a certain network into several clusters will cut connections in the network and lead to some costs inevitably. Optimal clustering results can be achieved by minimizing the clustering costs.

For a given network $G$, let $\mathcal{C} = \{U_1, U_2, \ldots, U_k\}$ be the community structures detected from $G$. Term $\overline{U_i} = \mathcal{U} - U_i$ is defined to be the complement of set $U_i$ in $G$. Various cost measure of partition $\mathcal{C}$ can be used, e.g., \textit{cut} and \textit{normalized cut} as introduced in Section~\ref{subsec:chap8_sec2_spectral}:
\begin{align}
&cut(\mathcal{C}) =  \frac{1}{k}\sum_{i = 1}^k S(U_i, \overline{U_i}) = \frac{1}{k}\sum_{i = 1}^k \sum_{u \in U_i, v \in \overline{U_i}} S(u, v),\\
&ncut(\mathcal{C}) = \frac{1}{k}\sum_{i = 1}^k \frac{S(U_i, \overline{U_i})}{S(U_i, \cdot)} = \frac{1}{k} \sum_{i = 1}^k\frac{cut(U_i, \overline{U}_i)}{S(U_i, \cdot)},
\end{align}
where term $S(u, v)$ denotes the HNMP-Sim between $u, v$ and $S(U_i, \cdot) = S(U_i, \mathcal{U}) = S(U_i, U_i) + S(U_i, \overline{U}_i)$. 

For all users in $\mathcal{U}$, their clustering result can be represented in the \textit{result confidence matrix} $\mb{H}$, where $\bf{H} = [\bf{h}_1,$ $\bf{h}_2,$ $\ldots,$ $\bf{h}_n]^T$, $n = |\mathcal{U}|$, $\mb{h}_i = (h_{i,1}, h_{i,2}, \ldots, h_{i,k})$ and $h_{i,j}$ denotes the confidence that $u_i \in \mathcal{U}$ is in cluster $U_j \in \mathcal{C}$. The optimal $\bf{H}$ that can minimize the normalized-cut cost can be obtained by solving the following objective function \cite{L07}:
\begin{align}
\min_{\mb{H}}\ \ &\mbox{Tr} (\mb{H}^T\mb{L}\mb{H}),\\
s.t. \ \ &\mb{H}^T\mb{D}\mb{H} = \mb{I}.
\end{align}
where $\mb{L} = \mb{D} - \mb{S}$, diagonal matrix $\mb{D}$ has ${D}(i,i) = \sum_j {S}(i,j)$ on its diagonal, and $\mb{I}$ is an identity matrix.


\subsubsection{Discrepancy based Clustering of Multiple Networks}

Besides the shared information due to common network construction purposes and similar network features \cite{sdm15}, anchor users can also have unique information (e.g., social structures) across aligned networks, which can provide us with a more comprehensive knowledge about the community structures formed by these users. Meanwhile, by maximizing the consensus (i.e., minimizing the ``\textit{discrepancy}'') of the clustering results about the anchor users in multiple partially aligned networks, model {\ourmcd} will be able to refine the clustering results of the anchor users with information in other aligned networks mutually. The clustering results achieved in $G^{(1)}$ and $G^{(2)}$ can be represented as $\mathcal{C}^{(1)} = \{U^{(1)}_1, U^{(1)}_2,$ $\cdots,$ $U^{(1)}_{k^{(1)}}\}$ and $\mathcal{C}^{(2)} = \{U^{(2)}_1, U^{(2)}_2, \cdots, U^{(2)}_{k^{(2)}}\}$ respectively.

Let $u_i$ and $u_j$ be two anchor users in the network, whose accounts in $G^{(1)}$ and $G^{(2)}$ are $u^{(1)}_i$, $u^{(2)}_i$, $u^{(1)}_j$ and $u^{(2)}_j$ respectively. If users $u^{(1)}_i$ and $u^{(1)}_j$ are partitioned into the same cluster in $G^{(1)}$ but their corresponding accounts $u^{(2)}_i$ and $u^{(2)}_j$ are partitioned into different clusters in $G^{(2)}$, then it will lead to a \textit{discrepancy} \cite{bigdata15, ijcnn16} between the clustering results of $u^{(1)}_i$, $u^{(2)}_i$, $u^{(1)}_j$ and $u^{(2)}_j$ in aligned networks $G^{(1)}$ and $G^{(2)}$.

\begin{defn} 
(Discrepancy): The discrepancy between the clustering results of $u_i$ and $u_j$ across aligned networks $G^{(1)}$ and $G^{(2)}$ is defined as the difference of confidence scores of $u_i$ and $u_j$ being partitioned in the same cluster across aligned networks. Considering that in the clustering results, the confidence scores of $u^{(1)}_i$ and $u^{(1)}_j$ ($u^{(2)}_i$ and $u^{(2)}_j$ ) being partitioned into $k^{(1)}$ ($k^{(2)}$) clusters can be represented as vectors $\mb{h}_i^{(1)}$ and $\mb{h}_j^{(1)}$ ($\mb{h}_i^{(2)}$ and $\mb{h}_j^{(2)}$) respectively, while the confidences that $u_i$ and $u_j$ are in the same cluster in $G^{(1)}$ and $G^{(2)}$ can be denoted as $\mb{h}_i^{(1)} (\mb{h}_j^{(1)})^T$ and $\mb{h}_i^{(2)} (\mb{h}_j^{(2)})^T$. Formally, the discrepancy of the clustering results about $u_i$ and $u_j$ is defined to be \begingroup\makeatletter\def\f@size{8}\check@mathfonts$d_{ij}(\mathcal{C}^{(1)}, \mathcal{C}^{(2)}) = \left(\mb{h}_i^{(1)} (\mb{h}_j^{(1)})^T - \mb{h}_i^{(2)} (\mb{h}_j^{(2)})^T\right )^2$\endgroup \\ if $u_i, u_j$ are both anchor users; and $d_{ij}(\mathcal{C}^{(1)}, \mathcal{C}^{(2)}) = 0$ otherwise.
Furthermore, the discrepancy of $\mathcal{C}^{(1)}$ and $\mathcal{C}^{(2)}$ will be:
\begin{align}
d(\mathcal{C}^{(1)}, \mathcal{C}^{(2)}) &= \sum_i^{n^{(1)}} \sum_j^{n^{(2)}} d_{ij}(\mathcal{C}^{(1)}, \mathcal{C}^{(2)}),
\end{align}
where ${n^{(1)}} = |\mathcal{U}^{(1)}|$ and ${n^{(2)}} = |\mathcal{U}^{(2)}|$. In the definition, non-anchor users are not involved in the discrepancy calculation.
\end{defn} 

However, considering that $d(\mathcal{C}^{(1)}, \mathcal{C}^{(2)})$ is highly dependent on the number of anchor users and anchor links between $G^{(1)}$ and $G^{(2)}$, minimizing $d(\mathcal{C}^{(1)}, \mathcal{C}^{(2)})$ can favor highly consented clustering results when the anchor users are abundant but have no significant effects when the anchor users are very rare. To solve this problem, model {\ourmcd} proposes to minimize the \textit{normalized discrepancy} instead.

\begin{defn} 
(Normalized Discrepancy) The normalized discrepancy measure computes the  differences of clustering results in two aligned networks as a fraction of the discrepancy with regard to the number of anchor users across partially aligned networks:
\begin{equation}
nd(\mathcal{C}^{(1)}, \mathcal{C}^{(2)}) = \frac{d(\mathcal{C}^{(1)}, \mathcal{C}^{(2)})}{\left( \left | A^{(1,2)} \right | \right) \left( \left | A^{(1,2)} \right | - 1 \right)}.
\end{equation}
\end{defn}

Optimal consensus clustering results of $G^{(1)}$ and $G^{(2)}$ will be $\hat{\mathcal{C}^{(1)}}, \hat{\mathcal{C}^{(2)}}$:
\begin{equation}
\hat{\mathcal{C}}^{(1)}, \hat{\mathcal{C}}^{(2)} = \arg \min_{\mathcal{C}^{(1)}, \mathcal{C}^{(2)}} nd(\mathcal{C}^{(1)}, \mathcal{C}^{(2)}).
\end{equation}

Similarly, the normalized-discrepancy objective function can also be represented with the \textit{clustering results confidence matrices} $\mb{H}^{(1)}$ and $\mb{H}^{(2)}$ as well. Meanwhile, considering that the networks studied in this section are partially aligned, matrices $\mb{H}^{(1)}$ and $\mb{H}^{(2)}$ contain the results of both anchor users and non-anchor users, while non-anchor users should not be involved in the discrepancy calculation according to the definition of discrepancy. The introduced model proposes to prune the results of the non-anchor users with the following \textit{anchor transition matrix} first.

\begin{defn} 
(Anchor Transition Matrix): Binary matrix $\mb{T}^{(1,2)}$ (or $\mb{T}^{(2,1)}$) is defined as the anchor transition matrix from networks $G^{(1)}$ to $G^{(2)}$ (or from $G^{(2)}$ to $G^{(1)}$), where $\mb{T}^{(1,2)} = ({\mb{T}^{(2,1)}})^T$, $\mb{T}^{(1,2)}(i, j) = 1$ if $(u^{(1)}_i, u^{(2)}_j) \in A^{(1,2)}$ and $0$ otherwise. The row indexes of $\mb{T}^{(1,2)}$ (or $\mb{T}^{(2,1)}$) are of the same order as those of $\mb{H}^{(1)}$ (or $\mb{H}^{(2)}$). Considering that the constraint on anchor links is ``\textit{one-to-one}'' in this section, as a result, each row/column of $\mb{T}^{(1,2)}$ and $\mb{T}^{(2,1)}$ contains at most one entry filled with $1$. 
\end{defn} 

\begin{figure}[t]
\centering
    \begin{minipage}[l]{0.8\columnwidth}
      \centering
      \includegraphics[width=1.0\textwidth]{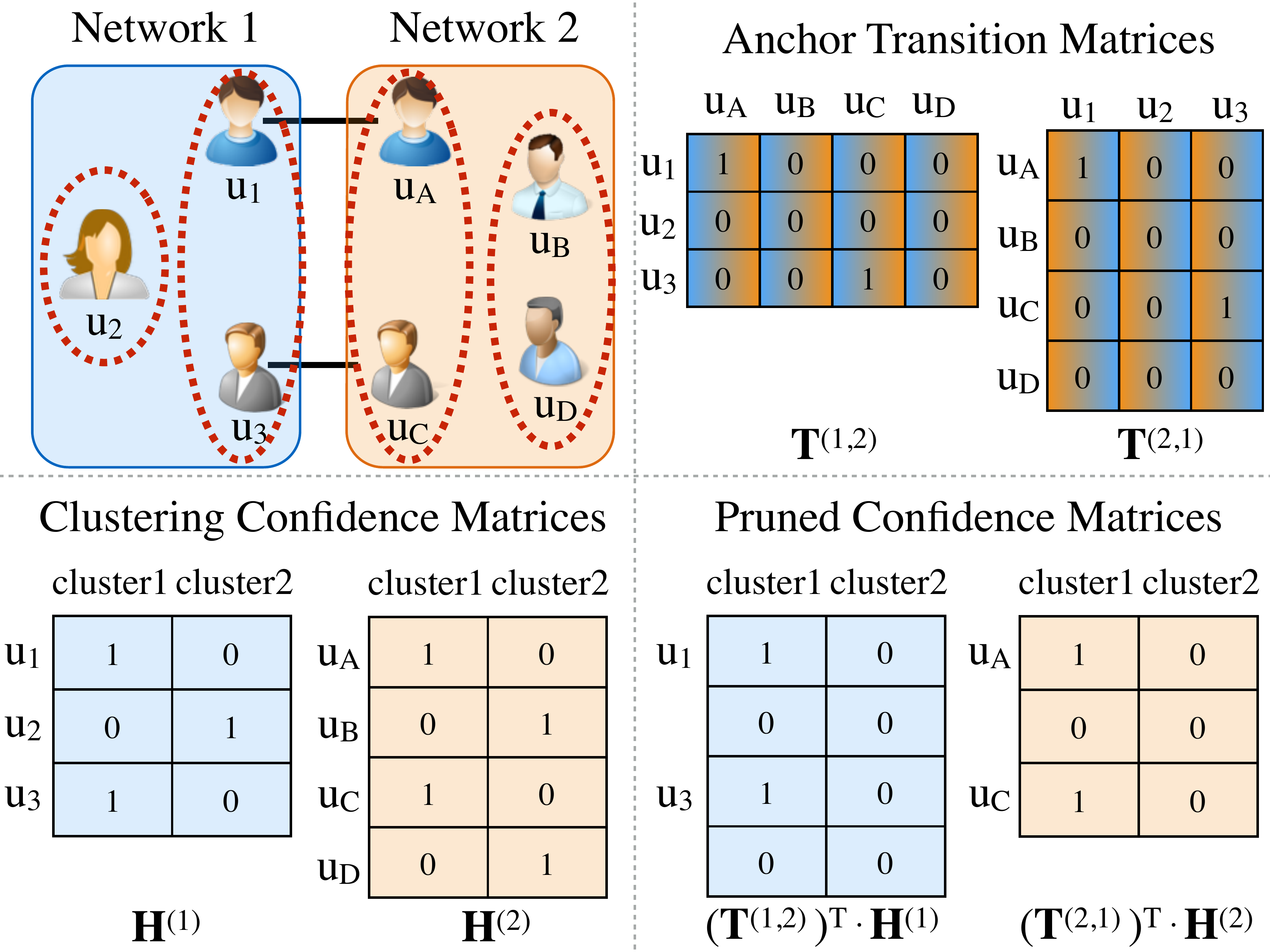}
    \end{minipage}
\caption{An example to illustrate the clustering discrepancy.}\label{fig:chap8_sec4_fig_discrepancy_example}
\end{figure}

Furthermore, the objective function of inferring clustering confidence matrices, which can minimize the normalized discrepancy can be represented as follows
\begin{align}
\min_{\mb{H}^{(1)}, \mb{H}^{(2)}} & \frac{\left \| \bar{\mb{H}}^{(1)} \left(\bar{\mb{H}}^{(1)} \right )^T - \bar{\mb{H}}^{(2)} \left ( \bar{\mb{H}}^{(2)} \right)^T \right \|^2_F}{\left \| \mb{T}^{(1,2)} \right \|^2_F \left (\left \| \mb{T}^{(1,2)} \right \|^2_F  - 1\right)}, \\
s.t. \ \ &(\mb{H}^{(1)})^T\mb{D}^{(1)}\mb{H}^{(1)} = \mb{I}, (\mb{H}^{(2)})^T\mb{D}^{(2)}\mb{H}^{(2)} = \mb{I}.
\end{align}
where $\mb{D}^{(1)}$, $\mb{D}^{(2)}$ are the corresponding diagonal matrices of HNMP-Sim matrices of networks $G^{(1)}$ and $G^{(2)}$ respectively.


\begin{algorithm}[t]
\caption{Curvilinear Search Method ($\mathcal{CSM}$)}
\label{alg:chap8_sec4_update}
\begin{algorithmic}[1]
	\REQUIRE $\mb{X}_k$ $C_k$, $Q_k$ and function $\mathcal{F}$\\
\qquad	parameters $\mb{\epsilon} = \{\rho, \eta, \delta, \tau, \tau_m, \tau_M\}$\\
\ENSURE  $\mb{X}_{k + 1}$, $C_{k + 1}$, $Q_{k + 1}$\\

\STATE	$\mb{Y}(\tau) = \left( \mb{I} + \frac{\tau}{2}\mb{A}\right)^{-1} \left( \mb{I} - \frac{\tau}{2}\mb{A}\right) \mb{X}_k$
\WHILE{$\mathcal{F}\left( \mb{Y}(\tau)\right) \ge \mb{C}_k + \rho \tau \mathcal{F}^\prime \left( (\mb{Y}(0)) \right)$}
\STATE	{$\tau = \delta \tau$}
\STATE	$\mb{Y}(\tau) = \left( \mb{I} + \frac{\tau}{2}\mb{A}\right)^{-1} \left( \mb{I} - \frac{\tau}{2}\mb{A}\right) \mb{X}_k$
\ENDWHILE
\STATE	{$\mb{X}_{k + 1} = \mb{Y}_k(\tau)$\\$Q_{k + 1} = \eta Q_k + 1$\\ $C_{k + 1} = \left( \eta Q_k C_k + \mathcal{F}(\mb{X}_{k + 1}) \right)/Q_{k + 1}$\\$\tau = \max \left( \min(\tau, \tau_M), \tau_m \right)$}
\end{algorithmic}
\end{algorithm}


\subsubsection{Joint Mutual Clustering of Multiple Networks}\label{subsec:joint}
Normalized-Cut objective function favors clustering results that can preserve the characteristic of each network, however, normalized-discrepancy objective function favors consensus results which are mutually refined with information from other aligned networks. Taking both of these two issues into considerations, the optimal \textit{Mutual Community Detection} results $\mathcal{\hat{C}}^{(1)}$ and $\mathcal{\hat{C}}^{(2)}$ of aligned networks $G^{(1)}$ and $G^{(2)}$ can be achieved as follows: \begingroup\makeatletter\def\f@size{7}\check@mathfonts
\begin{align}
\arg \min_{\mathcal{C}^{(1)}, \mathcal{C}^{(2)}} \alpha ncut(\mathcal{C}^{(1)}) + \beta ncut(\mathcal{C}^{(2)}) + \theta nd(\mathcal{C}^{(1)}, \mathcal{C}^{(2)})
\end{align}\endgroup
where $\alpha$, $\beta$ and $\theta$ represents the weights of these terms and, for simplicity, $\alpha$, $\beta$ are both set as $1$ in {\ourmcd}.

By replacing $ncut(\mathcal{C}^{(1)})$, $ncut(\mathcal{C}^{(2)})$, $nd(\mathcal{C}^{(1)}, \mathcal{C}^{(2)})$ with the objective equations derived above, the joint objective function can be rewritten as follows: \begingroup\makeatletter\def\f@size{7}\check@mathfonts
\begin{align}
\min_{\mb{H}^{(1)}, \mb{H}^{(2)}}\ \ \alpha &\mbox{Tr} (({\mb{H}^{(1)}})^T\mb{L}^{(1)}\mb{H}^{(1)}) + \beta \mbox{Tr} (({\mb{H}^{(2)}})^T\mb{L}^{(2)}\mb{H}^{(2)}) \\
&+ \theta  \frac{\left \| \bar{\mb{H}}^{(1)} \left(\bar{\mb{H}}^{(1)} \right )^T - \bar{\mb{H}}^{(2)} \left ( \bar{\mb{H}}^{(2)} \right)^T \right \|^2_F}{\left \| \mb{T}^{(1,2)} \right \|^2_F \left (\left \| \mb{T}^{(1,2)} \right \|^2_F  - 1\right)},\\
s.t. \ \ &({\mb{H}^{(1)}})^T\mb{D}^{(1)}\mb{H}^{(1)} = \mb{I}, ({\mb{H}^{(2)}})^T\mb{D}^{(2)}\mb{H}^{(2)} = \mb{I},
\end{align}\endgroup
where $\mb{L}^{(1)} = \mb{D}^{(1)} - \mb{S}^{(1)}$, $\mb{L}^{(2)} = \mb{D}^{(2)} - \mb{S}^{(2)}$ and matrices $\mb{S}^{(1)}$, $\mb{S}^{(2)}$ and $\mb{D}^{(1)}$, $\mb{D}^{(2)}$ are the HNMP-Sim matrices and their corresponding diagonal matrices defined before.


\begin{algorithm}[t]
\caption{Mutual Community Detector ({\ourmcd})}
\label{alg:chap8_sec4_framework}
\begin{algorithmic}[1]
	\REQUIRE aligned network: $\mathcal{G}$ = $\{\{G^{(1)}$, $G^{(2)}\}$, $\{A^{(1,2)},$ $A^{(2,1)}\}\}$;\\
\qquad	number of clusters in $G^{(1)}$ and $G^{(2)}$: $k^{(1)}$ and $k^{(2)}$;\\
\qquad	HNMP Sim matrices weight: $\mb{\omega}$;\\
\qquad	parameters: $\mb{\epsilon} = \{\rho, \eta, \delta, \tau, \tau_m, \tau_M\}$;\\
\qquad	function $\mathcal{F}$ and consensus term weight $\theta$
\ENSURE  $\mb{H}^{(1)}$, $\mb{H}^{(2)}$\\
\STATE	Calculate HNMP Sim matrices, $\mb{S}_i^{(1)}$ and $\mb{S}^{(2)}_i$
\STATE	$\mb{S}^{(1)} = \sum_i \omega_i S_i^{(1)}$, $\mb{S}^{(2)} = \sum_i \omega_i S_i^{(2)}$
\STATE	Initialize $\mb{X}^{(1)}$ and $\mb{X}^{(2)}$ with Kmeans clustering results on $\mb{S}^{(1)}$ and $\mb{S}^{(2)}$
\STATE	Initialize $C^{(1)}_0 = 0, Q^{(1)}_0 = 1$ and $C^{(2)}_0 = 0, Q^{(2)}_0 = 1$
\STATE	$converge = False$
\WHILE{$converge = False$}
\STATE	/* update $\mb{X}^{(1)}$ and $\mb{X}^{(2)}$ with $\mathcal{CSM}$ */\\
	$\mb{X}^{(1)}_{k+1}$, $C^{(1)}_{k+1}$, $Q^{(1)}_{k+1}$ = $\mathcal{CSM}(\mb{X}^{(1)}_k, C^{(1)}_k, Q^{(1)}_k, \mathcal{F}, \mb{\epsilon})$\\
		$\mb{X}^{(2)}_{k+1}$, $C^{(2)}_{k+1}$, $Q^{(2)}_{k+1}$ = $\mathcal{CSM}(\mb{X}^{(2)}_k, C^{(2)}_k, Q^{(2)}_k, \mathcal{F}, \mb{\epsilon})$
\IF{$\mb{X}^{(1)}_{k+1}$ and $\mb{X}^{(2)}_{k+1}$ both converge}
\STATE	$converge = True$ 
\ENDIF
\ENDWHILE
\STATE	{$\mb{H}^{(1)} = \left((\mb{D}^{(1)})^{-\frac{1}{2}} \right)^T \mb{X}^{(1)}$, $\mb{H}^{(2)} = \left((\mb{D}^{(2)})^{-\frac{1}{2}} \right)^T \mb{X}^{(2)}$}
\end{algorithmic}
\end{algorithm}

The objective function is a complex optimization problem with orthogonality constraints, which can be very difficult to solve because the constraints are not only non-convex but also numerically expensive to preserve during iterations. {\ourmcd} adopts curvilinear search method (i.e., Algorithm~\ref{alg:chap8_sec4_update}) with Barzilai-Borwein step \cite{WY10} to solve the problem, where the learning process can also converge quickly. The pseudo-code of the {\ourmcd} model is available in Algorithm~\ref{alg:chap8_sec4_framework}, which will call Algorithm~\ref{alg:chap8_sec4_update} for updating the variables iteratively.


\subsection{Large-Scale Network Synergistic Community Detection}\label{sec:chap8_sec5_large}

The community detection algorithm proposed in the previous section involves very complicated matrix operations, and works well for small-sized network data. However, when being applied to handle real-world online social networks involving millions even billions of users, they will suffer from the time complexity problem a lot. The problem to be introduced here follows the same formulation as the one introduced in Section~\ref{sec:chap8_sec4_mutual}, but the involved networks are of far larger sizes in terms of both node number and the social connection number. Synergistic partitioning across multiple large-scale social networks is very difficult for the following challenges:

\begin{itemize}
\item \textit{Social Network}: Distinct from generic data, usually contains intricate interactions, and multiple heterogeneous networks mean that the relationships across multiple networks should be taken into consideration. 

\item \textit{Network Scale}: Network size implies it is difficult for stand-alone programs to apply traditional partitioning methods and it is a difficult task to parallelize the existing stand-alone network partitioning algorithms. 

\item \textit{Distributed Framework}: For distributed algorithms, load balance should be taken into consideration and how to generate balanced partitions is another challenge.
\end{itemize}

\begin{algorithm}[t]
\caption{Edge Weight based Matching ($\mathcal{EWM}$)}
\scriptsize
\label{alg:chap8_sec5_partition}
\begin{algorithmic}[1]
	\REQUIRE Network $G_h$\\
	\qquad Maximum weight of a node $maxV W = n/k$
\ENSURE A coarser network $G_{h+1}$


\STATE	{$\mb{map}()$ Function:}
\FOR		{node $i$ in current data bolck}
\IF		{$match[i] == -1$}
\STATE	{$maxIdx = -1$}
\STATE	{$sortByEdgeWeight(NN(i))$}
\FOR		{$v_j \in NN(i)$}
\IF		{$match[j] == -1$ and $VW(i)+VW(j)<maxVW$}
\STATE	{maxIdx = j}
\ENDIF
\STATE	{$match[i]=maxIdx$}
\STATE	{$match[maxIdx] = i$}
\ENDFOR
\ENDIF
\ENDFOR

\STATE	{$\mb{reduce}()$ Function:}
\STATE	{new $newNodeID[n+1]$}
\STATE	{new $newVW[n+1]$}
\STATE	{set $idx=1$}
\FOR		{$i \in \{1, 2, \cdots, n\}$}
\IF		{$i < match[i]$}
\STATE	{set $newNodeID[match[i]] = idx$}
\STATE	{set $newNodeID[i] = idx$}
\STATE	{set $newVW[i] = newVW[match[i]]=VW(i)+VW(match[i[)$}
\STATE	{$idx ++$}
\ENDIF
\ENDFOR

\end{algorithmic}
\end{algorithm}

To address the challenges, in this section, we will introduce a network structure based distributed network partitioning framework, namely {\ourscalable} \cite{bigdata14}. The {\ourscalable} model identifies the anchor nodes among the multiple networks, and selects a network as the datum network, then divides it into k balanced partitions and generate $\langle$anchor node ID, partition ID$\rangle$ pairs as the main objective. Based on the objective, {\ourscalable} coarsens the other networks (called as synergistic networks) into smaller ones, which will further divides the smallest networks into $k$ balanced initial partitions, and tries to assign same kinds of anchor nodes into the same initial partition as many as possible. Here, anchor nodes of same kind means that they are divided into same partition in the datum network. Finally, {\ourscalable} projects the initial partitions back to the original networks.


\subsubsection{Distributed Multilevel k-way Partitioning}

In this section, we describe the heuristic framework for synergistic partitioning among multiple large scale social networks, and we call the framework {\ourscalable}. For large-sized networks, data processing in {\ourscalable} can be roughly divided into two stages: datum generation stage and network alignment stage.

When got the anchor node set $\mathcal{A}^{(1,2)}$ between networks $G^{(1)}$ and $G^{(2)}$, the {\ourscalable} framework will apply a distributed multilevel $k$-way partitioning method onto the datum network to generate $k$ balanced partitions. During this process, the anchor nodes are ignored and all the nodes are treated identically. We call this process datum generation stage. When finished, partition result of anchor nodes will be generated, {\ourscalable} stores them in a set-$Map \langle anidx,pidx \rangle$, where $anidx$ is anchor node ID and $pidx$ represents the partition ID the anchor node belongs to. After the datum generation stage, synergistic networks will be partitioned into k partitions according to the $Map\langle anidx, pidx \rangle$ to make the synergistic networks to align to the datum network, and during this process \textit{discrepancy} and \textit{cut} are the objectives to be minimized. We call this process network alignment stage.


Algorithms guaranteed to find out near-optimal partitions in a single network have been studied for a long period. But most of the methods are stand-alone, and performance is limited by the server's capacity. Inspired by the multilevel $k-way$ partitioning (MKP) method proposed by Karypis and Kumar \cite{KK98, KK96} and based on our previous work \cite{AXY09}, {\ourscalable} uses MapReduce \cite{DG08} to speedup the MKP method. As the same with other multilevel methods, MapReduce based MKP also includes three phases: coarsening, initial partitioning and un-coarsening.

Coarsening phase is a multilevel process and a sequence of smaller approximate networks $G_i = (\mathcal{V}_i, \mathcal{E}_i)$ are constructed from the original network $G_0 = (\mathcal{V}, \mathcal{E})$ and so forth, where $|\mathcal{V}_i| < |\mathcal{V}_{i-1}|, i \in \{1, 2, \cdots, n\}$. To construct coarser networks, node combination and edge collapsing should be performed. The task can be formally defined in terms of matching inside the networks \cite{BJ93}. A intra-network matching can be represented as a set of node pairs $\mathcal{M} = \{(v_i, v_j)\}, i \neq j$ and $(v_i, v_j) \in \mathcal{E}$, in which each node can only appear for no more than once. For a network $G_i$ with a matching $\mathcal{M}_i$, if $(v_j, v_k) \in \mathcal{M}_i$ then $v_j$ and $v_k$ will form a new node $v_q \in \mathcal{V}_{i+1}$ in network $G_{i+1}$ coarsen from $G_i$. The weight of $v_q$ equals to the sum of weight $v_j$ and $v_k$, besides, all the links connected to $v_j$ or $v_k$ in $G_i$ will be connected to $v_q$ in $G_{i+1}$. The total weight of nodes will remain unchanged during the coarsening phase but the total weight of edges and number of nodes will be greatly reduced. Let's define $W(\cdot)$ to be the sum of edge weight in the input set and $N(\cdot)$ to be the number of nodes/components in the input set. In the coarsening process, we have
\begin{align}
W(\mathcal{E}_{i+1}) &= W(\mathcal{E}_{i}) - W(\mathcal{M}_{i}),\\
N(\mathcal{V}_{i+1}) &= N(\mathcal{V}_{i}) - N(\mathcal{M}_{i}).
\end{align}

Analysis in \cite{KK95} shows that for the same coarser network, smaller edge-weight corresponds to smaller edge-cut. With the help of MapReduce framework, {\ourscalable} uses a local search method to implement an edge-weight based matching (EWM) scheme to collect larger edge weight during the coarsening phase. For the convenience of MapReduce, {\ourscalable} designs an emerging network representation format: each line contains essential information about a node and all its neighbors (NN), such as node ID, vertex weight (VW), edge weight (W), et al. The whole network data are distributed in distributed file system, such as HDFS \cite{SKRC10}, and each data block only contains a part of node set and corresponding connection information. Function $map()$ takes a data block as input and searches locally to find node pairs to match according to the edge weight. Function $reduce()$ is in charge of node combination, renaming and sorting. With the new node IDs and matching, a simple MapReduce job will be able to update the edge information and write the coarser network back onto HDFS. The complexity of EWM is $O(|\mathcal{E}|)$ in each iteration and pseudo code about EWM is shown in Algorithm~\ref{alg:chap8_sec5_partition}.

After several iterations, a coarsest weighted network $G_s$ consisting of only hundreds of nodes will be generated. For the network size of $G_s$, stand-alone algorithms with high computing complexity will be acceptable for initial partitioning. Meanwhile, the weights of nodes and edges of coarser networks are set to reflect the weights of the finer network during the coarsening phase, so $G_s$ contains sufficient information to intelligently satisfy the balanced partition and the minimum edge-cut requirements. Plenty of traditional bisection methods are quite qualified for the task. In {\ourscalable}, it adopts the KL method with an $O(|\mathcal{E}|^3)$ computing complexity to divide $G_s$ into two partitions and then take recursive invocations of KL method on the partitions to generate balanced $k$ partitions.

Un-coarsening phase is inverse processing of coarsening phase. With the initial partitions and the matching of the coarsening phase, it is easy to run the un-coarsening process on the MapReduce cluster.

\begin{algorithm}[t]
\caption{Synergistic Partitioning ($\mathcal{SP}$)}
\scriptsize
\label{alg:chap8_sec5_synergistic}
\begin{algorithmic}[1]
	\REQUIRE Network $G_h$\\
	\qquad Anchor Link Map $Map<anidx, pidx>$\\
	\qquad Maximum weight of a node $maxV W = n/k$
\ENSURE A coarser network $G_{h+1}$


\STATE	{Call Synergistic Partitioning-Map Function}
\STATE	{Call Synergistic Partitioning-Reduce Function}

\end{algorithmic}
\end{algorithm}

\begin{algorithm}[t]
\caption{Synergistic Partitioning-Map}
\scriptsize
\label{alg:chap8_sec5_synergistic_map}
\begin{algorithmic}[1]
	\REQUIRE Network $G_h$\\
	\qquad Anchor Link Map $Map<anidx, pidx>$\\
	\qquad Maximum weight of a node $maxV W = n/k$
\ENSURE A coarser network $G_{h+1}$


\STATE	{$\mb{map}()$ Function:}
\FOR		{node $i$ in current data bolck}
\IF		{$match[i] == -1$}
\STATE	{set $flag = false$}
\STATE	{$sortByEdgeWeight(NN(i))$}
\IF		{$v_i \in Map<anidx, pidx>$}

\FOR		{$v_j \in NN(i)$ \& $match[j] == -1$}
\IF		{$v_j \in Map<anidx, pidx>$ \& $Map.get(v_i)==Map.get(v_j)$ \& $VW(i)+VW(j)<maxVW$}
\STATE	{$match[i] = j, match[j] = i$}
\STATE	{$flag=true$, break}
\ENDIF
\ENDFOR

\IF		{$flag== false$, no suitable anchor node}
\FOR		{$v_j \in NN(i)$ \& $match[j] == -1$ \& $VW(v_i) + VW(v_j) < maxVW$}
\STATE	{$indirectNeighbor = NN(v_j)$}
\STATE	{$sortByEdgeWeight(NN(i))$}
\FOR		{$v_k \in indirectNeighbor$}
\IF		{$v_k \in Map<anidx, pidx>$ \& $Map.get(v_i)==Map.get(v_k)$}
\STATE	{$match[i] = j, match[j] = i$}
\STATE	{$flag=true$, break}
\ENDIF
\ENDFOR

\IF		{$flag == true$}
\STATE	{break}
\ENDIF
\ENDFOR
\ENDIF
\ELSE	
\STATE	{$sortByEdgeWeight(NN(i))$}
\FOR		{$v_j \in NN(v_i)$ \& $v_j \notin Map<anidx, pidx>$ \& $VW(i) + VW(j) < maxVW$ \& $match[j] == -1$}
\STATE	{$match[i] = j$, $match[j] =i$, break}
\ENDFOR
\ENDIF
\ENDIF
\ENDFOR

\end{algorithmic}
\end{algorithm}

\begin{algorithm}[t]
\caption{Synergistic Partitioning-Reduce}
\scriptsize
\label{alg:chap8_sec5_synergistic_reduce}
\begin{algorithmic}[1]
	\REQUIRE Network $G_h$\\
	\qquad Anchor Link Map $Map<anidx, pidx>$\\
	\qquad Maximum weight of a node $maxV W = n/k$
\ENSURE A coarser network $G_{h+1}$


\STATE	{$\mb{reduce}()$ Function:}
\STATE	{new $newNodeID[n+1]$}
\STATE	{new $newVW[n+1]$}
\STATE	{set $idx=1$}
\FOR		{$i \in newNodeID[]$}
\IF		{$i < match[i]$}
\STATE	{set $newNodeID[match[i]] = idx$}
\STATE	{set $newNodeID[i] = idx$}
\STATE	{set $newVW[i] = newVW[match[i]]=VW(i)+VW(match[i[)$}
\STATE	{$idx ++$}
\ENDIF
\ENDFOR
\STATE	{new $newPurity[idx+1]$}
\STATE	{new $newPidx[idx+1]$}
\FOR		{$i \in [1, idx]$}
\STATE	{$newPurity[i]=\frac{purity[i]*VW(i) + purity[j]*VW(j)}{VW(i) + VW(j)}$}
\STATE	{$newPidx[i] = \max\{pidx[i], pidx[match[i]]\}$}
\ENDFOR

\end{algorithmic}
\end{algorithm}


\subsubsection{Distributed Synergistic Partitioning Process}

In this section, we will talk about the synergistic partitioning process in {\ourscalable} based on the synergistic networks with the knowledge of partition results of anchor nodes from datum network. The synergistic partitioning is also a MKP process but quite different from general MKP methods.

In the coarsening phase, anchor nodes are endowed with higher priority than non-anchor nodes. When choosing nodes to pair, {\ourscalable} assumes that anchor nodes and non-anchor nodes have different tendencies. Let $G^d$ be the datum network. For an anchor node $v_i$ in another aligned networks, at the top of its preference list, it would like to matched with another anchor node $v_i$, which has the same partition ID in the datum network, i.e., $pidx(G^d, v_i) = pidx(G^d, v_j)$ (here $pidx(G^d, v_i)$ denotes the community label that $v_j$ belongs to in $G^d$). Second, if there is no appropriate anchor node, it would try to find a non-anchor node to pair. When planing to find a non-anchor node to pair, the anchor node, assuming to be $v_i$, would like to find a correct direction, and it would prefer to match with the non-anchor node $v_j$, which has lots of anchor nodes as neighbors with the same $pidx$ with $v_i$. When being matched together, the new node will be given the same $pidx$ as the anchor node. To improve the accuracy of synergistic partitioning among multiple social networks, an anchor node will never try to combine with another anchor node with different $pidx$.

For a non-anchor node, it would prefer to be matched with an anchor node neighbor which belongs to the dominant partition in the non-anchor node's neighbors. Here, dominant partition in a node's neighbors means the number of anchor nodes with this partition ID is the largest. Next, a non-anchor node would choose a general non-anchor node to pair with. At last, a non-anchor node would not like to combine with an anchor node being part of the partitions which are in subordinate status. After combined together, the new node will be given the same $pidx$ as the anchor node. To ensure the balance among the partitions, about $\frac{1}{3}$ of the nodes in the coarsest network are unlabeled.

\begin{figure}[t]
\centering
    \begin{minipage}[l]{0.8\columnwidth}
      \centering
      \includegraphics[width=1.0\textwidth]{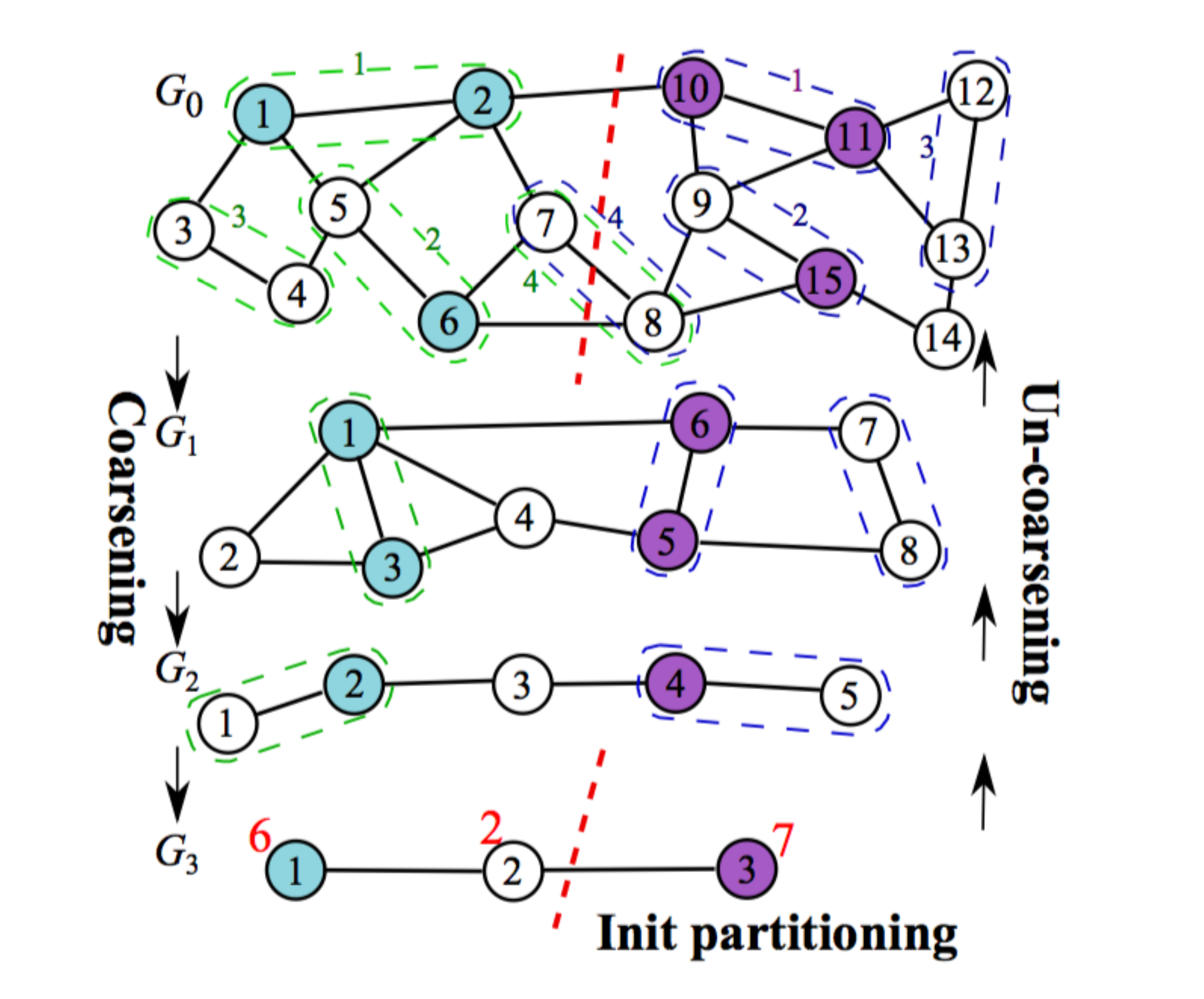}
    \end{minipage}
\caption{An Example of Synergistic Partition Process. In coarsening phase, the networks are stored in two servers, $V_1^i = \{v^i(j)|j \le |V^i|/2\}$ are stored on a sever and the others are on the other server. Anchor nodes are with colors, and different colors represent different partitions. Node pairs encircled by dotted chains represent the matchings. Numbers on chains mean the order of pairing.}\label{fig:chap8_sec5_network_coarsening_uncoarsening}
\end{figure}

In addition to minimizing both the discrepancy and cut discussed before, {\ourscalable} also tries to balance the size of partitions are the objectives in synergistic partitioning process. However, when put together, it is impossible to achieve them simultaneously. So, {\ourscalable} tries to make a compromise among them and develop a heuristic method to tackle the problems. 

\begin{itemize}
\item First, according to the conclusion smaller edge-weight corresponds to smaller edge-cut and the pairing tendencies, {\ourscalable} proposes a modified EWM (MEWM) method to find a matching in the coarsening phase, of which the edge-weight is as large as possible. At the end of the coarsening phase, there is no impurity in any node, meaning that each node contains no more than one type of anchor nodes. Besides, a ``\textit{purity}'' vector attribute and a $pidx$ attribute are added to each node to represent the percentage of each kind of anchor nodes swallowed up by it and the $pidx$ of the new node, respectively. 

\item Then, during the initial partitioning phase, {\ourscalable} treats the anchor nodes as labeled nodes and use a modified label propagation algorithm to deal with the non-anchor nodes in the coarsest network. 

\item At the end of the initial partitioning phase, {\ourscalable} will be able to generate balanced $k$ partitions and to maximize the number of same kind of anchor nodes being divided into same partitions. 

\item Finally, {\ourscalable} projects the coarsest network back to the original network, which is the same as traditional MKP process. 
\end{itemize}

The pseudo code of coarsening phase in synergistic partitioning process is available in Algorithm~\ref{alg:chap8_sec5_synergistic}, which will call the $Map()$ and $Reduce()$ functions in Algorithms~\ref{alg:chap8_sec5_synergistic_map} and \ref{alg:chap8_sec5_synergistic_reduce} respectively.



\section{Information Diffusion}\label{chapter:information_diffusion}\label{sec:diffusion}

Social influence can be widely spread among people, and information exchange has become one of the most important social activities in the real world. The creation of the Internet and online social networks has rapidly facilitated the communication among people. Via the interactions among users in online social networks, information can be propagated from one user to other users. For instance, in recent years, online social networks have become the most important social occasion for news acquisition, and many outbreaking social events can get widely spread in the online social networks at a very fast speed. People as the multi-functional ``sensors'' can detect different kinds of signals happening in the real world, and write posts to report their discoveries to the rest of the world via the online social networks. 

In this section, we will study the information diffusion process in the online social networks. \textit{Diffusion} denotes the spreading process of certain entities (like information, idea, innovation, even heat in physics and disease in bio-medical science) through certain channels among the target object group in a system. The entities to be spread, the channels available, the target object group and the system can all affect the diffusion process and lead to different diffusion observations. Therefore, different types of diffusion models have been proposed already, which will be introduced in this chapter.

Depending on the system where the diffusion process is originally studied, the diffusion models can be divided into (1) information diffusion models in social networks \cite{KKT03, cikm16}, (2) viral spreading in the bio-medical system \cite{BMZ83, DGG01}, and (3) heat diffusion in physical system \cite{N99, B96_2}. We will take the information diffusion in online social networks as one example. The channels for information diffusion belong to certain sources, like online world diffusion channels and offline world diffusion channels, or diffusion channels in different social networks. Meanwhile, depending on the diffusion channels and sources available, the diffusion models include (1) single-channel diffusion model \cite{asonam16_intertwined, KKT03}, (2) single source multi-channel diffusion model \cite{icdcs17}, (3) multi-source single-channel diffusion model \cite{ZNDZT16, iri16}, and (4) multi-source multi-channel diffusion model \cite{pakdd15, cikm16, bibm16}. Based on the categories of topics to be spread in the online social networks, the diffusion models can be categorized into (1) single topic diffusion \cite{KKT03, pakdd15}, (2) multiple intertwined topics concurrent diffusion \cite{cikm16, asonam16_intertwined, KOW08, DMS10, BKS07}.

In the following part of this section, we will introduce different kinds of diffusion models proposed to depict how information propagates among users in online social networks. We will first talk about the classic diffusion models proposed for the single-network single channel scenario, including the \textit{threshold based models}, \textit{cascades based models}, \textit{heat diffusion based models} and \textit{viral diffusion based models}. After that, several different cross-network diffusion models will be introduced, including the \textit{network coupling based diffusion model}, \textit{multi-source multi-channel diffusion model}, and \textit{cross-network random walk based diffusion model}.

\subsection{Traditional Information Diffusion Models}

The ``\textit{diffusion}'' phenomenon has been observed in different disciplines, like social science, physics, and bio-medical science. Various diffusion models have been proposed in these areas already. In this part, we will provide a brief introduction to these models, and introduce how to apply or adapt them for describe information diffusion process in online social networks.

Let $G = (\mathcal{V}, \mathcal{E})$ represent the network structure, based on which we want to study the information diffusion problem. Formally, given a user node $u \in \mathcal{V}$, we can represent the set of neighbors of $u$ as $\Gamma(u)$. Each user node in the network $G$ will have an indicator denoting whether the user has been activated or not. We will use notation $s(u) =1$ to denote that user $u$ has been activated, and $s(u) =0$ to represent that $u$ is still inactive. Initially, all the users are inactive to a certain information. Information can be propagated from an initial influence seed user set $\mathcal{S} \subset \mathcal{V}$ who are exposed to and activated by the information at the very beginning. At a timestamp in the diffusion process, given user $u$'s neighbor, we can represent the subset of the active neighbors as $\Gamma^a(u) = \{v | v \in \Gamma(u), s(v) = 1\}$. The set of inactive neighbors can be represented as $\Gamma^i(u) = \Gamma(u) \setminus \Gamma^a(u)$. Generally, the information diffusion process will stop if no new activation is available.

\subsubsection{Linear Threshold (LT) Models}

In this subsection, we will introduce the threshold models, and will use \textit{linear threshold model} as an example to illustrate such a kind of models. Several different variants of the \textit{linear threshold models} will be briefly introduced here as well.

Generally, the \textit{threshold models} assume that individuals have a unique threshold indicating the minimum amount of required information for them to be activated by certain information. Information can propagate among the users, and the information amount is determined by the closeness of the users. Close friends can influence each other much more than regular friends and strangers. If the information propagated from other users in the network surpass the threshold of a certain user, the user will turn to an activated status and also start to influence other users. Therefore, the threshold values can determine the performance of users in the online social networks. Depending on the setting of the thresholds as well as the amount of information propagated among the users, the \textit{threshold models} have different variants.

\noindent \textbf{LT Model}

In the \textit{linear threshold} (LT) model \cite{KKT03}, each user has a unique threshold denoting the minimum required information to active the user. Formally, the threshold of user $u$ can be represented as $\theta_u \in [0, 1]$. In the simulation experiments, the threshold values are normally selected from the uniform distribution $U(0, 1)$. Meanwhile, for each user pair, like $u, v \in \mathcal{V}$, information can be propagated between them. As mentioned before, close friends will have larger influence on each other compared with regular friends and strangers. Formally, the amount of information users $u$ can send to $v$ is denoted as weight $w_{u, v} \in [0, 1]$. Generally, the total amount of informations can send out is bounded. For instance, in the LT model, the total amount of information user $u$ can send out is bounded by $1$, i.e., $\sum_{v \in \Gamma(u)} w_{u, v} \le 1$. Different ways have been proposed to define the specific value of the weight $w_{u, v}$ value, and in many of the cases $w_{u, v}$ can be different from $w_{v, u}$ since the information each user can send out can be different. However, in many other cases, to simplify the setting, for the same user pair, $w_{u, v}$ and $w_{v, u}$ are usually assigned with the same value. For instance, in some LT models, Jaccard's Coefficient is applied to calculate the closeness between the user pairs which will be used as the weight value.

In the LT model, the information sent from the neighbors to user $u$ can be aggregated with linear summation. For instance, the total amount of information user $u$ can receive from his/her neighbors can be denoted as $\sum_{v \in \Gamma(u)} w(v, u) s(v)$ or $\sum_{v \in \Gamma^a{u}} w(v, u)$. To check whether a user can be activated or not, LT model will only need to check whether the following equation holds or not,
\begin{equation}
\sum_{v \in \Gamma^a{u}} w(v, u) \ge \theta_u.
\end{equation}
It denotes whether the received information surpasses the activation threshold of user $u$ or not. Here, we also need to notice that inactive neighbors will not send out information, and only the active neighbors can send out information. The information provided so far shows the critical details of the LT model. Next, we will show the general framework of the LT model to illustrate how it works.

In the LT model, the initial activated seed user set can be represented as $\mathcal{S}$, users in which can start the propagation of information to their neighbors. Generally, information propagates within the network step by step.
\begin{itemize}
\item \textit{Diffusion Starts}: At step $0$, only the seed users in $\mathcal{S}$ are active, and all the remaining users have inactive status.

\item \textit{Diffusion Spreads}: At step $t (t > 0)$, for each user $u$, if the information propagated from $u$'s active neighbors is greater than the threshold of $u$, i.e., $\sum_{v \in \Gamma^a{u}} w(v, u) \ge \theta_u$, $u$ will be activated with status $s(u) = 1$. All the activated users will remain active in the coming rounds, and can send out information to the neighbors. Active users cannot be activated again.

\item \textit{Diffusion Ends}: If no new activation happens in step $t$, the diffusion process will stop.
\end{itemize}

Specifically, in the diffusion process, at step $t$, we don't need to check all the users to see whether they will be activated or not. The reason is that, in the diffusion process, for most of the inactive users, if the status of their neighbors are not changed in the previous step, i.e., step $t-1$, the influence they can receive in step $t$ will still be the same as in step $t-1$. And they will remain the same status as they are in the previous step, i.e., ``\textit{inactive}''. Let $\mathcal{V}^a(t-1)$ denote the set of users who are recently activated in step $t-1$, we can represent the set of users they can influence as $\bigcup_{u \in \mathcal{V}^a(t-1)} \Gamma(u)$. In step $t$, these recently activated users will make changes to the information their neighbors can receive. Therefore, we only need to check whether the status of inactive users in the set $\bigcup_{u \in \mathcal{V}^a(t-1)} \Gamma(u)$ will meet the activation criterion or not.

After the diffusion process stops, a group of users with the active status will indicates the influence these seed users spread to, which can be represented as set $\mathcal{V}^a$. Generally, there will exist a mapping: $\sigma: \mathcal{S} \to |\mathcal{V}^a|$, which is formally called the influence function. Given the \textit{influence function}, with different seed user sets as the input, the influence they can achieve is usually different. Choosing the optimal seed user who can lead to the maximum influence is named as the \textbf{influence maximization} problem.

\noindent \textbf{Other Threshold Models}

The LT model assumes the cumulative effects of information propagated from the neighbors, and can illustrate the basic information diffusion process among users in the online social networks. The LT model has been well analyzed, and many other variant models have been proposed as well. Depending on the assignment of the threshold and weight values, many other different diffusion models can all be reduced to a special case of the LT model.

\begin{itemize}
\item \textbf{Majority Threshold Model}: Different from the LT mode, in \textit{majority threshold model} \cite{C08}, an inactive user $u$ can be activated if majority of his/her neighbors are activated. The \textit{majority threshold model} can be reduced to the LT model in the case that: (1) the influence weight between any friends $(u, v)$ in the network is assigned with value $1$; (2) the threshold of any user $u$ is set as $\frac{1}{2} D(u)$, where $D(u)$ denotes the degree of node $u$ in the network. For the nodes with large degrees, like the central node in the star-structured diagram, their activation will lead to the activation of lots of surrounding nodes in the network.

\item \textbf{k-Threshold Model}: Another diffusion model similar to the LT model is called the \textit{k-threshold diffusion model} \cite{C08}, in which users can be activated of at least $k$ of his/her neighbors are active. The \textit{k-threshold model} is equivalent to the LT model with settings (1) the influence weight between any friend pairs $(u, v)$ in the network is assigned with value $1$; and (2) the activation thresholds of all the users are assigned with a shared value $k$. For each user $u$, if $k$ of his/her neighbors have been activated, $u$ will be activated. 

Depending on the values of $k$, the \textit{k-threshold model} will have different performance. When $k=1$, a user will be activated of at least one of his/her neighbor is active. In such a case, all the users in the same connected components with the initial seed users will be activated finally. When $k$ is a very large value and even greater than the large node degree, e.g., $k > \max_{u \in \mathcal{V}} D(u)$, no nodes can be activated. When $k$ is a medium value, some of the users will be activated as the information propagates, but the other users with less than $k$ neighbors will never be activated.
\end{itemize}

\subsubsection{Independent Cascade (IC) Model}

An information cascade occurs when a people observe the actions of others and then engage in the same acts. Cascade clearly illustrates the information propagation routes, and the activating actions performed for users to their neighbors. In the cascade model, the information propagation dynamics is carried out in a step-by-step fashion. At each step, users can have trials to activate their neighbors to change their opinions with certain probabilities. If they succeed, the neighbors will change their status to follow the initiators. In the case that multiple users can all have the change to activate certain target user, the activation trials are performed sequentially in an arbitrary order.

Depending on the activation trials and users' reactions to the activation trials, different cascade models have been proposed already. In this section, we will talk about the cascade based models and use the \textit{independent cascade} (IC) model as an example to illustrate the model architecture.

\noindent \textbf{IC Model}

In the diffusion process, about one certain target user, multiple activation trials can be performed by his/her neighbors. In the \textit{independent cascade} model \cite{KKT03}, each activation is performed independently regardless of the historical unsuccessful trials. The activation trials are performed step by step. When user $u$ who has been activated in the previous step and tries to activate user $v$ in the current step, the success probability is denoted as $p_{u, v} \in [0,1]$. Generally, if users $u$ and $v$ are close friends, the activation probability will be larger compared with regular friends and strangers. The specific activation probability values is usually correlated with the social closeness between users $u$ and $v$, which can also be defined based the Jaccard's Coefficient in the simulation. The activation trials will only happen among the users who are friends. If $u$ succeeds in activating $v$, then user $v$ will change his/her status to ``\textit{active}'' and will remain in the status in the following steps. However, if $u$ fails to activate $v$, $u$ will lose the chance and cannot perform the activation trials any more. 

In the IC mode, we can represent the initial seed user as set $\mathcal{S} \subset \mathcal{V}$, who will spread the information to the remaining users. We illustrate the general information propagation procedure as follows:
\begin{itemize}

\item \textit{Diffusion Starts}: In the initial, the seed users will send out the information and start to activate their neighbors. For the users in set $\mathcal{S}$, the activation trials will start from them in a random order. For instance, if we pick user $u \in \mathcal{S}$ as the first user, $u$ will activate his/her inactive friends in $\Gamma(u)$ in a random order as well.

\item \textit{Diffusion Spreads}: In step $t$, only the users who have just been activated in the previous step can activate other users. We can denote the users who have just been activated in the previous as set $\mathcal{S}(t-1)$. Users in set $\mathcal{S}(t-1)$ will start to perform activation trials. For the users who are activated by these users, they will remain active in the following steps and will be added to the set $\mathcal{S}(t)$, who will start the activation trials in the next step.

\item \textit{Diffusion Ends}: If no activation happens in a step, the diffusion process stops.

\end{itemize}

In IC model, the activation trials are performed by flipping a coin with certain probabilities, whose result is uncertain. Even with the same provided initial seed user set $\mathcal{S}$, the number of users who will be activated by the seed users can be different if we running the IC model twice. Formally, we can represent the set of activated users by the seed users as $\mathcal{V}^a \subset \mathcal{V}$. Therefore, in the experimental simulations, we usually run the diffusion model multiple times and calculate the average number of activated users, i.e., $|\mathcal{V}^a|$, to denote the expected influence achieved by the seed user set $\mathcal{S}$. 

\noindent \textbf{Other Cascade Models}

Generally, the independent activation assumption renders the IC model the simplest cascade based diffusion models. In the real world, the diffusion process will be more complicated. For the users, who have been failed to be activated by many other users, it probably indicates that the user is not interested in the information. Viewed in such a perspective, the probability for the user to be activated will decrease as more activation trials have been performed. In this part, we will introduce another cascade based diffusion model, \textit{decreasing cascade model} (DC) \cite{KKT05}.

To illustrate the DC model more clearly and show it difference compared with the IC model, we use notation $P(u \to v | \mathcal{T})$ to represent the probability for user $u$ to activate $v$ given a set of users $\mathcal{T}$ have performed and failed the activation trials to $v$ already. Let $\mathcal{T}$, $\mathcal{T}'$ denote two historical activation trial user set, where $\mathcal{T} \subseteq \mathcal{T}'$. In the IC model, we have
\begin{equation}
P(u \to v | \mathcal{T}) = P(u \to v | \mathcal{T}').
\end{equation}
In other words, every activation trial is independent with each other, and the activation probability will not be changed as more activation trials have been performed.

As introduced at the beginning of this subsection, the fact that users in set $\mathcal{T}$ fail to activate $v$ indicates that $v$ probably is not interested in the information, and the change for $v$ to be activated afterwards will be lower. Furthermore, as more activation trials, e.g., users in $\mathcal{T}'$ are performed, the probability for $u$ to active $v$ will be decreased, i.e.,
\begin{equation}
P(u \to v | \mathcal{T}) \ge P(u \to v | \mathcal{T}').
\end{equation}

Intuitively, this restriction states that a contagious node's probability of activating some $v$ decreases if more nodes have already attempted to activate $v$, and $v$ is hence more ``marketing-saturated''. The DC model incorporates the IC model as a special case, and is more general in information diffusion process modeling than the IC model.

\subsubsection{Epidemic Diffusion Model}

The threshold and cascade based diffusion models introduced in the previous part mostly assume that ``once a user is activated, he/she will remain the active status forever''. However, in the real world, these activated users can change their minds and the activated users can still have the chance to recover to the original status. In the bio-medical science, diffusion models have been studied for many years to model the spread of disease, and several \textit{epidemic diffusion models} have been introduced already. In the disease propagation, people who are susceptible to the disease can be get infected by other people. After some time, many of these infected people can get recovered and become immune to the disease, while many other users can get recovered and get susceptible to the disease again. Depending on the people's reactions to the disease after recovery, several different \textit{epidemic diffusion models} \cite{PCVV15} have been proposed already.

In this subsection, we will introduce the \textit{epidemic diffusion models}, and try to use them to model the diffusion of information in the online social networks.

\noindent \textbf{Susceptible-Infected-Recovered (SIR) Diffusion Model}

The SIR model was proposed by W. O. Kermack and A. G. McKendrick in 1927 to model the infectious diseases, which consider a fixed population with three main categories: \textit{susceptible} (S), \textit{infected} (I), and \textit{recovered} (R). As the disease propagates, the individual status can change among \{S, I, R\} following flow:
\begin{equation}
S \to I \to R.
\end{equation}
In other words, the individuals who are susceptible to the disease can get infected, while those infected individuals also have the chance to recover from the disease as well.

In this part, we will use the SIR model to describe the information cascading process in online social networks. Let $\mathcal{V}$ denote the set of users in the network. We introduce the following notations to represent the number of users in different categories:

\begin{itemize}
\item $S(t)$: the number of users who are \textit{susceptible} to the information at time $t$, but have not gotten \textit{infected} yet.

\item $I(t)$: the number of users who are currently \textit{infected} by the information, and can spread the information to others in the \textit{susceptible} catetory.

\item $R(t)$: the number of users who have been infected and already recovered from the information infection. The users are immune to the information will not be infected again.
\end{itemize}

Based on the above notations, we have the following equations hold in the SIR model.
\begin{align}
&S(t) + I(t) + R(t) = |\mathcal{V}|,\\
&\frac{\mathrm{d}S(t)}{\mathrm{d}t} + \frac{\mathrm{d}I(t)}{\mathrm{d}t} + \frac{\mathrm{d}R(t)}{\mathrm{d}t} = 0,\\
\end{align}
where,
\begin{equation}
\begin{cases}
&\frac{\mathrm{d}S(t)}{\mathrm{d}t} = - \beta S(t) I (t),\\
&\frac{\mathrm{d}I(t)}{\mathrm{d}t} = \beta S(t) I (t) - \gamma I(t),\\
&\frac{\mathrm{d}R(t)}{\mathrm{d}t} = \gamma I(t).
\end{cases}
\end{equation}

In the above equations, the parameters $\beta$ denotes the infection rate of these \textit{susceptible} users by the \textit{infected} users in unit time, and $\gamma$ represents the recovery rate. Generally, all the users in the social network will belong to these three categories, and the total number of users in these three categories will sum to $|\mathcal{V}|$ at any time in the diffusion process. Therefore, we can also get the derivatives of the summation with regarding to the time parameter $t$ will be $0$. At a unit time, the number of users transit from the \textit{susceptible} status to the \textit{infection} status depends on the available susceptible and infected users at the same time. For each infected user, the number of users he/she can infect is proportional to the available susceptible users, which can be denoted ad $\beta S(t)$. For all the infected users, the total number of users can get infected will be $\beta S(t) I(t)$. For the number of users who are recovered in unit time, it depends on the number of total infected users $I(t)$ as well as the recovery rate $\gamma$, which can be represented as $\gamma I(t)$. Meanwhile, as to the number of infected user changes in unit time is determined by both the number of susceptible users who get infected as well as the infected users who get recovered. 

We have parameters $\beta, \gamma \ge 0$, and the numbers $S(t), I(t),\\ R(t) \ge 0$ to be positive at any time. Therefore, we can know that (1) $\frac{\mathrm{d}S(t)}{\mathrm{d}t} \le 0$, and users in the \textit{susceptible} group is non-increasing; (2) $\frac{\mathrm{d}R(t)}{\mathrm{d}t}  \ge 0$, and users in the \textit{recovered} group is non-decreasing; while (3) the sign of term $\frac{\mathrm{d}I(t)}{\mathrm{d}t}$ can be either positive, zero or negative depending on the parameters $\beta, \gamma$ and the users in the \textit{susceptible} and \textit{infected} groups:
\begin{itemize}
\item \textit{positive}: if $\beta S(t) > \gamma$;

\item \textit{zero}: if $\beta S(t) = \gamma$ or $I(t) = 0$;

\item \textit{negative}: $\beta S(t) < \gamma$.
\end{itemize}

\noindent \textbf{Susceptible-Infected-Susceptible (SIS) Diffusion Model}

In some cases, the users cannot get immune to the information and don't exist the \textit{recovery} status actually. For the users, who get infected, they can go to the \textit{susceptible} status and can get \textit{infected} again in the future. To model such a phenomenon, another diffusion model very similar to the SIR model has been proposed, which is called the Susceptible-Infected-Susceptible (SIS) model.

In the SIS model, the individual status flow is provided as follows:
\begin{equation}
S \to I \to S.
\end{equation}

Such a status flow will continue, and individuals will switch their status between \textit{susceptible} and \textit{infected} in the information diffusion process. Therefore, the absolute number changes of individuals in these two categories will be the same in unit time.
\begin{align}
&\frac{\mathrm{d}S(t)}{\mathrm{d}t} = - \beta S(t) I (t) + \gamma I(t),\\
&\frac{\mathrm{d}I(t)}{\mathrm{d}t} = \beta S(t) I (t) - \gamma I(t).
\end{align}

\noindent \textbf{Susceptible-Infected-Recovered-Susceptible (SIRS)\\ Diffusion Model}

The Susceptible-Infected-Recovered-Susceptible (SIRS) diffusion model to be introduced in this part is another type of epidemic model, where the individuals in the \textit{recovery} category can lose the immunity and transit to the \textit{susceptible} category and have the potential to be infected again. Therefore, the individual status flow will be
\begin{equation}
S \to I \to R \to S.
\end{equation}

We can denote the rate of individuals who lose the immunity as $f$, and the total number of individuals who may lose the immunity will be $f\cdot R(t)$. Therefore, we can have the derivative of the individual numbers belonging to different categories as
\begin{align}
&\frac{\mathrm{d}S(t)}{\mathrm{d}t} = - \beta S(t) I (t) + f R(t),\\
&\frac{\mathrm{d}I(t)}{\mathrm{d}t} = \beta S(t) I (t) - \gamma I(t),\\
&\frac{\mathrm{d}R(t)}{\mathrm{d}t} = \gamma I(t) - f R(t).
\end{align}

Besides these epidemic diffusion models introduced in this subsection, there also exist many different version of the epidemic diffusion models, which considers many other factors in the diffusion process, like the birth/death of individuals. It is also very common in the real-world online social networks, since new users will join in the social network, and existing users will also delete their account and get removed from the social network. Involving such factors will make the diffusion model more complex, and we will not introduce them here due to the limited space. More information about these different epidemic diffusion models is available in \cite{N02, PCVV15}.

\subsubsection{Heat Diffusion Models}

Heat diffusion is a well observed physical phenomenon. Generally, in a medium, heat will always diffuses from regions with a high temperature to the region with a lower temperature. Recently, many works have applied the heat diffusion to model the information propagation in online social networks. In this subsection, we will talk about the \textit{heat diffusion model} and introduce how to adapt it to model the information diffusion process in online social networks.

\noindent \textbf{General Heat Diffusion}

Throughout a geometric manifold, let function $f(x, t)$ denote the temperature at location $x$ at time $t$, and we can represent the initial temperature at different locations as $f_0(x)$. The heat flows with initial conditions can be described by the following second order differential equation
\begin{equation}
\begin{cases}
\frac{\partial f(x, t)}{\partial t} - \Delta f(x, t) = 0\\
f(x, 0) = f_0(x),
\end{cases}
\end{equation}
where $\Delta f(x, t)$ is a \textit{Laplace-Beltrami operator} on function $f(x, t)$. 

Many existing works on the heat diffusion studies are mainly focused on the heat kernel matrix. Formally, let $\mb{K}_t$ denote the heat kernel matrix at timestamp $t$, which describes the heat diffusion among different regions in the medium. In the matrix, entry $K_t(x, y)$ denotes the heat diffused from the original position $y$ to position $x$ at time $t$. However, it is very difficult to represent the medium as a regular geometry with a known dimension. In the next part, we will introduce how to apply the heat diffusion observations to model the information diffusion in the network-structured graph data. 

\noindent \textbf{Heat Diffusion Model}

Given a homogeneous network $G = (\mathcal{V}, \mathcal{E})$, for each node $u \in \mathcal{V}$ in the network, we can represent the information at $u$ in timestamp $t$ as $f(u, t)$. The initial information available at each of the node can be denoted as $f(u, 0)$. The information can be propagated among the nodes in the network if there exists a pipe (i.e., a link) between them. For instance, with a link $(u, v) \in \mathcal{E}$ in the network, information can be propagated between $u$ and $v$.

Generally, in the diffusion process, the amount of information propagated between different nodes in the network depends on (1) the difference of information available at them, and (2) the thermal conductivity-the heat diffusion coefficient $\alpha$. For instance, at timestamp $t$, we can represent the amount of information reaching nodes $u, v \in \mathcal{V}$ as $f(u, t)$ and $f(v, t)$. If $f(u, t) > f(v, t)$, information tends to propagate from $u$ to $v$ in the network, and the amount of information propagated is $\alpha \cdot \left(f(u, t) - f(v, t) \right)$, and the propagation direction will be reversed if $f(u, t) < f(v, t)$. The information amount changes at node $u$ at timestamps $t$ and $t + \Delta t$ can be represented as
\begin{equation}
\frac{f(u, t + \Delta t) - f(u, t)}{\Delta t} = - \sum_{v \in \Gamma(u)} \alpha \cdot \left(f(u, t) - f(v, t) \right).
\end{equation}

Let's use vector $\mb{f}(t)$ to represent the amount of information available at all the nodes in the network at timestamp $t$. The above information amount changes can be rewritten as
\begin{equation}
\frac{\mb{f}(t + \Delta t) - \mb{f}(t)}{\Delta t} = \alpha \mb{H} \mb{f}(t),
\end{equation}
where in the matrix $\mb{H} \in \mathbb{R}^{|\mathcal{V}| \times |\mathcal{V}|}$, entry $H(u, v)$ has value
\begin{equation}
H(u, v) =
\begin{cases}
1, &\mbox{ if } (u, v) \in \mathcal{E} \lor (v, u) \in \mathcal{E},\\
-D(u), & \mbox{ if } u = v,\\
0, & \mbox{ otherwise},
\end{cases}
\end{equation}
where $D(u)$ denotes the degree of node $u$ in the network.

In the limit case $\Delta t \to 0$, we can rewrite the equation as
\begin{equation}
\frac{\mathrm{d} \mb{f}(t)}{\mathrm{d}t} = \alpha \mb{H} \mb{f}(t). 
\end{equation}
Solving the function, we can represent the amount of information at each node in the network as
\begin{align}
\mb{f}(t) &= \exp^{t \alpha \mb{H}} \mb{f}(0)\\
&= \left( \mb{I} + \alpha t \mb{H} + \frac{\alpha^2 t^2}{2!} \mb{H}^2 + \frac{\alpha^3 t^3}{3!} \mb{H}^3 + \cdots \right) \mb{f}(0),
\end{align}
where term $\exp^{t \alpha \mb{H}}$ is called the diffusion kernel matrix, which can be expanded according to Taylor's theorem.

\subsection{Intertwined Diffusion Models}\label{sec:chap9_sec3_intertwined}

For the models introduced in the previous section, they are all proposed for modeling the diffusion of information in online social networks involving one single type of connections in propagating one type of information only. However, in the real world, multiple types of information can be propagated within the network simultaneously, relationships among which can be quite intertwined, including \textit{competitive}, \textit{complimentary} and \textit{independent}. Furthermore, within the networks, even the network structure is homogeneous but the social links among users may be associated with polarities indicating the relationship among the users. For instance, for some of the social links, they denote friendship, while for some of the links, they indicate the user pairs are enemies. Formally, the social network structure with polarities associated with the social links are called signed networks, where the link polarities can affect the information diffusion in them greatly.

In this section, we will introduce the \textit{intertwined diffusion models} to describe the information propagation process about both (1) the information entities with intertwined relationships, and (2) for network structures with links attaching different polarities. The models to be introduced in this section are based on \cite{asonam16_intertwined, icdcs17} respectively.

\subsubsection{Intertwined Diffusion Models for Multiple Topics}

Traditional information diffusion studies mainly focus on one single online social network and has extensive concrete applications in the real world, e.g., product promotion \cite{DMS10, NN12} and opinion spread \cite{CCCKLRSWWY11}. In the traditional viral marketing setting \cite{DR01, KKT03}, only one product/idea is to be promoted. However, in the real scenarios, the promotions of multiple products can co-exist in the social networks at the same time, which is referred to as the \textit{intertwined information diffusion problem}.

The relationships among the products to be promoted in the network can be very complicated. For example, in Figure~\ref{fig:chap9_sec3_relationships}, we show $4$ different products to be promoted in an online social network and HP printer is our target product. At the product level, the relationships among these products can be:

\begin{itemize}

\item \textit{independent}: promotion activities of some products (e.g., HP printer and Pepsi) can be \textit{independent} of each other.

\item \textit{competing}: products having common functions will \textit{compete} for the market share \cite{BKS07, CNWZ07} (e.g., HP printer and Canon printer). Users who have bought a HP printer are less likely to buy a Canon printer again.

\item \textit{complementary}: product cross-sell is also very common in marketing \cite{NN12}. Users who have bought a certain product (e.g., PC) will be more likely to buy another product (e.g., HP printer) and the promotion of PC is said to be \textit{complementary} to that of HP printer.

\end{itemize}

\begin{figure}[t]
\centering
    \begin{minipage}[l]{0.63\columnwidth}
      \centering
      \includegraphics[width=\textwidth]{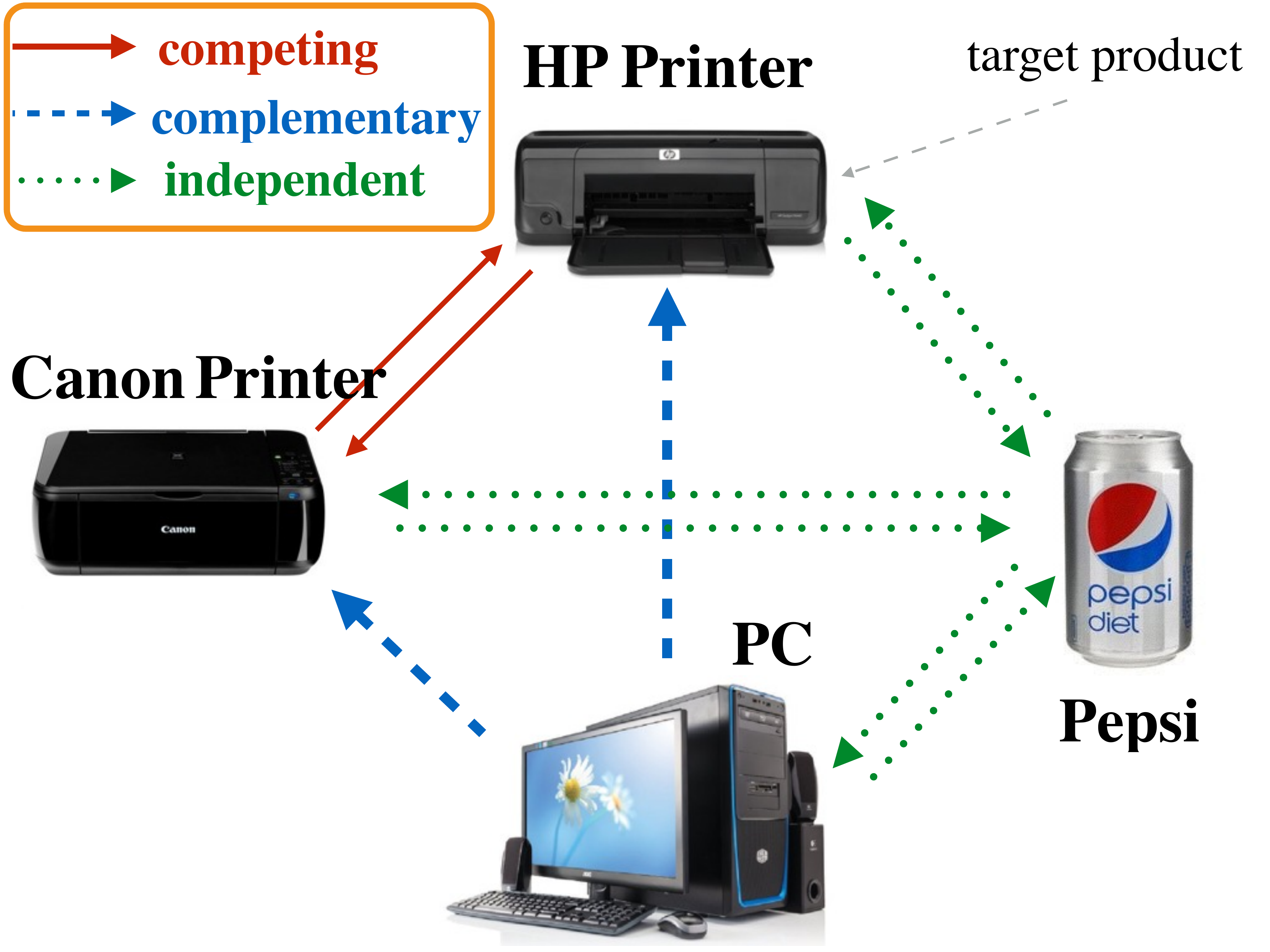}
    \end{minipage}
  \caption{Intertwined relationships among products.}\label{fig:chap9_sec3_relationships}
\end{figure}

In this section, we will study the information diffusion problem in online social networks, where multiple products are being promoted simultaneously. The relationships among these product can be obtained in advance via effective market research, which can be \textit{independent}, \textit{competitive} or \textit{complementary}. A novel information diffusion model inter\underline{T}wined \underline{L}inear \underline{T}hreshold ({\diffusion}) will be introduced in this section. {\diffusion} quantifies the \textit{impacts} among products with the \textit{intertwined threshold updating strategy} and can handle the intertwined diffusions of these products at the same time.

\noindent \textbf{Diffusion Setting Description and Concept Definition}

\begin{defn} 
(Social Network): An online \textit{social network} can be represented as $G = (\mathcal{V}, \mathcal{E})$, where $\mathcal{V}$ is the set of users and $\mathcal{E}$ contains the interactions among users in $\mathcal{V}$. The set of $n$ different products to be promoted in network $G$ can be represented as $\mathcal{P} = \{p^1, p^2, \cdots, p^n\}$.
\end{defn}

\begin{defn} 
(User Status Vector): For a given product $p^j \in \mathcal{P}$, users who are influenced to buy $p^j$ are defined to be ``\textit{active}'' to $p^j$, while the remaining users who have not bought $p^j$ are defined to be ``\textit{inactive}'' to $p^j$. User $u_i$'s status towards all the products in $\mathcal{P}$ can be represented as ``\textit{user status vector}'' $\mb{s}_i = (s_i^1, s_i^2, \cdots, s_i^n)$, where $s_i^j$ is $u_i$'s status to product $p^j$. Users can be activated by multiple products at the same time (even competing products), i.e., multiple entries in \textit{status vector} $\mb{s}_i$ can be ``\textit{active}'' concurrently.
\end{defn} 

\begin{defn} 
(Independent, Competing and Complementary Products): Let $P(s_i^j = 1)$ (or $P(s^j_i)$ for simplicity) denote the probability that $u_i$ is activated by product $p^j$ and $P(s_i^j | s_i^k)$ be the conditional probability given that $u_i$ has been activated by $p^k$ already. For products $p^j, p^k \in \mathcal{P}$, the promotion of $p^k$ is defined to be (1) \textit{independent} to that of $p^j$ if $\forall u_i \in \mathcal{V}$, $P(s_i^j | s_i^k) = P(s_i^j)$, (2) \textit{competing} to that of $p^j$ if $\forall u_i \in \mathcal{V}$, $P(s_i^j | s_i^k) < P(s_i^j)$, and (3) \textit{complementary} to that of $p^j$ if $\forall u_i \in \mathcal{V}$, $P(s_i^j | s_i^k) > P(s_i^j)$.
\end{defn} 

\noindent \textbf{TLT Diffusion Model}

To depict the intertwined diffusions of multiple independent/competing/complementary products, a new information diffusion model {\diffusion} is introduced in \cite{asonam16_intertwined}. In the existence of multiple products $\mathcal{P}$, user $u_i$'s influence to his neighbor $u_k$ in promoting product $p^j$ can be represented as $w^j_{i,k} \ge 0$. Similar to the traditional LT model, in {\diffusion}, the influence of different products can propagate within the network step by step. User $u_i$'s threshold for product $p^j$ can be represented as $\theta^j$ and $u_i$ will be activated by his neighbors to buy product $p^j$ if 
\begin{equation}
\sum_{u_l \in \Gamma_{out}(u_i)} w^j_{l,i} \ge \theta^j_i.
\end{equation}

Different from traditional LT model, in {\diffusion}, users in online social networks can be activated by multiple products at the same time, which can be either \textit{independent}, \textit{competing} or \textit{complementary}. As shown in Figure~\ref{fig:chap9_sec3_relationships}, we observe that users' chance to buy the HP printer will be (1) unchanged given that they have bought Pepsi (i.e., the \textit{independent} product of HP printer), (2) increased if they own PCs (i.e., the \textit{complementary} product of HP printer), and (3) decreased if they already have the Canon printer (i.e., the \textit{competing} product of HP printer).

To model such a phenomenon in {\diffusion}, the following \textit{intertwined threshold updating strategy} has been introduced in \cite{asonam16_intertwined}, where users' \textit{thresholds} to different products will change \textit{dynamically} as the influence of other products propagates in the network.

\begin{defn} 
(Intertwined Threshold Updating Strategy): Assuming that user $u_i$ has been activated by $m$ products $p^{\tau_1}$, $p^{\tau_2}$, $\cdots$, $p^{\tau_m} \in \mathcal{P} \setminus \{p^j\}$ in a sequence, then $u_i$'s \textit{threshold} towards product $p^j$ will be updated as follows:
\begin{align}
(\theta^j_i)^{\tau_1} &= \theta^j_i \frac{P(s^j_i)}{P(s^j_i | s^{\tau_1}_i)}, (\theta^j_i)^{\tau_2} = (\theta^j_i)^{\tau_1} \frac{P(s^j_i | s^{\tau_1}_i)}{P(s^j_i | s^{\tau_1}_i, s^{\tau_2}_i)}, \cdots\\
(\theta^j_i)^{\tau_m} &= (\theta^j_i)^{\tau_{m-1}} \frac{P(s^j_i | s^{\tau_1}_i, \cdots, s^{\tau_{m-1}}_i)}{P(s^j_i | s^{\tau_1}_i, \cdots, s^{\tau_{m-1}}_i, s^{\tau_{m}}_i)},
\end{align}
where $(\theta^j_i)^{\tau_k}$ denotes $u_i$'s threshold to $p^j$ after he has been activated by $p^{\tau_1}$, $p^{\tau_2}$, $\cdots$, $p^{\tau_k}$, $k \in \{1, 2, \cdots, m\}$.
\end{defn} 

In this section, we do not focus on the order of products that activate users \cite{CCCKLRSWWY11} and to simplify the calculation of the \textit{threshold updating strategy}, we assume only the most recent activation has an effect on updating current thresholds, i.e., 
\begin{equation}
\frac{P(s^j_i | s^{\tau_1}_i, \cdots, s^{\tau_{m-1}}_i)}{P(s^j_i | s^{\tau_1}_i, \cdots, s^{\tau_{m-1}}_i, s^{\tau_{m}}_i)} \approx \frac{P(s^j_i)}{P(s^j_i | s^{\tau_{m}}_i)}=\phi_i^{\tau_{m} \to j}.
\end{equation}

\begin{defn} 
(Threshold Updating Coefficient): Term $\phi_i^{l \to j} = \frac{P(s^j_i)}{P(s^j_i | s^l_i)}$ is formally defined as the ``\textit{threshold updating coefficient}'' of product $p^l$ to product $p^j$ for user $u_i$, where
\begin{equation}
\phi_i^{l \to j}
\begin{cases} 
< 1,  & \mbox{if }p^l\mbox{ is \textit{complementary} to }p^j, \\
= 1,  & \mbox{if }p^l\mbox{ is \textit{independent} to }p^j, \\
> 1,  & \mbox{if }p^l\mbox{ is \textit{competing} to }p^j.
\end{cases}
\end{equation}
\end{defn} 

The \textit{intertwined threshold updating strategy} can be rewritten based on the \textit{threshold updating coefficients} as follows:
\begin{equation}
(\theta^j_i)^{\tau_m} \approx \theta^j_i \cdot \phi_i^{\tau_{1} \to j} \cdot \phi_i^{\tau_{2} \to j} \cdots \phi_i^{\tau_{m} \to j}.
\end{equation}

\subsubsection{Diffusion Models for Signed Networks}\label{subsec:chap9_sec3_signed}

In recent years, signed networks \cite{zhang2016trust, DBLP:journals/corr/TangCAL15} have gained increasing attention because of their ability to represent diverse and contrasting social relationships. Some examples of such contrasting relationships include friends vs enemies \cite{WS12}, trust vs distrust \cite{YTYXL13}, positive attitudes vs negative attitudes \cite{YCZC13}, and so on. These contrasting relationships can be represented as links of different polarities, which result in signed networks. Signed social networks can provide a meaningful perspective on a wide range of social network studies, like \textit{user sentiment analysis} \cite{WPLP14}, \textit{social interaction pattern extraction} \cite{LHK10_2}, \textit{trustworthy friend recommendation} \cite{LHK10}, and so on.

Information dissemination is common in social networks~\cite{SZ11}. Due to the extensive social links among users, information on certain topics, e.g., politics, celebrities and product promotions, can propagate leading to a large number of nodes reporting the same (incorrect) observations rapidly in online social networks. In particular, the links in signed networks are of different polarities and can denote trust and distrust relationships among users~\cite{LXCGSL14}, which will inevitably have an impact on information propagation. 

\begin{figure}[t]
\centering
    \begin{minipage}[l]{0.6\columnwidth}
      \centering
      \includegraphics[width=\textwidth]{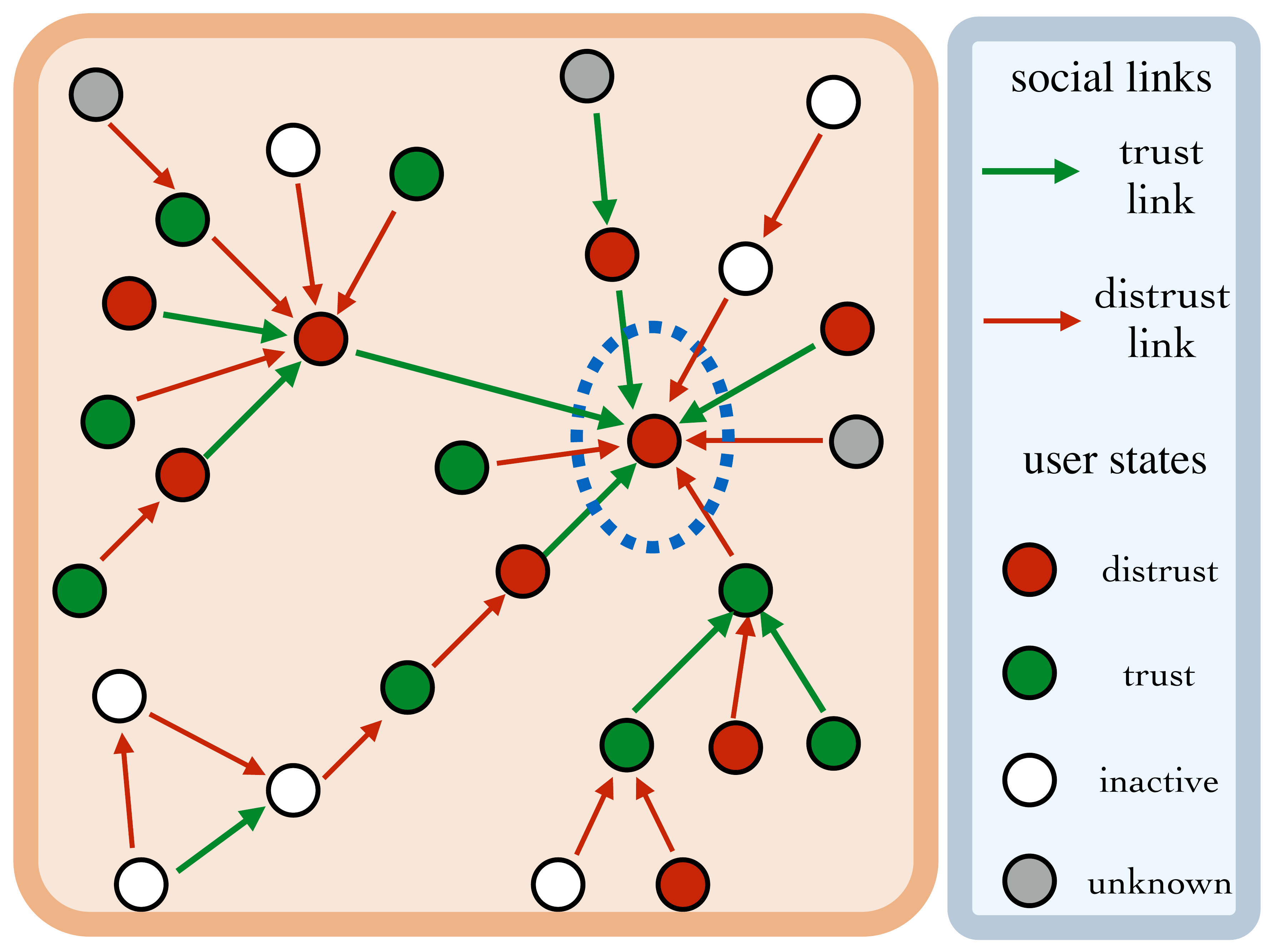}
    \end{minipage}
  \caption{Example of the information diffusion problem in signed networks.}\label{fig:chap9_sec3_example}
\end{figure}

In  Figure~\ref{fig:chap9_sec3_example}, an example is provided to help illustrate the information diffusion problem in signed networks more clearly. In the example, users are connected to one another with signed links, depending on their trust and distrust relations. It is noteworthy that the conventions used for the direction of information diffusion in this network are slightly different from traditional influence analysis, because they represent signed links. For instance, if Alice trusts (or follows) Bob, a directed edge exists from Alice to Bob, but the information diffusion direction will be from Bob to Alice. Via the signed links, inactive users in the network can get infected by certain information propagated from their neighbors with either a positive or negative opinion about the information (i.e., the green or red states in the figure). Considering the fact that it is often difficult to  directly identify  all the user infection states in real settings,  we allow for the possibility of some user states in the network  to  be unknown. Activated users can propagate the information to other users. In general, if a user is activated with a positive or negative opinion about the information, she might activate one or more of her incoming neighbors to trust or distrust the information, depending on the sign of the incoming link. 

In this subsection, we will focus on studying the information diffusion problem in signed networks. The edges in the network are directed and signed, and they represent trust or distrust relationships. For example, when node $i$ trusts or distrusts node $j$, we will have a corresponding positive or negative link from node $i$ to node $j$. In this setting, nodes are associated with states corresponding to a prevailing opinion about the truth of a fact. These states can be  drawn from $\{-1, +1, 0, ? \}$, where $+1$ indicates their agreement with a specific fact, $-1$ indicates their disagreement, $0$ indicates the fact that they have no opinion of the fact at hand, and $?$ indicates their opinion is unknown. The last of these states is necessary to model the fact that the states of many nodes in large-scale networks are often unknown.  Note that the use of multiple states of nodes in the network is different from traditional influence analysis.  Users are influenced with varying opinions of the fact in question, based on their observation of their neighbors (i.e., states of neighborhood nodes), and their trust or distrust of their neighbor's opinions (i.e., signs of links with them). This model is essentially a signed version of influence propagation models, because the sign of the link plays a critical role in how a specific bit of information is transmitted.

Most existing information diffusion models are designed for unsigned networks. In signed networks, information diffusion is also related to actor-centric  trust and distrust, in which notions of node states and  the signs on links play an important role. To depict how information propagates in the signed networks, a new diffusion model, namely \textit{asy\underline{M}metric \underline{F}lipping \underline{C}ascade} ({\sic}), has been introduced for signed networks in \cite{icdcs17}.

Traditional social networks are unsigned in the sense that the links are assumed, by default, to be positive links. Signed social networks are a generalization of this basic concept.

\begin{defn}
(Weighted Signed Social Network): Formally, a \textit{weighted signed social network} can be represented as a graph $G = (\mathcal{V}, \mathcal{E}, s, w)$, where $\mathcal{V}$ and $\mathcal{E}$ represents the nodes (users) and directed edges (social links), respectively.  In signed networks, each social link has its own polarity (i.e., the sign) and is associated with a weight indicating the intimacy among users, which can be represented with the mappings $s: \mathcal{E} \to \{-1, +1\}$ and $w: \mathcal{E} \to [0,1]$ respectively.
\end{defn}

As discussed in before,  we interpret the signs from a trust-centric point of view.  Information propagated among users is highly associated with the intimacy scores~\cite{sdm15} among them: information tends to propagate among close users. To represent the information diffusion process in trust-centric networks, the concept of \textit{weighted signed diffusion network} was defined as follows:

\begin{defn} (Weighted Signed Diffusion Network): Formally, given a \textit{signed social network} $G$, its corresponding \textit{weighted signed diffusion network} can be represented as $G_D = (\mathcal{V}_D, \mathcal{E}_D, s_D, w_D)$, where $\mathcal{V}_D = \mathcal{V}$ and $\mathcal{E}_D = \{(v, u)\}_{(u, v) \in \mathcal{E}}$. Diffusion links in $\mathcal{E}_D$ share the same sign and weight mappings as those in $\mathcal{E}$, which can be obtained via mappings $s_D: \mathcal{E}_D \to \{-1, +1\}$, $s_D(v, u) = s(u, v), \forall (v, u) \in \mathcal{E}_D$ and $w_D: \mathcal{E}_D \to [0,1]$, $w_D(v, u) = w(u, v), \forall (v, u) \in \mathcal{E}_D$. For any directed diffusion link $(u, v) \in \mathcal{E}_D$, we can represent its sign and weight to be $s_D(u, v)$ and $w_D(u, v)$ respectively.
\end{defn}

Note that we have reversed the direction of the links because of the trust-centric interpretation, in which information diffuses from A to B, when B trusts A. However, in networks with other semantic interpretations, this reversal does not need to be performed. The overall algorithm is agnostic to the specific preprocessing performed in order to fit a particular semantic interpretation of the signed network.

\noindent \textbf{MFC Diffusion Model}

The IC model, which  assumes that social links are all of the same polarity, works for unsigned networks, but it cannot be applied to signed networks with node states to reflects beliefs of different polarities. To overcome such a shortcoming, a novel diffusion model, {\sic}, will be introduced in this section.

\setlength{\textfloatsep}{0pt}
\begin{algorithm}[t]
\small
\caption{MFC Information Diffusion Model}
\label{alg:chap9_sec3_sic}
\begin{algorithmic}[1]
	\REQUIRE input rumor initiators $\mathcal{I}$ with states $\mathcal{S}$\\
	\qquad diffusion network $G_D = (\mathcal{V}_D, \mathcal{E}_D, s_D, w_D)$
\ENSURE  infected diffusion network $G_I$

\STATE	{initialize infected user set $\mathcal{U} = \mathcal{I}$, state set $\mathcal{S}_{\mathcal{U}} = \mathcal{S}$}
\STATE	{let recently infected user set $\mathcal{R} = \mathcal{I}$}
\WHILE	{$\mathcal{R} \neq \emptyset$}
\STATE	{new recently infected user set $\mathcal{N} = \emptyset$}
\FOR		{$u \in \mathcal{R}$}
\STATE	{let the set of users that $u$ can activate to be $\Gamma(u)$}
\FOR		{$v \in \Gamma(u)$}
\IF		{$s(v) = 0$ or \big($s_D(u, v) = +1$ and $s(u) \neq s(v)$\big)}
\IF		{$s_D(u, v) = +1$}
\STATE	{$p = \min\{1.0, \alpha \cdot w_D(u, v)\}$}
\ELSE
\STATE	{$p = w_D(u, v)$}
\ENDIF
\IF	{u activates $v$ with probability $p$}
\STATE	{$\mathcal{U} = \mathcal{U} \cup \{v\}$, $\mathcal{S}_{\mathcal{U}} = \mathcal{S}_{\mathcal{U}} \cup \{s(v) = s(u) \cdot s_D(u, v)\}$}
\STATE	{$\mathcal{N} = \mathcal{N} \cup \{v\}$}
\ENDIF
\ENDIF
\ENDFOR
\ENDFOR
\STATE	{$\mathcal{R} = \mathcal{N}$}
\ENDWHILE
\STATE	{extract infected diffusion network $G_I$ consisting of infected users $\mathcal{U}$}

\end{algorithmic}
\end{algorithm}

The signs associated with diffusion links denote the ``positive'' and ``negative'' relationships, e.g., trust and distrust, among users. In everyday life, people tend to believe information from people they trust and not believe the information from those they distrust. For example, if someone we trust says that ``Hillary Clinton will be the new president'', we believe it to be true. However, if someone we distrust says the same thing, we might not believe it. In addition, when receiving contradictory messages, information obtained from the trusted  people is usually given higher weights. In other words, the effects of trust and distrust diffusion links are asymmetric in activating users. For instance, when various actors assert that ``Hillary Clinton will be the new president", we may tend to  follow those we trust, even though the distrusted ones also say it. In addition, if someone we distrust says that ``Hillary Clinton will be the new president", we may think it to be false and will not believe it. However, after being activated to distrust it, if we are exposed to contradictory information from a trusted party, we might be willing to change our minds. To model such cases, which are unique to signed and state-centric networks, \cite{icdcs17} proposes to follow a number of basic principles in the {\sic} model, (1) the effects of positive links in activating users is boosted to give them higher weights in activating users, and (2) users who are activated already will stay active in the subsequential rounds but their activation states can be flipped to follow the people they trust.

In {\sic}, users have $3$ unique known states in the information diffusion process: $\{+1, -1, 0\}$ (i.e., trust, distrust and inactive respectively). Users with unknown states are automatically taken into account during the model construction process by assuming states as necessary. For simplicity, we use $s(\cdot)$ to represent both the sign of links as well as the states of users. If user $u$ trusts the information, then user $u$ is said to have a positive state $s(u) = +1$ towards the information. The initial states of all users in {\sic} are assigned a value of $0$ (i.e., inactive to the information). A set of information seed users $\mathcal{I} \subseteq \mathcal{V}$ activated by the information at the very beginning will have their own attitudes towards the information based on their judgements, which can be represented with $\mathcal{S} = \{+1, -1\}^{|\mathcal{I}|}$. Information seed users in $\mathcal{I}$ spread the information to other users in signed networks step by step. At step $\tau$, user $u$ (activated at $\tau-1$) is given only one chance to activate (1) inactive neighbor $v$, as well as (2) active neighbor $v$ but $v$ has different state from $u$ and $v$ trusts $u$, with the boosted success probability $\overline{w_D}(u, v)$, where $\overline{w_D}(v, u) \in [0, 1]$ can be represented as
\begin{equation}
\overline{w_D}(v, u) =
\begin{cases}
\min\{\alpha \cdot w_D(v, u), 1\}  & \mbox{if }s_D(v, u) = +1, \\
w_D(v, u), & \mbox{otherwise.}
\end{cases}
\end{equation}
In the above equation, parameter $\alpha > 1$ denotes the boosting of information from $u$ to $v$ and is called the \textit{asymmetric boosting coefficient}.

If $u$ succeeds, $v$ will become active in step $\tau + 1$, whose states can be represented as $s(v) = s(u) \cdot s(u, v)$. For example, if user $u$ thinks the information to be real (i.e., $s(u) = +1$) and $v$ trusts $u$ (i.e., $s(u, v) = +1$), once $v$ get activated by $u$ successfully, the state of $v$ will be $s(v) = +1$ (i.e., believe the information to be true). Otherwise, $v$ will keep its original state (either inactive or activated) and $u$ cannot make any further attempts to activate $v$ in subsequent rounds. All activated users will stay active in the following rounds and the process continues until no more activations are possible.

\begin{figure}[t]
\centering
    \begin{minipage}[l]{0.8\columnwidth}
      \centering
      \includegraphics[width=\textwidth]{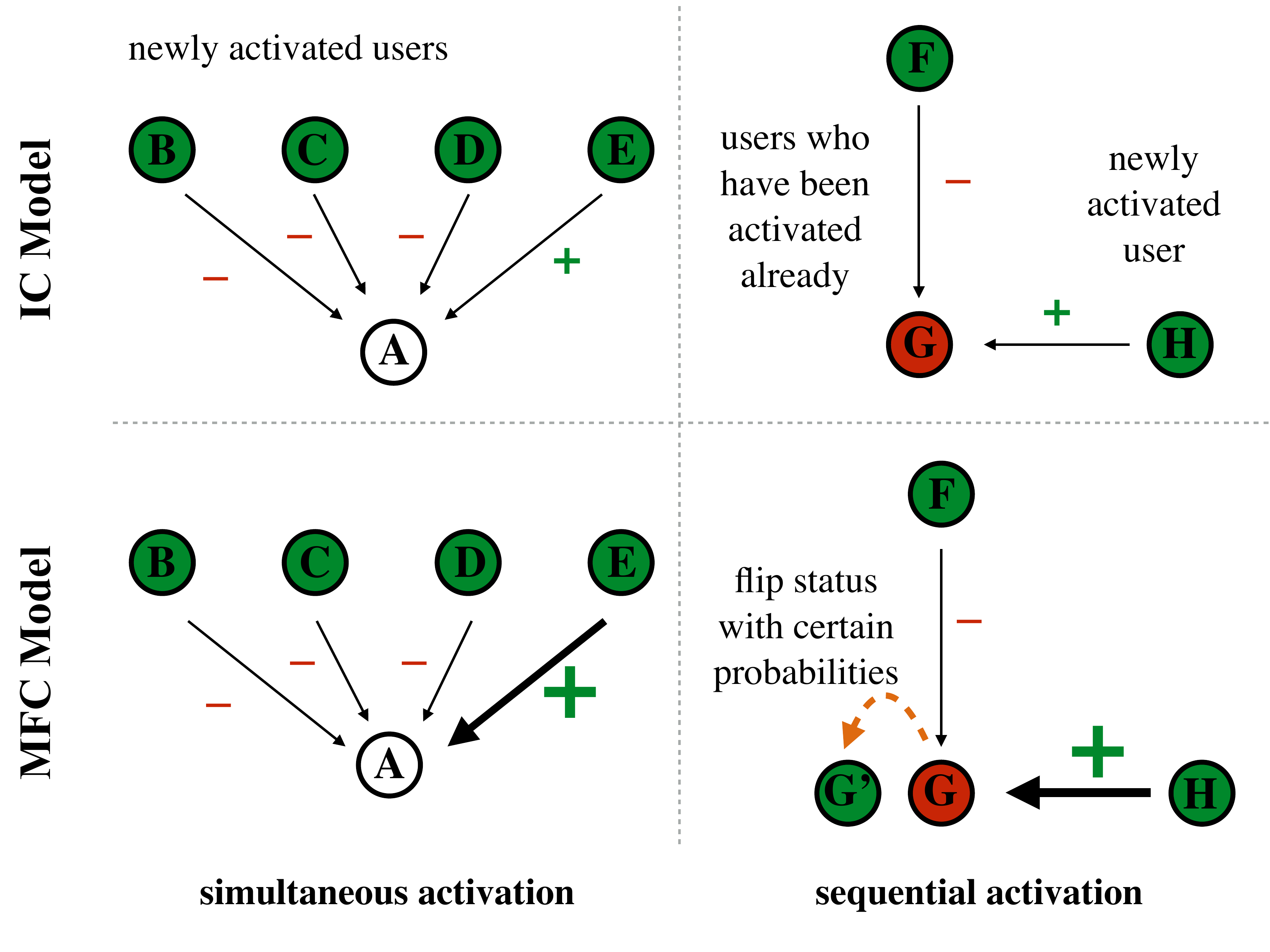}
    \end{minipage}
  \caption{Example of the binary tree transformation.}\label{fig:chap9_sec3_case}
\end{figure}

{\sic} can model the information diffusion process in signed social networks much better than traditional diffusion models, such as IC. To illustrate the advantages of {\sic}, we also give an example in Figure~\ref{fig:chap9_sec3_case}, where two different cases: ``simultaneous activation'' (i.e., the left two plots) and ``sequential activation'' (i.e., the right two plots) are shown. In the ``simultaneous activation'' case, multiple users ($B$, $C$, $D$ and $E$) are all just activated at step $\tau$, who all think a information to be true and at step $\tau+1$, $B$-$E$ will activate their inactive neighbor $A$. Among these users, $A$ trusts $E$ and distrusts the remaining users. In traditional IC models, signs on links are ignored and $B$-$E$ are given equal chance to activate $A$ in random order with activation probabilities $w_D(\cdot, A), \cdot \in \{B, C, D, E\}$. However, in the {\sic} model, signs of links are utilized and the activation probability of positive diffusion $(E, A)$ will be boosted and can be represented as $\min\{\alpha \cdot w_D(E, A), 1\}$. As a result, user $A$ is more likely to be activated by $E$ in {\sic}. Meanwhile, in the sequential activation case, once a user (e.g., $F$) succeeds in activating $G$, $G$ will remain active and other users (e.g., H) cannot reactivate $A$ any longer in traditional IC model. However, in the {\sic} model, we allow users to flip their activation state by people they trust. For example, if $G$ has been activated by $F$ with state $s(G) = -1$ already, the trusted user $H$ can still have the chance to flip $G$'s state with probability $\min\{\alpha \cdot w_D(H, G), 1\}$. The pseudo-code of the {\sic} diffusion model is provided in Algorithm~\ref{alg:chap9_sec3_sic}.

\subsection{Inter-Network Information Diffusion via Network Coupling}

The information diffusion models introduced in the previous sections are mostly based one single network, assuming that information will only propagate within the network only. However, in the real-world, users are involved in multiple social sites simultaneously, and cross-site information diffusion is happening all the time. Users as the bridges, they can receive information from one social sites, and share with their friends in another network. Sometimes, due to the social network settings, the activities happing in one social site (e.g., Foursquare) can be reposted to other social sites (e.g., Twitter) automatically.

In this section and the following two sections, we will study the information diffusion across multiple social sites. Several different existing cross-network information diffusion models will be introduced. Generally, different networks will great different information diffusion sources, and interactions available among users in each of the sources can all propagate information among users. Two cross-network information diffusion models based on \textit{network coupling} and \textit{random walk}. Meanwhile, in each of the diffusion sources, there usually exist different types of diffusion channels, since users can interact with each other via different types of services provided by the network service providers. A new diffusion model named {\muse} will introduced to depict how information belonging to different topics diffuses via multiple channels across multiple sources.

Cross-network information sharing and reposting renders the inter-network information diffusion ubiquitous and very common in the real-world online social networks. By involving in multiple online social networks simultaneously, users can also be exposed to more information from multiple social sites at the same time. Generally, once a user has been activated in one of the social site, the user account owner will receive the information and can diffuse it to other users in the other networks. The network coupling model proposes to combine multiple social networks together, and treat the information diffusion in each of the networks independently. 

\subsubsection{Single Network Diffusion Model}

Formally, let $G^{(1)}, G^{(2)}, \cdots, G^{(k)}$ denote the $k$ online social networks that we are focusing on in the information diffusion model, whose network structures are all homogeneous involving users and friendship links only. For each of the network, e.g., $G^{(i)}$, we can represent its structure as $G^{(i)} = (\mathcal{V}^{(i)}, \mathcal{E}^{(i)})$, where $\mathcal{V}^{(i)}$ denotes the set of users in the network. Information diffusion process in network $G^{(i)}$ can be modeled with some existing models. In this part, we will use the LT model as the base diffusion model for each of the networks.

Based on network $G^{(i)}$, each user $u$ in the network is associated with a threshold $\theta_u^{(i)}$ indicating the minimal amount of required information to activate the users. Meanwhile, the amount of information sent between the users (e.g., $u$ and $v$) can be denoted as weight $w_{u, v}^{(i)}$, whose value can be determined in the same way as the LT model introduced before. For an inactive user $u$, he/she can be activated iff the amount of information propagated from their friends is greater than $u$'s threshold, i.e.,
\begin{equation}
\sum_{v \in \Gamma(u)} I(v, t) \cdot w^{(i)}_{v, u} \ge \theta_u^{(i)},
\end{equation}
where $\Gamma(u)$ represents the neighbors of user $u$ and $I(v, t)$ indicates whether $v$ has been activated or not at time $t$.

\subsubsection{Network Coupling Scheme}

Generally, in the real world, among these $k$ different online social sites $G^{(1)}, G^{(2)}, \cdots, G^{(k)}$, if there exists one network $G^{(i)}$, in which the above equation holds, user $u$ will become active. In other words, to determine whether user $u$ has been activated or not, we need to check his/her status in all these $k$ networks one by one. To reduce the activation checking works, in the lossy network coupling scheme, the activation checking criterion is relaxed to
\begin{equation}
\sum_{i = 1}^k \alpha^{(i)} \cdot \sum_{v \in \Gamma(u)} I(v, t) \cdot w^{(i)}_{v, u} \ge \sum_{i = 1}^k \alpha^{(i)} \cdot \theta_u^{(i)},
\end{equation}
where $\alpha^{(1)}, \alpha^{(2)}, \cdots, \alpha^{(k)} > 0$ denote the parameters representing the importance of different networks. 

\begin{theorem}
Given the $k$ networks, $G^{(1)}, G^{(2)}, \cdots, G^{(k)}$, if equation
\begin{equation}
\sum_{i = 1}^k \alpha^{(i)} \cdot \sum_{v \in \Gamma(u)} I(v, t) \cdot w^{(i)}_{v, u} \ge \sum_{i = 1}^k \alpha^{(i)} \cdot \theta_u^{(i)},
\end{equation}
holds, user $u$ will be activated.
\end{theorem}

\begin{proof}
The theorem can be proven with by contradiction. Let's assume the equation holds but $u$ has not been activated in networks $G^{(1)}, G^{(2)}, \cdots, G^{(k)}$, then we have
\begin{equation}
\sum_{v \in \Gamma(u)} I(v, t) \cdot w^{(i)}_{v, u} < \theta_u^{(i)},
\end{equation}
hold for all these networks. 

By times both sides of the inequality with a positive weight $\alpha^{(i)}$, and sum the equations across all these $k$ networks, we have
\begin{equation}
\sum_{i = 1}^k \alpha^{(i)} \cdot \sum_{v \in \Gamma(u)} I(v, t) \cdot w^{(i)}_{v, u} < \sum_{i = 1}^k \alpha^{(i)} \cdot \theta_u^{(i)},
\end{equation}
which contradicts the equation in the theorem.

Therefore, if the new activation criterion holds, user $u$ will be activated.
\end{proof}

The relaxed activation criterion is actually a sufficient but not necessary condition when determining whether $u$ is activated or not. In some cases, $u$ has already been activated in some of the networks, but the criterion cannot meet, which will lead to some latency in status checking. One way to solve the problem is to assign an appropriate weight $\alpha^{(i)}$ by increasing its value proportion to $\sum_{v \in \Gamma(u)} I(v, t) \cdot w^{(i)}_{v, u} - \theta_u^{(i)}$. In the special case that user $u$ can be activated in network $G^{(i)}$ already, we can assign the weight $\alpha^{(i)}$ with a very large value, where $\alpha^{(i)} \gg \alpha^{(j)}, j \in \{1, 2, \cdots, k\}, j \neq i$ is way larger compared with the remaining networks. So far, there don't exist any methods to adjust the parameters automatically, and heuristics are applied in most of the cases.

\subsection{Random Walk based Diffusion Model}\label{sec:chap9_sec5_random_walk}

Different online social networks usually have their own characteristics, and users tend to have different status regarding the same information. For instance, information about personal entertainments (like movies, pop stars) can be widely spread among users in Facebook, and users interested in them will be activated very easily and also share the information to their friends. However, such a kind of information is relatively rare in the professional social network LinkedIn, where people seldom share personal entertainment to their colleagues, even though they may have been activated already in Facebook. What's more, the structures of these online social networks are usually heterogeneous, containing many different kinds of connections. Besides the direct follow relationships among the users, these diverse connections available among the users may create different types of communication channels for information diffusion. To model such an observation in information diffusion across multiple heterogeneous online social sites, in this part, we will introduce a new information diffusion model, {\ipath}, based on random walk. 

\subsubsection{Intra-Network Propagation}
The traditional research works on homogeneous networks assume that information can only be spread by the social links among users. If user $v$ follows user $u$, $(v,u) \in \mathcal{E}$ (where $\mathcal{E}$ is the edge set), the message can spread from $u$ to $v$, i.e. $u \to v$. However in a heterogeneous network, multi-typed and interconnected entities, such as images, videos and locations, can create various information propagation relations among users. For instance, if user $u$ recommends a good restaurant to his friend $v$ by checking in at this place, information will flow from $u$ to $v$ through the location entity $l$, which can be expressed by $u \xrightarrow[l]{check-in} v$. Similarly, we can represent the information diffusion routes among users via other information entities, which can be formally represented as the diffusion route set $\mathcal{R} = \{r_1, r_2, \dots, r_m\}$, where $m$ is the route number.

According to each diffusion route, we can represent the connections among users as an adjacency matrix. We can take the source network $G^{(s)}=(\mathcal{V}^{(s)}, \mathcal{E}^{(s)})$ as an example. For any diffusion route $r_i \in \mathcal{R}$, the adjacency matrix of $r_i$ will be $\mb{A}^{(s)}_i \in \mathbb{R}^{|\mathcal{V}^{(s)}|\times|\mathcal{V}^{(s)}|}$, where $A^{(s)}_i(u,v)$ is a binary-value variable and $A^{(s)}_i(u, v)=1$ iff $u$ and $v$ are connected with each other via relation $r_i$. The weighted diffusion matrix can be represented as the normalization of $\mb{W}^{(s)}_i = \mb{A}^{(s)}_i \mb{D}^{-1}$, where $\mb{D}^{-1}$ is a diagonal matrix with $D(u,u)=\sum_{v}^{|\mathcal{V}^{(s)}|} A^{(s)}_i(v,u)$, denoting the in-degree of $u$. The entry ${W}^{(s)}_i(u,v)$ denotes the probability of going from $v$ to $u$ in one step. In a similar way, we can represent the weighted diffusion matrices for other relations, which altogether can be represented as $\{\mb{W}^{(s)}_1, \mb{W}^{(s)}_2, \dots, \mb{W}^{(s)}_m\}$. To fuse the information diffused from different relations, {\ipath} will linearly combine these weighted matrices as follows:
\begin{equation}\label{eq:aggregate}
	\mb{W}_s = \lambda_1 \times \mb{W}^{(s)}_1+ \lambda_2 \times \mb{W}^{(s)}_2+ \cdots + \lambda_m \times \mb{W}^{(s)}_m,
\end{equation} 
where $\lambda_i$ denotes the aggregation weight of matrix corresponding to relation $r_i$. In real scenarios, different relations play different roles in the information propagation for different users. However, to simplify the settings, in {\ipath}, all these relations are treated to be equally important, and the aggregated matrix $\mb{W}^{(s)}$ takes the average of all these weighted diffusion matrices. In a similar way, we can define the weight matrix $\mb{W}^{(t)}$ of the target network $G^{(t)}$. 

\subsubsection{Inter-Network Propagation}
Across the aligned networks, information can propagate not only within networks but also across networks.
Based on the known anchor links between networks $G^{(t)}$ and $G^{(s)}$, i.e., set $\mathcal{A}^{(s,t)}$, we can define the binary adjacency matrix $\mb{A}^{(s \rightarrow t)} \in \mathbb{R}^{|\mathcal{V}^{(s)}|\times|\mathcal{V}^{(t)}|}$, where $A^{(s \rightarrow t)} (u,v) = 1$ if $(u^{(s)}, v^{(t)}) \in \mathcal{A}^{(s,t)}$. {\ipath} assumes that each anchor user in $G^{(s)}$ only has one corresponding account in $G^{(t)}$. Therefore $\mb{A}^{(s \rightarrow t)}$ has been normalized and the weight matrix $\mb{W}^{(s \rightarrow t)}  =  \mb{A}^{(s \rightarrow t)}$, denoting the chance of information propagating from $G^{(s)}$ to $G^{(t)}$. Furthermore, we can represent the weighted diffusion matrix from networks $G^{(t)}$ to $G^{(s)}$ as $\mb{W}^{(t \rightarrow s)} = (\mb{W}^{(s \rightarrow t)})^\top$, considering that the anchor links are undirected. 

\subsubsection{The IPATH Information Propagation Model}

\begin{figure}
	\centering
	\includegraphics[width=0.2\textwidth]{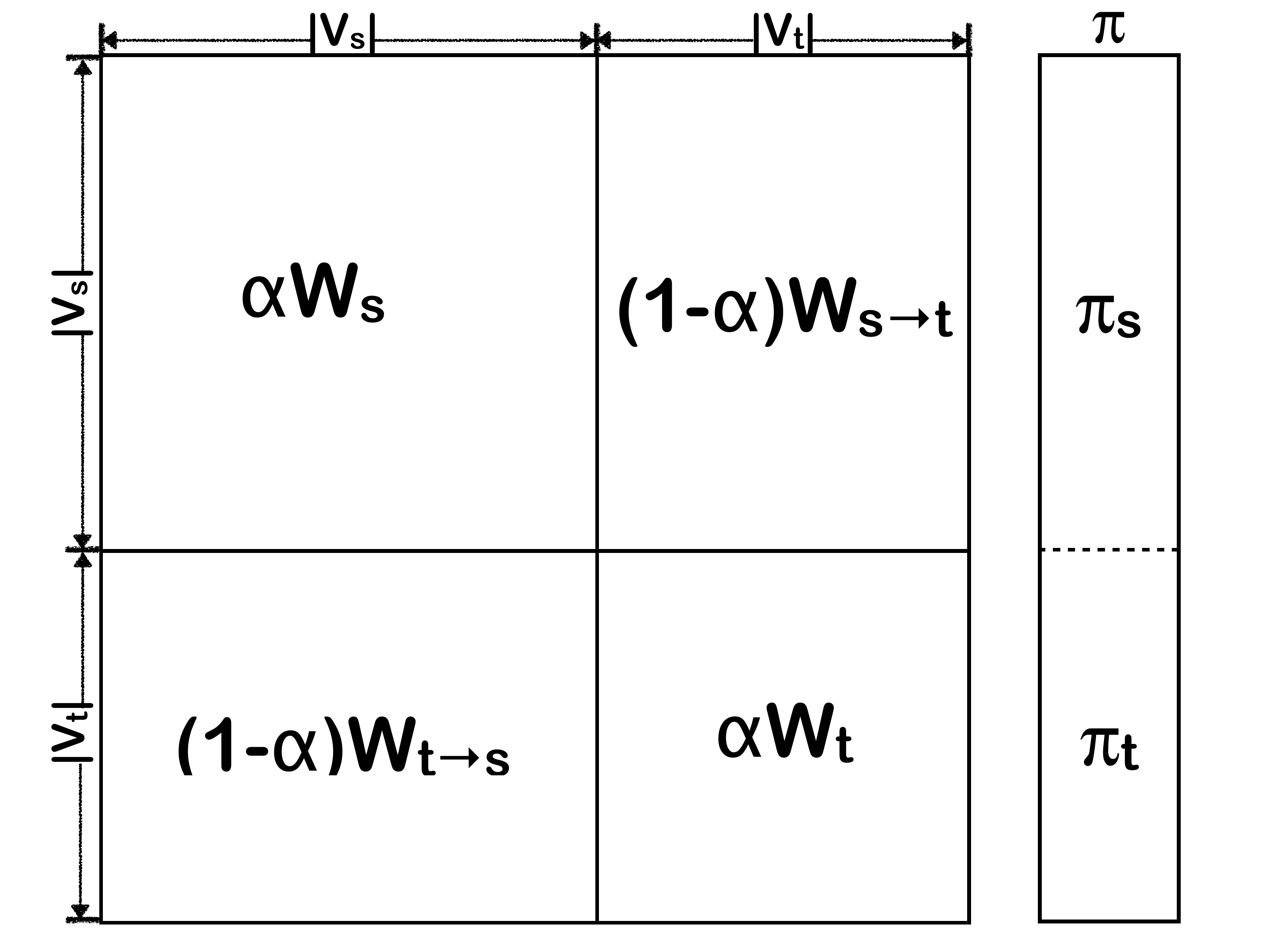}
	\caption{The weight matrix and the information distribution vector}
	\label{fig:chap9_sec5_matrix}
\end{figure}

Both the intra-network propagation relations, represented by weight matrices $\mb{W}^{(s)}$ and $\mb{W}^{(t)}$ in networks $G^{(s)}$ and $G^{(t)}$ respectively, and the inter-network propagation relations, represented by weight matrix $\mb{W}^{(s \rightarrow t)}$ and $\mb{W}^{(t \rightarrow s)}$, have been constructed already in the previous subsection. 
As shown in Figure \ref{fig:chap9_sec5_matrix}, to model the cross-network information diffusion process involving both the intra- and inter-network relations simultaneously, {\ipath} proposes to combine these weighted diffusion matrices to build an integrated matrix  $\mb{W} \in \mathbb{R}^{(|\mathcal{V}^{(s)}|+|\mathcal{V}^{(t)}|)^2}$. In the integrated matrix $\mb{W}$, the parameter $\alpha \in [0, 1]$ denotes the probability that the message stay in the original network, thus $1-\alpha$ represents the chance of being transmitted across networks (i.e., the probability of activated anchor user passing the influence to the target network).
In real scenarios, the probabilities for different users to repost information across aligned networks can be quite diverse. However, to simplify the problem setting, in {\ipath}, these probabilities are unified with parameter $\alpha$.

Let vector $\mb{\pi}_k \in \mathbb{R}^{(|\mathcal{V}^{(s)}|+|\mathcal{V}^{(t)}|)}$ represent the information that users in $G^{(s)}$ and $G^{(t)}$ can receive after $k$ steps. As shown in Figure \ref{fig:chap9_sec5_matrix}, vector $\pi_k$ consists of two parts $\pi_k = [\pi_k^{(s)}, \pi_k^{(t)}]$, where $\pi_k^{(s)} \in \mathbb{R}^{|\mathcal{V}^{(s)}| }$ and $\pi_k^{(t)} \in \mathbb{R}^{|\mathcal{V}^{(t)}| }$. The initial state of the vector can be denoted as $\pi_0$, which is defined based on the seed user set $\mathcal{Z}$ with function $g(\cdot)$ as follows:
\begin{equation}
\begin{aligned}
\pi_{0} &= g(\mathcal{Z}), \mbox{where } \pi_0[u] &= 
\begin{cases}
     1  \quad \text{if } u \in \mathcal{Z},\\
     0  \quad \text{otherwise.}\\
\end{cases}
\end{aligned}
\end{equation}
Seed set $\mathcal{Z}$ can also be represented as $\mathcal{Z} = g^{-1}(\pi_0)$. Users from $G^{(s)}$ and $G^{(t)}$ both have the chance of being selected as seeds, but when the structure information of $G^{(t)}$ is hard to obtain, the seed users will be only chosen from $G^{(s)}$. 
In {\ipath}, the information diffusion process is modeled by \textit{random walk}, because it is widely used in which the total probability of the diffusing through different relations remains constant 1 \cite{TFP06, GLSVT11}. Therefore, in the information propagation process, vector $\pi$ will be updated stepwise with the following equation:
\begin{equation}
\label{eq:update}
	{\pi}^{(k+1)} =(1- \alpha)\times \mb{W} {\pi}_k + \alpha \times \pi_0,
\end{equation}
where constant $\alpha$ denotes the probability of returning to the initial state.
By keeping updating $\pi$ according to (\ref{eq:update}) until convergence, we can present the stationary state of vector $\pi$ to be $\pi^*$,
\begin{equation}
\label{eq:convergence}
\pi^* =\alpha[\mb{I} -(1-\alpha) \mb{W}]^{-1}\pi_0,
\end{equation}
where matrix $\mb{I} \in \{0,1\}^{(|\mathcal{V}^{(s)}|+|\mathcal{V}^{(t)}|)\times(|\mathcal{V}^{(s)}|+|\mathcal{V}^{(t)}|)}$ is an identity matrix. The value of entry $\pi^*[u]$ denotes the activation probability of $u$, and user $u$ will be activated if $\pi^*[u] \geq \theta$, where $\theta$ denotes the threshold of accepting the message. In {\ipath}, parameter $\theta$ is randomly sampled from range $[0, \theta_{bound}]$. The threshold bound $\theta_{bound}$ is a small constant value, as the amount of information each user can get at the stationary state in {\ipath} can be very small (which is set as 0.01 in the experiments). In addition, we can further represent the activation status of user $u$ as vector $\pi'$, where 
\begin{equation}
\label{eq:threshold}
\pi'[u] = 
\begin{cases}
     1  \quad \text{if } \pi^*[u] \geq \theta,\\
     0  \quad \text{otherwise.}\\
\end{cases}
\end{equation}
In Equation (\ref{eq:threshold}), $\pi'[u] = 1$ denotes that user $u$ is activated. In practice, the value of $\pi^*[u]$ is usually in $[0,1]$ when the networks are sparse and the size of the seed set is small, and it can be represented approximately as following: 
\begin{equation} \pi'[u]  \approx \lfloor  \pi^*[u] - \theta+1 \rfloor. \end{equation}

Based on this, we define the mapping function $h$ between two vectors, where the floor function is applied to each element in the vector, i.e.,
\begin{equation} \pi' = h(\pi^*) = \lfloor \pi^* + \mathbf{c} \rfloor, \end{equation}
where $\mathbf{c}$ is a constant vector where each entry equals to $1-\theta$. To calculate the final number of activated users in $G^{(t)}$, we define a $(|\mathcal{V}^{(s)}|+|\mathcal{V}^{(t)}|)$-dimension constant vector  $\mathbf{b}= [0,0,\cdots,0,1,1,\cdots,1]$, where the number of $0$ is $|\mathcal{V}^{(s)}|$ and the number of $1$ is $|\mathcal{V}^{(t)}|$. Thus the influence function of the {\ipath} model can be denoted as
\begin{equation}
\label{eq:objective}
\begin{aligned}
\hspace{-5pt}\sigma(\mathcal{Z}) = \mathbf{b}\cdot h(\pi^*) = \mathbf{b} \cdot h\left( a[I-(1-a)\mb{W}]^{-1} \cdot g(\mathcal{Z}) \right),
\end{aligned}
\end{equation}
which can effectively compute the number of users who could be activated by the model based on the seed user set $\mathcal{Z}$.

\subsection{MUSE Model across Online and Offline World}\label{sec:chap9_sec6_muse}

Besides the online world, information can actually propagate within the online and offline world simultaneously. In this section, we will use the workplace as one example to illustrate the information diffusion via both the online and offline world simultaneously. On average, people nowadays need to spend more than $30\%$ of their time at work everyday. According to the statistical data in \cite{K00}, the total amount of time people spent at workplace in their life is tremendously large. For instance, a young man who is 20 years old now will spend $19.1\%$ of his future time working \cite{K00}. Therefore, workplace is actually an easily neglected yet important social occasion for effective communication and information exchange among people in our social life.

Besides the traditional offline contacts, like face-to-face communication, telephone calls and messaging, to facilitate the cooperation and communications among employees, a new type of online social networks named Enterprise Social Networks (ESNs) has been launched inside the firewalls of many companies \cite{kdd15, cikm15}. A representative example is Yammer, which is used by over $500,000$ leading businesses around the world, including $85\%$ of the Fortune $500$\footnote{https://about.yammer.com/why-yammer/}. Yammer provides various online communication services for employees at workplace, which include instant online messaging, write/reply/like posts, file upload/download/share, etc. In summary, the communication means existing among employees at workplaces are so diverse, which can generally be divided into two categories \cite{TQBGB10}: (1) offline communication means, and (2) online virtual communication means. 

\begin{figure}[t]
\centering
    \begin{minipage}[l]{0.8\columnwidth}
      \centering
      \includegraphics[width=\textwidth]{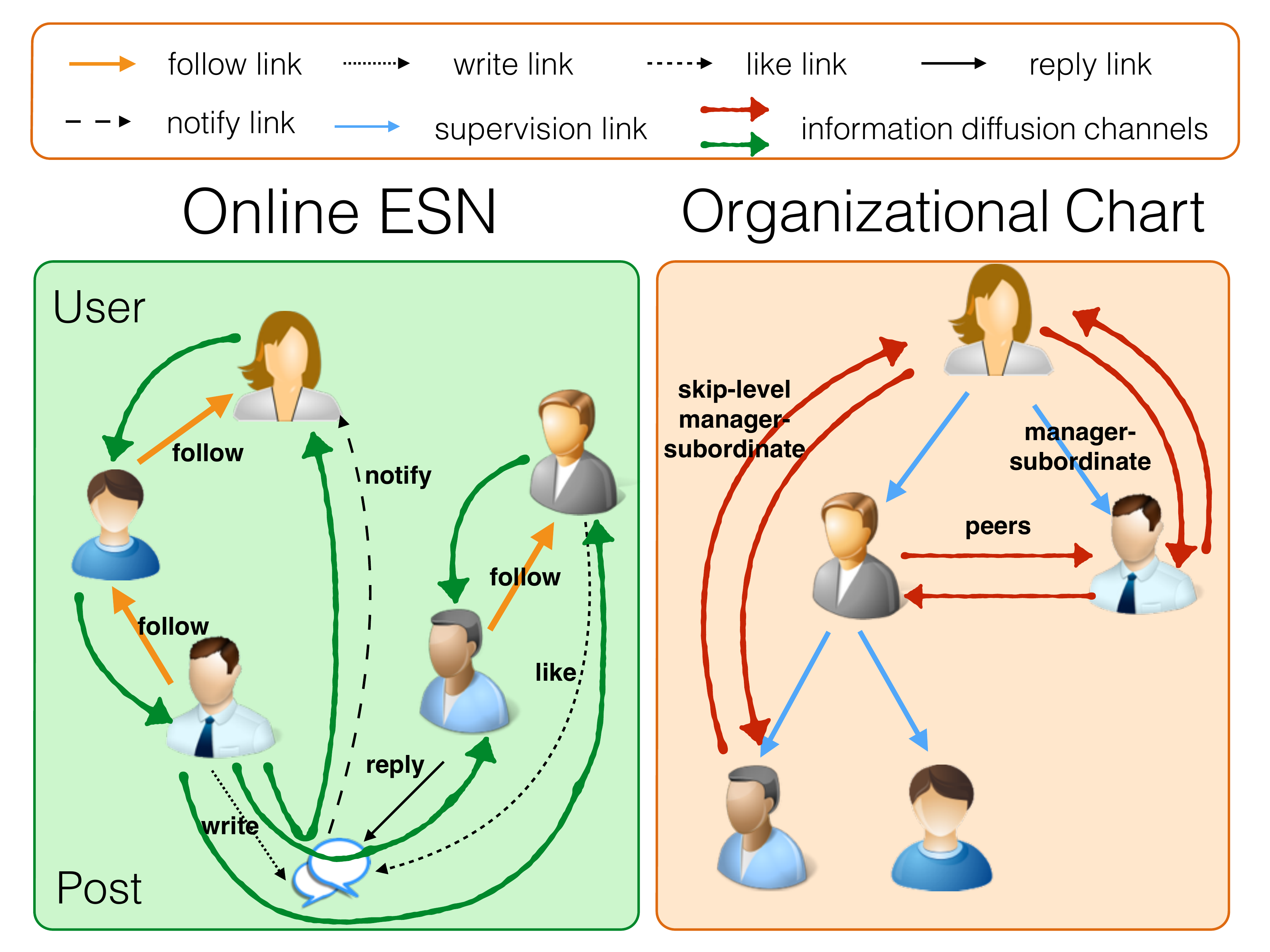}
    \end{minipage}
  \caption{An example of information diffusion at workplace.}\label{fig:chap9_sec6_example}
\end{figure}

In this section, we will study how information diffuses via both online and offline communication means among employees at workplace. To help illustrate the problem more clearly, we also give an example in Figure~\ref{fig:chap9_sec6_example}. The left plot of Figure~\ref{fig:chap9_sec6_example} is about an online ESN, employees in which can perform various social activities. For instances, employees can follow each other, can write/reply/like posts online, and posts written by them can also @certain employees to send notifications, which create various online information diffusion channels (i.e., the green lines) among employees. Meanwhile, the relative management relationships among the employees in the company can be represented with the organizational chart (i.e., the right plot), which is a tree-structure diagram connecting employees via supervision links (from managers to subordinates). Colleagues who are physically close in the organizational chart (e.g., peers, manager-subordinates) may have more chance to meet in the offline workplace. For example, subordinates need to report to their managers regularly, peers may co-operate to finish projects together, which can form various offline information diffusion channels (i.e., the red lines) among employees at workplace.

\begin{defn}
(Enterprise Social Networks (ESNs)): Online \textit{enterprise social networks} are a new type of online social networks used in enterprises to facilitate employees' communications and daily work, which can be represented as \textit{heterogeneous information networks} $G = (\mathcal{V}, \mathcal{E})$, where $\mathcal{V} = \bigcup_i \mathcal{V}_i$ is the set of different kinds of nodes and $\mathcal{E} = \bigcup_j \mathcal{E}_j$ is the union of complex links in the network.
\end{defn}

In this section, we will use Yammer as an example of online ESNs. Yammer can be represented as $G = (\mathcal{V}, \mathcal{E})$, where node set $\mathcal{V} = \mathcal{U} \cup \mathcal{O} \cup \mathcal{P}$ and $\mathcal{U}$, $\mathcal{O}$ and $\mathcal{P}$ are the sets of users, groups and posts respectively; link set $\mathcal{E} = \mathcal{E}_{s} \cup \mathcal{E}_{j} \cup \mathcal{E}_{w} \cup \mathcal{E}_{r} \cup \mathcal{E}_{l}$ denoting the union of social, group membership, write, reply and like links in Yammer respectively. In this section, we regard different group participation as the target activity, information about which can diffuse among employees at the workplace. Groups in ESNs are usually of different themes (e.g., new products, state-of-art techniques, daily-life entertainments), which are treated as different information topics in this section.

\begin{defn}
(Organizational Chart): \textit{Organizational chart} is a diagram outlining the structure of an organization as well as the relative ranks of employees' positions and jobs, which can be represented as a rooted tree $C = (\mathcal{N}, \mathcal{L}, root)$, where $\mathcal{N}$ denotes the set of employees and $\mathcal{L}$ is the set of directed \textit{supervision links} from managers to subordinates in the company, $root$ usually represents the CEO by default.
\end{defn}

Each employee in the company can create exactly one account in Yammer with valid employment ID, i.e., there is \textit{one-to-one} correspondence between the users in Yammer and employees in the organization chart. For simplicity, in this section, we assume the user set in online ESN to be identical to the employee set in the organizational chart (i.e., $\mathcal{U} = \mathcal{N}$) and we will use ``Employee'' to denote individuals in both online ESN and offline organizational chart by default.

To address all the above challenges, we will introduce a novel information diffusion model {\muse} (\underline{M}ulti-source \underline{M}ulti-channel \underline{M}ulti-topic diff\underline{U}sion \underline{SE}lection) proposed in \cite{cikm16}. {\muse} extracts and infers sets of online, offline and hybrid (of online and offline) diffusion channels among employees across online ESN and offline organizational structure. Information propagated via different channels can be aggregated effectively in {\muse}. Different diffusion channels will be weighted according to their importance learned from the social activity log data with optimization techniques and top-K effective diffusion channels will be selected in {\muse} finally.

\subsubsection{Preliminary}

In this section, a novel information diffusion model {\muse} will be proposed to depict the information propagation process of multiple topics via different diffusion channels across the online and offline world at workplace. We denote the set of topics diffusing in the workplace as set $\mathcal{T}$. Three different diffusion sources will be our main focus in this section: online source, offline source and the hybrid source (across online and offline sources). The diffusion channel set of all these three sources can be represented as $\mathcal{C}^{(on)}$, $\mathcal{C}^{(off)}$ and $\mathcal{C}^{(hyb)}$ respectively, whose sizes are $\left | \mathcal{C}^{(on)} \right | = k^{(on)}$, $\left | \mathcal{C}^{(off)} \right | = k^{(off)}$, $\left | \mathcal{C}^{(hyb)} \right | = k^{(hyb)}$.

In {\muse}, a set of users are activated initially, whose information will propagate in discrete steps within the network to other users. Let $v$ be an employee at workplace who has been activated by topic $t \in \mathcal{T}$. For instance, at step $\tau$, $v$ will send a amount of $w^{(on), i}(v, u, t)$ information on topic $t$ to $u$ via the $i_{th}$ channel in the online source (i.e., channel $c^{(on),i} \in \mathcal{C}^{(on)}$), where $u$ is an employee following $v$ in channel $c^{(on),i}$. The amount of information that $u$ receives from $v$ via all the channels in the online source at step $\tau$ can be represented as vector $\mb{w}^{(on)}(v, u, t) = [w^{(on), 1}(v, u, t), w^{(on), 2}(v, u, t), \cdots, w^{(on), k^{(on)}}(v, u, t)]$. Similarly, we can also represent the vectors of information $u$ receives from $v$ through channels in offline source and hybrid source as vectors $\mb{w}^{(off)}(v, u, t)$ and $\mb{w}^{(hyb)}(v, u, t)$ respectively.

Meanwhile, users in {\muse} are associated thresholds to different topics, which are selected at random from the uniform distribution in range $[0, 1]$. Employee $u$ can get activated by topic $t$ if the information received from his active neighbors via diffusion channels of all these three sources can exceed his \textit{activation threshold} $\theta(u, t)$ to topic $t$,
\begin{equation}
f\left (\mb{w}^{(on)}(\cdot, u, t), \mb{w}^{(off)}(\cdot, u, t),
 \mb{w}^{(hyb)}(\cdot, u, t)\right ) \ge \theta(u, t),
 \end{equation}
where aggregation function $f(\cdot)$ maps the information $u$ receives from all the channels to $u$'s \textit{activation probability} in range $[0,1]$. Here, the vector $\mb{w}^{(on)}(\cdot, u, t) = [w^{(on),1}(\cdot, u, t),\\ w^{(on),2}(\cdot, u, t), \cdots, w^{(on),k^{(on)}}(\cdot, u, t)]$, where $w^{(on),i}(\cdot, u, t)$ denotes the information received from all the employees $u$ follows in channel $c^{(on),i}$, i.e., 
\begin{equation}
w^{(on),i}(\cdot, u, t) = \sum_{v \in \Gamma_{out}^{(on),i}(u)} w^{(on),i}(v, u, t).
 \end{equation}
Vectors $\mb{w}^{(off)}(\cdot, u, t)$ and $\mb{w}^{(hyb)}(\cdot, u, t)$ can be represented in a similar way. Once being activated, a user will stay active in the remaining rounds and each user can be activated at most once. Such a process will end if no new activations are possible.

Considering that individuals' \textit{activation thresholds} $\theta(u, t)$ to topic $t$ is is pre-determined by the uniform distribution, next we will focus on studying the information received via channels of the \textit{online}, \textit{offline} and \textit{hybrid} sources and the \textit{aggregation function} $f(\cdot)$ in details.

\subsubsection{Online and Offline Diffusion Channels Extraction}

Both online ESNs and offline organizational chart provide various communication means for employees to contact each other, where individuals who have no social connections can still pass information via many other connections. Each connection among employees can form an information diffusion channel across online ESN and offline organizational chart. In {\muse}, various diffusion channels among employees will be extracted based on a set of \textit{social meta paths} \cite{SHYYW11} extracted across the online and offline world.

In enterprise social networks, individuals can (1) get information from employees they follow (i.e., their followees) and (2) people that their ``followees'' follow (i.e., $2_{nd}$ level followees), and obtain information from employees by (3) viewing and replying their posts, (4) viewing and liking their posts, as well as (5) getting notified by their posts (i.e., explicitly @ certain users in posts). {\muse} proposes to extract $5$ different \textit{online social meta paths} from the online ESN, whose physical meanings, representations and abbreviated notations are listed as follows:
\begingroup\makeatletter\def\f@size{6}\check@mathfonts
\begin{itemize}
\item Followee: $Employee \xleftarrow{Social^{-1}} Employee$, whose notation is $\Phi_1$.
\item Followee-Followee: $Employee \xleftarrow{Social^{-1}} Employee \xleftarrow{Social^{-1}} Employee$, whose notation is $\Phi_2$.
\item Reply Post: $Employee \xleftarrow{Reply^{-1}} Post \xleftarrow{Write} Employee$, whose notation is $\Phi_3$.
\item Like Post: $Employee \xleftarrow{Like^{-1}} Post \xleftarrow{Write} Employee$, whose notation is $\Phi_4$.
\item Post Notification: $Employee \xleftarrow{Notify} Post \xleftarrow{Write} Employee$, whose notation is $\Phi_5$.
\end{itemize}\endgroup

Meanwhile, in offline workplace, the most common social interaction should happen between close colleagues, e.g., peers, manager-subordinate, and skip-level manager-subordinates, etc. The physical meaning and notations of offline social meta paths extracted in this section are listed as follows:\begingroup\makeatletter\def\f@size{6}\check@mathfonts
\begin{itemize}
\item Manager: $Employee \xleftarrow{Supervision} Employee$, whose notation is $\Omega_1$.
\item Subordinate: $Employee \xleftarrow{Supervision^{-1}} Employee$, whose notation is $\Omega_2$.
\item Peer: $Employee \xleftarrow{Supervision} Employee \xleftarrow{Supervision^{-1}} Employee$, whose notation is $\Omega_3$.
\item 2nd-Level Manager: $Employee \xleftarrow{Supervision} Employee\\ \xleftarrow{Supervision} Employee$, whose notation is $\Omega_4$.
\item 2nd-Level Subordinate: $Employee \xleftarrow{Supervision^{-1}} Employee \\ \xleftarrow{Supervision^{-1}} Employee$, whose notation is $\Omega_5$.
\end{itemize}\endgroup

Besides the pure online/offline diffusion channels, information can also propagate across both online and offline world simultaneously. Consider, for example, two employees $v$ and $u$ who are not connected by any diffusion channels in online ESN or offline workplace, $v$ can still influence $u$ by activating $u$'s manager via online contacts and the manager will further propagate the influence to $v$ via offline interactions. To capture such relationships among the employees, a set of hybrid social meta path extracted in this {\muse}, together with their physical meanings, notations are listed as follows: \begingroup\makeatletter\def\f@size{6}\check@mathfonts
\begin{itemize}
\item Followee-Manager: $Employee \xleftarrow{Social^{-1}} Employee \xleftarrow{Supervision} Employee$, whose notation is $\Psi_1$,
\item Followee-Subordinate: $Employee \xleftarrow{Social^{-1}} Employee \\ \xleftarrow{Supervision^{-1}} Employee$, whose notation is $\Psi_2$,
\item Manager-Followee: $Employee \xleftarrow{Supervision} Employee \xleftarrow{Social^{-1}} Employee$, whose notation is $\Psi_3$,
\item Subordinate-Followee: $Employee \xleftarrow{Supervision^{-1}}  Employee \\ \xleftarrow{Social^{-1}} Employee$, whose notation is $\Psi_4$,
\item Followee-Peer: $Employee \xleftarrow{Social^{-1}} Employee \xleftarrow{Supervision} Employee \xleftarrow{Supervision^{-1}} Employee$, whose notation is $\Psi_5$,
\item Peer-Followee: $Employee \xleftarrow{Supervision} Employee \xleftarrow{Supervision^{-1}} Employee \xleftarrow{Social^{-1}} Employee$, whose notation is $\Psi_6$,
\end{itemize} \endgroup

The direction of the links denotes the information diffusion direction and end of the diffusion links (i.e., the first employee of the above paths) represents the target employee to receive the information. Each of the above \textit{social meta path} defines a information diffusion channel among individuals across the online and offline world.

Furthermore, let $\mathcal{P}^{(on)}_{\Phi_i}(v \rightsquigarrow \cdot)$ and $\mathcal{P}^{(on)}_{\Phi_i}(\cdot \rightsquigarrow u)$ be the sets of path instances of $\Phi_i$ going out from $v$ and going into $u$ respectively, with which we can define the amount of information propagating from $v$ to $u$ via diffusion channel $c^{(on), i} = \Phi_i$ to be
\begin{align}
w^{(on),i}(v, u, t) = \frac{2 \left | \mathcal{P}^{(on)}_{\Phi_i}(v \rightsquigarrow u) \right | \cdot I(v, t)}{\left | \mathcal{P}^{(on)}_{\Phi_i}(v \rightsquigarrow \cdot) \right | + \left | \mathcal{P}^{(on)}_{\Phi_i}(\cdot \rightsquigarrow u) \right |},
\end{align}
where binary function $I(v, t) = 1$ if $v$ has been activated by topic $t$ and $0$ otherwise.

Similarly, based on offline social meta path, e.g., $\Omega_i$, and hybrid diffusion channel, e.g., $\Psi_i$, the amount of information on topic $t$ propagating from employee $v$ to $u$ can be represented as follows respectively: \begingroup\makeatletter\def\f@size{8}\check@mathfonts
\begin{align}
&w^{(off),i}(v, u, t) &= \frac{2 \left | \mathcal{P}^{(off)}_{\Omega_i}(v \rightsquigarrow u) \right | \cdot I(v, t)}{\left | \mathcal{P}^{(off)}_{\Omega_i}(v \rightsquigarrow \cdot) \right | + \left | \mathcal{P}^{(off)}_{\Omega_i}(\cdot \rightsquigarrow u) \right |},\\
&w^{(hyb),i}(v, u, t) &= \frac{2 \left | \mathcal{P}^{(hyb)}_{\Psi_i}(v \rightsquigarrow u) \right | \cdot I(v, t)}{\left | \mathcal{P}^{(hyb)}_{\Psi_i}(v \rightsquigarrow \cdot) \right | + \left | \mathcal{P}^{(hyb)}_{\Psi_i}(\cdot \rightsquigarrow u) \right |}.
 \end{align}\endgroup

\subsubsection{Channel Aggregation}

Different diffusion channels deliver various amounts of information among employees via the online communications in ESN and offline contacts. In this subsection, we will focus on aggregating information propagated via different channels with the information aggregation function $f(\cdot): \mathbb{R}^{n\times1} \to [0, 1]$, which can map the amount of information received by employees to their activation probabilities. Generally, any function that can map real number to probabilities in range $[0, 1]$ can be applied and without loss of generality, we will use the logistic function $f(x) = \frac{e^x}{1 + e^x}$ \cite{DS86} in this section. 

Based on the information on topic $t$ received by $u$ via the online, offline and hybrid diffusion channels, we can represent $u$'s activation probability to be: \begingroup\makeatletter\def\f@size{6}\check@mathfonts
\begin{align}
&f\left (\mb{w}^{(on)}(\cdot, u, t), \mb{w}^{(off)}(\cdot, u, t),
 \mb{w}^{(hyb)}(\cdot, u, t)\right )\\
&= \frac{e^{\left (g(\mb{w}^{(on)}(\cdot, u, t)) + g(\mb{w}^{(off)}(\cdot, u, t)) + g(\mb{w}^{(hyb)}(\cdot, u, t)) + \theta_0\right)}}{1 + e^{\left(g(\mb{w}^{(on)}(\cdot, u, t)) + g(\mb{w}^{(off)}(\cdot, u, t)) + g(\mb{w}^{(hyb)}(\cdot, u, t)) + \theta_0\right)}},
\end{align}\endgroup
where function $g(\cdot)$ linearly combines the information in different channels belonging to certain sources and $\theta_0$ denotes the weight of the constant factor. Terms $g(\mb{w}^{(on)}(\cdot, u, t))$, $g(\mb{w}^{(off)}(\cdot, u, t))$ and $g(\mb{w}^{(hyb)}(\cdot, u, t))$ can be represented as follows \begingroup\makeatletter\def\f@size{6}\check@mathfonts
\begin{align}
g(\mb{w}^{(on)}(\cdot, u, t)) &= \sum_{i=1}^{k^{(on)}} \alpha_i \cdot \sum_{v \in \Gamma_{out}^{(on), i}(u)} {w}^{(on), i}(v, u, t),\\
g(\mb{w}^{(off)}(\cdot, u, t)) &= \sum_{i=1}^{k^{(off)}} \beta_i \cdot \sum_{v \in \Gamma_{out}^{(off), i}(u)} {w}^{(off), i}(v, u, t),\\
g(\mb{w}^{(hyb)}(\cdot, u, t)) &= \sum_{i=1}^{k^{(hyb)}} \gamma_i \cdot \sum_{v \in \Gamma_{out}^{(hyb), i}(u)} {w}^{(hyb), i}(v, u, t),
\end{align}\endgroup
where $\alpha_i$, $\beta_i$, $\gamma_i$ are the weights of different \textit{online}, \textit{offline} and \textit{hybrid} diffusion channels respectively and $\sum_{i=1}^{k^{(on)}} \alpha_i + \sum_{i=1}^{k^{(off)}} \beta_i + \sum_{i=1}^{k^{(hyb)}} \gamma_i + \theta_0= 1$. Depending of roles of different diffusion channels, the weights can be

\begin{itemize}
\item $>\ 0$, if positive information in the channel will increase employees' activation probability;
\item $=\ 0$, if positive information in the channel will not change employees' activation probability;
\item $<\ 0$, if positive information in the channel will decrease employees' activation probability.
\end{itemize}

In {\muse}, weights of certain diffusion channels can be negative. As a result, the likelihood for a node to become active will no longer grow monotonically in the {\muse} diffusion model. The optimal weights of different diffusion channels can be learned from the group participation log data (i.e., the target social activity diffusing at workplace). Different diffusion channels will be ranked according to their importance and top-$k$ diffusion channels which can increase individuals' activation probabilities will be selected in the next subsection.

\subsubsection{Channel Weighting and Selection}\label{subsec:chap9_sec6_weighting}

In Yammer, users can create and join groups of their interests, which can be about very diverse topics, e.g., products (e.g., iPhone, Windows, Android, etc.), people (e.g., Bill Gates, Leslie Lamport, etc.), projects (e.g., Project Complete, Meeting, ect.) and personal life issues (e.g., Diablo Games, Work Life Balance, etc.). The users' participation in groups log data can be represented as a set of tuples $\{(u, t)\}_{u, t}$, where tuple $(u, t)$ represents that user $u$ gets activated by topic $t$ (of groups). Such a tuple set can be split into three parts according to ratio 3:1:1 in the order of the timestamps, where 3 folds are used as the training set, 1 fold is used as the validation set and 1 fold as the test set. We will use the training set data to calculate the activation probabilities of individuals getting activated by topics in both the validation set and test set, while validation set is used to learn the weights of different diffusion channels and test set is used to evaluate the learned model. 

Let $\mathcal{V} = \{(u, t)\}_{u, t}$ be the validation set. Based on the amount of information propagating among employees in the workplace calculated with the training set, we can infer the probability of user $u$'s (who has not been activated yet) get activated by topic $t$, for $\forall (u, t) \in \mathcal{V}$, which can be represented with matrix $\mb{F} \in \mathbb{R}^{|\mathcal{U}| \times |\mathcal{T}|}$, where $\mb{F}(i, j)$ denotes the inferred activation probability of tuple $(u_i, t_j)$ in the validation set. Meanwhile, based on the validation set itself, we can get the ground-truth of users' group participation activities, which can be represented as a binary matrix $\mb{H} \in \{0, 1\}^{|\mathcal{U}| \times |\mathcal{T}|}$. In matrix $\mb{H}$, only entries corresponding tuples in the validation set are filed with value $1$ and the remaining entries are all filled with $0$. The optimal weights of information delivered in different diffusion channels (i.e., $\mb{\alpha}^*$,  $\mb{\beta}^*$, $\mb{\gamma}^*$, $\theta_0^*$) can be obtained by solving the following objective function \begingroup\makeatletter\def\f@size{7}\check@mathfonts
\begin{align}
&\mb{\alpha}^*,  \mb{\beta}^*, \mb{\gamma}^*, \theta_0^* = \arg \min_{\mb{\alpha}, \mb{\beta}, \mb{\gamma}, \mb{\theta}_0} \left \| \mb{F} - \mb{H} \right \|_F^2\\
&s.t. \sum_{i=1}^{k^{(on)}} \alpha_i + \sum_{i=1}^{k^{(off)}} \beta_i + \sum_{i=1}^{k^{(hyb)}} \gamma_i + \theta_0 = 1.
\end{align}\endgroup

The final objective function is not convex and can have multiple local optima, as the aggregation function (i.e., the logistic function) is not convex actually. {\muse} proposes to solve the objective function and handle the non-convex issue by using a two-stage process to ensure the robust of the learning process as much as possible.

(1) Firstly, the above objective function can be solved by using the method of Lagrange multipliers \cite{B96}, where the corresponding Lagrangian function of the objective function can be represented as \begingroup\makeatletter\def\f@size{7}\check@mathfonts
\begin{align}
&\mathcal{L}(\mb{\alpha}, \mb{\beta}, \mb{\gamma}, \theta_0, \eta) \\
&= \left \| \mb{F} - \mb{H} \right \|_F^2 + \eta (\sum_{i=1}^{k^{(on)}} \alpha_i + \sum_{i=1}^{k^{(off)}} \beta_i + \sum_{i=1}^{k^{(hyb)}} \gamma_i + \theta_0 - 1),\\
&=\mbox{Tr}(\mb{F}\mb{F}^\top - \mb{F}\mb{H}^\top - \mb{H}\mb{F}^\top + \mb{H}\mb{H}^\top)\\
&+ \eta (\sum_{i=1}^{k^{(on)}} \alpha_i + \sum_{i=1}^{k^{(off)}} \beta_i + \sum_{i=1}^{k^{(hyb)}} \gamma_i + \theta_0 - 1).
\end{align}\endgroup

By taking the partial derivatives of the Lagrange function with regards to variable $\alpha_i, i \in \{1, 2, \cdots, k^{(on)}\}$, we can get \begingroup\makeatletter\def\f@size{7}\check@mathfonts
\begin{align}
&\frac{\partial \mathcal{L}(\mb{\alpha}, \mb{\beta}, \mb{\gamma}, \theta_0, \eta)}{\partial \alpha_i}\\
&= \frac{\partial \mbox{Tr}(\mb{F}\mb{F}^\top)}{\partial \alpha_i} - \frac{\partial \mbox{Tr}(\mb{F}\mb{H}^\top)}{\partial \alpha_i} - \frac{\partial \mbox{Tr}(\mb{H}\mb{F}^\top)}{\partial \alpha_i} + \frac{\partial \mbox{Tr}(\mb{H}\mb{H}^\top)}{\partial \alpha_i} \\
&+ \frac{\partial \eta (\sum_{i=1}^{k^{(on)}} \alpha_i + \sum_{i=1}^{k^{(off)}} \beta_i + \sum_{i=1}^{k^{(hyb)}} \gamma_i + \theta_0 - 1)}{\partial \alpha_i}.
\end{align}\endgroup
Term \begingroup\makeatletter\def\f@size{7}\check@mathfonts
\begin{align}
&\frac{\partial \eta (\sum_{i=1}^{k^{(on)}} \alpha_i + \sum_{i=1}^{k^{(off)}} \beta_i + \sum_{i=1}^{k^{(hyb)}} \gamma_i + \theta_0 - 1)}{\partial \alpha_i} = \eta \\
&\frac{\partial \mbox{Tr}(\mb{F}\mb{F}^\top)}{\partial \alpha_i} = \sum_{j=1}^{\left | \mathcal{U} \right |} \sum_{l=1}^{\left | \mathcal{T} \right |} \frac{\partial \mb{F}^2(j,l)}{\partial \alpha_i} \sum_{j=1}^{\left | \mathcal{U} \right |} \sum_{l=1}^{\left | \mathcal{T} \right |} \Big( 2 f \big(\mb{w}^{(on)}(\cdot, u_j, t_l), \\
&\mb{w}^{(off)}(\cdot, u_j, t_l), \mb{w}^{(hyb)}(\cdot, u_j, t_l) \big) \Big) \cdot \Big(\frac{e^y}{(1+e^y)^2} \cdot \frac{\partial y}{\partial \alpha_i} \Big),
\end{align}\endgroup
where the introduced term $y$ denotes $y = g(\mb{w}^{(on)}(\cdot, u_j, t_l)) + g(\mb{w}^{(off)}(\cdot, u_j, t_l)) + g(\mb{w}^{(hyb)}(\cdot, u_j, t_l)) + \theta_0$ and its derivative is $\frac{\partial y}{\partial \alpha_i} = \frac{\partial g(\mb{w}^{(on)}(\cdot, u_j, t_l))}{\partial \alpha_i} = \sum_{v \in \Gamma_{out}^{(on), i}(u)} {w}^{(on), i}(v, u_j, t_k)$.
Similarly, we can obtain terms $\frac{\partial \mbox{Tr}(\mb{F}\mb{H}^\top)}{\partial \alpha_i}$, $\frac{\partial \mbox{Tr}(\mb{H}\mb{F}^\top)}{\partial \alpha_i}$, and $\frac{\partial \mbox{Tr}(\mb{H}\mb{H}^\top)}{\partial \alpha_i}$. By making $\frac{\partial \mathcal{L}(\mb{\alpha}, \mb{\beta}, \mb{\gamma}, \theta_0, \eta)}{\partial \alpha_i} = 0$, we can obtain an equation involving variables $\alpha_i$, $\beta_i$, $\gamma_i$, $\theta_0$ and $\eta$. Furthermore, we can calculate the partial derivatives of the Lagrange function with regards to variable $\beta_i$, $\gamma_i$, $\theta_0$ and $\eta$ respectively and make the equation equal to $0$, which will lead to an equation group about variables $\alpha_i$, $\beta_i$, $\gamma_i$, $\theta_0$ and $\eta$. The equation group can be solved with open source toolkits, e.g., SciPy Nonlinear Solver\footnote{http://docs.scipy.org/doc/scipy-0.14.0/reference/optimize.nonlin.html}, effectively. By giving the variables with different initial values, multiple solutions (i.e., multiple local optimal points) can be obtained by resolving the objective function.

(2) Secondly, the local optimal points obtained are further applied to the objective function and the one achieving the lowest objective function value is selected as the final results (i.e., the weights of different channels).

According to the learned weights, different diffusion channels can be ranked according to their importance in delivering information to activate employees in the workplace. Considering that, some diffusion channels may not perform very well in information propagation (e.g., those with negative or zero learned weights), top-$k$ channels that can increase employees' activation probabilities are selected as the effective channels used in {\muse} model finally. In other words, $k$ equals to the number of diffusion channels with positive weights learnt from the above objective function. Such a process is formally called diffusion channel weighting and selection in this section. The rational of channel weighting and selection is that: among all the diffusion channels, some channels can be useful but some may be not. $3$ different sets of diffusion channels are introduced in previous sections and we want to select the good ones.


\section{Network Embedding}\label{sec:embedding}

In the era of big data, information from diverse disciplines is generated at an extremely fast pace, lots of which are highly structured and can be represented as massive and complex networks. The representative examples include online social networks, like Facebook and Twitter, academic retrieval sites, like DBLP and Google Scholar, as well as bio-medical data, e.g., human brain networks. These networks/graphs are usually very challenging to handle due to their extremely large scale (involving millions even billions of nodes), complex structures (containing heterogeneous links) as well as the diverse attributes (attached to the nodes or links). For instance, the Facebook social network involves more than 1 billion active users; DBLP contains about 2.8 billions of papers; and human brain has more than 16 billion of neurons. 

Great challenges exist when handling these network structured data with traditional machine learning algorithms, which usually take feature vector representation data as the input. A general representation of heterogeneous networks as feature vectors is desired for knowledge discovery from such complex network structured data. In recent years, many research works propose to embed the online social network data into a lower-dimensional feature space, in which the user node is represented as a unique feature vector, and the network structure can be reconstructed from these feature vectors. With the embedded feature vectors, classic machine learning models can be applied to deal with the social network data directly, and the storage space can be saved greatly.

In this section, we will talk about the \textit{network embedding} problem, aiming at projecting the nodes and links in the network data in low-dimensional feature spaces. Depending on the application setting, exist graph embedding works can be categorized into the embedding of \textit{homogeneous networks}, \textit{heterogeneous networks}, and \textit{multiple aligned heterogeneous networks}. Meanwhile, depending on the models being applied, current embedding works can be divided into the \textit{matrix factorization based embedding}, \textit{translation based embedding}, and \textit{deep learning architecture based embedding}.

In the following parts in this section, we will first introduce the \textit{translation based graph embedding} models in Section~\ref{sec:chap11_sec2_translation}, which are mainly proposed for the multi-relational knowledge graphs, including {TransE} \cite{BUGWY13}, {TransH} \cite{WZFC14} and {TransR} \cite{LLSLZ15}. After that, in Section~\ref{sec:chap11_sec3_deep}, we will introduce three homogeneous network embedding models, including {DeepWalk} \cite{PAS14}, {LINE} \cite{TQWZYM15} and {node2vec} \cite{GL16}. Two embedding models for the heterogeneous networks will be introduced in Section~\ref{sec:chap11_sec4_hin}, which projects the nodes to feature vectors based on the heterogeneous information inside the networks \cite{CHTQAH15, CS16}. Finally, we will talk about the model proposed for the multiple aligned heterogeneous network \cite{icdm17} in Section~\ref{sec:chap11_sec5_aligned_hin}, where the anchor links are utilized to transfer information across different sites for mutual refinement of the embedding results synergistically. 

\subsection{Relation Translation based Graph Entity Embedding}\label{sec:chap11_sec2_translation}

Multi-relational data refers to the directed graphs whose nodes correspond to entities and links denote the relationships. The multi-relational data can be represented as a graph $G = (\mathcal{V}, \mathcal{E})$, where $\mathcal{V}$ denotes the node set and $\mathcal{E}$ represents the link set. For the link in the graph, e.g., $r = (h, t) \in \mathcal{E}$, the corresponding entity-relation can be represented as a triple $(h, r, t)$, where $h$ denotes the link initiator entity, $t$ denotes the link recipient entity and $r$ represents the link. The embedding problem studied in this section is to learn a feature representation of both entities and relations in the triples, i.e., $h$, $r$ and $t$.

Model {TransE} is the initial translation based embedding work, which projects the entity and relation into a common feature space. {TransH} improves {TransE} by considering the link cardinality constraint in the embedding process, and can achieve comparable time complexity. In the real-world multi-relational networks, the entities can have multiple aspects, and the different relations can express different aspects of the entity. Model {TransR} proposes to build the entity and relation embeddings in separate entity and relation spaces instead. Next, we will introduce the embedding models {TransE}, {TransH} and {TransR} one by one as follows, where the relation is more like a translation of entities in the embedding space. It is the reason why these models are called the \textit{translation based embedding models}.

\subsubsection{TransE}

The {TransE} \cite{BUGWY13} model is an energy-based model for learning low-dimensional embeddings of entities and relations, where the relations are represented as the \textit{translations} of entities in the embedding space. Given a entity-relation triple $(h, r, t)$, the embedding feature representation of the entities and relations can be represented as vectors $\mb{h} \in \mathbb{R}^k$, $\mb{r} \in \mathbb{R}^k$ and $\mb{t} \in \mathbb{R}^k$ ($k$ denotes the objective vector dimension). If the triple $(h, r, t)$ holds, i.e., there exists a link $r$ starting from $h$ to $t$ in the network, the corresponding embedding vectors $\mb{h} + \mb{r}$ should be as close to vector $\mb{t}$ as possible. 

Let $\mathcal{S}^+ = \{(h, r, t)\}_{r = (h, t) \in \mathcal{E}}$ represents the set of positive training data, which contains the triples existing in the networks. The {TransE} model aims at learning the embedding features vectors of the entities $h$, $t$ and the relation $r$, i.e., $\mb{h}$, $\mb{r}$ and $\mb{t}$. For the triples in the positive training set, we want to ensure the learnt embedding vectors $\mb{h} + \mb{r}$ is very close to $\mb{t}$. Let $d(\mb{h} + \mb{r}, \mb{t})$ denotes the distance between vectors $\mb{h} + \mb{r}$ and $\mb{t}$. The loss introduced for the triples in the positive training set can be represented as
\begin{equation}
\mathcal{L}(\mathcal{S}^+) = \sum_{(h, r, t) \in \mathcal{S}^+} d(\mb{h} + \mb{r}, \mb{t}).
\end{equation}

Here the distance function can be defined in different ways, like the $L_2$ norm of the difference between vectors $\mb{h} + \mb{r}$ and $\mb{t}$, i.e.,
\begin{equation}
d(\mb{h} + \mb{r}, \mb{t}) = \left\| \mb{h} + \mb{r} - \mb{t} \right\|_2.
\end{equation}

By minimizing the above loss function, the optimal feature representations of the entities and relations can be learnt. To avoid trivial solutions, like $\mb{0}$s for $\mb{h}$, $\mb{r}$ and $\mb{t}$, additional constraints that the $L_2$-norm of the embedding vectors of the entities should be $1$ will be added in the function. Furthermore, a negative training set is also sampled to differentiate the learnt embedding vectors. For a triple $(h, r, t) \in \mathcal{S}^+$, the corresponding sampled negative training set can be denoted as $\mathcal{S}^-_{(h, r, t)}$, which contains the triples formed by replacing the initiator entity $h$ or the recipient entity $t$ with random entities. In other words, the negative training set $\mathcal{S}^-_{(h, r, t)}$ can be represented as
\begin{equation}
\mathcal{S}^-_{(h, r, t)} = \{(h', r, t) | h' \in \mathcal{V}\} \cup \{(h, r, t')  | t' \in \mathcal{V}\}.
\end{equation}

The loss function involving both the positive and negative training set can be represented as \begingroup\makeatletter\def\f@size{6}\check@mathfonts
\begin{align}
&\mathcal{L}(\mathcal{S}^+, \mathcal{S}^-) = \\
&\sum_{(h, r, t) \in \mathcal{S}^+}  \sum_{(h', r, t') \in \mathcal{S}^-_{(h, r, t)}} \hspace{-10pt} \max \left(\gamma + d(\mb{h} + \mb{r}, \mb{t}) - d(\mb{h}' + \mb{r}, \mb{t}'), 0 \right),
\end{align}\endgroup
where $\gamma$ is a margin hyperparameter and $\max (\cdot, 0)$ will count the positive loss only.

The optimization is carried out by stochastic gradient descent (in minibatch mode). The embedding vectors of entities and relationships are initialized with a random procedure. At each iteration of the algorithm, the embedding vectors of the entities are normalized and a small set of triplets is sampled from the training set, which will serve as the training triplets of the minibatch. The parameters are then updated by taking a gradient step with constant learning rate.

\subsubsection{TransH}

{TransE} is a promising method proposed recently, which is very efficient while achieving state-of-the-art predictive performance. However, in the embedding process, {TransE} fail to consider the \textit{cardinality constraint} on the relations, like \textit{one-to-one}, \textit{one-to-many} and \textit{many-to-many}. The {TransH} model \cite{WZFC14} to be introduced in this part considers such properties on relations in the embedding process. Furthermore, different from the other complex models, which can handle these properties but sacrifice efficiency, {TransH} achieves comparable time complexity as {TransE}. {TransH} models the relation as a hyperplane together with a translation operation on it, where the correlation among the entities can be effectively preserved. 

In {TransH}, different from the embedding space of entities, the relations, e.g., $r$, is denoted as a transition vector $\mb{d}_r$ in the hyperplane $\mb{w}_r$ (a normal vector). For each of the triple $(h, r, t)$, the embedding vector $\mb{h}$, $\mb{t}$ are fist projected to the hyperplane $\mb{w}_r$, whose corresponding projected vectors can be represented as $\mb{h}_\perp$ and $\mb{t}_\perp$ respectively. The vectors $\mb{h}_\perp$ and $\mb{t}_\perp$ can be connected by the translation vector $\mb{d}_r$ on the hyperplane. Depending on whether the triple appears in the positive or negative training set, the distance $d(\mb{h}_\perp + \mb{d}_r, \mb{t}_\perp)$ should be either minimized or maximized.

Formally, given the hyperplane $\mb{w}_r$, the projection vectors $\mb{h}_{\perp}$ and $\mb{t}_{\perp}$ can be represented as
\begin{align}
\mb{h}_\perp = \mb{h} - \mb{w}_r^\top \mb{h} \mb{w}_r,\\
\mb{t}_\perp = \mb{t} - \mb{w}_r^\top \mb{t} \mb{w}_r.
\end{align}
Furthermore, the $L_2$ norm based distance function can be represented as \begingroup\makeatletter\def\f@size{8}\check@mathfonts
\begin{equation}
d(\mb{h}_\perp + \mb{d}_r, \mb{t}_\perp) = \left \| (\mb{h} - \mb{w}_r \mb{h} \mb{w}_r) + \mb{d}_r - (\mb{t} - \mb{w}_r \mb{t}\mb{w}_r) \right \|_2^2.
\end{equation}\endgroup

The variables to be learnt in the {TransH} model include the embedding vectors of all the entities, the hyperplane and translation vectors for each of the relations. To learn these variables simultaneously, the objective function of {TransH} can be represented as \begingroup\makeatletter\def\f@size{6}\check@mathfonts
\begin{align}
&\mathcal{L} (\mathcal{S}^+, \mathcal{S}^-) =\\
& \sum_{(h, r, t) \in \mathcal{S}^+}  \sum_{(h', r', t') \in \mathcal{S}^-_{(h, r, t)}} \hspace{-15pt} \max \left(\gamma + d(\mb{h}_\perp + \mb{d}_r, \mb{t}_\perp) - d(\mb{h}'_\perp + \mb{d}_r', \mb{t}'_\perp), 0 \right),
\end{align}\endgroup
where $\mathcal{S}^-_{(h, r, t)}$ denotes the negative set constructed for triple $(h, r, t)$. Different from {TransE}, {TransH} applies a different to sample the negative training triples with considerations of the relation \textit{cardinality constraint}. For the relations with \textit{one-to-many}, {TransH} will give more chance to replace the initiator node; and for the \textit{many-to-one} relations, {TransH} will give more chance to replace the recipient node instead. 

Besides the loss function, the variables to be learnt are subject to some constraints, like the embedding vector for entities is a normal vector; $\mb{w}_r$ and $\mb{d}_r$ should be orthogonal, and $\mb{w}_r$ is also a normal vector. We summarize the constraints of the {TransH} model as follows
\begin{align}
&\left \| \mb{h} \right \|_2 \le 1, \left \| \mb{t} \right \|_2 \le 1, \forall h, t \in \mathcal{V},\\
&\frac{|\mb{w}_r^\top \mb{d}_r |}{\left \| \mb{d}_r \right \|_2} \le \epsilon, \forall r \in \mathcal{E},\\
&\left \| \mb{w}_r \right\|_2 \le 1, \forall r \in \mathcal{E}.
\end{align}

The constraints can be relaxed as some penalty terms, which can be added to the objective function with a relatively large weight. The final objective function can be learnt with the stochastic gradient descent, and by minimizing the loss function, the model variables can be learned and we will get the final embedding results.

\subsubsection{TransR}

Both {TransE} and {TransH} introduced in the previous subsections assume embeddings of entities and relations within the same space $\mathbb{R}^k$. However, entities and relations are actually totally different objects, and they may be not capable to be represented in a common semantic space. To address such a problem, {TransR} \cite{LLSLZ15} is proposed, which models the entities and relations in distinct spaces, i.e., the entity space and relation space, and performs the translation in relation space.

In {TransR}, given a triple $(h, r, t)$, the entities $h$ and $t$ are embedded as vectors $\mb{h}, \mb{t} \in \mathbb{R}^{k_e}$, and the relation $r$ is embedded as vector $\mb{r} \in \mathbb{R}^{k_r}$, where the dimension of the entity space and relation space are not the same, i.e., $k_e \neq k_r$. To project the entities from the entity space to the relation space, a projection matrix $\mb{M}_r \in \mathbb{R}^{k_e \times k_r}$ is defined in {TransR}. With the projection matrix, the projected entity embedding vectors can be defined as
\begin{align}
&\mb{h}_r = \mb{h} \mb{M}_r, \\
&\mb{t}_r = \mb{t} \mb{M}_r.
\end{align}

The loss function is defined as
\begin{equation}
d(\mb{h}_r + \mb{r}, \mb{t}_r) = \left\| \mb{h}_r + \mb{r} - \mb{t}_r \right\|_2^2.
\end{equation}

The constraints involved in {TransR} include
\begin{align}
&\left \| \mb{h} \right \|_2 = 1, \left \| \mb{t} \right \|_2 = 1, \forall h, t \in \mathcal{V},\\
&\left \| \mb{h} \mb{M}_r \right \|_2 = 1, \left \| \mb{t} \mb{M}_r \right \|_2 = 1, \forall h, t \in \mathcal{V},\\
&\left \| \mb{w}_r \right\|_2 \le 1, \forall r \in \mathcal{E}.
\end{align}

The negative training set $\mathcal{S}^-$ in {TransR} can be obtained in a similar way as {TransH}, where the variables can be learnt with the stochastic gradient descent. We will not introduce the information here to avoid content duplication.

\subsection{Homogeneous Network Embedding}\label{sec:chap11_sec3_deep}

Besides the translation based network embedding models, in this section, we will introduce three embedding models for network data, including {DeepWalk}, {LINE} and {node2vec}. Formally, the networks studied in this part are all homogeneous networks, which is represented as $G = (\mathcal{V}, \mathcal{E})$. Set $\mathcal{V}$ denotes the set of nodes in the homogeneous network, and $\mathcal{E}$ represents the set of links among the nodes inside the network.

\subsubsection{DeepWalk}

\setlength{\textfloatsep}{0pt}
\begin{algorithm}[t]
\small
\caption{DeepWalk}
\label{alg:chap11_sec3_deepwalk}
\begin{algorithmic}[1]
	\REQUIRE Input homogeneous network $G =(\mathcal{V}, \mathcal{E})$\\
	\qquad	Window size $s$; Embedding size $d$\\
	\qquad	Walk length $l$; Walks per node $\gamma$
\ENSURE  Matrix of node representations $\mb{X} \in \mathbb{R}^{|\mathcal{V}| \times d}$

\STATE	{Initialize $\mb{X}$ with random values following the uniform distribution}
\STATE	{Build a binary tree $T$ from node set $\mathcal{V}$}
\FOR	{Round $i = 1$ to $\gamma$}
\STATE	{$\mathcal{O} = \mbox{shuffle}(\mathcal{V})$}
\FOR	{Node $u \in \mathcal{O}$}
\STATE	{$W_u  = \mbox{WalkGenerator}(G, u, l)$}
\STATE	{SkipGram($\mb{X}$, $W_u$, $w$)}
\ENDFOR
\ENDFOR
\STATE	{Return $\mb{X}$}

\end{algorithmic}
\end{algorithm}

The {DeepWalk} \cite{PAS14} algorithm consists of two main components: (1) a random walk generator, and (2) an update procedure. In the first step, the DeepWalk model randomly selects a node, e.g., $u \in \mathcal{V}$, as the root of a random walk $W_u$ from the nodes in the network. Random walk $W_u$ will sample the neighbors of the node last visited uniformly until the maximum length $l$ is met. In the second step, the sampled neighbors are used to update the representations of the nodes inside the graph, where \textit{SkipGram} \cite{MSCCD13} is applied here. 

The pseudo code of the {DeepWalk} algorithm is available in Algorithm~\ref{alg:chap11_sec3_deepwalk}, which illustrates the general architecture of the algorithm. In the algorithm, line 1 initializes the representation matrix $\mb{X}$ for all the nodes, and line 2 builds a binary tree involving all the nodes in the network as the leaves, which will be introduced in more detail in Section~\ref{subsec:chap11_sec3_softmax}. Lines 3-9 denote the main part of the {DeepWalk} algorithm, where the random walk starting randomly at each node is generated for $\gamma$ times by calling function \textit{WalkGenerator}. For each node $u$, a random walk $W_u$ is generated whose length is bounded by parameter $l$. The random walk will be applied to update the node representation with the \textit{SkipGram} function to be introduced in Section~\ref{subsec:chap11_sec3_skipgram}.

\noindent \textbf{Random Walk Generator}

The {random walk} model has been introduced in Section~\ref{subsec:closeness_measures}. Formally, the random walk starting at node $u \in \mathcal{V}$ can be represented as $W_u$, which actually denotes a stochastic process with random status $W_u^0$, $W_u^1$, $\cdots$, $W_u^k$. Formally, at the very beginning, i.e., step $0$, the random walk is at the initial node, i.e., $W_u^0 = u$. The status variable $W_u^k$ denotes the node where the node is at step $k$.

Random walk can capture the local network structures effectively, where the neighborhood and social connection closeness can affect the next nodes that the random walk will move to in the next step. Therefore, in the {DeepWalk}, random walk is applied to sample a stream of short random walks as the tool for extracting information from a network. Random walk can provide two very desirable properties, besides the ability to capture the local community structures. Firstly, the random walk based local exploration is easy to parallelize. Several random walks can simultaneously explore different parts of the same network in different threads, processes and machines. Secondly, with the information obtained from short random walks, it is possible to accommodate small changes in the network structure without the need for global recomputation.

\noindent \textbf{SkipGram Technique}\label{subsec:chap11_sec3_skipgram}

The updating procedure used in {DeepWalk} is very similar to the word appearance prediction in language modeling. In this part, we will first provide some basic knowledge about language modeling problem first, and then introduce the \textit{SkipGram} technique.

Formally, the objective of language modeling is to estimate the likelihood of a specific sequence of words appearing in a corpus. More specifically, given a sequence of words $(w_1, w_2, \\ \cdots, w_{n-1})$ where word $w_i \in \mathcal{V}$ ($\mathcal{V}$ denotes the vocabulary), the word appearing prediction problem aims at inferring the word $w_{n}$ that will appear next. An intuitive idea to model the problem is to maximize the estimation likelihood for the next word $w_{n}$ given $w_1, w_2, \cdots, w_{n-1}$, and the problem can be formally represented as
\begin{equation}
w^*_{n} = \arg_{w_{n} \in \mathcal{V}} P(w_{n} | w_1, w_2, \cdots, w_{n-1}).
\end{equation}
where term $P(w_{n} | w_1, w_2, \cdots, w_{n-1})$ denotes the conditional probability of having $w_{n}$ attached to the observed word sequence $w_1, w_2, \cdots, w_{n-1}$.

Meanwhile, in neural networks, the words will have a latent representation denoted as vector, like $\mb{x}_{w_i} \in \mathbb{R}^{d \times 1}$ for word $w_i \in \mathcal{V}$. Furthermore, computation of the above conditional probability is very challenging, especially as the observed word sequence goes longer, i.e., $n$ is large. Therefore, a window is proposed to limit the length of word sequence in probability computation. Term ${s}$ is denoted as the size of the window. Therefore, the above objective function can be rewritten as
\begin{equation}
w^*_{n} = \arg_{w_{n} \in \mathcal{V}} P(w_{n} | \mb{x}_{w_{n-s}}, \mb{x}_{w_{n-s+1}}, \cdots, \mb{x}_{w_{n-1}}).
\end{equation}

A recent relaxation to the above problem in language modeling turns the prediction problem on its head. Three big changes are applied to the model: (1) instead of predicting the objective word with the context, the relaxation predicts the context with the objective word instead; (2) the context denotes the words appearing before and after the objective word limited by the window size $s$, and (3) the order of words is removed and the context denotes a set of words instead. Formally, the objective function can be rewritten as  \begingroup\makeatletter\def\f@size{7}\check@mathfonts
\begin{equation}
w^*_{n} = \arg_{w_{n} \in \mathcal{V}} P(\{w_{n-s}, w_{n-s+1}, \cdots, w_{n+s}\} \setminus \{w_{n}\}|\mb{x}_{w_{n}}).
\end{equation}\endgroup

SkipGram is a language model that maximize the co-occurrence probability of words appearing in the time window $s$ in a sentence. Here, when applying the \textit{SkipGram} technique to the {DeepWalk} model, the nodes $u \in \mathcal{V}$ in the network can be regarded as the words $w$ denoted in the equations aforementioned. Meanwhile, for the nodes sampled by the random walk model within the window size $s$ before and after node $v$, they will be treated as the words appearing ahead of and after node $v$. Furthermore, SkipGram assumes the appearance of the words (or nodes for networks) to be independent, and the above probability equations can be rewritten as follows:  \begingroup\makeatletter\def\f@size{6}\check@mathfonts
\begin{equation}
\hspace{-1pt} P(\{u_{n-s}, u_{n-s+1}, \cdots, u_{n+s}\} \setminus \{u_{n}\}|\mb{x}_{u_{n}}) = \hspace{-10pt} \prod_{i = n-s, i \neq n}^{n+s} P(u_i | \mb{x}_{u_n}),
\end{equation}\endgroup
where $u_{n-s}, u_{n-s+1}, \cdots, u_{n+s}$ denotes the sequence of nodes sampled by the random walk model. 

The learning process of the SkipGram algorithm is provided in Algorithm~\ref{alg:chap11_sec3_skipgram}, where we will enumerate all the co-locations of nodes in the sampled node series $u_{n-s}, u_{n-s+1},\\ \cdots, u_{n+s}$ by a random walk $W_u$ (starting from node $u$ in the network). With gradient descent, the representation of nodes with their neighbors representations can be updated with stochastic gradient descent. The derivatives are estimated with the back-propagation algorithm. However, in the equation, we need to have the conditional probabilities of the nodes and their representations. A concrete representation of the probability can be a great challenging problem. As proposed in \cite{MSCCD13}, such a distribution can be learnt with some existing models, like logistic regression. However, since the labels used here denote the nodes in the network, it will lead to a very large label space with $|\mathcal{V}|$ different labels, which renders the learning process extremely time consuming. To solve such a problem, some techniques, like Hierarchical Softmax, have been proposed which represents the nodes in the network as a binary tree and can lower done the probability computation time complexity from $O(|\mathcal{V}|)$ to $O(\log |\mathcal{V}|)$.

\noindent \textbf{Hierarchical Softmax}\label{subsec:chap11_sec3_softmax}

\setlength{\textfloatsep}{0pt}
\begin{algorithm}[t]
\small
\caption{SkipGram}
\label{alg:chap11_sec3_skipgram}
\begin{algorithmic}[1]
	\REQUIRE Representations of nodes: $\mb{X}$\\
	\qquad	Random walk starting from node $u$: $W_u$\\
	\qquad	Window size $s$\\
\ENSURE  Updated matrix of node representations $\mb{X}$

\FOR	{Each node $u_i \in W_u$}
\STATE	{$W_u$ will generate a sampled sequence before and after $u_j$ bounded by window size $s$: $(u_{i - s}, \cdots, u_{i+s})$}
\FOR	{Each node $u_j \in (u_{i - s}, \cdots, u_{i+s})$}
\STATE	{$J(\mb{X}) = - \log P(u_j | \mb{x}_{u_i})$}
\STATE	{$\mb{X} = \mb{X} - \alpha \frac{J(\mb{X})}{\partial \mb{X}}$}
\ENDFOR
\ENDFOR
\STATE	{Return $\mb{X}$}
\end{algorithmic}
\end{algorithm}

In the SkipGram algorithm, calculating probability $P(u_i | \mb{x}_{u_n})$ is infeasible. Therefore, in the {DeepWalk} model, \textit{hierarchical softmax} is used to factorize the conditional probability. In \textit{hierarchical softmax}, a binary tree is constructed, where the number of leaves equals to the network node set size, and each network node is assigned to a leaf node. The prediction problem is turned into a path probability maximization problem. If a path $(b_0, b_1, \cdots, b_{\left \lceil \log |\mathcal{V}| \right \rceil})$ is identified from the tree root to the node $u_k$, i.e., $b_0 = \mbox{root}$ and $b_{\left \lceil \log |\mathcal{V}| \right \rceil} = u_k$, then the probability can be rewritten as
\begin{equation}
P(u_i | \mb{x}_{u_n}) = \prod_{l = 1}^{\left \lceil \log |\mathcal{V}| \right \rceil} P(b_l | \mb{x}_{u_n}),
\end{equation}
where $P(b_l | \mb{x}_{u_n})$ can be modeled by a binary classifier denoted as
\begin{equation}
P(b_l | \mb{x}_{u_n}) = \frac{1}{1+e^{- \mb{x}_{b_l} \cdot \mb{x}_{u_n}}}.
\end{equation} 
Here the parameters involved in the learning process include the representations for both the nodes in the network as well as the nodes in the constructed binary trees. 

\subsubsection{LINE}\label{subsec:chap11_sec3_line}

To handle the real-world information networks, the embedding models need to have several requirements: (1) preserve the \textit{first-order} and \textit{second-order} proximity between the nodes, (2) scalable to large sized networks, and (3) able to handle networks with different links: \textit{directed} and \textit{undirected}, \textit{weighted} and \textit{unweighted}. In this part, we will introduce another homogeneous network embedding model, named {LINE} \cite{TQWZYM15}.

\noindent \textbf{First-order Proximity}

In the network embedding process, the network structure should be effectively preserved, where the node closeness is defined as the node \textit{proximity} concept in {LINE}. The \textit{first-order proximity} in a network denotes the \textit{local} pairwise proximity between nodes. For a link $(u, v) \in \mathcal{E}$ in the network, the \textit{first-order proximity} denotes the weight of link $(u, v)$ in the network (or $1$ if the network is unweighted). Meanwhile, if link $(u, v)$ doesn't exist in the network, the \textit{first-order proximity} between them will be $0$ instead. To model the \textit{first-order proximity}, for a given link $(u, v) \in \mathcal{E}$ in the network $G$, {LINE} defines the joint probability between nodes $u$ and $v$ as
\begin{equation}
p_1(u, v) = \frac{1}{1+ e^{- \mb{x}_{u} \cdot \mb{x}_v}},
\end{equation} 
where $\mb{x}_{u}, \mb{x}_{v} \in \mathbb{R}^d$ denote the vector representations of nodes $u$ and $v$ respectively. 

Function $p_1(\cdot, \cdot)$ defines the proximity distribution in the space of $\mathcal{V} \times \mathcal{V}$. Meanwhile, given a network $G$, the \textit{empirical proximity} between nodes $u$ and $v$ can be denoted as
\begin{equation}
\hat{p_1}(u, v) = \frac{w_{(u,v)}}{\sum_{(u,v) \in \mathcal{E}} w_{(u, v)}}.
\end{equation}
To preserve the \textit{first-order proximity}, {LINE} defines the objective function for the network embedding as
\begin{equation}
J_1 = d(p_1(\cdot, \cdot), \hat{p_1}(\cdot, \cdot)),
\end{equation}
where function $d(\cdot, \cdot)$ denotes the distance between between the introduced proximity distribution and the empirical proximity distribution. By replacing the distance function $d(\cdot, \cdot)$ with the KL-divergence and omitting some constants, the objective function can be rewritten as
\begin{equation}\label{equ:chap11_sec3_first_order}
J_1 = - \sum_{(u, v) \in \mathcal{E}} w_{(u,v)} \log p_1(u, v).
\end{equation}
By minimizing the objective function, {LINE} can learn the feature representation $\mb{x}_u$ for each node $u \in \mathcal{V}$ in the network.

\noindent \textbf{Second-order Proximity}

In the real-world social networks, the links among the nodes can be very sparse, where the \textit{first-order proximity} can hardly preserve the complete structure information of the network. {LINE} introduce the concept of \textit{second-order proximity}, where denotes the similarity between the neighborhood structure of nodes. Given a user pair $(u, v)$ in the network, the more common neighbors shared by them, the closer users $u$ and $v$ are in the network. Besides the original representation $\mb{x}_u$ for node $u\in \mathcal{V}$, the nodes are also associated with a feature vector representing its context in the network, which is denoted as $\mb{y}_u \in \mathbb{R}^d$.

Formally, for a given link $(u, v) \in \mathcal{E}$, the probability of context $\mb{y}_v$ generated by node $u$ can be represented as
\begin{equation}
p_2(v | u) = \frac{e^{\mb{x}_u^\top \cdot \mb{y}_v}}{\sum_{v' \in \mathcal{V}} e^{\mb{x}_u^\top \cdot \mb{y}_{v'}}}.
\end{equation}
Slightly different from \textit{first-order proximity}, the \textit{second-order} empirical proximity is denoted as 
\begin{equation}
\hat{p_2}(v | u) = \frac{w_{(u, v)}}{D(u)}.
\end{equation}
By minimizing the difference between the introduced proximity distribution and the empirical proximity distribution, the objective function for the \textit{second-order} proximity can be represented as
\begin{equation}
J_2 = \sum_{u \in \mathcal{V}} \lambda_u d(p_2(\cdot | u), \hat{p_2}(\cdot | u)),
\end{equation}
where $\lambda_u$ denotes the prestige of node $u$ in the network. Here, by replacing the distance function $d(\cdot | \cdot)$ with the KL-divergence and setting $\lambda_u = D(u)$, the \textit{second-order proximity} based objective function can be represented as
\begin{equation}\label{equ:chap11_sec3_second_order}
J_2 = - \sum_{(u, v) \in \mathcal{E}} w_{(u, v)} \log p_2(v | u).
\end{equation}

\noindent \textbf{Model Optimization}

Instead of combining the \textit{first-order proximity} and \textit{second-order proximity} into a joint optimization function, {LINE} learns the embedding vectors based on Equations~\ref{equ:chap11_sec3_first_order} and \ref{equ:chap11_sec3_second_order} respectively, which will be further concatenated together to obtain the final embedding vectors. 

In optimizing objective function \ref{equ:chap11_sec3_second_order}, {LINE} needs to calculate the conditional probability $P(\cdot | u)$ for all nodes $u \in \mathcal{V}$ in the network, which is computational infeasible. To solve the problem, {LINE} uses the negative sampling approach instead. For each link $(u, v) \in \mathcal{E}$, {LINE} samples a set of negative links according to some noisy distribution. 

Formally, for link $(u, v) \in \mathcal{E}$, the set of negative links sampled for it can be represented as $\mathcal{L}^-_{(u, v)} \subset \mathcal{V} \times \mathcal{V}$. The objective function defined for link $(u, v)$ can be represented as
\begin{equation}\label{equ:chap11_sec3_negative_sampling}
\log \sigma(\mb{y}_v^\top \cdot \mb{x}_u) + \sum_{(u, v') \in \mathcal{L}^-_{(u, v)}} \log \sigma (-\mb{y}_{v'}^\top \cdot \mb{x}_u),
\end{equation}
where $\sigma(\cdot)$ is the sigmoid function. The first term in the above equation denotes the observed links, and the second term represents the negative links drawn from the noisy distribution. Similar approach can also be applied to solve the objective function in Equation~\ref{equ:chap11_sec3_first_order} as well. The new objective function can be solved with the asynchronous stochastic gradient algorithm (ASGD), which samples a mini-batch of links and then update the parameters. 

\subsubsection{node2vec}\label{subsec:chap11_sec3_node2vec}

In {LINE}, the closeness among nodes in the networks is preserved based on either the \textit{first-order proximity} or the \textit{second-order proximity}. In a recent work, {node2vec} \cite{GL16}, the authors propose to preserve the proximity between nodes with a sampled set of nodes in the network. 

\noindent \textbf{node2vec Framework}

Model {node2vec} is based on the \textit{SkipGram} in language modeling, and the objective function of {node2vec} can be formally represented as
\begin{equation}
\max \sum_{u \in \mathcal{V}} \log P(\Gamma(u) | \mb{x}_u).
\end{equation}
where $\mb{x}_u$ denotes the latent feature vector learnt for node $u$ and $\Gamma(u)$ represents the neighbor set of node $u$ in the network. 

To simplify the problem and make the problem solvable, some assumptions are made to approximate the objective function into a simpler form. 
\begin{itemize}
\item \textit{Conditional Independence Assumption}: Given the latent feature vector $\mb{x}_u$ of node $u$, by assuming the observation of node in set $\Gamma(u)$ to be independent, the probability equation can be rewritten as
\begin{equation}
P(\Gamma(u) | \mb{x}_u) = \prod_{v \in \Gamma(u)} P(v | \mb{x}_u).
\end{equation} 

\item \textit{Symmetric Node Effect}: Furthermore, by assuming the source and neighbor nodes have a symmetric effect on each other in the feature space, the conditional probability $P(v | \mb{x}_u)$ can be rewritten as
\begin{equation}
P(v | \mb{x}_u) = \frac{e^{\mb{x}_v^\top \cdot \mb{x}_u}}{\sum_{v' \in \mathcal{V}} e^{\mb{x}_{v'}^\top} \cdot \mb{x}_u}. 
\end{equation}
\end{itemize}

Therefore, the objective function can be simplified as
\begin{equation}
\max_{\mb{X}} \sum_{u \in \mathcal{V}} [-\log Z_u + \sum_{v' \in \Gamma(u)} \mb{x}_{v'}^\top \cdot \mb{x}_u],
\end{equation}
where $Z_u = \sum_{v' \in \mathcal{V}} e^{\mb{x}_{v'}^\top \cdot \mb{x}_u}$. Term $Z_u$ will be different for different nodes $u \in \mathcal{V}$, which is expensive to compute for large networks, and {node2vec} proposes to apply the negative sampling technique instead. The main issue discussed in {node2vec} is about sampling the neighborhood set $\Gamma(u)$ from the network.

\noindent \textbf{BFS and DFS}

In the \textit{SkipGram}, neighborhood set $\Gamma(u)$ denotes the direct neighbors of $u$ in the network, i.e., the \textit{first-order proximity} of network local structures. Besides the local structure, {node2vec} can also capture other network structures with set $\Gamma(u)$ depending on the sampling strategy being applied. To fairly compared different sampling strategies, the neighborhood set $\Gamma(u)$ is usually limited with size $k$, i.e., $|\Gamma(u)| = k$. Two extreme sampling strategies for the neighborhood set $\Gamma(u)$ are
\begin{itemize}

\item \textit{BFS}: BFS samples the nodes directly connected to node $u$ and involve them in the neighborhood set $\Gamma(u)$ first, and then go to the second layer, where the nodes are two hopes away from $u$ in the network, until the size $k$ is met. Generally, the $\Gamma(u)$ sampled via BFS can sufficiently characterize the local neighborhood structure of the network. The {node2vec} model learnt based on BFS sampling strategy provides a micro-view of the network structure.

\item \textit{DFS}: DFS samples the nodes which are sequentially reachable from $u$ at an increasing distance and involve them into the neighborhood set $\Gamma(u)$ first. In DFS, the sampled nodes reflect a more global neighborhood of the network. The {node2vec} model learnt based on BFS sampling strategy provides a macro-view of the network neighborhood structure of the network, which can be essential for inferring the communities based on homophily. 

\end{itemize}

However, the BFS and DFS sampling strategy may also suffer from some shortcomings. For BFS, only a small proportion of the network is explored surrounding node $u$ in the sampling. Meanwhile, for DFS, the sampled nodes far away from the source node $u$ tend to involve complex dependencies relationships.

\noindent \textbf{Random Walk based Search}

To overcome the shortcomings of BFS and DFS, {node2vec} proposes to apply random walk to sample the neighborhood set $\Gamma(u)$ instead. Given a random walk $W$, the node $W$ resides at in step $i$ can be represented as variable $s_i \in \mathcal{V}$. The complete sequence of nodes that $W$ has resides at can be represented as $s_0, s_1, \cdots, s_k$, where $s_0$ denotes the initial node starting the walk. The transitional probability from node $u$ to $v$ in $W$ in the $i_{th}$ step can be represented as
\begin{equation}
P(s_i = v | s_{i-1} = u) = 
\begin{cases}
w_{(u,v)} & \mbox{ if } (u, v) \in \mathcal{E},\\
0, & \mbox{ otherwise},
\end{cases}
\end{equation}
where $w_{(u,v)}$ denotes the normalized weight of link $(u,v)$ in the network ($w_{(u, v)} = 1$ if the network is unweighted). 

Traditional random walk model doesn't take account for the network structure and can hardly explore different network neighborhoods. {node2vec} adapts the random walk model and introduce the $2_{nd}$ order random walk model with parameters $p$ and $q$, which will help guide the walk. In {node2vec}, let's assume the walk just traversed link $(t, u)$ and can go to node $v$ in the next step. Formally, the transitional probability of link $(u, v)$ is adjusted with parameter $\alpha_{p,q}(t, v)$ (i.e., $w_{(u,v)} = \alpha_{p,q}(t, v) \cdot w_{(u,v)}$), where
\begin{equation}
\alpha_{p,q}(t, v) = \begin{cases}
\frac{1}{p}, & \mbox{ if } d_{t,v} = 0,\\
1, & \mbox{ if } d_{t,v} = 1,\\
\frac{1}{q}, & \mbox{ if } d_{t,v} = 2,
\end{cases}
\end{equation}
where $d_{t,v}$ denotes the shortest distance between nodes $t$ and $v$ in the network. Since the walk can go from $t$ to $u$, and then from $u$ to $v$, the distance from $t$ to $v$ will be at most $2$. 

Parameters $p$ and $q$ control the walk transition sequence effectively, where parameter $p$ is also called the \textit{return parameter} and $q$ is called the \textit{in-out parameter} in {node2vec}. 
\begin{itemize}
\item \textit{Return Parameter $p$}: In the case that $d_{t,v} = 0$, i.e., $t = v$, the probability adjusting parameter $\frac{1}{p}$ controls the chance to returning to the node $t$. By assigning $p$ with a large value, the random walk model will have a lower chance to go back to node $t$ that the model has just visited. Meanwhile, by assigning $p$ with a small value, the random walk model will backtrack a step and keep exploring the local nodes that it has visited already.

\item \textit{In-out Parameter $q$}: In the case that $d_{t,v} = 2$, nodes $t$ and $v$ are not directly connected but are reachable via the intermediate node $u$. Therefore, parameter $q$ controls the chance of exploring the structure that are far away from the visited nodes. If $q > 1$, the random walk model is biased to explore nodes that are closer to $t$, since $\frac{1}{q}$ is smaller than the probability of visiting nodes in case that $d_{t,v} = 1$. Meanwhile, if $q < 1$, the random walk will be inclined to visit nodes that are far away from $t$ in the network instead.
\end{itemize}

\subsection{Heterogeneous Network Embedding}\label{sec:chap11_sec4_hin}

The embedding modes introduced in the previous section are proposed for homogeneous networks, which will encounter great challenges when applied to the heterogeneous networks. In this section, we will introduce the recent development of embedding problems for heterogeneous networks, including HNE (Heterogeneous Information Network Embedding) \cite{CHTQAH15}, Path-Augmented Heterogeneous Network Embedding \cite{CS16}, and HEBE (HyperEdge Based Embedding) \cite{GLTJNH16}.

\subsubsection{HNE: Heterogeneous Information Network Embedding}

Generally, the data available in the online social networks doesn't exist in isolation, and different types of data may co-exist simultaneously. For instances, in the posts and articles written by users online, there may exist both text and image. The co-existence interactions of text and image in the same articles can be formed either explicitly or implicitly with the linkages between text and images. Meanwhile, there also exist correlations between the text data as well as image data due to the hyperlinks among the text and common tags/categories shared by different images. The {HNE} \cite{CHTQAH15} model is proposed a heterogeneous information network involving text and image.

\noindent \textbf{Terminology Definition and Problem Formulation}

The network studied in HNE involves both text and images, which can be represented as the {Text-Image Heterogeneous Information Network} as follows:

\begin{defn}
(Text-Image Heterogeneous Information Network): Let $G = (\mathcal{V}, \mathcal{E})$ denote the heterogeneous information network involving text and image as the nodes, as well as diverse categories of links among them. Formally, the node set $\mathcal{V}$ can be decomposed into two disjoint subsets $\mathcal{V} = \mathcal{V}_T \cup \mathcal{I}$, where $\mathcal{T}$ denotes the text set and $\mathcal{I}$ represents the image set. Meanwhile, among the text, image as well as between text and images, there may exist different kinds of connections, which can be denoted as sets $\mathcal{E}_{T,T}$, $\mathcal{E}_{I,I}$, and $\mathcal{E}_{T,I}$ respectively in the link set $\mathcal{E}$.
\end{defn}

Furthermore, the text and image nodes are also summarized by unique content information. For instance, for each image $i_k \in \mathcal{I}$, it can be represented as a tensor $\mb{X}_k \in \mathbb{R}^{d_I \times d_I \times 3}$, where $d_I$ denotes the dimension of the image in RGB color space. Meanwhile, for each text $t_k \in \mathcal{T}$, it can be represented as a raw feature vector $\mb{z}_k \in \mathbb{R}^{d_T}$, where $d_T$ denotes the dimension of the text represented with the bag-of-words vectors normalized by TF-IDF. For the images involved in set $\mathcal{I}$, the connections among them can be represented as matrix $\mb{A}_{I,I} \in \{+1, -1\}^{|\mathcal{I}| \times |\mathcal{I}|}$, where entry $A_{I,I}(j,k) = +1$ if there exist a link connecting nodes $i_j$ and $i_k$ in the network; and  $A(i,j) = -1$ otherwise. In a similar way, the adjacency matrices $A_{T,T}$ and $A_{I,T}$ can be defined to represent the connections among texts as well as those between images and texts.

For all the connections among nodes in set $\mathcal{V}$, they can be represented with matrix $\mb{A} \in \{+1, -1\}^{|\mathcal{V}| \times |\mathcal{V}|}$, where entry $A(i,j) = +1$ if the corresponding nodes are connected by a link in the network; and  $A(i,j) = -1$ otherwise.

To handle the diverse information in the {Text-Image Heterogeneous Information Network}, a good way is to learn the feature vector representations of nodes inside the network. Formally, the network embedding problem studied here includes the learning of mappings $\mb{U}: \mb{X} \to \mathbb{R}^r$ and $\mb{V}: \mb{z} \to \mathbb{R}^r$ which will project the images and texts into a shared feature space of dimension $r$. Furthermore, the network structure can be preserved in the embedding process, where connected nodes will be projected to a close region. 

\noindent \textbf{HNE Model}

For each image $i_k \in \mathcal{I}$, {HNE} proposes to transform its representation from 3-way tensor $\mb{X}_k$ into a column vector $\mb{x}_k \in \mathbb{R}^{d_I'}$, where $d_I'$ denotes the dimension of the feature vector space. Different methods can be applied in the transformation. For instance, a simple way to do the transformation is to stack the column vectors of the image and append them together, in which case $d_I'$ will be equal to $d_I \times d_I \times 3$. Some other advanced techniques have also been proposed, like feature extraction of the images as well as pre-embedding of images, which will not be introduced here since they are not part of the network embedding problem studied in this section.

Formally, the linear mapping functions for the image and text data are denoted as matrices $\mb{U}: \mb{x} \to \mathbb{R}^r$ and $\mb{V}: \mb{z} \to \mathbb{R}^r$, which projects the data into a feature space of dimension $r$. The embedding process of image $i_j \in \mathcal{I}$ and text $t_k \in \mathcal{T}$ can be denoted as
\begin{align}
\tilde{\mb{x}}_j &= \mb{U}^\top \mb{x}_j,\\
\tilde{\mb{z}}_k &= \mb{V}^\top \mb{z}_k,
\end{align}
where vectors $\tilde{\mb{x}}_k$ and $\tilde{\mb{z}}_k$ denotes the embedded feature representation of image $i_k$ and text $t_k$ respectively.

The similarity between the embedded feature representation of images and texts can be defined as
\begin{align}
s(\mb{x}_j, \mb{x}_k) &= \tilde{\mb{x}}_j^\top \tilde{\mb{z}}_k  = \mb{x}_j^\top (\mb{U} \mb{U}) \mb{x}_k = \mb{x}_j^\top \mb{M}_{I,I} \mb{x}_k,\\
s(\mb{z}_j, \mb{z}_k) &= \tilde{\mb{z}}_j^\top \tilde{\mb{z}}_k  = \mb{z}_j^\top (\mb{V} \mb{V}) \mb{z}_k = \mb{z}_j^\top \mb{M}_{T,T} \mb{z}_k.
\end{align}
respectively. Furthermore, since the images and texts are embedded into a common feature space, the similarity between the nodes of different categories can be represented as
\begin{equation}
s(\mb{x}_j, \mb{z}_k) = \tilde{\mb{x}}_j^\top \tilde{\mb{z}}_k  = \mb{x}_j^\top (\mb{U} \mb{V}) \mb{z}_k = \mb{x}_j^\top \mb{M}_{I,T} \mb{z}_k.
\end{equation}
In the above equations, via the positive semi-definite matrices $\mb{M}_{I,I}$, $\mb{M}_{T,T}$, $\mb{M}_{I,T}$ the similarity of the texts and images can be effectively captured. 

Meanwhile, based on the network structure, the empirical similarities of the nodes in the networks can be denoted by their structures. For instance, the empirical similarity between images $i_j, i_k \in \mathcal{I}$ can be denoted as
\begin{equation}
\hat{s}(\mb{x}_j, \mb{x}_k) = A_{I,I}(j,k).
\end{equation}
The loss function introduced by the image pair $i_j, i_k$ is defined as
\begin{equation}
L(\mb{x}_j, \mb{x}_k) = \log \left(1 + e^{(- A_{I,I}(j,k) s(\mb{x}_j, \mb{x}_k))} \right).
\end{equation}

In a similar way, the loss functions for the text pairs, and image-text pairs can be defined. By combining the loss functions together, the objective function of {HNE} can be represented as \begingroup\makeatletter\def\f@size{8}\check@mathfonts
\begin{align}
&\min_{\mb{U}, \mb{V}} \ \ \frac{1}{N_{I,I}} \sum_{i_j, i_k \in \mathcal{I}} L(\mb{x}_j, \mb{x}_k) + \frac{\lambda_1}{N_{T,T}} \sum_{t_j, t_k \in \mathcal{T}} L(\mb{z}_j, \mb{z}_k) \\
&+ \frac{\lambda_2}{N_{I,T}} \sum_{i_j \in \mathcal{I}, t_k \in \mathcal{T}} L(\mb{x}_j, \mb{z}_k) + \lambda_3 (\left\| \mb{U} \right\|_F^2 + \left\| \mb{V} \right\|_F^2),
\end{align}\endgroup
where ${N_{I,I}} = | \mathcal{I} \times \mathcal{I} \setminus \{(i_j, i_j)\}_{i_j \in \mathcal{I}}|$ denotes the number of image pairs, and $\lambda_1$, $\lambda_2$, $\lambda_3$ denote the weights of the loss terms introduced by texts, image-text, and the regularization term respectively. The function can be solved alternatively with coordinate descent by fixing one variable and updating the other variable. More detailed information about the solution is available in \cite{CHTQAH15}.

\subsubsection{Path-Augmented Heterogeneous Network Embedding}

For most of the embedding models, they are based on the assumptions that the node feature representations can be learnt with the neighborhood. Here, the neighborhood denotes either the set of nodes directed connected to the target node or the nodes accessible to the target node via random walk. In \cite{CS16}, a new heterogeneous network embedding model has been introduced, which uses the meta path to exploit the rich information information in heterogeneous networks.

In the path augmented network embedding model, a set of meta paths are defined based on the heterogeneous network schema. For the node pairs in the network which are connected based on each of the meta paths, their correlation is represented with a meta path augmented adjacency matrix. For instance, based on the $r_{th}$ type of meta path, the corresponding adjacency matrix can be denoted as $\mb{M}^r$. In heterogeneous networks, some of the meta paths will lots of concrete meta path instances connecting nodes. For instance, in the online social networks, the meta path ``User $\xrightarrow{write}$ Post $\xrightarrow{contain}$ Word $\xleftarrow{contain}$ Post $\xleftarrow{write}$ User'' will have lots of instances, since users write lots of posts and each post will contain many words. Therefore, matrix $\mb{M}^r$ is usually normalized to ensure $\sum_{i,j} M^r(i,j) = 1$. 

The learning framework used here is very similar to those introduced {LINE} and {node2vec} in Sections~\ref{subsec:chap11_sec3_line} and \ref{subsec:chap11_sec3_node2vec}. The proximity between nodes $n_i, n_j \in \mathcal{V}$ based on the $r_{th}$ meta path can be denoted as
\begin{equation}
P(n_j | n_i; r) = \frac{e^{\mb{x}_i^\top \mb{x}_j}}{\sum_{j' \in DST(r)} e^{\mb{x}_i^\top \mb{x}_j}},
\end{equation}
where $\mb{x}_i$ and $\mb{x}_j$ denote the embedding vectors of nodes $n_i$ and $n_j$ respectively, and $DST(r)$ denotes the set of all possible nodes that are in the destination side of path $r$.

In the real world, set $DST(r)$ is usually very large, which renders the above conditional probability very expensive to compute. In \cite{CS16}, the authors propose to follow the techniques proposed in the existing works, and applies negative sampling to reduce the computation costs. Formally, the approximated objective function can be represented as \begingroup\makeatletter\def\f@size{8}\check@mathfonts
\begin{align}
&\log \tilde{P}(n_j | n_i; r) \\
&\approx log \sigma({\mb{x}_i^\top \mb{x}_j}) + \sum_{l = 1}^k \mathbb{E}_{n_{j'} \sim P^r_n(n_{j'})}[\log \sigma (-{\mb{x}_i^\top \mb{x}_{j'}} - b_r)],
\end{align}\endgroup
where $j'$ denotes the negative node sampled from the pre-defined noise distribution, $k$ denotes the number of sampled nodes, and $b_r$ is the bias term added for the $r_{th}$ meta path. The embedding vectors $\mb{x}_{n_i}$ for node $n_i$ in the network as well as the bias terms $b_r$ for the $r_{th}$ meta path can be learnt with the stochastic gradient descent method

\subsubsection{HEBE: HyperEdge Based Embedding}

The embedding models proposed so far mostly only consider the \textit{single typed} objective interactions, while the \textit{strongly typed} objects involving multiple kinds of interactions among different objectives has achieved an increasing interest in recent years. In this part, we will introduce a new embedding framework HEBE (HyperEdge Based Embedding) which captures strongly-typed objective interactions as a whole in the embedding process \cite{GLTJNH16}.

\noindent \textbf{Terminology Definition and Problem Formulation}

In HEBE, the subgraph centered with one certain type of target object in the whole network is defined as an \textit{event}. Depending on the number of node types involved in the \textit{event}, they can be further categorized into \textit{homogeneous event} and \textit{heterogeneous event}

\begin{defn}
(Event): Formally, the objects involved in the network can be represented as set $\mathcal{X} = \{\mathcal{X}_t\}_{t=1}^T$, where $\mathcal{X}_t$ denotes the set of objects belonging to the $t_{th}$ type. An event $Q_i$ is denoted as a subset of nodes involved in it and can be represented as $(\mathcal{V}_i, w_i)$, where $\mathcal{V}_i$ denotes the set of involved objects and $w_i$ is the occurrence number of event $Q_i$ in the network. The object set $\mathcal{V}_i$ can be further divided into several subsets $\mathcal{V}_i = \bigcup_{t = 1}^T \mathcal{V}_i^t$ depending on the object categories.
\end{defn}
In the above event definition, links connecting the nodes in the network are involved by default, which are not mentioned here for simplicity reasons. For event $Q_i = (\mathcal{V}_i, w_i)$, if more than one type of nodes are covered, it will be called a homogeneous event; otherwise, it is a heterogeneous event. 

Formally, the set of events involved in the network can be represented as \textit{event data} $\mathcal{D} = \{Q_i\}_i^N$. In the embedding problem, the objective is to learn a function $f: \mathcal{X} \to \mathbb{R}^d$ to project the different types of objects involved in the \textit{event data} $\mathcal{D}$ into a shared feature space of dimension $d$. Meanwhile, the \textit{proximity} of each event should be preserved. Here, the \textit{proximity} of an event is defined as the likelihood of observing a target object given all other participating objects in the same event.

\noindent \textbf{Objective Function Introduction}

Given an event $Q_i = (\mathcal{V}_i, w_i)$, let $u \in \mathcal{V}_i$ denote an object involved in the event. The remaining nodes in the event can be denoted as the context of $u$, i.e., $\mathcal{C} = \mathcal{V}_i \setminus \{u\}$. Let's assume object $u$ belongs to category $\mathcal{X}_1$ (i.e., $u \in \mathcal{X}_1$), the probability of predicting the target object $u$ given its context $\mathcal{C}$ is defined as
\begin{equation}
P(u | \mathcal{C}) = \frac{e^{S(u, \mathcal{C})}}{\sum_{v \in \mathcal{X}_1} e^{S(v, \mathcal{C})}},
\end{equation}
where $S(u, \mathcal{C})$ denotes the similarity between $u$ and context $\mathcal{C}$ and can be calculated by summing the inner products of object pairs in $\{u\} \times \mathcal{C}$. 

The loss function defined in HEBE is based on the Kullback-Leibler (KL) divergence between the conditional probability $P(\cdot | \mathcal{C})$ and the emperical probability $\hat{P}(\cdot | \mathcal{C})$, which can be defined as
\begin{equation}
\mathcal{L} = - \sum_{t=1}^T \sum_{\mathcal{C}_t \in \mathcal{P}_t} \lambda_{\mathcal{C}_t} KL(P(\cdot | \mathcal{C}), \hat{P}(\cdot | \mathcal{C})),
\end{equation}
where $\lambda_{\mathcal{C}_t}$ denotes the weight of context $\mathcal{C}_t$ and is defined as the occurrence of it in the event data $\mathcal{D}$
\begin{equation}
\lambda_{\mathcal{C}_t} = \sum_{i = 1}^N \frac{w_i \mb{I}(\mathcal{C}_t \in \mathcal{V}_i)}{|\mathcal{P}_{i,t}|}.
\end{equation}
In the above equation, $\mathcal{P}_t$ denotes the sample space of context $\mathcal{C}_t$ and $\mathcal{P}_{i,t}$ is the constraint sample space by object set $\mathcal{V}_i$. Function $\mb{I}(\cdot)$ is a binary function which takes value $1$ if the condition holds. By replacing $\lambda_{\mathcal{C}_t}$, the loss function can be rewritten as follows
\begin{equation}
\mathcal{L} = - \sum_{i = 1}^N w_i \sum_{t=1}^T \frac{1}{|\mathcal{P}_{i,t}|} \sum_{\mathcal{C}_t \in \mathcal{P}_t} P(\cdot | \mathcal{C}),
\end{equation}

\noindent \textbf{Learning Algorithm Description}

The conditional probability involved in the loss function is very hard to calculate especially in the case that the object set $\mathcal{X}_1$ that $u$ belongs to is very big. To address the problem, HEBE proposes to use the \textit{noise pairwise ranking} (NPR) to approximate the probability calculation instead. 

Formally, the conditional probability function can be rewritten as
\begin{equation}
P(u | \mathcal{C}) = \left(1 + \sum_{v \in \mathcal{X}_1 \setminus \{u\}} e^{S(v, \mathcal{C}) - S(u, \mathcal{C})} \right)^{-1}.
\end{equation}
Instead of enumerating all the nodes $v \in \mathcal{X}_1 \setminus \{u\}$, a small set of noise samples are selected from $\mathcal{X}_1 \setminus \{u\}$, where an individual noise sample can be denoted as $v_n$. HEBE propose to maximize the following probability instead 
\begin{equation}
P(u > u_n | \mathcal{C}) = \sigma(-S(v_n, \mathcal{C}) + S(u, \mathcal{C})).
\end{equation}
It is shown that
\begin{equation}
P(u | \mathcal{C}) > \prod_{v_n \neq u} P(u > v_n | \mathcal{C}).
\end{equation}
And the conditional probability can be approximated as follows
\begin{equation}
P(u | \mathcal{C}) \propto \mathbb{E}_{v_n \sim P_n} \log P(u > v_n | \mathcal{C}),
\end{equation}
where $P_n$ denotes the noise distribution and it is set as $P_n \propto D(u)^{\frac{3}{4}}$ with regarding to the degree of $u$.
By replacing the probability into the loss function, the loss function will be  \begingroup\makeatletter\def\f@size{8}\check@mathfonts
\begin{equation}
\tilde{\mathcal{L}} = - \sum_{i = 1}^N w_i \sum_{t=1}^T \frac{1}{|\mathcal{P}_{i,t}|} \sum_{\mathcal{C}_t \in \mathcal{P}_t} \mathbb{E}_{v_n \sim P_n} \log P(u > v_n | \mathcal{C}).
\end{equation}\endgroup
The objective function can be solved with the asynchronous stochastic gradient descent (ASGD) algorithm.

\subsection{Emerging Network Embedding across Networks}\label{sec:chap11_sec5_aligned_hin}

We have introduce several network embedding models in the previous sections already. However, when applied to handle real-world social network data, these existing embedding models can hardly work well. The main reason is that the network internal social links are usually very sparse in online soical networks \cite{TQWZYM15}, which can hardly preserve the complete network structure. For a pair of users who are not directed connected, these models will not be able determine the closeness of these users' feature vectors in the embedding space. Such a problem will be more severe when it comes to the \textit{emerging social networks} \cite{sdm15}, which denote the newly created online social networks containing very few social connections.

In this section, we will study the emerging network embedding problem across multiple aligned heterogeneous social networks simultaneously. In the concurrent embedding process, the emerging network embedding problem aims at distilling relevant information from both the emerging and other aligned mature networks to derive compliment knowledge and learn a good vector representation for user nodes in the emerging network. Formally, the studied problem can be formulated as follows.

Given two aligned networks $\mathcal{G} = ((G^{(1)}, G^{(2)}), (\mathcal{A}^{(1,2)}))$, where $G^{(1)}$ is an emerging network and $G^{(2)}$ is a mature network. In the {emerging network embedding} problem, we aim at learning a mapping function $f^{(i)}: \mathcal{U}^{(i)} \to \mathbb{R}^{d^{(i)}}$ to project the user node in $G^{(i)}$ to a feature space of dimension $d^{(i)}$ ($d^{(i)} \ll |\mathcal{U}|^{(i)}$). The objective of mapping functions $f^{(i)}$ is to ensure the embedding results can preserve the network structural information, where similar user nodes will be projected to close regions. Furthermore, in the embedding process, {emerging network embedding} also wants to transfer information between $G^{(2)}$ and $G^{(1)}$ to overcome the information sparsity problem in $G^{(1)}$.

To solve the problem, in this section, we will introduce a novel multiple aligned heterogeneous social network embedding framework, named {\dime} proposed in \cite{icdm17}. To handle the heterogeneous link and attribute information in the networks in a unified analytic, {\dime} introduces the \textit{aligned attribute augmented heterogeneous network} concept. From these networks a set of meta paths are introduced to represent the diverse connections among users in online social networks,  and a set of \textit{meta proximity} measures are defined for each of the meta paths denoting the closeness among users. These meta proximity information will be fed into a deep learning framework, which takes the input information from multiple aligned heterogeneous social networks simultaneously, to achieve the embedding feature vectors for all the users in these aligned networks. Based on the connection among users, framework {\dime} aims at embedding close user nodes to a close area in the lower-dimensional feature space for each of the social network respectively. Meanwhile, framework {\dime} also poses constraints on the feature vectors corresponding to the shared users across networks to map them to a relatively close region as well. In this way, information can be transferred from the mature networks to the emerging network and solve the \textit{information sparsity} problem.

\subsubsection{Proposed Methods}

For each attributed heterogeneous social network, the closeness among users can be denoted by the friendship links among them, where friends tend to be closer compared with user pairs without connections. Meanwhile, for the users who are not directly connected by the friendship links, few existing embedding methods can figure out their closeness, as these methods are mostly built based on the direct friendship link only. In this section, the potential closeness scores among the users can be computed with the heterogeneous information in the networks based on meta path concept \cite{SAH12}, which are formally called the \textit{meta proximity} in \cite{icdm17}.

\noindent \textbf{Friendship based Meta Proximity}\label{subsec:chap11_sec5_friendship_proximity}

In online social networks, the friendship links are the most obvious indicator of the social closeness among users. Online friends tend to be closer with each other compared with the user pairs who are not friends. Users' friendship links also carry important information about the local network structure information, which should be preserved in the embedding results. Based on such an intuition, the \textit{friendship based meta proximity} concept can be represented as follows.

\begin{defn}
(Friendship based Meta Proximity): For any two user nodes $u^{(1)}_i, u^{(1)}_j$ in an online social network (e.g., $G^{(1)}$), if $u^{(1)}_i$ and $u^{(1)}_j$ are friends in $G^{(1)}$, the \textit{friendship based meta proximity} between $u^{(1)}_i$ and $u^{(1)}_j$ in the network is $1$, otherwise the \textit{friendship based meta proximity} score between them will be $0$ instead. To be more specific, the \textit{friendship based meta proximity} score between users $u^{(1)}_i, u^{(1)}_j$ can be represented as $p^{(1)}(u^{(1)}_i,u^{(1)}_j) \in \{0, 1\}$, where term $p^{(1)}(u^{(1)}_i,u^{(1)}_j) = 1$ iff $(u^{(1)}_i, u^{(1)}_j) \in \mathcal{E}^{(1)}_{u,u}$.
\end{defn}

Based on the above definition, the \textit{friendship based meta proximity} scores among all the users in network $G^{(1)}$ can be represented as matrix $\mb{P}^{(1)}_{\Phi_0} \in \mathbb{R}^{|\mathcal{U}^{(1)}| \times |\mathcal{U}^{(1)}|}$, where entry ${P}^{(1)}_{\Phi_0}(i,j)$ equals to $p^{(1)}(u^{(1)}_i, u^{(1)}_j)$. Here $\Phi_0$ denotes the simplest meta path of length $1$ in the form $\mbox{U} \xrightarrow{\mbox{follow}} \mbox{U}$, and its formal definition will be introduced in the following subsection.

When network $G^{(1)}$ is an emerging online social network which has just started to provide services for a very short time, the friendship links among users in $G^{(1)}$ tend to be very limited (majority of the users are isolated in the network with few social connections). In other words, the \textit{friendship based meta proximity} matrix $\mb{P}^{(1)}_{\Phi_0}$ will be extremely sparse, where very few entries will have value $1$ and most of the entries are $0$s. With such a sparse matrix, most existing embedding models will fail to work. The reason is that the sparse friendship information available in the network can hardly categorize the relative closeness relationships among the users (especially for those who are even not connected by friendship links), which renders these existing embedding models may project all the nodes to random regions.

To overcome such a problem, besides the social links, {\dime} proposes to calculate the potential proximity scores for the users with the diverse link and attribute information available in the heterogeneous networks. To handle the diverse links and attributes simultaneously in a unified analytic, {\dime} will treat the attributes as nodes as well and introduce the \textit{attribute augmented network}. If a node has certain attributes, a new type of link ``\textit{have}'' will be added to connected the node and the newly added attribute node. By extending the meta path definition introduced in Section~\ref{sec:meta_path} to incorporate the attribute information, set of different \textit{social meta path} $\{\Phi_0, \Phi_1, \Phi_2, \cdots, \Phi_7\}$ can be extracted from the network, whose notations, concrete representations and the physical meanings are illustrated in Table~\ref{tab:chap8_sec4_meta_path}. Here, meta paths $\Phi_0-\Phi_4$ are all based on the user node type and follow link type; meta paths $\Phi_5-\Phi_7$ involve the user, post node type, attribute node type, as well as the \textit{write} and \textit{have} link type. Based on each of the meta paths, there will exist a set of concrete meta path instances connecting users in the networks. For instance, given a user pair $u$ and $v$, they may have been checked-in at 5 different common locations, which will introduce $5$ concrete meta path instance of meta path $\Phi_7$ connecting $u$ and $v$ indicating their strong closeness (in location check-ins). In the next subsection, we will introduce how to calculate the proximity score for the users based on these extracted meta paths.

\noindent \textbf{Heterogeneous Network Meta Proximity}\label{subsec:chap11_sec5_meta_path_proximity}

The set of \textit{attribute augmented social meta paths} $\{\Phi_0, \Phi_1, \Phi_2, \\ \cdots, \Phi_7\}$ extracted in the previous subsection create different kinds of correlations among users (especially for those who are not directed connected by friendship links). With these \textit{social meta paths}, different types of proximity scores among the users can be captured. For instance, for the users who are not friends but share lots of common friends, they may also know each other and can be close to each other; for the users who frequently checked-in at the same places, they tend to be more close to each other compared with those isolated ones with nothing in common. Therefore, these meta paths can help capture much broader network structures compared with the local structure captured by the \textit{friendship based meta proximity} talked about in subsection~\ref{subsec:chap11_sec5_friendship_proximity}. In this part, we will introduce the method to calculate the proximity scores among users based on these \textit{social meta paths}.

Similar to the meta paths shown in Table~\ref{subsec:metapath_proximity}, all the social meta paths extracted from the networks can be represented as set $\{\Phi_1, \Phi_2, \cdots, \Phi_7\}$. Given a pair of users, e.g., $u^{(1)}_i$ and $u^{(1)}_j$, based on meta path $\Phi_k \in \{\Phi_1, \Phi_2, \cdots, \Phi_7\}$, the set of meta path instances connecting $u^{(1)}_i$ and $u^{(1)}_j$ can be represented as $\mathcal{P}_{\Phi_k}^{(1)}(u^{(1)}_i, u^{(1)}_j)$. Users $u^{(1)}_i$ and $u^{(1)}_j$ can have multiple meta path instances going into/out from them. Formally, all the meta path instances going out from user $u^{(1)}_i$ (or going into $u^{(1)}_j$), based on meta path $\Phi_k$, can be represented as set $\mathcal{P}_{\Phi_k}^{(1)}(u^{(1)}_i, \cdot)$ (or $\mathcal{P}_{\Phi_k}^{(1)}(\cdot, u^{(1)}_j)$). The proximity score between $u^{(1)}_i$ and $u^{(1)}_j$ based on meta path $\Phi_k$ can be represented as the following \textit{meta proximity} concept formally.

\begin{defn}(Meta Proximity): Based on social meta path $\Phi_k$, the meta proximity between users $u^{(1)}_i$ and $u^{(1)}_j$ in network $G^{(1)}$ can be represented as 
\begin{equation}
p^{(1)}_{\Phi_k}(u^{(1)}_i, u^{(1)}_j) = \frac{2|\mathcal{P}_{\Phi_k}^{(1)}(u^{(1)}_i, u^{(1)}_j)|}{|\mathcal{P}_{\Phi_k}^{(1)}(u^{(1)}_i, \cdot)| + |\mathcal{P}_{\Phi_k}^{(1)}(\cdot, u^{(1)}_j)|}.
\end{equation}
\end{defn}

\textit{Meta proximity} considers not only the meta path instances between users but also penalizes the number of meta path instances going out from/into $u^{(1)}_i$ and $u^{(1)}_j$ at the same time. It is also reasonable. For instance, sharing some common location check-ins with some extremely active users (who have tens thousand checkins) may not necessarily indicate closeness with them, since they may have common check-ins with so many other users due to his very large check-in record volume.

With the above meta proximity definition, the meta proximity scores among all users in the network $G^{(1)}$ based on meta path $\Phi_k$ can be denoted as matrix $\mb{P}^{(1)}_{\Phi_k} \in \mathbb{R}^{|\mathcal{U}^{(1)}| \times |\mathcal{U}^{(1)}|}$, where entry ${P}^{(1)}_{\Phi_k}(i,j) = p^{(1)}_{\Phi_k}(u^{(1)}_i, u^{(1)}_j)$. All the meta proximity matrices defined for network $G^{(1)}$ can be represented as $\{\mb{P}^{(1)}_{\Phi_k}\}_{\Phi_k}$. Based on the meta paths extracted for network $G^{(2)}$, similar matrices can be defined as well, which can be denoted as $\{\mb{P}^{(2)}_{\Phi_k}\}_{\Phi_k}$.

\begin{figure*}
	\centering
	\includegraphics[width=0.8\textwidth]{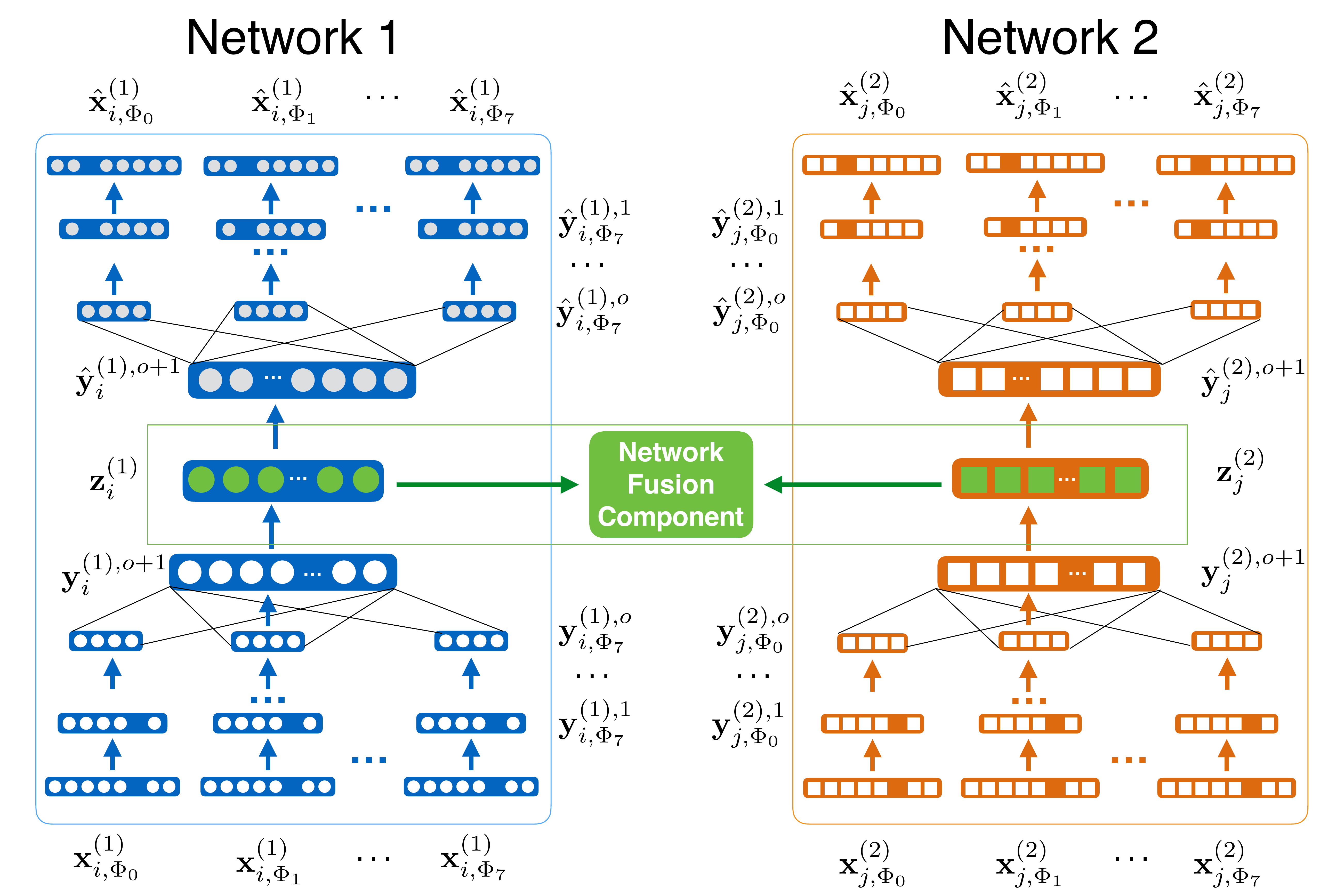}
	\caption{The {\dime} Framework.}
	\label{fig:chap11_sec5_deep}
\end{figure*}

\noindent \textbf{Deep {\dimesh} Model}\label{subsec:chap11_sec5_single}

With these calculated \textit{meta proximity} introduced in the previous section, we will introduce the embedding framework {\dime} next. {\dime} is based on the \textit{aligned auto-encoder model}, which extends the traditional \textit{deep auto-encoder model} to the \textit{multiple aligned heterogeneous networks} scenario. In this part, we will talk about the embedding model component for one heterogeneous information network in Section~\ref{subsec:chap11_sec5_single}, which takes the various meta proximity matrices as the input. {\dime} effectively couples the embedding process of the emerging network with other aligned mature networks, where cross-network information exchange and result refinement is achieved via the loss term defined based on the anchor links, which will be introduced in the next part.

When applying the auto-encoder model for one single homogeneous network node embedding, e.g., for $G^{(1)}$, the model can be learned with the node meta proximity feature vectors, i.e., rows corresponding to users in matrix $\mb{P}^{(1)}_{\Phi_0}$ (introduced in Section~\ref{subsec:chap11_sec5_friendship_proximity}). In the case that $G^{(1)}$ is heterogeneous, multiple node \textit{meta proximity} matrices have been defined before (i.e., $\{\mb{P}^{(1)}_{\Phi_0}, \mb{P}^{(1)}_{\Phi_1}, \cdots, \mb{P}^{(1)}_{\Phi_7}\}$), how to fit these matrices simultaneously to the auto-encoder models is an open problem. In this part, we will introduce the single-heterogeneous-network version of framework {\dime}, namely {\dimesh}, which will be used as an important component of framework {\dime} as well. For each user node in the network, {\dimesh} computes the embedding vector based on each of the proximity matrix independently first, which will be further fused to compute the final latent feature vector in the output hidden layer.

As shown in the architecture in Figure~\ref{fig:chap11_sec5_deep} (either the left component for network 1 or the right component for network 2), about the same instance, {\dimesh} takes different feature vectors extracted from the meta paths $\{\Phi_0, \Phi_1, \cdots, \Phi_7\}$ as the input. For each meta path, a series of separated encoder and decoder steps are carried out simultaneously, whose latent vectors are fused together to calculate the final embedding vector $\mb{z}^{(1)}_i \in \mathbb{R}^{d^{(1)}}$ for user $u_i^{(1)} \in \mathcal{V}^{(1)}$. In the {\dimesh} model, the input feature vectors (based on meta path $\Phi_k \in \{\Phi_0, \Phi_1, \cdots, \Phi_7\}$) of user $u_i$ can be represented as $\mb{x}^{(1)}_{i, \Phi_k}$, which denotes the row corresponding to users $u_i^{(1)}$ in matrix $\mb{P}^{(1)}_{\Phi_k}$ defined before. Meanwhile, the latent representation of the instance based on the feature vector extracted via meta path $\Phi_k$ at different hidden layers can be represented as $\{\mb{y}^{(1),1}_{i, \Phi_k}, \mb{y}^{(1),2}_{i, \Phi_k}, \cdots, \mb{y}^{(1),o}_{i, \Phi_k}\}$. 


One of the significant difference of model {\dimesh} from traditional auto-encoder model lies in the (1) combination of various hidden vectors $\{\mb{y}^{(1),o}_{i, \Phi_0}, \mb{y}^{(1),o}_{i, \Phi_1}, \cdots, \mb{y}^{(1),o}_{i, \Phi_7}\}$ to obtain the final embedding vector $\mb{z}^{(1)}_i$ in the encoder step, and (2) the dispatch of the embedding vector $\mb{z}^{(1)}_i$ back to the hidden vectors in the decoder step. As shown in the architecture, formally, these extra steps can be represented as \begingroup\makeatletter\def\f@size{7}\check@mathfonts
\begin{align}\hspace{-5pt}
\begin{cases}
& \hspace{-7pt} \mbox{\# extra encoder steps}\\
& \hspace{-7pt} \mb{y}^{(1), o+1}_i =  \sigma(\sum_{\Phi_k \in \{\Phi_0, \cdots, \Phi_7\}} \mb{W}^{(1),o+1}_{\Phi_k} \mb{y}^{(1),o}_{i, \Phi_k} + \mb{b}^{(1),o+1}_{\Phi_k}),\\
& \hspace{-7pt} \mb{z}^{(1)}_i = \sigma(\mb{W}^{(1),o+2} \mb{y}^{(1), o+1}_i + \mb{b}^{(1),o+2}).\\
& \hspace{-7pt} \mbox{\# extra decoder steps }\\
& \hspace{-7pt} \hat{\mb{y}}^{(1), o+1}_i =  \sigma(\hat{\mb{W}}^{(1),o+2} \mb{z}^{(1)}_i + \hat{\mb{b}}^{(1),o+2}),\\
& \hspace{-7pt} \hat{\mb{y}}^{(1),o}_{i, \Phi_k} = \sigma(\hat{\mb{W}}^{(1),o+1}_{\Phi_k} \hat{\mb{y}}^{(1), o+1}_i + \hat{\mb{b}}^{(1),o+1}_{\Phi_k}).
\end{cases}
\end{align}\endgroup

What's more, since the input feature vectors are extremely sparse (lots of the entries have value $0$s), simply feeding them to the model may lead to some trivial solutions, like $\mb{0}$ vectors for both $\mb{z}^{(1)}_i$ and the decoded vectors $\hat{\mb{x}}_{i, \Phi_k}^{(1)}$. To overcome such a problem, another significant difference of model {\dimesh} from traditional auto-encoder model lies in the loss function definition, where the loss introduced by the non-zero features will be assigned with a larger weight. In addition, by adding the loss function for each of the meta paths, the final loss function in {\dimesh} can be formally represented as \begingroup\makeatletter\def\f@size{8}\check@mathfonts
\begin{equation}
\mathcal{L}^{(1)} = \sum_{\Phi_k \in \{\Phi_0, \cdots, \Phi_7\}}\sum_{u_i \in \mathcal{V}} \left \| \left( \mb{x}^{(1)}_{i, \Phi_k} - \hat{\mb{x}}^{(1)}_{i,\Phi_k} \right) \odot \mb{b}^{(1)}_{i,\Phi_k} \right\|_2^2,
\end{equation}\endgroup
where vector $\mb{b}^{(1)}_{i,\Phi_k}$ is the weight vector corresponding to feature vector $\mb{x}^{(1)}_{i,\Phi_k}$. Entries in vector $\mb{b}^{(1)}_{i,\Phi_k}$ are filled with value $1$s except the entries corresponding to non-zero element in $\mb{x}^{(1)}_{i, \Phi_k}$, which will be assigned with value $\gamma$ ($\gamma > 1$ denoting a larger weight to fit these features). In a similar way, the loss function for the embedding result in network $G^{(2)}$ can be formally represented as $\mathcal{L}^{(2)}$.

\noindent \textbf{Deep {\dime} Framework}\label{subsec:chap11_sec5_framework}

Even through {\dimesh} has incorporate all these heterogeneous information in the model building, the meta proximity calculated based on which can help differentiate the closeness among different users. However, for the emerging networks which just start to provide services, the information sparsity problem may affect the performance of {\dimesh} significantly. In this part, we will introduce {\dime}, which couples the embedding process of the emerging network with another mature aligned network. By accommodating the embedding between the aligned networks, information can be transferred from the aligned mature network to refine the embedding result in the emerging network  effectively. The complete architecture of {\dime} is shown in Figure~\ref{fig:chap11_sec5_deep}, which involve the {\dimesh} components for each of the aligned networks, where the information transfer component aligns these separated {\dimesh} models together.

To be more specific, given a pair of aligned heterogeneous networks $\mathcal{G} = ((G^{(1)}, G^{(2)}), \mathcal{A}^{(1,2)})$ ($G^{(1)}$ is an emerging network and $G^{(2)}$ is a mature network), the embedding results can be represented as matrices  $\mb{Z}^{(1)} \in \mathbb{R}^{|\mathcal{U}^{(1)}| \times d^{(1)}}$ and $\mb{Z}^{(2)} \in \mathbb{R}^{|\mathcal{U}^{(2)}| \times d^{(2)}}$ for all the user nodes in $G^{(1)}$ and $G^{(2)}$ respectively. The $i_{th}$ row of matrix $\mb{Z}^{(1)}$ (or the $j_{th}$ row of matrix $\mb{Z}^{(2)}$) denotes the encoded feature vector of user $u^{(1)}_i$ in $G^{(1)}$ (or $u^{(2)}_j$ in $G^{(2)}$). If $u^{(1)}_i$ and $u^{(2)}_j$ are the same user, i.e., $(u^{(1)}_i, u^{(2)}_j) \in \mathcal{A}^{(1,2)}$, by placing vectors $\mb{Z}^{(1)}(i,:)$ and $\mb{Z}^{(2)}(j,:)$ in a close region in the embedding space, the information from $G^{(2)}$ can be used to refine the embedding result in $G^{(1)}$.

Information transfer is achieved based on the anchor links, and we only care about the anchor users. To adjust the rows of matrices $\mb{Z}^{(1)}$ and $\mb{Z}^{(2)}$ to remove non-anchor users and make the same rows correspond to the same user, {\dime} introduces the binary inter-network transitional matrix $\mb{T}^{(1,2)} \in \mathbb{R}^{|\mathcal{U}^{(1)}| \times |\mathcal{U}^{(2)}|}$. Entry $T^{(1,2)}(i,j) = 1$ iff the corresponding users are connected by anchor links, i.e., $(u^{(1)}_i, u^{(2)}_j) \in \mathcal{A}^{(1,2)}$. Furthermore, the encoded feature vectors for users in these two networks can be of different dimensions, i.e., $d^{(1)} \neq d^{(2)}$, which can be accommodated via the projection $\mb{W}^{(1,2)} \in \mathbb{R}^{d^{(1)} \times d^{(2)}}$. 

Formally, the introduced \textit{information fusion loss} between networks $G^{(1)}$ and $G^{(2)}$ can be represented as
\begin{equation}
\mathcal{L}^{(1,2)} = \left\| (\mb{T}^{(1,2)})^\top \mb{Z}^{(1)} \mb{W}^{(1,2)} - \mb{Z}^{(2)}  \right\|_F^2.
\end{equation}
By minimizing the \textit{information fusion loss} function $\mathcal{L}^{(1,2)}$, the anchor users' embedding vectors from the mature network $G^{(2)}$ can be used to adjust his embedding vectors in the emerging network $G^{(1)}$. Even through in such a process the embedding vector in $G^{(2)}$ can be undermined by $G^{(1)}$, it will not be a problem since $G^{(1)}$ is the target network and {\dime} only care about the embedding result of the emerging network $G^{(1)}$ in \cite{icdm17}. 

The complete objective function of framework include the loss terms introduced by the component {\dimesh} for networks $G^{(1)}$, $G^{(2)}$, and the \textit{information fusion loss}, which can be denoted as
\begin{equation}
\mathcal{L}(G^{(1)}, G^{(2)}) = \mathcal{L}^{(1)} + \mathcal{L}^{(2)} + \alpha \cdot \mathcal{L}^{(1,2)}  + \beta \cdot \mathcal{L}_{reg}.
\end{equation}
Parameters $\alpha$ and $\beta$ denote the weights of the \textit{information fusion loss} term and the regularization term. In the objective function, term $\mathcal{L}_{reg}$ is added to the above objective function to avoid overfitting, which can be formally represented as \begingroup\makeatletter\def\f@size{6}\check@mathfonts
\begin{align}\hspace{-7pt}
\begin{cases}
&\hspace{-7pt} \mathcal{L}_{reg} = \mathcal{L}_{reg}^{(1)} + \mathcal{L}_{reg}^{(2)} + \mathcal{L}_{reg}^{(1,2)},\\
&\hspace{-7pt} \mathcal{L}_{reg}^{(1)} = \sum_{i}^{o^{(1)}+2} \sum_{\Phi_k \in \{\Phi_0, \cdots, \Phi_7\}} \left( \left\| \mb{W}^{(1),i}_{\Phi_k} \right\|_F^2 + \left\| \hat{\mb{W}}^{(1),i}_{\Phi_k} \right\|_F^2 \right),\\
&\hspace{-7pt} \mathcal{L}_{reg}^{(2)} = \sum_{i}^{o^{(2)}+2} \sum_{\Phi_k \in \{\Phi_0, \cdots, \Phi_7\}} \left( \left\| \mb{W}^{(2),i}_{\Phi_k} \right\|_F^2 + \left\| \hat{\mb{W}}^{(2),i}_{\Phi_k} \right\|_F^2 \right),\\
&\hspace{-7pt} \mathcal{L}_{reg}^{(1,2)} =  \left\| \mb{W}^{(1,2)} \right\|_2^2.
\end{cases}
\end{align}\endgroup

To optimize the above objective function, we utilize Stochastic Gradient Descent (SGD). To be more specific, the training process involves multiple epochs. In each epoch, the training data is shuffled and a minibatch of the instances are sampled to update the parameters with SGD. Such a process continues until either convergence or the training epochs have been finished.

\section{Conclusion and Future Developments}\label{sec:future_works}

In this paper, we have introduced the current research works on broad learning and its applications on social media studies. This paper has covered $5$ main research directions about broad learning based social media studies: (1) \textit{network alignment}, (2) \textit{link prediction}, (3) \textit{community detection}, (4) \textit{information diffusion} and (5) \textit{network embedding}. These problems introduced in this chapter are all very important for many concrete real-world social network applications and services. A number of nontrivial algorithms have been proposed to resolve these problems, which have been talked about in great detail in this paper respectively.

Both the \textit{broad learning} and \textit{social media mining} are very promising research directions, and some potential future development directions are illustrated as follows.

\begin{enumerate}

\item {\textbf{Scalable Broad Learning Algorithms}}: Data generated nowadays is usually of very large scale, and fusion of such big data from multiple sources together will render the problem more challenging. For instance, the online social networks (like Facebook) usually involve millions even billions of active users, and the social data generated by these users in each day will consume more than 600 TB storage space (in Facebook). One of the major future development about the \textit{broad learning based social media mining} is to develop scalable data fusion and mining algorithms that can handle such a \textbf{large volume} (of \textbf{big data}) challenge. One tentative approach is to develop information fusion algorithms based on distributed platforms, like Spark and Hadoop \cite{bigdata14}, and handle the data with a large distributed computing cluster. Another method to resolve the scalability challenge is from the model optimization perspective. Optimizing existing learning models and proposing new approximated learning algorithms with lower time complexity are desirable in the future research projects. In addition, applications of the latest deep learning models to fuse and mine the large-scale datasets can be another alternative approach for the scalable \textit{broad learning} on social networks.

\item {\textbf{Multiple Sources Fusion and Mining}}: Current research works on multiple source data fusion and mining mainly focus on aligning entities in one single pair of data sources (i.e., two sources), where information exchange between the sources mainly rely on the anchor links between these aligned entities. Meanwhile, when it comes to fusion and mining of multiple (more than two) sources, the problem setting will be quite different and become more challenging. For example, in the alignment of more networks, the transitivity property of the inferred anchor links needs to be preserved \cite{icdm15}. Meanwhile, in the information transfer from multiple external aligned sources to the target source, the information sources should be weighted differently according to their importance. Therefore, the \textbf{diverse variety} of the multiple sources will lead to more research challenges and opportunities, which is also a great challenge in \textbf{big data} studies. New information fusion and mining algorithms for the multi-source scenarios can be another great opportunity to explore broad learning in the future.

\item {\textbf{Broader Learning Applications}}: Besides the research works on social network datasets, the third potential future development of broad learning and mining lies its broader applications on various categories of datasets, like enterprise internal data \cite{kdd15, cikm15, cikm16, wsdm17}, geo-spatial data \cite{mdm16, iri16, sigspatial16}, knowledge base data, and pure text data. Some prior research works on fusing enterprise context information sources, like enterprise social networks, organizational chart and employee profile information have been done already \cite{kdd15, cikm15, cikm16, wsdm17}. Several interesting problems, like organizational chart inference \cite{kdd15}, enterprise link prediction \cite{cikm15}, information diffusion at workplace \cite{cikm16} and enterprise employee training \cite{wsdm17}, have been studied based on the fused enterprise internal information. In the future, these areas are still open for exploration. Applications of broad learning techniques in other application problems, such as employee training, expert location and project team formation, will be both interesting problems awaiting for further investigation. In addition, analysis of the correlation of different traveling modalities (like shared bicycles \cite{mdm16, iri16, sigspatial16}, bus and metro train) with the city zonings in smart city; and fusing multiple knowledge bases, like Douban and IMDB, for knowledge discovery and truth finding are both good application scenarios for broad learning research works.

\end{enumerate}
\balance
\bibliographystyle{plain}
\bibliography{reference}

\end{document}